\documentclass[11pt]{amsart}
\usepackage[margin=1in]{geometry}
\usepackage{amsfonts, float}
\usepackage{amssymb,amsmath}
\usepackage{dsfont}
\usepackage{euscript}
\usepackage{xcolor}
\usepackage{cite}

\usepackage[all]{xy}
\usepackage{multirow}

\usepackage{fancyvrb}

 \newtheorem{thm}{Theorem}[section]
 \newtheorem{cor}[thm]{Corollary}
 \newtheorem{lem}[thm]{Lemma}
 \newtheorem{ex}[thm]{Example}
 \newtheorem{prop}[thm]{Proposition}

 \theoremstyle{definition}
 \newtheorem{defn}[thm]{Definition}
 \theoremstyle{remark}

 \newtheorem{assumption}[thm]{Assumption}
 \numberwithin{equation}{section}

\newcommand{\idty}{\mathds{1}}

 \def\idtyty{{\mathchoice {\mathrm{1\mskip-4mu l}} {\mathrm{1\mskip-4mu l}} %
{\mathrm{1\mskip-4.5mu l}} {\mathrm{1\mskip-5mu l}}}}

\newcommand{\bR}{{\mathbb R}}

\newcommand{\bZ}{{\mathbb Z}}
\newcommand{\cA}{{\mathcal A}}
\newcommand{\cB}{{\mathcal B}}

\newcommand{\cD}{{\mathcal D}}

\newcommand{\cF}{{\mathcal F}}
\newcommand{\cK}{{\mathcal K}}
\newcommand{\cG}{{\mathcal G}}
\newcommand{\cH}{{\mathcal H}}
\newcommand{\cP}{{\mathcal P}}
\newcommand{\cS}{{\mathcal S}}

\newcommand{\supp}{\operatorname{supp}}

\newcommand{\caA}{{\mathcal A}}

\newcommand{\bbE}{{\mathbb E}}

\newcommand{\tr}{\mathrm{tr}}
\newcommand{\Tr}{\mathrm{Tr}}

\newcommand{\id}{\mathrm{id}}

\newcommand{\vertiii}[1]{{\left\vert\kern-0.25ex\left\vert\kern-0.25ex\left\vert #1 
    \right\vert\kern-0.25ex\right\vert\kern-0.25ex\right\vert}} 
\newcommand{\ketbra}[1]{\vert #1\rangle\langle #1\vert}

\newcommand{\dom}{\mathop{\rm dom}}
    
\newcommand{\spec}{\mathrm{spec}}

\newcommand{\be}{\begin{equation}}
\newcommand{\ee}{\end{equation}}
\newcommand{\bea}{\begin{eqnarray}}
\newcommand{\eea}{\end{eqnarray}}
\newcommand{\beann}{\begin{eqnarray*}}
\newcommand{\eeann}{\end{eqnarray*}}

\newcommand{\Rl}{\mathbb{R}}
\newcommand{\Ir}{\mathbb{Z}}

\newcommand{\eq}[1]{(\ref{#1})}

\newcommand{\diam}{{\rm diam}}

\title[Quasi-Locality Bounds I]{Quasi-Locality Bounds for Quantum Lattice Systems.
Part I. Lieb-Robinson Bounds, Quasi-Local Maps, and Spectral Flow Automorphisms}

\author[B. Nachtergaele]{Bruno Nachtergaele}
\thanks{Based upon work supported by the National Science Foundation under Grant DMS-1515850 and DMS-1813149.}
\address{Department of Mathematics and Center for Quantum Mathematics and Physics\\
University of California, Davis\\
Davis, CA 95616, USA}
\email{bxn@math.ucdavis.edu}

\author[R. Sims]{Robert Sims}
\address{Department of Mathematics \\
University of Arizona\\
Tuscon, AZ 85721, USA}
\email{rsims@math.arizona.edu}

\author[A. Young]{Amanda Young}
\address{Department of Mathematics\\
University of Arizona\\
Tucson, AZ 85721, USA}
\email{amyoung@math.arizona.edu}

\begin{document}
\date{\today }
\begin{abstract}
Lieb-Robinson bounds show that the speed of propagation of information under the Heisenberg dynamics in a wide class of non-relativistic quantum lattice systems is
essentially bounded. We review work of the past dozen years that has turned this fundamental result into a powerful tool for analyzing quantum lattice systems. We introduce
a unified framework for a wide range of applications by studying quasi-locality properties of general classes of maps defined on the algebra of local observables of quantum
lattice systems. We also consider a number of generalizations that include systems with an infinite-dimensional Hilbert space at each lattice site and Hamiltonians that may
involve unbounded on-site contributions. These generalizations require replacing the operator norm topology with the strong operator topology in a number of basic results 
for the dynamics of quantum lattice systems. The main results in this paper form the basis for a detailed proof of the stability of gapped ground state phases of frustration-free 
models satisfying a Local Topological Quantum Order condition, which we present in a sequel to this paper.
\end{abstract}

\maketitle

\tableofcontents

%
%

\section{Introduction}\label{sec:intro}

Quantum many-body theory comes in two flavors. The first is the relativistic version generically referred to as Quantum Field Theory (QFT), used for particle physics, and the second is non-
relativistic many-body theory, which serves as the basic framework for most of condensed matter physics. The close physical and mathematical similarities between the two have long been 
recognized and exploited with great success. Bogoliubov's theory of superfluidity and the BCS theory of superconductivity serve as definitive proof that 
quantum fields are a useful and even fundamental concept for understanding non-relativistic many-body systems. Condensed matter theorists
have developed field theory techniques that are now omnipresent in the subject \cite{abrikosov:1975,itzykson:1989,Tsvelik:1995,martin:2002,Mussardo:2010,altland:2010,fradkin:2013}.
See \cite{seiringer:2014,hainzl:2016} for reviews on recent progress in the mathematics of Bogoliubov's theory and the BSC theory of superfluidity.

The absence of Lorentz-invariance (and the associated constant speed of 
light $c$ leading to the all-important property of {\em locality} in the sense of Haag \cite{haag:1996}) in non-relativistic many-body theories is the most obvious difference between the two perspectives. Given the importance of this invariance in QFT, which plays an essential role in deriving many of the fundamental properties, and the strong constraints it imposes on its mathematical structure, one would expect that its absence would prevent any close analogy between the relativistic and the non-relativistic setting to hold true. Contrary to this expectation, successful applications of QFT to problems of condensed matter physics have been numerous. Quantum Field Theories have provided accurate descriptions as effective theories describing important aspects such as excitation spectra and derived quantities. This typically involves a scaling limit of some type. Conformal Field Theories have been spectacularly effective in describing and classifying second order phase transitions. Also here a scaling limit is often implied.

The quasi-locality properties that are the subject of study in this paper partly explain the closer-than-expected similarities between QFT and the non-relativistic many-body theory of condensed matter systems. More importantly, they make it possible to prove that much of the mathematical structure of QFT can be found in non-relativistic many-body systems in an approximate sense. Instead of asymptotic statements and qualitative comparisons, we can prove quantitative estimates: the quasi-locality of the dynamics is characterized by an approximate light-cone with errors that can be bounded explicitly. These are the Lieb-Robinson bounds, which have been an essential ingredient in a large number of breakthrough results in the past dozen years.

Although the result of Lieb and Robinson dates back to the early 70's \cite{lieb:1972,robinson:1976}, the impetus for the recent flurry of 
activity and major applications came from the work of Hastings 
on the Lieb-Schultz-Mattis Theorem in arbitrary dimension \cite{hastings:2004}. The possibility of adapting some of the major results of QFT to (non-relativistic) quantum lattice systems
was anticipated by others. For example, Wreszinski studied the connection between the Goldstone Theorem \cite{landau:1981}, charges, and spontaneous symmetry breaking
\cite{wreszinski:1987}. A rigorous proof of a non-relativistic Exponential Clustering Theorem, long known in QFT \cite{ruelle:1962, fredenhagen:1985}, did not appear until the 
works \cite{nachtergaele:2006a,hastings:2006}.  The time-evolution of quantum spin systems turns local observables into quasi-local ones.
Lieb-Robinson bounds were applied to approximate such quasi-local observables by strictly local ones with an error bound in \cite{bravyi:2006a,nachtergaele:2006,nachtergaele:2013}. 
These (sequences of) strictly local approximations are what is used in many concrete applications and also have conceptual appeal. Further extensions of Lieb-Robinson bounds and a sampling of interesting applications are discussed in Section \ref{sec:lrb}.

Apart from offering a review of the state of the art of quasi-locality estimates, in this paper we also extend existing results in the literature in several directions. First, for most of the results we allow the quantum system at each lattice site to be described by an arbitrary infinite-dimensional Hilbert space. For many results, the single-site Hamiltonians may be arbitrary densely defined self-adjoint operators. Another generalization in comparison to the existing literature, made necessary by the consideration of unbounded Hamiltonians, is that time-dependent perturbations are assumed to be continuous with respect to the strong operator topology instead of the operator norm topology. In order to handle this more general situation, a number of technical issues 
need to be addressed related to the continuity properties of operator-valued functions and of the dynamics generated by strongly continuous time-dependent interactions. These technical issues cascade through the better part of the paper. We will understand if the reader is surprised by the length of the paper, since we were taken aback ourselves as we were completing the manuscript. Many proofs can be shortened if one is only interested in particular cases. Indeed, in many cases results for more restricted cases exist in the literature. There are also places, however, where the published results in the literature provide only weaker estimates or have incomplete proofs.

This paper has seven sections in addition to this introduction. In Section \ref{sec:dynamics} we review the construction and basic properties of quantum dynamics in the context of this work. This includes a careful presentation of analysis  with operator-valued functions using the strong operator topology. Section \ref{sec:lrb} is devoted to Lieb-Robinson bounds
and their application to proving the existence of the thermodynamic limit of the dynamics. We also derive an estimate on the dependence of the dynamics on the interactions and introduce a 
notion of convergence of interactions that implies the convergence of the infinite-volume dynamics. Section \ref{sec:str-loc} is devoted to the approximation of quasi-local observables by strictly local ones by means of suitable maps called conditional expectations. Because they are needed for our applications, the continuity properties of a class of such maps are studied in detail. A general notion of quasi-local maps is introduced in Section \ref{sec:quasilocal-maps}, and we study the properties of several operations involving such maps that are used extensively in applications. In Section \ref{sec:spectral-flow} we construct an auxiliary dynamics called the spectral flow (also called the quasi-adiabatic evolution), which is the main tool 
in recent proofs of the stability of spectral gaps and gapped ground state phases. A first application of the spectral flow is the notion of automorphic equivalence, discussed in Section \ref{sec:equivalence_of_phases}, which allows us to give a precise definition of a gapped ground state phase as an equivalence class for a certain equivalence relation on families
of quantum lattice  models. Section \ref{sec:appendix} is an appendix in which we collect a number of arguments that are used throughout the paper.

Our original motivation for this work was to supply all the tools needed for the results of \cite{QLBII}. However, this work can now be read as a stand-alone review article about quasi-locality estimates for quantum lattice systems. Since the sequel of this paper, \cite{QLBII}, will be devoted to applying the quasi-locality bounds and the spectral flow results from this work to prove stability of gapped ground state phases, the examples and applications here will be chosen in support of the presentation of the general results.

Throughout this paper we focus on so-called bosonic lattice systems, for which observables with disjoint support commute. Virtually all results carry over to lattice fermion systems with only minor changes. This is discussed in some detail in \cite{nachtergaele:2016b}. Another extension of quasi-locality techniques not covered in this paper is the case of so-called {\em extended
operators}. An important example are the half-infinite string operators that create the elementary excitations in models with topological order such as the Toric Code model \cite{kitaev:2006}. Lieb-Robinson bounds for such non-quasi-local operators are used in \cite{cha:2018}.

%
%
%

\section{Some basic properties of quantum dynamics} \label{sec:dynamics}

In this paper, the primary object of study is the Heisenberg dynamics acting on a suitable algebra of observables for a finite or 
infinite lattice system. For finite systems, this dynamics is expressed with a unitary propagator $U(t,s)$, $s\leq t\in \Rl$, on a 
separable Hilbert space $\cH$. However, in some cases, for example when one is interested in the excitation spectrum and dynamics of 
perturbations with respect to a thermal equilibrium state of the system, the generator of the dynamics is not semi-bounded and the Hilbert 
space may be non-separable. Therefore, in general, we will not assume that $\cH$ is separable or that the Hamiltonian is bounded below.

As described in the introduction, we consider finite and infinite lattice systems with interactions that are sufficiently local. We 
allow for an infinite-dimensional Hilbert space at each site of the lattice. However, we impose conditions on the interactions that permit us to 
prove quasi-locality bounds of Lieb-Robinson type (in terms of the operator norm) for bounded local observables. 
This means that we will allow for the possibility of unbounded `spins' and unbounded single-site Hamiltonians, but require that the
interaction be given by bounded self-adjoint operators that satisfy a suitable decay condition at large distances (see Section~\ref{sec:lrb} for more details). We do not consider
lattice oscillator systems with harmonic interactions in this paper, since one should not expect bounds in terms of the operator norm for this class of 
systems (see \cite{nachtergaele:2009a,amour:2010}). An interesting model that does fit in the framework presented here is the so-called quantum rotor model, 
which has an unbounded Hamiltonian for the quantum rotor at each site, but the interactions between rotors are described by a bounded potential 
\cite{lee:1986,sachdev:1999,klein:1992}.

We will use the so-called {\em interaction picture} to describe the dynamics of Hamiltonians with unbounded on-site terms. This requires that we 
also consider time-dependent interactions. Time-dependent Hamiltonians are, of course, of interest in their own right, for instance in applications 
of quantum information theory. Therefore, we begin with a discussion of the Schr\"odinger equation for the class of time-dependent Hamiltonians 
considered in this work.

Let $\mathcal{H}$ be a complex Hilbert space and $\mathcal{B}( \mathcal{H})$ denote the bounded linear operators on $\mathcal{H}$. Let $I \subseteq \bR$ 
be a finite or infinite interval. In this section we review some basic properties of the dynamics of a quantum system with a time-dependent 
Hamiltonian of the form
\be
H(t)\psi = H_0\psi + \Phi(t)\psi, \quad\psi\in\cD
\label{Ht}\ee
where $H_0$ is a time-independent self-adjoint operator with dense domain $\cD\subset\cH$, and for
$t\in I$, $\Phi(t)^*=\Phi(t) \in \cB(\cH)$ and $t\mapsto \Phi(t)$ is continuous in the strong operator topology.
This means that for all $\psi\in \cH$, the function $t\mapsto \Phi(t)\psi$ is continuous in the Hilbert space norm. 
From these assumptions, it follows that for all $t\in I$, $H(t)$ is self-adjoint with time-independent dense domain $\cD$, 
see \cite{weidmann:1980}[Theorem 5.28].

We will often consider operator-valued and vector-valued functions of one or more real (or complex) variables and impose various continuity 
assumptions, which we now briefly review. An operator-valued function is said to be {\em norm continuous} ({\em norm differentiable}) if it is continuous (differentiable) 
in the operator norm, and {\em strongly continuous} ({\em strongly differentiable}) if it is continuous (differentiable) in the strong operator topology. 
With a slight abuse of terminology, we will refer to Hilbert space-valued functions as {\em strongly continuous} ({\em strongly differentiable}) if they are 
continuous (differentiable) in the Hilbert space norm. For transparency, when we consider maps defined on a linear space of operators, we will indicate the relevant topology and continuity assumptions explicitly.

The dynamics of a system described in \eqref{Ht} is determined by the Schr\"odinger equation:
\be
\frac{d}{dt} \psi(t) = -i H(t)\psi(t), \quad \psi(t_0)=\psi_0\in \cD , \quad t_0 \in I. 
\label{SE}\ee
For bounded $H(t)$, through a standard construction, we will see that there exists a family of unitaries $U(t,s)\in\cB(\cH)$, $s,t \in I$, that is 
jointly strongly continuous with $\psi(t)=U(t,t_0)\psi_0$ being the unique solution of \eq{SE} for all $\psi_0 \in \cH$. It follows that the family 
$U(t,s)$ has the co-cycle property: for $r \leq s \leq t \in I$, $U(t,r)=U(t,s)U(s,r)$ and $U(t,t)=\idty$. In the case that $H(t) = H_0 + \Phi(t)$ where $H_0$ is an arbitrary unbounded self-adjoint operator and $\Phi(t)$ is bounded, we will make use of the well-known {\em interaction picture dynamics} to construct 
an analogous unitary co-cycle. This co-cycle will, in particular, generate the unique weak solution of the 
Schr\"odinger equation. To this end, we first discuss some other aspects of strongly continuous operator-valued functions that we will need.

%
%

\subsection{Properties of continuity, measurability, and integration in $\mathcal{B}( \mathcal{H})$}\label{sec:review}

In this section we review some terminology and discuss a number of properties of operator-valued functions that will be used extensively in the rest of the paper.

Let $I \subset \mathbb{R}$ be a finite or infinite interval and $A: I  \to \mathcal{B}(\mathcal{H})$ be strongly continuous, i.e.,  for all $\psi \in \mathcal{H}$, $t \mapsto A(t) \psi$ is 
continuous with respect to the Hilbert space norm. 
By the Uniform Boundedness Principle,
if $A$ is strongly continuous, then $A$ is locally bounded, meaning if $J \subset I$ is compact, then
\begin{equation} \label{local_bound}
M_J:= \sup_{t \in J} \| A(t) \| < \infty.
\end{equation}
The strong continuity of $t\mapsto A(t)$ implies that $t\mapsto \Vert A(t)\psi\Vert $
is continuous for all $\psi\in\cH$, and by the above, the map $t \mapsto \Vert A(t)\Vert$ is locally bounded. However, strong continuity does not imply that
$t\mapsto \Vert A(t)\Vert$
is continuous, see \cite[Section 2]{nachtergaele:2014} for a counterexample. 

We note that in this paper we use the notations $\Vert A(t)\Vert$ and $\Vert A\Vert (t)$
interchangeably. For ease of later reference, we now state a simple proposition. 

\begin{prop} \label{prop:st_cont_prods} Let $I \subset \mathbb{R}$ be a finite or infinite interval and $\mathcal{H}$ and $\mathcal{K}$ be Hilbert spaces.
\begin{enumerate}
\item[(i)] If $A, B : I \to \mathcal{B}( \mathcal{H})$ are strongly continuous (strongly differentiable), then $(t,s) \mapsto A(t)B(s)$ is jointly strongly continuous
(separately strongly differentiable).
\item[(ii)] If $A: I \to \mathcal{B}(\mathcal{H})$ and $B : I \to \mathcal{B}( \mathcal{K})$ are strongly continuous (strongly differentiable), then $(t,s) \mapsto 
A(t) \otimes B(s)$  is jointly strongly continuous (separately strongly differentiable).
\item[(iii)] If $A:I\to\cB(\cH)$ is strongly continuous, then the function $t\mapsto \Vert A(t)\Vert $ is lower semicontinuous, measurable, locally bounded and, hence,
locally integrable.
\end{enumerate}
\end{prop}

It is clear that an analogue of Proposition~\ref{prop:st_cont_prods} holds when {\it strongly} is replace with
{\it norm} in the statements above. Moreover, an argument similar to the one found in Proposition~\ref{prop:st_cont_prods} (i) below shows that if
$A: I \to \mathcal{B}( \mathcal{H})$  and $\psi : I \to \mathcal{H}$ are both strongly continuous (strongly differentiable), 
then $(t,s) \mapsto A(t) \psi(s)$ is jointly strongly continuous (separately strongly differentiable). As will be clear from the proof, we note that the conclusions of part (iii) of this proposition continue to hold even for weakly continuous $A(t)$.

\begin{proof} We prove the statements above in the case of strong continuity; the strong differentiability claims follow similarly.

For (i), let $\psi \in \cH$ and $t_0, \, s_0 \in I$ be fixed. Given that $A$ and $B$ are strongly continuous, and hence locally bounded, it follows that $A(t)B(s)\psi \to A(t_0)B(s_0)\psi$ as $(t,s)\to(t_0,s_0)$ since
\[
\|A(t)B(s)\psi - A(t_0)B(s_0)\psi\|
\leq 
\|A(t)\| \|(B(s)-B(s_0))\psi\| + \|(A(t)-A(t_0))B(s_0)\psi\|.
\] 

To prove (ii), we first show that if $t\mapsto A(t)\in \cB(\cH)$ is strongly continuous and $\cK$ is another Hilbert space, then the map $t\mapsto A(t)\otimes\idty \in \cB(\cH\otimes\cK)$ is also strongly continuous. To see this, note that for all $\psi \in \cH \otimes \cK$, there exist two countable sets of vectors $\psi_n \in \cH$ and $\phi_n\in \cK$ with $\{\phi_n\}_n$ orthonormal such that $\psi = \sum_n \psi_n\otimes \phi_n$, and $\sum_n \|\psi_n\|^2 = \|\psi\|^2$. Fix $s\in I$, and let $J\subseteq I$ be a compact interval that contains a neighborhood of $s$. Using the orthonormality of $\phi_n$, we find that for all $t \in J$
\[\|((A(t)-A(s))\otimes \idtyty) \psi\|^2 = \sum_n \|(A(t)-A(s))\psi_n\|^2. \]
For any $\epsilon > 0$, choose $N$ large enough so that $\sum_{n>N} \|\psi_n\|^2 \leq \epsilon/8M_J^2,$ where $M_J>0$ satisfies \eqref{local_bound}. By the strong continuity of $A(t)$, there exists a $\delta>0$ such that for all $t\in(s-\delta, \, s+\delta) \subseteq J$ and $1\leq n \leq N$, one has $\|(A(t)-A(s))\psi_n\|^2 < \epsilon/2N.$ Putting these together, if $|t-s|< \delta$, then
\be \label{sc_tensor}
\|((A(t)-A(s))\otimes \idtyty) \psi\|^2 \leq \sum_{n=1}^N\|(A(t)-A(s))\psi_n\|^2 +\sum_{n>N}(\|A(t)\|+\|A(s)\|)^2\|\psi_n\|^2 < \epsilon.
\ee
Since $A(t)\otimes B(s) = (A(t)\otimes \idtyty)(\idtyty \otimes B(s))$, by (i), the tensor product of two strongly continuous maps $t\mapsto A(t)\in \cB(\cH)$ and $s\mapsto B(s)\in \cB(\cK)$ is jointly strongly continuous.

To prove (iii), we start by noting that, by virtue of the strong continuity of $A(t)$, the function $\Vert A(t)\Vert$ can be expressed as a supremum of continuous functions:
$$
\Vert A(t) \Vert =\sup\{ \langle \phi, A(t)\psi\rangle \mid \phi,\psi\in\cH , \Vert \phi\Vert = \Vert \psi\Vert =1\}.
$$
Recall that a function $f:I\to\Rl$ is lower semicontinuous iff for all $s\in\Rl$, $f^{-1} ((s,\infty)) $ is an open subset of $I$. Now, if $f$ is the supremum of a family of functions 
$\{f_\alpha\}$, we have that $ f^{-1} ((s,\infty))  = \bigcup_\alpha f_\alpha^{-1} ((s,\infty)) $. In our case, the $f_\alpha$ are indexed by a pair of unit vectors in $\cH$ and each 
$f_\alpha$ is continuous. Therefore, $f_\alpha^{-1} ((s,\infty))$ is open for all $\alpha$ and so is  $\bigcup_\alpha f_\alpha^{-1} ((s,\infty))$. This shows the lower semicontinuity.

Since we have that  $ f^{-1} ((s,\infty))$ is open, this set is also Borel measurable, for all $s\in \Rl$. By a standard lemma in measure theory \cite{folland:1999}, this implies that $f$ is measurable.

We already noted above that $\Vert A(t)\Vert $ is bounded on compact intervals by the Uniform Boundedness Principle. This concludes the proof.
\end{proof}

We will make frequent use of integrals of vector-valued and operator-valued functions.
It is straightforward to define such integrals in the weak sense. In fact, if 
$A: I \to \mathcal{B}( \mathcal{H})$ is locally bounded and weakly measurable, i.e. for 
all $\phi, \psi \in \mathcal{H}$, $t \mapsto \langle \phi, A(t) \psi \rangle$ is measurable, then
for any compact $J \subset I$ the integral of $A$ over $J$ is defined as the operator
$B_J \in \mathcal{B}( \mathcal{H})$
corresponding to the bounded sesquilinear form
\begin{equation}
\langle \phi, B_J \psi \rangle = \int_J  \langle \phi, A(t) \psi \rangle \, dt \, .
\end{equation}
We routinely use the notation $B_J = \int_J  A(t) \, dt$ to denote this operator.
For strongly continuous functions $A(t)$, the same integral can be interpreted in the strong sense:
\be
\left(\int_J  A(t) \, dt\right) \psi = \int_J  A(t)\psi \, dt,
\ee
where the RHS is understood to be the Riemann integral of a strongly continuous, Hilbert space-valued 
function. Since the range of any strongly continuous, Hilbert space-valued function
belongs to a separable subspace (even if $\cH$ is not separable, see, e.g., \cite[Section V.4]{yosida:1980}),
this integral also exists in the sense of Bochner. 

The following well-known inequalities hold for all $A : I \to \mathcal{B}( \mathcal{H})$, strongly continuous and
$J \subset I$ compact:
\begin{equation}
| \langle \phi, B_J \psi \rangle | \leq \int_J | \langle \phi, A(t) \psi \rangle | dt \leq
\| \phi \| \int_J \| A(t) \psi \| dt,
\end{equation}
and thus
\begin{equation}\label{int1}
\left\| \left(\int_J A(t)dt \right) \psi \right\| \leq \int_J \| A(t) \psi\|  dt.
\end{equation}
In particular, we obtain
\begin{equation}\label{int2}
\left\| \int_J A(t) dt \right\| \leq \int_J \| A(t) \| dt \leq |J| \sup_{t\in J}\Vert A(t)\Vert.
\end{equation}
The  first inequality extends to infinite $J$ if, e.g., $A(t) = B(t) w(t)$, 
with $B(t)$ strongly continuous and bounded and $w\in L^1(J)$.
Finally, it is easy to see that if $A$ is strongly continuous, then $B(t) =\int_{t_0}^t A(s) ds$ is 
norm continuous, and strongly differentiable with $\frac{d}{dt}B(t) = A(t)$. As such, the fundamental theorem 
of calculus also holds in the strong sense.

%
%

\subsection{Dynamical equations and the Dyson series} \label{sec:Dyson}

In this section, we review some well-known facts about Dyson series and from them
obtain the Schr\"odinger dynamics generated by a bounded, time-dependent Hamiltonian. 
A standard result in this direction can be summarized as follows.
Let $\mathcal{H}$ be a Hilbert space, $I\subset\mathbb{R}$ be a finite or infinite interval, 
and $H: I \to \mathcal{B}(\mathcal{H})$ be strongly continuous and pointwise self-adjoint,
i.e. $H(t)^*= H(t)$ for all $t \in I$. Under these assumptions (see, e.g. Theorem X.69 of 
\cite{reed:1975}) for each $t_0\in I$ and every initial condition $\psi_0 \in \mathcal{H}$, 
the time-dependent Schr\"odinger equation
\begin{equation} \label{tdseqn}
i\frac{d}{dt}\psi(t)=H(t)\psi(t), \quad \psi(t_0)=\psi_0,
\end{equation}
has a unique solution in the sense that there is a unique, strongly differentiable function $\psi : I \to \mathcal{H}$ 
which satisfies (\ref{tdseqn}). This solution can be characterized in terms of a two-parameter family of 
unitaries $\{ U(t,s) \}_{s,t \in I}\subset \mathcal{B}(\mathcal{H})$ such that  
\begin{equation}
\psi(t)=U(t,s)\psi(s) \quad \mbox{for all } s,t \in I \, .
\end{equation}
These unitaries are often referred to as propagators, and an explicit construction of them is given by the Dyson series.
Specifically, for any $s,t \in I$ and each $\psi \in \mathcal{H}$ the Hilbert space-valued series
\begin{equation}\label{dyson}
U(t,s)\psi =  \psi + \sum_{n=1}^{\infty} (-i)^n \int_{s}^t \int_{s}^{t_1} \cdots \int_{s}^{t_{n-1}} H(t_1) \cdots H(t_n)  \psi \, dt_n \cdots dt_1\,
\end{equation}
is easily seen to be absolutely convergent in norm. One checks that $U(t,s)$, as defined in (\ref{dyson}), satisfies 
the differential equation
\begin{equation}\label{se_propagator}
\frac{d}{dt}U(t,s) =-iH(t)U(t,s), \quad U(s,s) = \idtyty
\end{equation}
which is to be understood in the sense of strong derivatives. Of course, under the stronger assumption that 
$H : I \to \mathcal{B}( \mathcal{H})$ is norm continuous, then \eq{se_propagator} also holds in norm.

The additional observation we want to make here is that $U(t,s)$ is 
not only the unique strong solution of (\ref{se_propagator}); it is also the case that 
any bounded weak solution of (\ref{se_propagator}) necessarily coincides with $U(t,s)$. By {\em weak solution}, we mean that for all $\phi,\psi\in\mathcal{H}$ and
any $s,t \in I$, $U(t,s)$ satisfies 
\begin{equation} \label{Uts_weak_sol}
 \frac{d}{dt}\langle\phi,U(t,s)\psi\rangle =-i\langle\phi, H(t)U(t,s)\psi\rangle , \quad U(s,s) = \idtyty \, .
\end{equation}
A proof of this fact is contained in the following proposition.

\begin{prop} \label{prop:sols}
Let $A: I \to \mathcal{B}(\mathcal{H})$ be strongly continuous, and consider the differential equation
\begin{equation}\label{basic_de}
\frac{d}{dt} V(t)= A(t) V(t), \quad V(t_0) = V_0 \in \mathcal{B}( \mathcal{H}), \quad  t_0 \in I. 
\end{equation}
The following statements hold:

(i) There is a unique strong solution $V : I \to \mathcal{B}( \mathcal{H})$ of \eqref{basic_de}, and $V$ is norm continuous.

(ii) Any locally norm-bounded, weak solution of \eqref{basic_de} coincides with the strong solution.

(iii) Let $\mathcal{D} \subset \mathcal{H}$ be dense.
Suppose $V: I \to \mathcal{B}(\mathcal{H})$ is strongly continuous and satisfies
\begin{equation}\label{de_on_vectors}
\frac{d}{dt} V(t)\psi= A(t) V(t)\psi,  \quad V(t_0) \psi= V_0\psi
\end{equation}
for all $\psi \in \mathcal{D}$ and $t \in I$. Then, $V$ is the unique strong solution.

(iv) If $V_0$ is invertible, the strong solution $V$ of (\ref{basic_de}) is invertible for all $t\in I$.
Moreover, in this case, the inverse of $V$ is the unique strong solution of
\begin{equation}\label{inverse_de}
\frac{d}{dt} V^{-1}(t)= -V^{-1}(t)A(t), \quad V^{-1}(t_0) = V^{-1}_0 \in \mathcal{B}( \mathcal{H}).
\end{equation}

(v) If $H : I \to \mathcal{B}( \mathcal{H})$ is strongly continuous and pointwise self-adjoint, then the strong solution $V$ of
(\ref{basic_de}) with the choice $A = -iH$ and $V_0 = \idtyty$ is unitary for all $t \in I$. Moreover, the map
$U : I \times I \to \mathcal{B}( \mathcal{H})$ given by $U(t,s) =V(t)V(s)^*$ is the 
unique strong solution to \eqref{se_propagator}.
\end{prop}

Before giving the proof of this proposition, we first comment on the content of part (v) and its relationship to the Dyson series from \eqref{dyson}. 
A simple consequence of (i) is that the mapping $U(t,s)$, defined in (v), is jointly norm continuous. 
As stated in (v), this $U(t,s)$ is the unique strong solution of (\ref{se_propagator}); it is also strongly 
differentiable in $s$, and by (iv), this strong derivative is
\begin{equation} \label{strong_adjoint_de}
\frac{d}{ds}U(t,s) = i U(t,s) H(s) \, .
\end{equation}
In addition, one readily checks that $U(t,s)^{-1} = V(s) V(t)^*$, for all $s,t \in I$, and thus 
$U(t,s)^{-1} = U(t,s)^* = U(s,t)$, for all $s,t \in I$, so that $U(t,s)$ is a two-parameter family of unitaries. 
This family of unitaries satisfies the co-cycle property: if $r \leq s \leq t$, then
\begin{equation} \label{co-cycle_rel}
U(t,s)U(s,r) = U(t,r) \quad \mbox{and} \quad U(s,s) = \idtyty.
\end{equation} 
Finally, arguing as in the proof of (i) below, one sees that
the Dyson series (\ref{dyson}) is a strong solution of (\ref{se_propagator}). Combining this with the
uniqueness proven in (v), we conclude that $U(t,s) =V(t)V(s)^*$ must coincide with the Dyson series constructed in (\ref{dyson}).

\begin{proof}
(i) Define a map $V: I \to \mathcal{B}( \mathcal{H})$ via a Dyson series, i.e. for any $\psi \in \mathcal{H}$ and each $t \in I$, set
\begin{equation}\label{dyson2}
V(t)\psi =  V_0\psi + \sum_{n=1}^{\infty}  \int_{t_0}^t \int_{t_0}^{t_1} \cdots \int_{t_0}^{t_{n-1}} A(t_1) \cdots A(t_n)  V_0\psi \, dt_n \cdots dt_1\, .
\end{equation}
We now argue that $V$ is the unique strong solution of (\ref{basic_de}). 

First, we show that $V$ is well-defined.
The integrals appearing as terms in this series are well-defined due to the strong continuity of $A$. More precisely, for any $n \geq 1$, the product $A(t_1) \cdots A(t_n)$ is jointly strongly continuous 
in the variables $t_1, \cdots, t_n$, and thus the integrands are locally integrable. 
Next, for any $t \geq t_0$, the bound
\begin{eqnarray} \label{dyson_est}
\| V(t) \psi \| & \leq & \| V_0 \| \| \psi \| +  \| V_0 \| \| \psi \| \sum_{n=1}^{\infty} \int_{t_0}^t \int_{t_0}^{t_1} \cdots \int_{t_0}^{t_{n-1}} \| A \| (t_1) \cdots \| A \| (t_n) \, dt_n \cdots dt_1 \nonumber \\
& \leq & \| V_0 \| \| \psi \| \sum_{n=0}^{\infty} \frac{1}{n!} \left( \int_{t_0}^{t} \| A \|(s) \, ds \right)^n
\end{eqnarray}  
holds. Here, we note that we are using the alternate notation $\|A\|(t)$ for $\|A(t)\|.$  As it is clear that a similar argument holds for $t<t_0$, we see that $V$ is well-defined as an absolutely convergent (in norm) series. 

Next, we prove that $V$ is a strong solution of (\ref{basic_de}). To see this, define recursively a sequence of
operators $\{ V_n \}_{n \geq 1}$, $V_n : I \to \mathcal{B}(\mathcal{H})$ by setting
\begin{equation} \label{recursive}
V_1(t) \psi = A(t) V_0 \psi \quad \mbox{ and } \quad V_n(t) \psi = A(t) \int_{t_0}^t V_{n-1}(t_1) \psi \, dt_1 \quad \mbox{for any } n \geq 2 \mbox{ and } t \in I \, .
\end{equation}
With respect to this notation, it is clear that
\begin{equation}
V(t) \psi = V_0 \psi + \sum_{n=1}^{\infty} \int_{t_0}^t V_n(t_1) \psi \, d t_1 
\end{equation}
and moreover, for any $h \neq 0$, 
\begin{eqnarray}
\frac{V(t+h) - V(t)}{h} \psi & = & \sum_{n=1}^{\infty} \frac{1}{h} \int_t^{t+h} V_n(t_1) \psi \, dt_1 \nonumber \\
& = & \sum_{n=1}^{\infty} V_n(t) \psi + \sum_{n=1}^{\infty} \frac{1}{h} \int_t^{t+h} \left(V_n(t_1) - V_n(t) \right) \psi \, dt_1 \, .
\end{eqnarray}
Using the recursive definition, i.e. (\ref{recursive}), it is clear that the first term on the right-hand-side above is $A(t)V(t) \psi$.
A dominated convergence argument, using an estimate like (\ref{dyson_est}), guarantees that the remainder term goes to zero in norm, and
hence $V$ is a strong solution. 

Finally, we prove uniqueness. Let $V_1$ and $V_2$ be two strong solutions of (\ref{basic_de}). 
For any $t\in I$, set $t_+ = \max\{ t, t_0\}$ and $t_- = \min \{t,t_0 \}$. Given $\psi \in \mathcal{H}$, we have that
\begin{eqnarray}
\Vert (V_1(t) - V_2(t)) \psi \Vert & = & \left\Vert \int_{t_0}^{t} \frac{d}{ds}   (V_1(s) - V_2(s)) \psi \, ds \right\Vert  \nonumber \\ 
& = &  \left\Vert \int_{t_0}^{t} A(s)(V_1(s) - V_2(s)) \psi\, ds\right\Vert  \nonumber \\
& \leq &  \int_{t_-}^{t_+} \| A \| (s) \Vert (V_1(s) - V_2(s)) \psi\Vert \, ds \nonumber
\end{eqnarray}
and uniqueness follows from Gronwall's Lemma. 

One also sees that this $V$ is norm continuous. In fact, let $t,t_0 \in I$ and $\psi \in \mathcal{H}$. Clearly
\begin{equation}
\| (V(t) - V(t_0)) \psi \| = \left\| \int_{t_0}^t \frac{d}{ds} V(s) \psi \, ds \right\| \leq \| \psi \| \int_{t_-}^{t_+} \| AV \|(s) \, ds
\end{equation}
and norm continuity of $V$ follows. 

(ii) The uniqueness statement in (ii) is proven similarly. In fact, let $V_1$ and $V_2$ be two 
locally norm-bounded weak solutions of (\ref{basic_de}). In this case, for any $\phi, \psi \in \mathcal{H}$ 
and each $t \in I$,
\begin{eqnarray}
\left| \langle \phi, \left( V_1(t) - V_2(t) \right) \psi \rangle \right| & = & \left| \int_{t_0}^{t} \frac{d}{ds} \langle \phi, \left( V_1(s) - V_2(s) \right) \psi \rangle \, ds \right|  \nonumber \\ & = &  \left| \int_{t_0}^{t} \langle \phi, A(s)( V_1(s) - V_2(s))\psi \rangle \, ds \right|  \nonumber \\
& \leq & \| \phi \| \| \psi \|  \int_{t_-}^{t_+} \| A \| (s) \| V_1 - V_2 \|(s) \, ds. \label{basic_de3}
\end{eqnarray}
where we have used the notation $t_\pm$ as above. 
Taking the supremum of \eqref{basic_de3} over all normalized $\phi, \psi \in \mathcal{H}$ gives
\begin{equation}
\| V_1 - V_2 \|(t) \leq  \int_{t_-}^{t_+} \| A \|(s) \| V_1 - V_2 \|(s) \, ds \, .
\end{equation}
Again, uniqueness follows from Gronwall's Lemma. Since any strong solution is also a locally bounded weak solution, 
the unique weak and unique strong solutions must coincide.

(iii) We will show that any $V$ satisfying the assumptions of (iii) is actually the locally bounded weak solution. 
To this end, note that for any $\phi\in\mathcal{H}$, (\ref{de_on_vectors}) implies 
\begin{equation}\label{dense_strong}
\frac{d}{dt} \langle\phi,V(t)\psi\rangle= \langle\phi,A(t) V(t)\psi\rangle
\end{equation}
holds for all $\psi\in\mathcal{D}$ and any $t \in I$.
Let $\psi\in\mathcal{H}$ and take any sequence $\{ \psi_n \}_{n \geq 1}$ in $\mathcal{D}$ with $\psi_n$ converging to $\psi$.
Consider the sequence of functions $f_n:I \to \mathbb{C}$ defined by
\begin{equation}
f_n(t)= \langle\phi,V(t)\psi_n\rangle \quad \mbox{for all } t \in I,
\end{equation}
and set $f(t) = \lim_{n \to \infty}f_n(t)$. Note that
these pointwise limits exist as $\psi_n$ converges to $\psi$ and $V$ is locally bounded. One also sees that
$f(t) = \langle\phi,V(t)\psi\rangle$ for all $t \in I$. From (\ref{dense_strong}), it is clear that $f_n^\prime(t) =  \langle\phi,A(t) V(t)\psi_n\rangle$.
Since $AV$ is strongly continuous, and hence locally bounded, the same argument shows that
$g:I \to \mathbb{C}$ with $g(t) = \lim_{n \to \infty} f_n^{\prime}(t)= \langle\phi,A(t) V(t)\psi\rangle$ is well-defined.  
Observing further that $f_n^{\prime}$ converges to $g$ uniformly on compact subsets of $I$, it is clear that the
conditions of \cite[Theorem 7.17]{rudin:1976} are satisfied. We conclude that $f^\prime(t)=g(t)$ for all $t \in I$ and
hence, $V$ is the unique locally bounded weak solution. By the result proven in (ii), $V$ also
coincides with the unique strong solution.

(iv) Arguing as in the proof of (i), the function $W: I \to \mathcal{B}( \mathcal{H})$ defined by setting 
\begin{equation}\label{dyson_inverse}
W(t)\psi =  V_0^{-1}\psi + \sum_{n=1}^{\infty}  (-1)^n\int_{t_0}^t \int_{t_0}^{t_1} \cdots \int_{t_0}^{t_{n-1}} V_0^{-1} A(t_n) \cdots A(t_1)  
\psi \, dt_n \cdots dt_1
\end{equation}
for any $\psi \in \mathcal{H}$ is a strong solution of the initial value problem
\begin{equation} \label{inver_de}
\frac{d}{dt} W(t)= -W(t)A(t), \quad W(t_0) = V_0^{-1} .
\end{equation} 
Now, with $V$ the strong solution of (\ref{basic_de}), consider the function $Y : I \to \mathcal{B}( \mathcal{H})$ given by
$Y(t)=V(t)W(t)$. One checks that 
\begin{equation} \label{comm_de}
\frac{d}{dt} Y(t)= A(t)Y(t)-Y(t)A(t), \quad Y(t_0) = \idty,
\end{equation}
holds in the strong sense. It is clear that $Y(t) =\idty$ solves the above initial value problem. 
A Gronwall argument, similar to those we have proven before, shows that this constant 
solution is the unique strong solution of (\ref{comm_de}) and thus, $W$ is a right inverse of $V$ for all $t \in I$. 
Noting that the function $Z : I \to \mathcal{B}( \mathcal{H})$ 
defined by $Z(t) = W(t) V(t)$ satisfies the trivial initial value problem:
\begin{equation}
\frac{d}{dt} Z(t)= 0, \quad Z(t_0) = \idty,
\end{equation}
we conclude $W(t)=V(t)^{-1}$ as claimed. In fact, uniqueness of the strong solution of (\ref{inver_de}) follows. 

(v) One sees that $V$ is unitary by noting that the adjoint of the operator defined by 
\eqref{dyson2} agrees with the Dyson series for $V^{-1}$ given in \eqref{dyson_inverse}. 
For each $s \in I$ fixed, the map $t \mapsto V(t)V(s)^*$ defines a strong solution of \eqref{se_propagator} and by (i)
it is unique. 
\end{proof}

We conclude this section with an estimate on the solution of certain dynamical equations that will be useful
in the proof of the Lieb-Robinson bound in Section~\ref{sec:lrb}.

\begin{lem} \label{lem:normbd}  Let $\cH$ be a Hilbert space, $I\subset\mathbb{R}$ be a finite or infinite interval, and
$A,B: I \to \mathcal{B}(\mathcal{H})$ be strongly continuous with $A$ pointwise self-adjoint, i.e. $A(t)^*=A(t)$ for all $t\in I$.
For each $t_0 \in I$ and $V_0 \in \mathcal{B}( \mathcal{H})$, the initial value problem
\begin{equation} \label{gen_norm_pres_de}
\frac{d}{dt} V(t) = -i[A(t), V(t)] + B(t) \quad \mbox{with} \quad V(t_0) = V_0 
\end{equation}
has a unique strong solution. In particular,
 \begin{equation} \label{norm_pres_sol_bd}
\| V \|(t) \leq \| V_0 \| + \int_{t_-}^{t_+} \| B \|(s) \, ds 
\end{equation}
where $t_+ = \max \{ t, t_0 \}$, $t_- = \min \{ t,t_0 \}$. Moreover, any locally bounded weak solution of (\ref{gen_norm_pres_de}) 
coincides with the strong solution and, therefore, satisfies the estimate in (\ref{norm_pres_sol_bd}). 
\end{lem}

\begin{proof} 
By Proposition \ref{prop:sols} (v), the unique strong solution of 
\begin{equation}
\frac{d}{dt} W(t)=-iA(t)W(t) \quad \mbox{with} \quad W(t_0) = \idtyty 
\end{equation}
is unitary for all $t \in I$. As a product of strongly differentiable maps, $V: I \to \mathcal{B}( \mathcal{H})$ given by
\begin{equation}
V(t) = W(t) \left( V_0 + \int_{t_0}^t W(s)^* B(s) W(s) \, ds \right) W(t)^*
\end{equation}
is strongly differentiable. In fact, a short calculation shows that this $V$ 
is a strong solution of (\ref{gen_norm_pres_de}), and moreover, the bound claimed in (\ref{norm_pres_sol_bd}) is clear. 
Arguments involving Gronwall's lemma, similar to those found in the proof of Proposition \ref{prop:sols} (i) and (ii), verify the claimed uniqueness results. 
\end{proof}

%
%

\subsection{The dynamics for a class of unbounded Hamiltonians} \label{sec:dyn_unbd}

%
%

\subsubsection{On the interaction picture dynamics} \label{sec:int_pic_dyn}

The following proposition is an important application of  Proposition \ref{prop:sols}.
As explained in the remarks of Section X.12 of \cite{reed:1975}, applying the {\it interaction
picture representation} to Hamiltonians with the form $H = H_0+ \Phi$, even if 
$\Phi$ is time-independent, leads one to study a dynamics with time-dependent Hamiltonians. 
In this situation, one often produces Hamiltonians that are strongly continuous, but not norm continuous.
This leads us to consider Hamiltonians of the form $H(t) = H_0 + \Phi(t)$, where $H_0$ is 
a self-adjoint operator with dense domain $\cD$, and $\Phi(t)$ is a bounded, pointwise self-adjoint operator that is strongly continuous in $t$. 

\begin{prop}\label{prop:tduprop} 
Let $\mathcal{H}$ be a Hilbert space and $H_0$ a self-adjoint operator
with dense domain $\mathcal{D} \subset \mathcal{H}$. Let $I \subset \mathbb{R}$ be a finite or infinite interval and
$\Phi: I \to \mathcal{B}( \mathcal{H})$ be strongly continuous and pointwise self-adjoint. 
Then, there is a two parameter family of unitaries $\{U(t,s) \}_{s, t \in I}$ associated to the self-adjoint 
operator $H(t) = H_0 + \Phi(t)$ for which:
\begin{enumerate}
	\item[(i)]$(t,s) \mapsto U(t,s)$ is jointly strongly continuous,
	\item[(ii)] $U(t,s)$ satisfies the co-cycle property (\ref{co-cycle_rel}),
	\item[(iii)] $U(t,s)$ generates the unique, locally bounded weak solutions of the Schr\"odinger equation associated to 
	$H(t)$, i.e. for any $t_0 \in I$ and $\psi_0 \in \mathcal{H}$, $\psi : I \to \mathcal{H}$ given by $\psi(t) = U(t,t_0) \psi_0$ satisfies 
	\begin{equation} \label{weak_SE}
	\frac{d}{dt} \langle \phi, \psi(t) \rangle = -i \langle H(t) \phi, \psi(t) \rangle \quad \mbox{with} \quad \psi(t_0) = \psi_0
	\end{equation}
	for all $\phi \in \mathcal{D}$ and $t \in I$.
\end{enumerate}
\end{prop}
\begin{proof}
Since $H_0$ is self-adjoint, Stone's theorem implies that $\{ e^{itH_0} \}_{t \in \mathbb{R}}$ is a 
strongly continuous, one-parameter unitary group. In this case, the map $\tilde{H}:I \to \mathcal{B}( \mathcal{H})$
given by 
\begin{equation} \label{ip_twist_ham}
\tilde{H}(t) = e^{it H_0} \Phi(t) e^{-itH_0}
\end{equation}
is clearly pointwise self-adjoint and strongly continuous. Using Proposition~\ref{prop:sols} (v), we conclude that
the unique strong solutions of 
\begin{equation}
\frac{d}{dt} \tilde{U}(t,s) = -i \tilde{H}(t) \tilde{U}(t,s) \quad \mbox{with} \quad \tilde{U}(s,s) = \idtyty
\end{equation}
form a two-parameter family of unitaries $\{ \tilde{U}(t,s) \}_{s,t \in I}$ which satisfy the co-cycle property (\ref{co-cycle_rel}). 
In terms of this family, we define $U: I \times I \to \mathcal{B}( \mathcal{H})$ by setting
\begin{equation} \label{full-propagator}
U(t,s) = e^{-itH_0} \tilde{U}(t,s) e^{isH_0} \, .
\end{equation}
One checks that $\{U(t,s) \}_{s,t \in I}$ is a two-parameter family of unitaries satisfying 
(i) and (ii) above. 

To prove (iii), let $t_0 \in I$ and $\psi_0 \in \mathcal{H}$. 
Define $\psi:I \to \mathcal{H}$ by setting $\psi(t) = U(t,t_0) \psi_0$. Observe that
for any $\phi \in \mathcal{D}$ and each $t \in I$
\begin{equation}
\langle \phi, \psi(t) \rangle = \langle e^{itH_0} \phi, \tilde{U}(t, t_0) e^{it_0H_0} \psi_0 \rangle, 
\end{equation}
with the right-hand-side being a differentiable function of $t$. One calculates that
\begin{eqnarray}
\frac{d}{dt} \langle \phi, \psi(t) \rangle  & = & \langle iH_0 e^{itH_0} \phi, \tilde{U}(t, t_0) e^{it_0H_0} \psi_0 \rangle +  \langle e^{itH_0} \phi, \frac{d}{dt}\tilde{U}(t, t_0) e^{it_0H_0} \psi_0 \rangle \nonumber \\
& = & \langle iH_0 \phi, \psi(t) \rangle + \langle e^{itH_0} \phi, -i e^{itH_0} \Phi(t) \psi(t) \rangle \nonumber \\
& = & -i \langle H(t) \phi, \psi(t) \rangle 
\end{eqnarray}
as claimed.

We need only justify uniqueness of the locally bounded weak solutions. Let $t_0 \in I$, 
$\psi_0 \in \mathcal{H}$, and suppose $\psi_1$ and $\psi_2$ are two locally bounded solutions of the initial value problem
(\ref{weak_SE}). Consider the functions  $\tilde\psi_1(t)=e^{it H_0}\psi_1(t)$ and 
$\tilde\psi_2(t)=e^{it H_0}\psi_2(t)$. It is easy to check that these functions are
locally bounded weak solutions of the Schr\"odinger equation associated to the bounded Hamiltonian
$\tilde{H}(t)$ in (\ref{ip_twist_ham}). As such, they are unique, which may be argued as in the proof
of Proposition~\ref{prop:sols}, and therefore, so too are $\psi_1$ and $\psi_2$.
\end{proof}

In this work, we define the Heisenberg dynamics on a suitable algebra of observables in terms of the strongly continuous 
propagator $U(t,s)$ whose existence is guaranteed by Proposition \ref{prop:tduprop}. We work under assumptions that 
guarantee the uniqueness of bounded weak solutions. Strictly speaking, the uniqueness of the weak solution and the 
possible absence of a strong solution to the Schr\"odinger equation in the Hilbert space will play no role in our analysis. 
More information about the solutions and their uniqueness could, however, be important for the 
unambiguous interpretation of our results. Additional results exist in the literature if one is willing to make additional 
assumptions on $H_0$ and $\Phi(t)$. For example, the following theorem establishes the existence of an invariant domain for
the generator and, consequently, the existence of a unique strong solution for the situation where $H_0$ is semi-bounded 
and $\Phi(t)$ is Lipschitz continuous, which is a common physical situation. As explained in the introduction, there are 
important applications of the methods in this paper to situations where these additional assumptions are not satisfied.

\begin{thm}\label{thm:strongsoln}
Let $H_0$ be a self-adjoint operator with dense domain $\mathcal{D}\subset\cH$ and suppose $H_0\geq 0$.
Suppose $\Phi:\Rl \to \cB(\cH)$ is pointwise self-adjoint and `Lipschitz' continuous in the sense that for any bounded interval
$I\subset \Rl$, there exists a constant $C$ such that for all $s,t\in I$, we have
\be
\Vert (H_0+\idty)^{-1} (\Phi(t)-\Phi(s))(H_0 + \idty)^{-1}\Vert \leq C | t-s|.
\ee
Then, there exists a strongly continuous propagator $U(t,s)$, such that $U(t,s)\mathcal{D}\subset \mathcal{D}$, for all $s\leq t \in I$, 
and such that $t\mapsto U(t,t_0) \psi_0$ is the unique strong solution of 
\be
\frac{d}{dt} \psi(t) = -i(H_0+\Phi(t)) \psi(t) \quad \mbox{with} \quad \psi(t_0) = \psi_0
\ee
for all $\psi_0\in\mathcal{D}$.
\end{thm}

In \cite[Theorem II.21]{simon:1971} Simon credits a version of this theorem to Yosida, who proved it in a more general Banach
space context \cite[Section XIV.4]{yosida:1980}, but with the Lipschitz condition replaced by a boundedness
condition on the derivative of $\Phi(t)$. Yosida gives credit to Kato \cite{kato:1953,kato:1956} and Kisy\'nski \cite{kisynski:1964}. 

%
%

\subsubsection{A Duhamel formula for bounded perturbations depending on a parameter}\label{sec:duhamel}

In this section, we consider families of Hamiltonians $H_{\lambda}(t)$ which depend on
a time-parameter $t \in I$ and an auxillary parameter $\lambda \in J$. For such families, we will
prove a version of the well-known Duhamel formula (Proposition~\ref{prop:duhamel}) 
and use it to derive various continuity properties of the corresponding dynamics (Proposition~\ref{properties_of_expiths}).

Let $H_0$ be a densely defined, self-adjoint operator on a Hilbert space $\mathcal{H}$ and
denote by $\mathcal{D} \subset \mathcal{H}$ the corresponding dense domain.
Let $I,J \subset \mathbb{R}$ be intervals and consider the family of Hamiltonians
$H_{\lambda}(t)$, $t \in I$ and $\lambda \in J$, acting on $\mathcal{D} \subset \mathcal{H}$ given by
\begin{equation} \label{fam_ham_tl}
H_{\lambda}(t) = H_0 + \Phi_{\lambda}(t)
\end{equation}
where for each $t \in I$ and $\lambda \in J$, $\Phi_{\lambda}(t)^*=\Phi_{\lambda}(t) \in \mathcal{B}( \mathcal{H})$. 
The self-adjointness of $H_{\lambda}(t)$ on the common domain, $\mathcal{D}$, is clear.
We will assume that $(t, \lambda) \mapsto \Phi_{\lambda}(t)$ is jointly strongly continuous. We will also assume that
for each fixed $t \in I$, the mapping $\lambda \mapsto \Phi_{\lambda}(t)$ is strongly differentiable 
and that the corresponding derivative, which we denote by $\Phi'_{\lambda}(t)$,
satisfies that the map $(t, \lambda) \mapsto \Phi'_{\lambda}(t)$ is jointly strongly continuous.

Under these assumptions, Proposition \ref{prop:tduprop} guarantees that for each $\lambda \in J$ there
exists a two parameter family of unitaries $\{ U_\lambda(t,s) \}_{s,t \in I}$ which 
generates the weak solutions of the Schr\"odinger equation associated to $H_\lambda(t)$, see (\ref{weak_SE}).
Our goal here is to show that for fixed $s,t \in I$, the map $\lambda \mapsto U_\lambda(t,s)$ is strongly 
differentiable and moreover, 
\be \label{derivativebound}
\left\Vert \frac{d}{d\lambda} U_\lambda(t,s)\right\Vert \leq \int_{\min(s,t)}^{\max(s,t)} \Vert \Phi^\prime_\lambda\Vert (r)dr .
\ee
We will obtain this bound as a corollary of the following proposition, which gives a Duhamel formula for the derivative in this setting. 
Although the Duhamel formula is well-known, we give an explicit proof here that allows us to clarify the continuity properties implied 
by our assumptions. In the proof we avoid taking derivatives with respect to $t$ or $s$ which, in general, 
are unbounded operators.

\begin{prop}[Duhamel Formula]\label{prop:duhamel}
Let $H_\lambda(t)$ be a family of self-adjoint operators as in (\ref{fam_ham_tl}) above 
and let $U_\lambda(t,s)$ denote the corresponding unitary propagator. Then, for all $s,t\in I$ with $s\leq t$, we have that
\be \label{duhamel}
\frac{d}{d\lambda} U_\lambda(t,s) = -i \int_s^t U_\lambda(t,r)\ \Phi^\prime_\lambda(r) U_\lambda(r,s)\, dr
\ee
where the derivative and the integral are to be understood in the strong sense. 
\end{prop}
With stronger assumptions, one can prove (\ref{duhamel}) holds in norm. In fact, arguing as below, if
\begin{enumerate}
	\item[(i)] the map $(t, \lambda) \mapsto \Phi_{\lambda}(t)$ is jointly norm continuous,
	\item[(ii)] for each $t \in I$, the map $\lambda \mapsto \Phi_\lambda(t)$ is norm differentiable; with derivative denoted by $\Phi'_{\lambda}(t)$, and
	\item[(iii)] the map $(t, \lambda) \mapsto \Phi_{\lambda}'(t)$ is jointly norm continuous,  
\end{enumerate}
then $\lambda \mapsto U_\lambda(t,s)$ is norm differentiable and its derivative satisfies (\ref{duhamel}).
\begin{proof}
Recall that the unitary propagator $U_{\lambda}(t,s)$, as defined in the proof of Proposition~\ref{prop:tduprop}, is
\be \label{expiths}
U_\lambda(t,s) = e^{-i t H_0} \tilde{U}_\lambda(t,s)e^{isH_0} 
\ee
where $\tilde{U}_{\lambda}(t,s)$ is the unique strong solution of 
\be \label{twisted_de_duh}
\frac{d}{dt} \tilde{U}_{\lambda}(t,s) = -i \tilde{\Phi}_{\lambda}(t) \tilde{U}_{\lambda}(t,s) \quad \mbox{with} \quad \tilde{U}_{\lambda}(s,s) = \idtyty \quad \mbox{and} \quad
\tilde{\Phi}_\lambda(t) = e^{i t H_0} \Phi_\lambda(t) e^{-i t H_0}.
\ee 
We first prove the analogue of (\ref{duhamel}) for $\tilde{U}_\lambda(t,s)$, i.e.
\be \label{duhamelUst}
\frac{d}{d\lambda} \tilde{U}_\lambda(t,s) = -i \int_s^t \tilde{U}_\lambda(t,r) \tilde{\Phi}^\prime_\lambda(r) \tilde{U}_\lambda(r,s) \, dr.
\ee
Given \eq{duhamelUst}, the $\lambda$-derivative of \eq{expiths} is easily seen to satisfy \eq{duhamel}. 

We now show \eq{duhamelUst}. The unique strong solution of (\ref{twisted_de_duh}) is given by the Dyson series
\be \label{dysonUstr}
\tilde{U}_\lambda(t,s) =  \idty + \sum_{n=1}^{\infty} (-i)^n \int_{s}^t \int_{s}^{t_1} \cdots \int_{s}^{t_{n-1}} \tilde{\Phi}_\lambda(t_1) \cdots \tilde{\Phi}_\lambda(t_n) \, dt_n \cdots dt_1\,.
\ee
To each $n \geq 1$, define a map $\Psi_{\lambda} : I^n \to \mathcal{B}( \mathcal{H})$ by setting
\begin{equation}
\Psi_{\lambda}(t_1, \cdots, t_n) = \tilde{\Phi}_{\lambda}(t_1) \cdots \tilde{\Phi}_{\lambda}(t_n) \quad \mbox{for any } (t_1, \cdots. t_n) \in I^n \, .
\end{equation} 
With $(t_1, \cdots, t_n) \in I^n$ fixed, our assumptions imply that $\lambda \mapsto \Psi_{\lambda}(t_1, \cdots, t_n)$ is
strongly differentiable and moreover
\begin{equation}
\frac{d}{d \lambda} \Psi_{\lambda}(t_1, \cdots, t_n)  =  \sum_{k=1}^n \tilde{\Phi}_{\lambda}(t_1) \cdots \tilde{\Phi}_{\lambda}(t_{k-1})\tilde{\Phi}_{\lambda}'(t_{k})\tilde{\Phi}_{\lambda}(t_{k+1})\cdots \tilde{\Phi}_{\lambda}(t_n).
\end{equation}
The joint strong continuity of $(t, \lambda) \mapsto \tilde{\Phi}_{\lambda}(t)$ and $(t, \lambda) \mapsto \tilde{\Phi}_{\lambda}'(t)$ can be
used to justify term-by-term differentiation of the Dyson series (\ref{dysonUstr}), and we obtain
\be \label{dyson-derived}
\frac{d}{d\lambda} \tilde{U}_\lambda(t,s)
= \sum_{n=1}^{\infty} (-i)^n
\int_{s}^t \cdots \int_{s}^{t_{n-1}} \frac{d}{d \lambda} \Psi_{\lambda}(t_1, \cdots, t_n) \, dt_n \cdots dt_1\,.
\ee

The proof of (\ref{duhamelUst}) is now completed by demonstrating that upon inserting the Dyson 
series for $\tilde{U}_{\lambda}(t,r)$ and $\tilde{U}_{\lambda}(r,s)$ into the integral on the right-hand-side 
of (\ref{duhamelUst}), the result simplifies to the expression on the right-hand-side of (\ref{dyson-derived}).  

Note that upon substitution of \eq{dysonUstr} into the right-hand-side of \eq{duhamelUst} we find
\bea
-i \int_s^t  \tilde{U}_\lambda(t,r)\tilde{\Phi}^\prime_\lambda(r)  \tilde{U}_\lambda(r,s) \, dr\nonumber
&=& \sum_{p,q\geq 0}(-i)^{p+q+1}\int_s^t \int_r^t \int_r^{t_1}\cdots \int_r^{t_{p-1}} \int_s^r \int_s^{t_{p+2}}\cdots \int_s^{t_{p+q}}\\
&&\times  \tilde{\Phi}_\lambda(t_1) \cdots \tilde{\Phi}_\lambda(t_p) \tilde{\Phi}^\prime_\lambda(r) \tilde{\Phi}_\lambda(t_{p+2}) \cdots \tilde{\Phi}_\lambda(t_{p+q+1})\\
&&\times dt_{p+q+1} \cdots dt_{p+2} dt_p \cdots dt_1 dr.\nonumber
\eea
Here $p$ (respectively $q$) is the index of the terms in the series for the first (respectively second) propagator, 
and we have taken as integration variables $t_1,\ldots, t_p$ and $t_{p+2},\ldots,t_{p+q+1}$. 
Each integrand above is the product of $n=p+q+1 \geq 1$ operators.
Since the goal is to re-write the above as in \eqref{dyson-derived}, we now re-index 
by writing $p = k-1$ and $q = n-k$ for  $n\geq 1$ and $1\leq k\leq n$.  One sees that
\bea
-i \int_s^t  \tilde{U}_\lambda(t,r) \tilde{\Phi}^\prime_\lambda(r)  \tilde{U}_\lambda(r,s)\, dr\nonumber
&=& \sum_{n=1}^\infty \sum_{k=1}^n (-i)^n
\int_s^t \int_r^t \int_r^{t_1}\cdots \int_r^{t_{k-2}} \int_s^r \int_s^{t_{k+1}}\cdots \int_s^{t_{n-1}}\\
&&\times  \tilde{\Phi}_\lambda(t_1) \cdots \tilde{\Phi}_\lambda(t_{k-1}) \tilde{\Phi}^\prime_\lambda(r) \tilde{\Phi}_\lambda(t_{k+1}) \cdots \tilde{\Phi}_\lambda(t_n)\label{dyson-resummed}\\
&&\times dt_n \cdots dt_{k+1} dt_{k-1} \cdots dt_1 dr.\nonumber
\eea
The identity \eqref{duhamelUst} now follows by comparing, term by term, the integration domains on 
the right-hand-sides of \eqref{dyson-derived} and \eqref{dyson-resummed}. That they are equal follows, e.g.,
by reordering the iterated integrals in (\ref{dyson-resummed}).
\end{proof}

For each $\lambda \in J$, the Heisenberg dynamics $\tau_{t,s}^{\lambda}$, $s,t \in I$, 
associated to the family of Hamiltonians in (\ref{fam_ham_tl}) is the co-cycle of automorphisms of $\mathcal{B}( \mathcal{H})$ given by
\begin{equation} \label{dyn_tl}
\tau_{t,s}^{\lambda}(A) = U_{\lambda}(t,s)^* A U_{\lambda}(t,s) \quad \mbox{for all } A \in \mathcal{B}( \mathcal{H}).
\end{equation} 
As it will be convenient for later applications, we summarize various continuity properties of this
dynamics in the following proposition.

\begin{prop}\label{properties_of_expiths}
Let $H_{\lambda}(t)$ be a family of Hamiltonians as described in (\ref{fam_ham_tl}).
The corresponding dynamics, as in (\ref{dyn_tl}) above, has the following properties:
\begin{enumerate}
\item[(i)] For each $\lambda \in J$ and $A \in \mathcal{B}( \mathcal{H})$, the map $(s,t) \mapsto \tau_{t,s}^{\lambda}(A)$ is
jointly strongly continuous. 
\item[(ii)] For each $s,t \in I$ and $A \in \mathcal{B}( \mathcal{H})$, the map $ \lambda \mapsto \tau^{\lambda}_{t,s}(A)$ is 
strongly differentiable (and hence strongly continuous). Moreover, one has the estimate
\be \label{norm_bd_dyn_der}
\left\Vert \frac{d}{d\lambda} \tau_{t,s}^\lambda(A)\right\Vert \leq 2\Vert A\Vert 
\int_{\rm \min(s,t)}^{\rm \max(s,t)}  \Vert \Phi^\prime_\lambda\Vert (r)dr.
\ee
\item[(iii)] For fixed $s, t \in I$ and $\lambda \in J$, the map $\tau_{t,s}^{\lambda}(\cdot): \mathcal{B}( \mathcal{H}) \to \mathcal{B}(\mathcal{H})$ 
is continuous on bounded sets when both its domain and codomain are equipped with the strong operator topology.
This continuity is uniform for $\lambda$ in compact subsets of $J$.
\end{enumerate}
\end{prop}

\begin{proof}
The statement in (i) follows from Proposition \ref{prop:tduprop} as $\tau_{t,s}^{\lambda}(A)$, see (\ref{dyn_tl}), is the product of jointly strongly continuous mappings.

To prove (ii), we use Duhamel's formula from Proposition \ref{prop:duhamel} to calculate the derivative. Specifically, note that if $s \leq t$, then
\be \label{dyn_duhamel}
\frac{d}{d\lambda} \tau_{t,s}^\lambda(A)= i\int_s^t  \tau^\lambda_{r,s} ([\Phi^\prime_\lambda(r), \tau^\lambda_{t,r} (A)])\, dr  .
\ee
An estimate of the form in (\ref{norm_bd_dyn_der}) is now clear.

To prove (iii), fix $s,t \in I$, and let $[a,b] \subset J$.  Without loss of generality, assume that $s \le t$. Let $\epsilon >0$. Since $(r, \lambda) \mapsto \Phi_{\lambda}'(r)$ is
jointly strongly continuous,
\begin{equation}
M := \sup_{(r, \lambda) \in [s,t] \times [a,b]} \| \Phi_{\lambda}'(r) \| < \infty.
\end{equation}
Take $\delta>0$ so that
\begin{equation}
2 \delta (t-s) M \leq \epsilon.
\end{equation}
By compactness, there is some $N \geq 1$ and numbers $\lambda_1, \cdots, \lambda_N \in [a,b]$ for which the balls of radius $\delta$ centered
at $\lambda_i$, $1 \leq i \leq N$, cover $[a,b]$. Using the result in (ii), we see that for every $\lambda \in [a,b]$ there is
some $1 \leq i \leq N$ for which 
\begin{equation}
\| \tau_{t,s}^{\lambda}(A) - \tau_{t,s}^{\lambda_i}(A) \| \leq \epsilon \| A \| \quad \mbox{for all } A \in \mathcal{B}( \mathcal{H}) \, .
\end{equation}

Now, to prove the continuity statement claimed, let $\{ A_n \}_{n \geq 1} \subseteq \mathcal{B}( \mathcal{H})$ be a bounded sequence that converges to $A \in \mathcal{B}( \mathcal{H})$ in the strong operator topology. Let $B < \infty$ be such that $\sup_{n \geq 1} \| A_n \| \leq B$. Using \eqref{dyn_tl} and the strong convergence of $A_n$ to $A$, it is easy to verify that for any $\psi \in \mathcal{H}$ and any $1 \leq i \leq N$, the sequence 
$\{ \tau_{t,s}^{\lambda_i}(A_n)\psi \}_{n \geq 1}$ converges to $\tau_{t,s}^{\lambda_i}(A)\psi$ in $\mathcal{H}$. Pick $n_0 \geq 1$ so that for all $n \geq n_0$ and each $1 \leq i \leq N$,
\begin{equation}
\Vert \tau^{\lambda_i}_{t,s}(A_n)\psi - \tau^{\lambda_i}_{t,s}(A)\psi\Vert \leq \epsilon B \Vert\psi\Vert.
\end{equation}
In this case, for any $\lambda \in [a,b]$ there is an $i$ for which
\beann
\Vert \tau^{\lambda}_{t,s}(A_n)\psi - \tau^{\lambda}_{t,s}(A)\psi\Vert &\leq& 
\Vert  \tau^{\lambda}_{t,s}(A_n)\psi - \tau^{\lambda_i}_{t,s}(A_n)\psi\Vert \\
&&+ \Vert \tau^{\lambda_i}_{t,s}(A_n)\psi\ - \tau^{\lambda_i}_{t,s}(A)\psi\Vert\\
&&+\Vert \tau^{\lambda_i}_{t,s}(A)\psi\ - \tau^{\lambda}_{t,s}(A)\psi\Vert \\
&\leq& 3\epsilon B  \Vert\psi\Vert
\eeann
whenever $n \geq n_0$. This proves that the strong convergence is uniform for 
$\lambda \in [a,b]$, or in other words, that the family of maps $\{ \tau^{\lambda}_{t,s}(\cdot)
\mid \lambda \in [a,b]\}$, for $s,t \in I$ fixed, is equicontinuous on bounded sets in $\mathcal{B}( \mathcal{H})$ with respect
to the strong operator topology. 
\end{proof}

%
%

\section{Lieb-Robinson bounds and infinite volume dynamics of lattice systems} \label{sec:lrb}

The scope of this paper is lattice models with possibly unbounded single-site Hamiltonians and bounded interactions
that, in general, may be time-dependent. This is the setting in which one expects to obtain Lieb-Robinson bounds 
with estimates in terms of the operator norm of the observables. A well-known example of this situation is the 
quantum rotor model. We will not consider lattice models with unbounded interactions in this work. The only systems 
with unbounded interactions that have been studied so far are oscillator lattice systems for which the interactions are 
quadratic \cite{cramer:2008} or bounded perturbations of quadratic interactions \cite{nachtergaele:2009a,amour:2010} .

In this paper,  the `lattice' in {\em lattice systems} is understood to be a countable metric space $(\Gamma, d)$ (not 
necessarily a lattice in the sense of the linear combinations with integer coefficients of a set of basis vectors in Euclidean 
space). Typically, $\Gamma$ is infinite (or more specifically, has infinite diameter), and 
models are given in terms of Hamiltonians for a family of finite subsets of $\Gamma$. After an initial analysis of the 
finite systems, we study the thermodynamic limit through sequences of increasing and absorbing finite volumes 
$\{ \Lambda_n \}$, i.e. $\Lambda_n \uparrow \Gamma$. Often, the goal is to obtain estimates for the finite systems 
defined on $\Lambda_n$ that are uniform in $n$. The definitions below prepare for this goal. We note that it is perfectly 
possible to consider a finite set $\Gamma$ and apply the results derived in this and the subsequent sections to finite 
systems. We note that some of the conditions we impose are trivially satisfied for finite systems.

The points of $\Gamma$, also called {\em sites} of the lattice, label a family of `small' systems, which are often, but 
not necessarily, identical copies of a given system such as a spin, a particle in a confining potential such as a harmonic 
oscillator, or a quantum rotor. The quantum many-body lattice systems of condensed matter physics are of this type. 
A wide range of interesting behaviors arises due to interactions between the component systems. It is a central feature 
of extended physical systems that interactions have a local structure, meaning that the strength of the interactions 
decreases with the distance between the systems. Often, each system only interacts directly with its nearest neighbors 
in the lattice. The mean-field approximation ignores the geometry of the ambient space and it often is a good first approximation. 
In more realistic models, however, the interactions between different components depends on the distance between them.
In this section we derive a fundamental property of the dynamics of quantum lattice systems that is intimately related 
to the local structure of the interactions. This property is referred to as {\em quasi-locality} and its basic feature is a bound
on the speed of propagation of disturbances in the system, which is known as a Lieb-Robinson bound.

Lieb and Robinson were the first to derive bounds of this type \cite{lieb:1972}. In the years following the original article, 
a number of further important results appeared, e.g., by Radin \cite{radin:1978} and in particular by Robinson \cite{robinson:1976} 
who gave a new proof of the theorem of Lieb and Robinson (which is included in \cite{bratteli:1997}). Robinson also showed
that Lieb-Robinson bounds can be used to prove the existence of the thermodynamic limit of the dynamics and 
used the bounds to derive fundamental locality properties of quantum lattice systems. It was only much later however, that 
Hastings who pointed out how the Lieb-Robinson bounds could be used to prove exponential clustering in gapped 
ground states in a paper where he provided the first generalization of the Lieb-Schultz-Mattis theorem to higher dimensions 
\cite{hastings:2004}. 
Mathematical proofs then followed by Nachtergaele and Sims \cite{nachtergaele:2006a}, Hastings and Koma \cite{hastings:2006}
and Nachtergaele, Ogata, and Sims \cite{nachtergaele:2006}. The new approach to proving Lieb-Robinson bounds developed in
these works leading to \cite{nachtergaele:2009b}, yields a better prefactor with a more accurate dependence on the support of the 
observables. This was important for certain applications such as the proof of the split property for gapped ground states in one 
dimension by Matsui \cite{matsui:2010,matsui:2013}.

Further extensions of the Lieb-Robinson bounds in several directions quickly followed: Lieb-Robinson bounds for lattice fermions 
\cite{hastings:2006, bru:2017, nachtergaele:2016b}, Lieb-Robinson bounds for irreversible quantum dynamics 
\cite{poulin:2010,nachtergaele:2011,han:2018}, a bound for certain long-range interactions \cite{gong:2014,richerme:2014, tran:2018},
anomalous or zero-velocity bounds for disordered and quasi-periodic systems \cite{burrell:2007,hamza:2012, damanik:2014, damanik:2015},	
propagation estimate for lattice oscillator systems \cite{buerschaper:2007,cramer:2008,nachtergaele:2009a, amour:2010} and 
other systems with unbounded interactions \cite{premont-schwarz:2010a}, including classical lattice systems \cite{butta:2007,raz:2009,islambekov:2012}.

The list of applications of Lieb-Robinson bounds includes a broad range of topics:
Lieb-Schultz-Mattis theorems \cite{hastings:2004,nachtergaele:2007}, the entanglement area law in one dimension \cite{Hastings:2007},
the quantum Hall effect \cite{hastings:2015,giuliani:2017,bachmann:2017b}, quasi-adiabatic evolution (spectral flow and automorphic equivalence) including stability and classification of gapped ground state phases \cite{hastings:2005, bravyi:2010, bravyi:2011,michalakis:2013,bachmann:2012,QLBII},
the stability of dissipative systems \cite{brandao:2015,lucia:2015}, quasi-particle structure of the excitation spectrum of gapped 
systems \cite{haegeman:2013, bachmann:2016}, a stability property of the area law of entanglement \cite{marien:2016},
the efficiency of quantum thermodynamic engines \cite{shiraishi:2017}, the adiabatic theorem and linear response theory for extended 
systems \cite{bachmann:2017c, bachmann:2018}, the design and analysis of quantum algorithms \cite{haah:2018}, and the list continues to grow: 
\cite{de-roeck:2015,kanda:2016,monaco:2019,buchholz:2017,grundling:2017,bachmann:2017b,teufel:2017,bachmann:2017c,bachmann:2017d,cha:2018a}. 

In order to express the locality properties of the interactions and the resulting dynamics, we introduce some additional 
structure on the discrete metric space $(\Gamma, d)$ in the next section.

%
%

\subsection{Lieb-Robinson estimates for bounded time-dependent interactions} \label{subsec:lrb}

\subsubsection{General setup}\label{sec:spatialstructure}

As described above, we will study quantum lattice models with possibly unbounded single-site Hamiltonians but bounded, in general, time-dependent interactions. In this section, 
we give the framework for quantum lattice systems and describe the bounded interactions of interest. We will consider the addition of unbounded on-site Hamiltonians in later sections. 

The lattice models we consider are defined over a countable metric space $(\Gamma, \, d)$. To each site $x \in \Gamma$, we associate a complex Hilbert space $\mathcal{H}_x$ and denote the algebra of all bounded linear operators on $\mathcal{H}_x$ by $\mathcal{B}(\mathcal{H}_x)$. Let $\cP_0(\Gamma)$ be the collection of all finite subsets of $\Gamma$. For any $\Lambda \in\cP_0(\Gamma)$, the Hilbert space of states and algebra of local observables over $\Lambda$ are denoted by
\begin{equation}
\mathcal{H}_{\Lambda} := \bigotimes_{x \in \Lambda} \mathcal{H}_x \quad \mbox{and} \quad \mathcal{A}_{\Lambda} := 
\bigotimes_{x \in \Lambda} \mathcal{B}(\mathcal{H}_x) =\mathcal{B}( \mathcal{H}_{\Lambda} ),
\label{HLambda}\end{equation}  
where we have chosen to define the tensor product of the algebras $\mathcal{B}(\mathcal{H}_x)$ so that the last equality holds (i.e. the spatial tensor product, corresponding to the minimal $C^*$-norm \cite{sakai:1971}).
For any two finite sets $\Lambda_0 \subset \Lambda\subset \Gamma$, each $A \in \mathcal{A}_{\Lambda_0}$ can be naturally identified with $A \otimes \idty_{\Lambda \setminus \Lambda_0} \in \mathcal{A}_{\Lambda}$. With respect to this identification, the algebra of {\em local observables} is then defined as the inductive limit
\begin{equation}
\mathcal{A}_{\Gamma}^{\rm loc} = \bigcup_{\Lambda \in \cP_0(\Gamma)} \mathcal{A}_{\Lambda},
\end{equation}
and the $C^*$-algebra of \emph{quasi-local observables}, which we denote by $\mathcal{A}_{\Gamma}$, is the completion of $\mathcal{A}_{\Gamma}^{\rm loc}$ with respect to the operator norm. We will use the phrase \emph{quantum lattice system} to mean the countable metric space $(\Gamma, d)$ and quasi-local algebra $\cA_\Gamma$. 

A model on a quantum lattice system is given in terms of an \emph{interaction} $\Phi$. In the time-independent case, this is a map $\Phi : \cP_0(\Gamma)\to \cA_\Gamma^{\rm loc}$ such that 
$
\Phi(Z)^* = \Phi(Z) \in \mathcal{A}_Z $ for all $Z \in \mathcal{P}_0( \Gamma).
$
The \emph{quantum lattice model} associated to $\Phi$ is the collection of all local Hamiltonians of the form
\be \label{tind_loc_hams}
H_\Lambda = \sum_{X \subset \Lambda} \Phi(X), \quad \Lambda \in \cP_0(\Gamma).
\ee

We will also consider time-dependent interactions. Let $I \subset \mathbb{R}$ be an interval. 
A map $\Phi : \cP_0(\Gamma) \times I \to \cA_\Gamma^{\rm loc}$ is said to be a \emph{strongly continuous interaction} if:
\begin{enumerate}
	\item[(i)] To each $t \in I$, the map $\Phi( \cdot, t) : \mathcal{P}_0( \Gamma) \to \mathcal{A}_{\Gamma}^{\rm loc}$ is an interaction.
	\item[(ii)]For each $Z \in \mathcal{P}_0( \Gamma)$, $\Phi(Z, \cdot) :I \to \mathcal{A}_Z$ is strongly continuous.
\end{enumerate}
\noindent Given such a strongly continuous interaction $\Phi$, we will often denote by $\Phi(t)$ the interaction $\Phi( \cdot, t)$ as in (i) above, and define the corresponding local Hamiltonians 
\begin{equation} \label{tdham}
H_{\Lambda}(t) = \sum_{Z \subset \Lambda} \Phi(Z,t) \quad \mbox{for } \Lambda \in \mathcal{P}_0( \Gamma).
\end{equation}
Analogous to the above, a corresponding time-dependent quantum lattice model may be defined. By our assumptions on the interaction, it is clear that for each $t \in I$, $H_{\Lambda}(t)$ is a bounded, self-adjoint operator on $\mathcal{H}_{\Lambda}$. Moreover, by Proposition~\ref{prop:st_cont_prods} $H_{\Lambda}: I \to \mathcal{A}_{\Lambda}$ is strongly continuous. In this case, Proposition \ref{prop:sols} demonstrates that there exists a two-parameter family of unitaries $\{ U_{\Lambda}(t,s) \}_{s,t \in I}\subset \mathcal{A}_{\Lambda}$, defined as the unique strong solution of the initial value problem
\begin{equation} \label{uprop}
\frac{d}{dt} U_{\Lambda}(t,s) = -i H_{\Lambda}(t) U_{\Lambda}(t,s),\quad U_{\Lambda}(s,s) = \idty \, , \text{ for all  } s,t \in I.
\end{equation}
In terms of these unitary propagators, we define a Heisenberg dynamics $\tau_{t,s}^{\Lambda}: \cA_\Lambda \to \cA_\Lambda$ by setting 
\begin{equation} \label{td_heis_dyn}
\tau_{t,s}^{\Lambda}(A) = U_{\Lambda}(t,s)^* A U_{\Lambda}(t,s)
\quad \mbox{for all } A \in \mathcal{A}_{\Lambda}.
\end{equation}
In some applications, including Theorem \ref{thm:lrb} below, we will also consider the inverse dynamics, 
\be \label{inv_dyn}
\hat{ \tau}^\Lambda_{t,s}(A) : = U_\Lambda(t,s) A U_\Lambda(t,s)^* = \tau^\Lambda_{s,t}(A),
\ee
where the final equality follows from Proposition \ref{prop:sols} (iv).

As discussed above, Lieb-Robinson bounds approximate the speed of propagation of dynamically evolved observables through a quantum lattice system, and this estimate is closely tied to the locality of the interaction in question. To quantify the locality of an interaction, we introduce the notion of an $F$-function. 
An \emph{$F$-function} on $(\Gamma, d)$ is a non-increasing function $F: [0, \infty) \to (0, \infty)$, satisfying the following two properties:

(i) $F$ is uniformly integrable over $\Gamma$, i.e.
\begin{equation} \label{F:int}
\| F \| = \sup_{x \in \Gamma} \sum_{y \in \Gamma} F(d(x,y)) < \infty,
\end{equation}

(ii) $F$ satisfies the convolution condition
\begin{equation} \label{F:conv}
C_F = \sup_{x,y \in \Gamma} \sum_{z \in \Gamma} \frac{F(d(x,z))F(d(z,y))}{F(d(x,y))} < \infty. 
\end{equation}
An equivalent formulation of (ii) is that there exists a constant $C<\infty$ such that 
\begin{equation} \label{Fconvbis}
\sum_{z \in \Gamma} F(d(x,z))F(d(z,y)) \leq C_F F(d(x,y)), \mbox{ for all } x,y\in\Gamma.
\end{equation}

Let $F$ be an $F$-function on $(\Gamma, d)$ and $g:[0,\infty)\to [0,\infty)$ be any non-decreasing, subadditive function, i.e. $g(r+s) \leq g(r) + g(s)$ for all $r,s\in [0,\infty)$. 
Then, the function 
\be \label{weightedF}
F_g(r) = e^{-g(r)} F(r),
\ee
also satisfies (i) and (ii) with $\| F_g \| \leq \| F \|$ and $C_{F_g} \leq C_F$. We call any $F$-function of this
form a \emph{weighted F-function}.

It is easy to produce examples of these $F$-functions when $\Gamma = \mathbb{Z}^{\nu}$ for some $\nu \geq 1$ 
and $d(x,y) = |x-y|$ is the $\ell^1$-distance. In fact, for any $\epsilon >0$ the function
\begin{equation}
F(r) = \frac{1}{(1+r)^{\nu + \epsilon}}
\end{equation}
is an $F$-function on $\bZ^\nu$. It is clear that this function is uniformly integrable, i.e. (\ref{F:int}) holds. Moreover, one may verify that
\begin{equation}
C_F \leq 2^{\nu + \epsilon} \| F \| \, .
\end{equation}
In the special case of $g(r) = a r$, for some $a\geq 0$, we obtain a very useful family of weighted $F$-functions, which we denote by $F_a$, given by $F_a(r)=e^{-ar}/(1+r)^{\nu+\epsilon}$.  See Appendix \ref{app:sec_def_F}-\ref{subsec:regsets} for other examples and properties of $F$-functions.

We use these $F$-functions to describe the decay of a given interaction.
Let $F$ be an $F$-function on $(\Gamma, d)$ and $\Phi:\cP_0(\Gamma) \to \caA_\Gamma^{\rm loc}$ be
an interaction. The {\it $F$-norm} of $\Phi$ is defined by
\begin{equation} \label{intnorm}
\| \Phi \|_F  = \sup_{x,y \in \Gamma} \frac{1}{F(d(x,y))} \sum_{\stackrel{Z \in \mathcal{P}_0( \Gamma):}{x,y \in Z}} \| \Phi(Z) \| .
\end{equation}
It is clear from the above equation that for all $x, \, y\in \Gamma$,
\be \label{Fnorm_sum_bd}
\sum_{\stackrel{Z \in \mathcal{P}_0( \Gamma):}{x,y \in Z}} \| \Phi(Z) \| \leq \|\Phi\|_F F(d(x,y)).
\ee
Note that for any $Z \in \mathcal{P}_0(\Gamma)$, there exist $x,y \in Z$ for which $d(x,y) = {\rm diam}(Z)$;
the latter being the diameter of $Z$. In this case, a simple consequence of (\ref{Fnorm_sum_bd}) is
\begin{equation} \label{Fnorm_term_bd}
\| \Phi(Z) \| \leq \sum_{\stackrel{Z' \in \mathcal{P}_0( \Gamma):}{x,y \in Z'}} \| \Phi(Z') \| \leq \| \Phi \|_F F({\rm diam}(Z)) .
\end{equation}
We will be mainly interested in situations where the quantity in (\ref{intnorm}) is finite. 
In this case, the bound (\ref{Fnorm_term_bd}) demonstrates that the $F$-function 
governs the decay of an individual interaction term, and moreover, the estimate 
(\ref{Fnorm_sum_bd}) generalizes this notion of decay by including all interaction terms 
containing a fixed pair of points $x$ and $y$. 

When $\Gamma$ is finite, then $\| \Phi \|_F$ is finite for any interaction $\Phi$ and any function $F$. 
For infinite $\Gamma$, the set of interactions $\Phi$ for which 
$\| \Phi \|_F  < \infty$ depends on $F$. It is easy to check that $\| \cdot \|_F $ is a norm on the set of interactions for which it is finite. In terms of this norm, we define the Banach space
\be
\cB_F = \{ \Phi : \cP_0(\Gamma)\to \cA_\Gamma^{\rm loc} \mid \Phi \mbox{ is an interaction and } \| \Phi \|_F < \infty\}.
\ee
Of course, $\cB_F$ depends on $\Gamma$ and on the single-site Hilbert spaces $\cH_x$, but that information will always be clear from the 
context.

We introduce an analogue of (\ref{intnorm}) for time-dependent interactions as follows. Consider a quantum lattice
system comprised of $(\Gamma, d)$ and $\mathcal{A}_{\Gamma}$. Let $I \subset \mathbb{R}$ be an interval and 
$\Phi: \mathcal{P}_0( \Gamma) \times I \to \mathcal{A}_{\Gamma}^{\rm loc}$ be a strongly continuous 
interaction. 
%
%
Given an $F$-function on $(\Gamma, d)$, 
we will denote by $\cB_F(I)$ the collection of all strongly continuous interactions $\Phi$ for which the mapping 
\be \label{tdintnorm}
\Vert \Phi (t)\Vert_F = \sup_{x,y \in \Gamma} \frac{1}{F(d(x,y))} \sum_{\stackrel{Z \in \mathcal{P}_0( \Gamma):}{x,y \in Z}} \| \Phi(Z,t) \|  , \quad \mbox{for } t \in I 
\ee
is locally bounded. As with the operator norm, we will sometimes use the alternate notation $\Vert \Phi \Vert_F (t)$ for the quantity defined in \eq{tdintnorm}.
The function $t \mapsto \Vert \Phi \Vert_F (t)$ is measurable since it is the supremum of a countable 
family of measurable functions. As such, $\Vert \Phi\Vert_F$ is locally integrable.
As in the time-independent case, \eqref{tdintnorm} implies that for all $t\in I$ and $x,\, y\in \Gamma$, 
\be \label{FnormBound}
\sum_{\stackrel{Z \in \mathcal{P}_0( \Gamma):}{x,y \in Z}} \| \Phi(Z,t) \| \leq \|\Phi\|_F(t) F(d(x,y)),
\ee
a bound which will appear in many of our estimates. 
See Appendix~\ref{app:sec:BasicBounds} for more useful estimates involving interactions $\Phi \in \cB_F(I)$.

\subsubsection{Lieb-Robinson estimates for bounded interactions} 
In Theorem \ref{thm:lrb}, we demonstrate that the finite volume Heisenberg dynamics $\tau_{t,s}^{\Lambda}$, as defined in
(\ref{td_heis_dyn}), associated to any $\Phi \in \mathcal{B}_F(I)$ satisfies a Lieb-Robinson bound.
Such bounds provide an estimate for the speed of propagation of dynamically evolved observables in a quantum lattice system. One can use these bounds to show that for small times the dynamically evolved observable is well approximated by a local operator. For this reason, Lieb-Robinson bounds and other similar results are often referred to as quasi-locality estimates.

Before we state the result, two more pieces of notation will be useful.
First, to each $X \in \mathcal{P}_0( \Gamma)$, we denote by $\partial^I_{\Phi}X \subset X$ the $\Phi$-boundary of
$X$:
\begin{equation}
\partial_{\Phi}^I X : = \left\{ x \in X : \exists Z \in \mathcal{P}_0( \Gamma) \mbox{ with } x \in Z, Z \cap ( \Gamma \setminus X) \neq \emptyset, \mbox{ and } \exists t \in I 
\mbox{ with } \Phi(Z,t) \neq 0 \right\} \, .
\end{equation} 
In some estimates, it may be useful to restrict the time interval used to define the $\Phi$-boundary. For instance, given $\Phi\in\cB_F(\Rl)$ one could find that $\partial^\Rl_{\Phi}X=X$ for some $X$, while $\partial^I_{\Phi}X$ is strictly smaller for a subinterval $I \subset \bR$. From now on we will drop the time-interval $I$ from the notation and simply write $\partial_{\Phi}X$. We note also that in many situations, not much is lost by using $X$ instead of $\partial_{\Phi}X$ in the following estimates.

Second, for $\Phi\in\cB_F(I)$, and $s,t\in I$, the quantity $I_{t,s}(\Phi)$  defined by
\be
I_{t,s}(\Phi) = C_F\int_{\min(t,s)}^{\max(t,s)} \| \Phi \|_F(r) \, dr,
\label{ItsPhi}\ee
will appear in many results we provide, including Theorem~\ref{thm:lrb}. 
Clearly, if $C_F\Vert \Phi (r) \Vert_F\leq M$, for all $r\in[\min(t,s),\max(t,s)]$, 
we have $I_{t,s}(\Phi) \leq |t-s| M$. For example we see that
\[
I_{t,s}(\Phi) \leq C_F |t-s| \vertiii{\Phi}_F,
\]
with
\be \label{SupFnorm}
\vertiii{\Phi}_{F} := \sup_{t\in I} \Vert \Phi (t) \Vert_{F}.
\ee

 \begin{thm}[Lieb-Robinson Bound] \label{thm:lrb} Let $\Phi \in \cB_F(I)$ and $X,Y \in \mathcal{P}_0( \Gamma)$ with $X \cap Y = \emptyset$.
 For any $\Lambda \in \mathcal{P}_0( \Gamma)$ with 
 $X \cup Y \subset \Lambda$ and any $A \in \mathcal{A}_X$ and $B \in \mathcal{A}_Y$, we have
 \begin{equation} \label{lrbest}
 \left\| \left[ \tau_{t,s}^{\Lambda}(A), B \right] \right\| \leq \frac{2 \| A \| \| B \|}{C_F} \left( e^{2 I_{t,s}(\Phi)} - 1 \right) D(X,Y)
 \end{equation}
for all $t,s \in I$. Here, $C_F$ is the constant in \eq{F:conv}, and the quantity $D(X,Y)$ is given by
 \begin{equation} \label{lrbmin}
 D(X,Y) = \min \left\{ \sum_{x \in X} \sum_{y \in \partial_{\Phi} Y} F(d(x,y)), \sum_{x \in \partial_{\Phi} X} \sum_{y \in Y} F(d(x,y))\right\}.
 \end{equation}
 \end{thm}
 
It is easy to see that with the definition $F_1(r)= C_F^{-1} F(r)$, $F_1$ is a new $F$-function in terms of which the bound \eq{lrbest} slightly 
simplifies in the sense that $C_{F_1}=1$. This is a general feature of our estimates involving $F$-functions and the associated norms on the 
interactions. In the following sections a variety of different $F$-functions will be used. Often, new  $F$-functions are obtained by elementary 
transformations of old ones, see, e.g., Section~\ref{subsec:regsets}. Instead of figuring out the normalization constants that make $C_F=1$ for each of the  $F$-functions, we note 
that the final result can be expressed with a renormalized $F$-function such that $C_F=1$.

Before moving on to the proof of the theorem, we make two simple remarks that are implicit in many applications of the Lieb-Robinson bounds.
First, one trivially has $ \left\| \left[ \tau_{t,s}^{\Lambda}(A), B \right] \right\| \leq 2 \| A \| \| B \|$. 
Second, in the case that $\Phi\in B_{F_g}(I)$, for a weighted $F$-function $F_g(r)=e^{-g(r)}F(r)$, we can further estimate
 \begin{eqnarray}
 D(X,Y)
 &\leq& \min \left\{ \sum_{x \in X} \sum_{y \in \partial_{\Phi} Y} F(d(x,y)), \sum_{x \in \partial_{\Phi} X} \sum_{y \in Y} F(d(x,y))\right\} 
 e^{-g(d(X,Y))}
 \nonumber\\
 &\leq& \min \{\vert \partial_{\Phi} X\vert,\vert \partial_{\Phi} Y\vert\} \Vert F\Vert e^{-g(d(X,Y))},
\end{eqnarray}
where $d(X,Y)$ is the distance between $X$ and $Y$. When $g(r) = ar$ for some $a>0$ (i.e. $\Phi \in B_{F_a}(I)$ with $F_a(r)=e^{-ar}F(r)$) it makes sense to define the quantity $v_a = 2a^{-1}C_{F_a} \vertiii{\Phi}_{F_a}$ which is often referred to as the {\em Lieb-Robinson velocity}, or more correctly a bound for the speed of propagation of any type of disturbance or signal in the system. In terms of $v_a$, \eq{lrbest} implies the more transparent estimate
\be
\left\| \left[ \tau_{t,s}^{\Lambda}(A), B \right] \right\| \leq 2 \| A \| \| B \| \Vert F\Vert C_{F_a}^{-1}
\min \{\vert \partial_{\Phi} X\vert,\vert \partial_{\Phi} Y\vert\} e^{a (v_a |t-s| - d(X,Y))}.
\label{lr_traditional}\ee
Note that the RHS of the bounds in \eq{lrbest} and \eq{lr_traditional} are expressed in terms of quantities defined over the system on $\Gamma$ and, in particular, these estimates are uniform in the choice of the finite set $\Lambda \subset \Gamma$. This fact will be vital in many applications. 

Before we prove Theorem \ref{thm:lrb}, we first prove a lemma. For this lemma and later use,
we define the `surface' of $X$ in the volume $\Lambda$, denoted $S_{\Lambda}(X)$, as follows:
\begin{equation}\label{surface_set}
S_{\Lambda}(X) = \{ Z \subset \Lambda : Z \cap X \neq \emptyset \mbox{ and } Z \cap ( \Lambda \setminus X) \neq \emptyset \}.
\end{equation}
It is simply the set of supports of the interaction terms that connect $X$ and $\Lambda\setminus X$.
We will also use the following notation, for $X,Y\in \mathcal{P}_0(\Gamma)$:
\be
\delta_Y(X)= \begin{cases}
 0 & \mbox{if } X\cap Y = \emptyset\\ 1 & \mbox{if }  X\cap Y \ne \emptyset.\end{cases}
\ee

\begin{lem} \label{lem:normpresapp}  Let $\Phi \in \mathcal{B}_F(I)$. Fix $Y \in \mathcal{P}_0( \Gamma)$, 
$B \in \mathcal{A}_Y$, and $\Lambda \in \mathcal{P}_0( \Gamma)$ with $Y \subset \Lambda$. 
For any $X \subset \Lambda$, the family mappings $g^{X,B}_{t,s} : \mathcal{A}_X \to \mathcal{A}_{\Lambda}$ for $t,s \in I$, defined by
\begin{equation}
g^{X,B}_{t,s}(A) = [ \tau_{t,s}^{\Lambda}(A), B ] 
\end{equation} 
are norm-continuous; more precisely, $(s,t) \mapsto g_{t,s}^{X,B}$ is jointly continuous in the norm on
$\mathcal{B}( \mathcal{A}_X, \mathcal{A}_{\Lambda})$.
Moreover, for fixed $t$ and $s$, the mapping $g_{t,s}^{X, B}$ satisfies
\begin{equation} \label{gennormbd}
\| g^{X,B}_{t,s} \| \leq 2 \| B \| \delta_Y(X) + 2 \sum_{Z \in S_{\Lambda}(X)} \int_{\min(t,s)}^{\max(t,s)} 
\| g^{Z,B}_{r,s}( \Phi(Z,r)) \| \, dr. 
\end{equation}
\end{lem}

The continuity of $g_{t,s}^{X,B}$ follows directly from the joint norm continuity of $(s,t) \mapsto U_{\Lambda}(t,s)$ as proven in 
Proposition~\ref{prop:sols} (v), see also statements following. In fact, for any $A \in \cA_X$, one has the
simple estimate
\begin{equation}
\| g_{t,s}^{X,B}(A) - g_{t_0,s_0}^{X,B}(A) \| \leq 2 \| B \| \| \tau_{t,s}^{\Lambda}(A) - \tau_{t_0,s_0}^{\Lambda}(A) \| \leq 4 \| A \| \| B \| \| U_{\Lambda}(t,s) - U_{\Lambda}(t_0,s_0) \| \, .
\end{equation}
In general, this continuity does not carry over to the thermodynamic limit. Of course, we always have
	\be
	\| g^{X,B}_{t,s}(A(t)) \|\leq  \| g^{X,B}_{t,s} \Vert \, \Vert A(t)\Vert.
	\label{normineqs}\ee

We also note that the map $g^{X,B}_{t,s}$ equals the restriction of $g^{\Lambda,B}_{t,s}$ to $\cA_X$. It is useful, however, to consider them as separate maps for each $X\subset \Lambda$, because the estimates for their norms depend crucially on $X$ through $S_\Lambda(X)$. Also note that each $g^{X,B}_{t,s}$ only depends on interaction terms $\Phi(Z,r)$ such that $Z\subset \Lambda$ and $r\in [\min(t,s),\max(t,s)]$. 

\begin{proof}[Proof of Lemma~\ref{lem:normpresapp}]
Fix $X \subset \Lambda$, $A \in \mathcal{A}_X$, and $s \in I$. Recall that the inverse dynamics is given by
\begin{equation}
\hat{\tau}^X_{t,s}(A) = U_X(t,s) A U_X(t,s)^*,
\end{equation}
where the unitary mappings $U_X(t,s)$ are defined as in (\ref{uprop}), see also (\ref{tdham}), with $\Lambda = X$. Consider the function $f_s: I \to \mathcal{A}_{\Lambda}$ given by $f_s(t) = g_{t,s}^{X,B}(\hat{ \tau}_{t,s}^Xs(A))$. It follows that $f_s(t) = [ \tau_{t,s}^{\Lambda} \circ \hat{\tau}_{t,s}^X (A), B ]$ is strongly differentiable in $t$ and a short calculation shows that
\begin{eqnarray}
\frac{d}{dt} f_s(t) & = & i \left[ \tau_{t,s}^{\Lambda} \left( \left[ H_{\Lambda}(t) - H_X(t), \hat{\tau}_{t,s}^X(A) \right] \right), B \right]  \nonumber \\
& = & i \sum_{Z \in S_{\Lambda}(X)} \left[ \left[ \tau_{t,s}^{\Lambda}( \Phi(Z,t)), \tau_{t,s}^{\Lambda} \circ \hat{\tau}_{t,s}^X(A) \right], B \right] \nonumber \\
& = & i \sum_{Z \in S_{\Lambda}(X)} \left[ \tau_{t,s}^{\Lambda}( \Phi(Z,t)), f_s(t) \right] - i \sum_{Z \in S_{\Lambda}(X)} \left[ \tau_{t,s}^{\Lambda} \circ \hat{\tau}_{t,s}^X(A), \left[ \tau_{t,s}^{\Lambda}( \Phi(Z,t)), B \right] \right], 
\end{eqnarray}
where: for the first equality we have used that the adjoint of the unitary propagator has a strong derivative which can be calculated using (\ref{strong_adjoint_de}), 
for the second equality we have used that ${\rm supp}(\hat{\tau}_{t,s}^X (A)) \subset X$, and for the last equality we used the Jacobi identity. Hence,
\be
\frac{d}{dt}f_s(t) = - i[C(t), \, f_s(t)] + D(t)
\ee
where
\begin{equation}
C(t) = - \sum_{Z \in S_{\Lambda}(X)} \tau_{t,s}^{\Lambda}( \Phi(Z,t)) \quad \mbox{and} \quad D(t) =  -i \sum_{Z \in S_{\Lambda}(X)} \left[ \tau_{t,s}^{\Lambda} \circ \hat{\tau}_{t,s}^X(A), \left[ \tau_{t,s}^{\Lambda}( \Phi(Z,t)), B \right] \right]. 
\end{equation}
Since $C$ and $D$ are finite sums and products of strongly continuous functions with $C(t)=C(t)^*$, they satisfy the assumptions on $A$ and $B$, respectively, in Lemma~\ref{lem:normbd} with $t_0=s$. Thus, we have
\begin{equation} \label{comp_bound}
\left\| \left[ \tau_{t,s}^{\Lambda} \circ \hat{\tau}_{t,s}^X(A), B \right] \right\| \leq \| [A,B] \| + 2 \| A \| \sum_{Z \in S_{\Lambda}(X)} \int_{{\rm min}(t,s)}^{ {\rm max}(t,s)} 
\left\| \left[ \tau_{r,s}^{\Lambda}( \Phi(Z,r)), B \right] \right\| \, dr . 
\end{equation}
As $f_s(t) = g_{t,s}^{X,B}( \hat{\tau}_{t,s}^X(A))$, the bound claimed in (\ref{gennormbd}) follows by applying \eqref{comp_bound} to $\tilde{A} = \tau_{t,s}^X(A)$.
\end{proof}

\begin{proof}[Proof of Theorem~\ref{thm:lrb}]
Below, we will prove that $\| [ \tau_{t,s}^{\Lambda}(A), B ] \| $ satisfies the estimate (\ref{lrbest}) with 
\begin{equation}
D(X,Y) = \sum_{x \in \partial_{\Phi}X} \sum_{y \in Y} F(d(x,y)). 
\end{equation}
Since we also have that
\begin{equation}
\| [ \tau_{t,s}^{\Lambda}(A), B ] \| = \| \tau_{t,s}^{\Lambda} \left( [A, \hat{ \tau}_{t,s}^\Lambda(B)] \right) \| = \| [ \tau_{s,t}^{\Lambda}(B), A] \| \, ,
\end{equation}
the bound in (\ref{lrbest}) with $D(X,Y)$ defined to be the minimum in (\ref{lrbmin}) is also clear.

Let $X$, $Y$, $\Lambda$, $A$, and $B$ be as in Theorem~\ref{thm:lrb}. An application of Lemma~\ref{lem:normpresapp} demonstrates that
\begin{equation} \label{iterandum}
\Vert [\tau_{t,s}^{\Lambda}(A),B]\Vert \leq 2 \Vert A\Vert \Vert B\Vert \delta_Y(X)
+2\Vert A\Vert \sum_{Z \in S_{\Lambda}(X)} \int_{{\rm min}(s,t)}^{{\rm max}(s,t)} \Vert [\tau_{r,s}^{\Lambda}(\Phi(Z,r)),B]\Vert \, dr
\end{equation}
for all $s,t \in I$. As such, it suffices to consider the case $s \leq t$. Applying the bound (\ref{normineqs}) to the
integrand in (\ref{iterandum}), it is clear that we may iteratively apply Lemma~\ref{lem:normpresapp}.
As a result, for any $N \geq 1$
\begin{equation}
\Vert [\tau_{t,s}^{\Lambda}(A),B]\Vert \leq 2 \Vert A\Vert \Vert B\Vert  \left( \delta_Y(X) +  \sum_{n=1}^N a_n(t) \right) + R_{N+1}(t) 
\end{equation}
where
\begin{eqnarray}
a_n(t) = 2^n \sum_{Z_1 \in S_{\Lambda}(X)} \sum_{Z_2 \in S_{\Lambda}(Z_1)} \cdots \sum_{Z_n \in S_{\Lambda}(Z_{n-1})}  \delta_Y(Z_n) \int_s^t \int_s^{r_1} \cdots \int_s^{r_{n-1}} \times \nonumber \\
\times \left( \prod_{j=1}^{n} \| \Phi (Z_j, r_j) \| \right)  dr_n dr_{n-1} \cdots d r_1 
\end{eqnarray}
and 
\begin{eqnarray}\label{def:rem}
R_{N+1}(t) = 2^{N+1} \sum_{Z_1 \in S_{\Lambda}(X)} \sum_{Z_2 \in S_{\Lambda}(Z_1)}  \cdots  \sum_{Z_{N+1} \in S_{\Lambda}(Z_N)} \int_s^t \int_s^{r_1} \cdots  \int_s^{r_{N}} \times \nonumber \\  \times \left( \prod_{j=1}^{N} \| \Phi (Z_j, r_j) \| \right) \Vert [\tau_{r_{N+1},s}^{\Lambda}(\Phi(Z_{N+1},r_{N+1})),B]\Vert dr_{N+1} d r_{N} \cdots dr_1.
\end{eqnarray}

The remainder term $R_{N+1}(t)$ is estimated as follows. 
First, we observe that
\begin{equation}
\Vert [\tau_{r_{N+1},s}^{\Lambda}(\Phi(Z_{N+1}, r_{N+1})),B]\Vert \leq 2 \| B \|\, \| \Phi(Z_{N+1},r_{N+1}) \|  \, .
\end{equation}
Next, we note that the sums above are in fact sums over chains of sets $(Z_1, Z_2, \cdots, Z_{N+1})$
which satisfy $Z_1 \cap \partial_{\Phi}X \neq \emptyset$ and $Z_j \cap Z_{j-1} \neq \emptyset$ for $2 \leq j \leq N+1$. As such, there are points $w_1, w_2, \cdots, w_{N+1} \in \Lambda$ such that $w_1 \in Z_1 \cap \partial_{\Phi}X$ and
$w_j \in Z_j \cap Z_{j-1}$ for all $2 \leq j \leq N+1$. A simple upper bound on these sums is then obtained by overcounting:
\begin{equation} \label{sumbd}
\sum_{Z_1 \in S_{\Lambda}(X)} \sum_{Z_2 \in S_{\Lambda}(Z_1)}  \cdots  \sum_{Z_{N+1} \in S_{\Lambda}(Z_N)} *\quad \leq 
\sum_{w_1 \in \partial_{\Phi}X}\sum_{w_2,\ldots, w_{N+2} \in \Lambda} \sum_{\stackrel{Z_1, \ldots, Z_{N+1} \subset \Lambda:}{w_k, w_{k+1} \in Z_k, k=1,\ldots, N+1}} *
\end{equation}
where $*$ denotes arbitrary non-negative quantities. We have also used that the set $Z_{N+1}$ must contain more than one point since $Z_{N+1} \in S_{\Lambda}(Z_N)$.
As $\Phi \in \mathcal{B}_F(I)$, \eqref{FnormBound} implies that
\begin{equation} \label{intbd1}
\sum_{\stackrel{Z_k \subset \Lambda:}{w_k, w_{k+1} \in Z_k}} \| \Phi(Z_k, r_k) \| \leq \| \Phi \|_F(r_k) F(d(w_k, w_{k+1}))
\end{equation}
for each $1 \leq k \leq N+1$. Using this bound as well as \eqref{F:int} and \eqref{F:conv}, we conclude that
\begin{eqnarray}
R_{N+1}(t) & \leq & 2 \| B \| 2^{N+1} \int_s^t \!\cdots\! \int_s^{r_N} \!\!\! \sum_{w_1 \in \partial_{\Phi}X}\sum_{w_2,\ldots, w_{N+2} \in \Lambda} 
\!\!\!\!\!\!\!\!\!  \sum_{\stackrel{Z_1, \ldots, Z_{N+1} \subset \Lambda:}{w_k, w_{k+1} \in Z_k, k=1,\ldots, N+1}}
\!\!\!\!\!\!  \prod_{j=1}^{N+1} \| \Phi(Z_j, r_j) \|  dr_{N+1} \cdots dr_1 \nonumber \\
& \leq &  2 \| B \| 2^{N+1} \int_s^t \cdots \int_s^{r_N} \sum_{w_1 \in \partial_{\Phi}X} 
\sum_{w_2,\ldots, w_{N+2} \in \Lambda} \prod_{j=1}^{N+1} \| \Phi \|_F(r_j) F(d(w_j,w_{j+1})) dr_{N+1} \cdots dr_1 \nonumber \\ 
& \leq & 2 \| B \| 2^{N+1} C_F^N  \sum_{w_1 \in \partial_{\Phi}X}  \sum_{w_{N+2} \in \Lambda} F(d(w_1, w_{N+2})) \int_s^t \cdots \int_s^{r_N} \prod_{j=1}^{N+1} \| \Phi \|_F(r_j)
dr_{N+1} \cdots dr_1 \nonumber \\ 
& \leq & \frac{2 \| B \| | \partial_{\Phi} X| \| F \|}{C_F} \frac{ \left( 2C_F \int_s^t \| \Phi \|_F(r) \, dr \right)^{N+1}}{(N+1)!}.
\end{eqnarray}
We note that, in the last inequality, we performed the integration over the simplex. 
Since $\| \Phi \|_F$ is locally integrable on $I$, this remainder clearly goes to 0 as $N \to \infty$.

A similar estimate can be applied to the terms $a_n(t)$. In fact, these terms are also sums over chains, however, there is a restriction: only those chains 
whose final link $Z_n$ satisfies $Z_n \cap Y \neq \emptyset$ contribute to the sum. Recalling that $I_{t,s}(\Phi) = C_F \int_s^t \|\Phi\|_F(r) dr$, the bound
\begin{equation}
a_n(t) \leq \frac{1}{C_F} \frac{\left( 2I_{t,s}(\Phi) \right)^n}{n!} \sum_{x \in \partial_{\Phi}X} \sum_{y \in Y} F(d(x,y)) 
\end{equation}
then follows as above. Since $\delta_Y(X) = 0$ and $n \geq 1$, the bound in (\ref{lrbest}) is now clear.
\end{proof}

%
%

\subsection{A class of unbounded Hamiltonians} \label{sec:class_ubd_ham}

As we now discuss, the methods in the previous subsection extend to models with unbounded
on-site terms. Consider a quantum lattice system comprised of $(\Gamma, d)$ and $\mathcal{A}_{\Gamma}$.
Let $I \subset \mathbb{R}$ be an interval, $F$ an $F$-function on $(\Gamma, d)$, and
$\Phi \in \mathcal{B}_F(I)$ a time-dependent interaction.
To each $z \in \Gamma$, fix a self-adjoint operator $H_z$ with dense domain $\mathcal{D}_z \subset \mathcal{H}_z$. 
For any $\Lambda \in \mathcal{P}_0(\Gamma)$ and $t \in I$, consider the finite-volume Hamiltonian
\begin{equation} \label{gentdham}
H_{\Lambda}(t) = \sum_{z \in \Lambda} H_z + \sum_{Z \subset \Lambda} \Phi(Z,t). 
\end{equation}
The non-interacting Hamiltonian
\begin{equation}
H_{\Lambda}^{(0)} = \sum_{z \in \Lambda} H_z
\end{equation} 
is essentially self-adjoint with domain 
\begin{equation}
\mathcal{D}_{\Lambda} = {\rm span} \{ \bigotimes_{z \in \Lambda} \psi_z \, | \, \psi_z \in \mathcal{D}_z \mbox{ for all } z \in \Lambda \}, \, 
\end{equation} 
see \cite[Theorem VIII.33 and Corollary]{reed:1980}. Since the time-dependent terms are bounded, it follows from \cite[Theorem 5.28]{weidmann:1980} that for each $t \in I$, $H_{\Lambda}(t)$ is essentially self-adjoint on $\mathcal{H}_{\Lambda}$ with domain $\mathcal{D}_{\Lambda}$.
We proceed by using the notation $H_\Lambda^{(0)}$ and $H_\Lambda(t)$
for the corresponding self-adjoint closures.

As $\Phi \in \mathcal{B}_F(I)$, it is a strongly continuous interaction, and so for any $\Lambda \in \mathcal{P}_0( \Gamma)$ Proposition~\ref{prop:tduprop} guarantees the existence of a finite volume unitary propagator corresponding to $H_{\Lambda}(t)$.
Let us briefly review this in order to motivate our definition of the finite volume dynamics.
By Stone's theorem, the non-interacting self-adjoint Hamiltonian $H_{\Lambda}^{(0)}$ generates a free-dynamics
\begin{equation}
\tau_t^{(0)}(A) = e^{it H_{\Lambda}^{(0)}} A e^{-it H_{\Lambda}^{(0)}} \quad \mbox{for all } A \in \mathcal{A}_{\Lambda} \mbox{ and all } t \in \mathbb{R} \, 
\end{equation}
in terms of a group of strongly continuous unitaries $U_{\Lambda}^{(0)}(t,0) = e^{-it H_{\Lambda}^{(0)}}$. In this case,
\begin{equation} \label{intHamt}
\tilde{H}_{\Lambda}(t) = \sum_{Z \subset \Lambda} \tau_t^{(0)} \left( \Phi(Z,t) \right) \quad \mbox{for all } t \in I,
\end{equation}
is pointwise self-adjoint with $\tilde{H}_{\Lambda} : I \to \mathcal{A}_{\Lambda}$ strongly continuous. By Proposition~\ref{prop:sols} (v),
there is a unique strong solution of the initial value problem
\begin{equation} \label{ip_bd_de}
\frac{d}{dt} \tilde{U}_{\Lambda}(t,s) = -i \tilde{H}_{\Lambda}(t) \tilde{U}_{\Lambda}(t,s) \quad \mbox{with} \quad \tilde{U}_{\Lambda}(s,s) = \idty \,
\end{equation}
for each $s \in I$. In terms of these solutions, we introduce 
\begin{equation} \label{ip_dyn_1}
U_{\Lambda}(t,s) = e^{-itH_{\Lambda}^{(0)}} \tilde{U}_{\Lambda}(t,s) e^{i s H_{\Lambda}^{(0)}} ,
\end{equation}
for any $s,t \in I$. As is demonstrated in the proof of Proposition~\ref{prop:tduprop}, the operators
$\{ U_{\Lambda}(t,s) \}_{s,t \in I}$ form a two-parameter family of unitaries. They satisfy the co-cycle
property (\ref{co-cycle_rel}), and generate the unique locally norm bounded weak solutions of the time-dependent Schr\"odinger equation
corresponding to $H_{\Lambda}(t)$. We use these unitaries to define a dynamics associated to $H_{\Lambda}(t)$; namely for any $s,t \in I$, we
take $\tau_{t,s}^{\Lambda} : \cA_{\Lambda} \to \cA_{\Lambda}$ as
\begin{equation} \label{ubdyn}
\tau_{t,s}^{\Lambda}(A) = U_{\Lambda}(t,s)^* A U_{\Lambda}(t,s).
\end{equation}
One readily checks that the family $\{\tau_{t,s}^\Lambda\mid t,s\in I\}$ of automorphisms on $\mathcal{A}_{\Lambda}$ satisfies the 
co-cycle property and that the following analogue of Theorem~\ref{thm:lrb} holds for this dynamics. 

 \begin{thm} \label{thm:lrb2} Let $\{ H_z \}_{z \in \Gamma}$ be a collection of densely defined self-adjoint operators, $\Phi \in \cB_F(I)$, and $\tau_{t,s}^{\Lambda}$ be the dynamics given in \eqref{ubdyn}. Let $X,Y \in \mathcal{P}_0( \Gamma)$ be disjoint sets. For any $\Lambda \in \mathcal{P}_0( \Gamma)$ with 
 $X \cup Y \subset \Lambda$ and any $A \in \mathcal{A}_X$ and $B \in \mathcal{A}_Y$, the bound
 \begin{equation} \label{lrbest2}
 \left\| \left[ \tau_{t,s}^{\Lambda}(A), B \right] \right\| \leq \frac{2 \| A \| \| B \|}{C_F} \left( e^{2I_{t,s}(\Phi)} - 1 \right) D(X,Y)
 \end{equation}
 holds for all $t,s \in I$. Here, $C_F$ is the constant in (\ref{F:conv}), and the quantities $I_{t,s}(\Phi)$ and $D(X,Y)$ are as discussed earlier; see \eqref{ItsPhi} and (\ref{lrbmin}), respectively.
 \end{thm}

\begin{proof}
By construction, it is clear that
\begin{equation} \label{fulldyn}
\tau_{t,s}^{\Lambda}(A) = \hat{\tau}_s^{(0)} \circ \tilde{\tau}_{t,s}^{\Lambda} \circ \tau_t^{(0)}(A) \, .
\end{equation}
Here $\hat{\tau}_s^{(0)}$ and $\tilde{\tau}_{t,s}^{\Lambda}$ are the inverse free-dynamics and the
interaction-picture dynamics, i.e.  
\begin{equation}
\hat{\tau}^{(0)}_s(A) = U_{\Lambda}^{(0)}(s,0)AU_{\Lambda}^{(0)}(s,0)^* \quad \mbox{and} \quad \tilde{\tau}^{\Lambda}_{t,s}(A) = \tilde{U}_{\Lambda}(t,s)^*A \tilde{U}_{\Lambda}(t,s) \, .
\end{equation}
In this case,
\begin{equation}
\| [ \tau_{t,s}^{\Lambda}(A), B ] \| = \| [ \tilde{ \tau}_{t,s}^{\Lambda}( \tau_t^{(0)}(A)), \tau_s^{(0)}(B) ] \| \, .
\end{equation}
Note that for all $t \in I$, $\tau_t^{(0)}(A) \in \cA_X$ and $\tau_t^{(0)}(B) \in \cA_Y$.
Moreover, the interaction-picture dynamics is generated by the strongly continuous interaction $\tilde{\Phi}$
with terms
\begin{equation}
\tilde{\Phi}(Z,t) = e^{itH_Z^{(0)}} \Phi(Z,t) e^{-itH_Z^{(0)}}
\end{equation}
for any $Z \in \mathcal{P}_0( \Gamma)$ and $t \in I$. Since $\tilde{\Phi}(Z,t)$ and $\Phi(Z,t)$ have the same 
support and the same norm, it is clear that $\| \tilde{\Phi} \|_F(t) = \| \Phi \|_F(t)$
for all $t \in I$. In this case, the bound in (\ref{lrbest2}) follows from Theorem~\ref{thm:lrb} applied to the interaction-picture dynamics $\tilde{\Phi}$.
\end{proof}

In Section~\ref{sec:duhamel}, we considered a family of Hamiltonians, see (\ref{fam_ham_tl}), with bounded interactions
which depend not only on time but also on an auxillary parameter. Within the context of quantum lattice models,
the corresponding finite-volume Hamiltonian may have the form
\be \label{time+para_ham}
H_{\Lambda}^\lambda(t) = \sum_{z \in \Lambda} H_z + \sum_{X \subset \Lambda}\Phi_\lambda(Z,t).
\ee
If $(\lambda,t)\mapsto \Phi_{\lambda}(X,t)$ is jointly strongly continuous and strongly differentiable with respect to $\lambda$, 
Proposition~\ref{properties_of_expiths} applies to the corresponding finite-volume dynamics
\begin{equation} \label{unbd_two_param_dyn}
\tau_{t,s}^{\Lambda, \lambda}(A) = U_{\Lambda}^{\lambda}(t,s)^*AU_{\Lambda}^{\lambda}(t,s) \, .
\end{equation}
To be clear, we note that the unitaries $U_{\Lambda}^{\lambda}(t,s)$, in (\ref{unbd_two_param_dyn}) above, are constructed as in (\ref{ip_dyn_1}) by first
solving the analogue of (\ref{ip_bd_de}) with
\begin{equation}
\tilde{H}_{\Lambda}^{\lambda}(t) = \sum_{Z \subset \Lambda} \tau_t^{(0)} \left( \Phi_{\lambda}(Z,t) \right) \quad \mbox{for all } t \in I \, .
\end{equation}
With no further assumptions on the interaction terms, the bound provided by Proposition~\ref{properties_of_expiths} on the $\lambda$-derivative of $\tau_{t,s}^{\Lambda, \lambda}(A)$, 
see (\ref{norm_bd_dyn_der}), will generally depend on the volume $\Lambda$. However, under the additional assumption that $\Phi_{\lambda}, \Phi'_{\lambda} \in \mathcal{B}_F(I)$,
one obtains a better, volume independent estimate on the derivative. In fact, arguing as in the proof of Proposition~\ref{properties_of_expiths} one finds that if $s \leq t$, then
\begin{equation}
\frac{d}{d\lambda} \tau_{t,s}^{\Lambda, \lambda}(A)= i \sum_{Z \subset \Lambda} 
\int_s^t  \tau^{\Lambda, \lambda}_{r,s} ([\Phi^\prime_\lambda(Z,r), \tau^{\Lambda, \lambda}_{t,r} (A)])\, dr ;
\end{equation}
compare with (\ref{dyn_duhamel}). Since $\Phi_{\lambda} \in \mathcal{B}_F(I)$, the Lieb-Robinson bound (\ref{lrbest2}) holds for the dynamics
$\tau_{t,s}^{\Lambda, \lambda}$. In this case, an application of Corollary~\ref{app:cor:int_bd} shows that for all $A\in \cA_X$
\be \label{norm_bd_dyn_der2}
\left\|\frac{d}{d\lambda}\tau_{t,s}^{ \Lambda,\lambda}(A)\right\| \leq 2\|A\| \|F\| |X|e^{2I_{t,s}(\Phi_\lambda)} \int_s^t \|\Phi_\lambda'\|_F(r)dr \, ,
\ee
the right-hand-side of which is independent of $\Lambda$.

%
%

\subsection{The infinite-volume dynamics} \label{sec:eivd}

In this section, we will prove several convergence and continuity results for the Heisenberg dynamics associated with interactions $\Phi\in \cB_F(I)$ that make use of Lieb-Robinson bounds. As is well-known, see e.g. \cite{bratteli:1997}, Lieb-Robinson bounds can be used to prove the existence of a dynamics in the thermodynamic limit for sufficiently short-range interactions. In Theorem~\ref{thm:existbd}, we show that given an interaction $\Phi \in \mathcal{B}_F(I)$ the dynamics corresponding to finite-volume restrictions of
$\Phi$ converge in the thermodynamic limit. To prove the existence of the thermodynamic limit, we will apply Theorem~\ref{thm:continuity} below, which establishes that the Heisenberg dynamics is continuous in the interaction space. For example, in the case of time-independent interactions, Theorem~\ref{thm:continuity} implies that
the difference between the dynamical evolution of a local observable $A$ with respect to two different interactions $\Phi, \Psi \in \cB_F$ is small if $\|\Phi-\Psi\|_F$ is small. The statement of this result for finite-volume Heisenberg dynamics is the content of Theorem \ref{thm:continuity}, with the analogous thermodynamic limit statement given in Corollary~\ref{cor:continuity}. Lastly, given a sequence of interactions which converge {\it locally in $F$-norm}, see Definition~\ref{def:lcfnorm} below, we show that the corresponding dynamics (which necessarily exist by Theorem~\ref{thm:existbd}) converge as well; this is the content of Theorem~\ref{thm:exist_iv_dyn}. In particular, this can be used to prove that the thermodynamic limit of the Heisenberg dynamics is unchanged by the addition of (sufficiently local) boundary conditions. If the interactions are norm continuous, the cocycle of automorphisms describing the infinite-volume dynamics is differentiable with a strongly continuous generator. This is shown in Theorem~\ref{thm:differentiable_dynamics}.

We now begin with the continuity statement. For this result, we will once again make use of the quantity $I_{t,s}(\Phi)$, which is defined in \eqref{ItsPhi}.

%
%

\begin{thm}\label{thm:continuity}
Consider a quantum lattice system comprised of $(\Gamma, d)$ and $\cA_{\Gamma}$.
Let $I \subset \mathbb{R}$ be an interval, $F$ be an $F$-function on $(\Gamma, d)$, and
 $\Phi, \Psi \in \cB_F(I)$ be time-dependent interactions. 
Fix a collection of densely defined, self-adjoint on-site Hamiltonians $\{ H_z \}_{z \in \Gamma}$, and
for any $\Lambda \in \mathcal{P}_0( \Gamma)$, define Hamiltonians
\begin{equation} \label{two_hams}
H_{\Lambda}^{( \Phi)}(t) = \sum_{z \in \Lambda} H_z + \sum_{Z \subset \Lambda} \Phi(Z,t) \quad \mbox{and} \quad H_{\Lambda}^{( \Psi)}(t) = \sum_{z \in \Lambda} H_z + \sum_{Z \subset \Lambda} \Psi(Z,t)
\end{equation}
as well as their corresponding dynamics, $\tau_{t,s}^{\Lambda}$ and $\alpha_{t,s}^{\Lambda}$, for each $s, t \in I$, respectively. 
\begin{enumerate}
	\item[(i)]  For any $X, \Lambda \in \mathcal{P}_0( \Gamma)$ with $X \subset \Lambda$, the bound
	\begin{equation}\label{cont_bd} 
	\| \tau_{t,s}^{\Lambda}(A) - \alpha_{t,s}^{\Lambda}(A) \| \leq \frac{2\| A \|}{C_F} e^{2\min(I_{t,s}(\Phi),I_{t,s}(\Psi))}
	I_{t,s}(\Phi-\Psi) \sum_{x \in X} \sum_{y \in \Lambda} F(d(x,y)) 
	\end{equation} 
	holds for all $A \in \mathcal{A}_X$ and $s,t \in I$.
	\item[(ii)] For any $X, \Lambda_0, \Lambda \in \mathcal{P}_0( \Gamma)$ with $X \subset \Lambda_0 \subset \Lambda$, the
	bound
	\begin{equation}\label{sub_diff_bd} 
	\| \tau_{t,s}^{\Lambda}(A) - \tau_{t,s}^{\Lambda_0}(A) \| \leq \frac{2\| A \|}{C_F} e^{2 I_{t,s}(\Phi)}
	I_{t,s}(\Phi)  \sum_{x \in X} \sum_{y \in \Lambda \setminus \Lambda_0} F(d(x,y)) 
	\end{equation} 
	holds for all $A \in \mathcal{A}_X$ and $s,t \in I$. 
\end{enumerate}
\end{thm}

Using \eqref{F:int}, the estimates in \eqref{cont_bd} and \eqref{sub_diff_bd} can be interpreted as bounds on the norm of the difference of two dynamics, thought of as maps from $\cA_X$ to $\cA_\Lambda$, that are uniform in $\Lambda$ but grow linearly in $|X|$. Since the dynamics are, of course, automorphisms, the bounds of Theorem~\ref{thm:continuity} are only nontrivial if the RHS of \eqref{cont_bd} and \eqref{sub_diff_bd} are smaller than $2\|A\|$. As is well known, this will be true for both cases if $|t-s|$ is sufficiently small. Additionally, the bound in \eqref{cont_bd} will be nontrivial if $\Vert \Phi-\Psi\Vert_F$ is small, and the bound in \eqref{sub_diff_bd} will be nontrivial if $d(X,\Lambda\setminus\Lambda_0)$ is small. 
Note that, even if a map is bounded on all of $\cA^{\rm loc}_\Gamma$, as is the case in Theorem \ref{thm:continuity}, norm  bounds for their local restriction can be very useful. 

\begin{proof}
To prove (i), we first consider the case that $H_z =0$ for all $z \in \Gamma$.
For any $\Lambda \in \mathcal{P}_0( \Gamma)$, $X \subset \Lambda$, $s,t \in I$, and
$A \in \mathcal{A}_X$ we write the corresponding difference as
\begin{equation}  \label{dyn_diff}
\tau_{t,s}^{\Lambda}(A) - \alpha_{t,s}^{\Lambda}(A) = \int_s^t \frac{d}{dr} \left( \tau_{r,s}^{\Lambda} \circ \alpha_{t,r}^{\Lambda}(A) \right) \, dr 
= i \int_s^t \tau_{r,s}^{\Lambda} \left( \left[ H_{\Lambda}^{(\Theta)}(r), \alpha_{t,r}^{\Lambda} (A) \right] \right) \, dr 
\end{equation} 
where we have introduced the local Hamiltonian
\begin{equation} \label{diff_ham_terms}
H_{\Lambda}^{(\Theta)}(r) = \sum_{Z \subset \Lambda} \Theta(Z,r) \mbox{ with } \Theta(Z,r) = \Phi(Z,r) - \Psi(Z,r).
\end{equation}
Note that the equality in \eqref{dyn_diff} is to be understood in the strong sense. When $s \leq t$, a simple norm bound shows then that
\begin{equation} \label{dyn_diff_est1}
\| \tau_{t,s}^{\Lambda}(A) - \alpha_{t,s}^{\Lambda}(A)  \| \leq 
\sum_{Z \subset \Lambda} \int_s^t \| [ \alpha_{t,r}^{\Lambda} (A), \Theta(Z, r) ] \| \, dr.
\end{equation}
An application of Corollary~\ref{app:cor:int_bd} with $R=0$, then gives
\begin{equation}
\sum_{\stackrel{Z \subset \Lambda:}{d(Z,X) = 0}} \int_s^t \| [ \alpha_{t, r}^{\Lambda} (A), \Theta(Z, r)  ] \|  \, dr 
 \leq \frac{2 \| A \|}{C_F} I_{t,s}( \Theta) \sum_{x \in X} \sum_{y \in \Lambda} F(d(x,y)) 
\end{equation}
and moreover, 
\begin{equation} \label{Cor7.2App}
\sum_{\stackrel{Z \subset \Lambda:}{d(Z,X)>0}} \int_s^t \| [ \alpha_{t, r}^{\Lambda} (A), \Theta(Z, r)  ] \| \, dr 
\leq \frac{2 \| A \|}{C_F} \left( e^{2 I_{t,s}(\Psi)} - 1 \right) I_{t,s}( \Theta) \sum_{x \in X} \sum_{y \in \Lambda \setminus X} F(d(x,y)) .
\end{equation}
The bound (\ref{cont_bd}) follows from observing that one could instead have estimated the difference $\alpha_{t,s}^{\Lambda}(A) - \tau_{t,s}^{\Lambda}(A)$, which corresponds to exchanging $\tau_{t, s}^{\Lambda}$ and $\alpha_{t,s}^{\Lambda}$ in (\ref{dyn_diff})-(\ref{Cor7.2App}). 

To extend the result to the situation with non-trivial $H_z$, we argue with the interaction picture dynamics as was done in the proof of Theorem~\ref{thm:lrb2}. Using (\ref{fulldyn}), it is clear that
\begin{equation}
\| \tau_{t,s}^{\Lambda}(A) - \alpha_{t,s}^{\Lambda}(A) \| = \| \tilde{\tau}_{t,s}^{\Lambda}( \tau_t^{(0)}(A)) - \tilde{\alpha}_{t,s}^{\Lambda}( \tau_t^{(0)}(A)) \| 
\end{equation}
and since $\tau_t^{(0)}(A)\in \cA_X$, the proof in the general situation reduces to the previous case.

The proof of (ii) is nearly identical. In fact, again in the case that $H_z=0$ for all $z \in \Gamma$, the analogue of \eqref{dyn_diff_est1} is
\begin{equation} \label{dyn_diff_est2}
\| \tau_{t,s}^{\Lambda}(A) - \tau_{t,s}^{\Lambda_0}(A)  \| \leq \sum_{Z \in S_\Lambda(\Lambda_0)} \int_s^t \| [ \tau_{t, r}^{\Lambda_0} (A), \Phi(Z, r) ] \|  \, dr, 
\end{equation}
where we have taken advantage of cancellations and used \eqref{surface_set}. The bound in (\ref{sub_diff_bd}) now follows from similar estimates to those used in Proposition~\ref{app:prop:int_bd_dist} and Corollary~\ref{app:cor:int_bd}. 
\end{proof}

%
%
%

A first application of Theorem~\ref{thm:continuity} is a proof that given any collection of self-adjoint on-sites $\{ H_z \}_{z \in \Gamma}$ 
and an interaction $\Phi \in \mathcal{B}_F(I)$ there is a corresponding infinite-volume dynamics on $\mathcal{A}_{\Gamma}$.
We obtain this infinite-volume dynamics as a limit of finite-volume dynamics. With this in mind, we will say that
a sequence $\{ \Lambda_n \}_{n \geq 1} \subset \mathcal{P}_0( \Gamma)$ is increasing and exhaustive if
$\Lambda_n \subset \Lambda_{n+1}$ for all $n \geq 1$ and given any $X \in \mathcal{P}_0( \Gamma)$, there is an $N \geq 1$ 
for which $X \subset \Lambda_N$. 
 
\begin{thm}\label{thm:existbd}
Under the assumptions of Theorem~\ref{thm:continuity}, for each $A \in \mathcal{A}_{\Gamma}^{\rm loc}$ and any $s,t \in I$,
\begin{equation}\label{eq:claim} 
\tau_{t,s}(A) = \lim_{\Lambda \uparrow \Gamma} \, \tau_{t,s}^{\Lambda} (A)
\end{equation} 
exists in norm and the convergence is uniform for $s$ and $t$ in compact subsets of $I$. The limit may be taken 
along any increasing, exhaustive sequence of finite subsets of $\Gamma$, and is 
independent of the sequence. Moreover, $\tau_{t,s}$ is a co-cycle of automorphisms of $\cA_\Gamma$. If there exists $M\geq 0$ such that $\Vert H_z\Vert \leq M$ for all $z\in \Gamma$, then $(t,s)\mapsto \tau_{t,s}(A)$ is norm continuous for all $A\in\cA_\Gamma$. 
\end{thm}

For unbounded $H_z$, the continuity of $\tau_{t,s}$ is limited by the continuity of the on-site dynamics $\tau^{(0)}_t$. 
In a suitable representation of  $\cA_\Gamma$ on a Hilbert space, one can retrieve weak continuity of the dynamics. 
See \cite{nachtergaele:2010} for an example. Continuity properties of $\tau_{t,s}$ and other families of maps will be further 
discussed in Section \ref{sec:cont_l_a}.

\begin{proof} 
Let $X\in \mathcal{P}_0(\Gamma)$, $A\in \cA_X$, and consider $\{ \Lambda_n \}_{n \geq 0}$ any increasing, 
exhaustive sequence of finite subsets of $\Gamma$. Since the sequence is exhaustive, there exists $N \geq 1$ for which $X \subset \Lambda_n$ for all $n \geq N$. 
For any integers $n, m$ with $N \leq m \leq n$, Theorem~\ref{thm:continuity} (ii) implies 
\begin{equation}
\| \tau_{t,s}^{\Lambda_n}(A) - \tau_{t,s}^{\Lambda_m}(A) \| \leq \frac{2 \|A \|}{C_F} e^{2I_{t,s}(\Phi)} I_{t,s}(\Phi) \sum_{x \in X} \sum_{y \in \Lambda_n \setminus \Lambda_m}F(d(x,y)) ,
\end{equation}
where, again, $I_{t,s}(\Phi)$ is as defined in \eq{ItsPhi}. By \eqref{F:int}, the RHS converges to 0 as $n, \, m\to \infty$. As such, for any $[a,b] \subset I$, the sequence of observables $\{ \tau_{t,s}^{\Lambda_n}(A) \}_{n \geq 0}$ is Cauchy in norm, uniformly for $s, t \in [a, b]$. 

The proof of the remaining facts in the statement of this theorem is standard and proceeds in the same way as is done, e.g., in \cite{simon:1993} for quantum spin models with time-independent interactions.
\end{proof}

Combining Theorem~\ref{thm:continuity} and Theorem~\ref{thm:existbd}, we obtain the
following useful estimates for the infinite volume dynamics.

\begin{cor}\label{cor:continuity}
Under the assumptions of Theorem~\ref{thm:continuity},
\begin{enumerate}
	\item[(i)] For any $X, \, Y\in \cP_0(\Gamma)$ such that $X\cap Y = \emptyset$, the bound
	\[
	\|[\tau_{t,s}(A), B]\| \leq \frac{2\|A\|\|B\|}{C_F}(e^{2I_{t,s}(\Phi)}-1)D(X,Y)
	\]
	holds for all  $A\in \cA_X$, $B\in \cA_Y$, and $t, \, s\in I$.
	\item[(ii)] For any $X \in \mathcal{P}_0( \Gamma)$, the bound
	\begin{equation}\label{cont_bd_iv} 
	\| \tau_{t,s}(A) - \alpha_{t,s}(A) \| \leq \frac{2 \| A \|}{C_F} e^{2\min(I_{t,s}(\Phi),I_{t,s}(\Psi))} I_{t,s}( \Phi - \Psi) \sum_{x \in X} \sum_{y \in \Gamma} F(d(x,y)) 
	\end{equation} 
	holds for all $A \in \mathcal{A}_X$ and $s,t \in I$.
	\item[(iii)] For any $X, \Lambda \in \mathcal{P}_0( \Gamma)$ with $X \subset \Lambda$, the
	bound
	\begin{equation}\label{sub_diff_bd_iv} 
	\| \tau_{t,s}(A) - \tau_{t,s}^{\Lambda}(A) \| \leq \frac{2\| A \|}{C_F} e^{2I_{t,s}(\Phi)} I_{t,s}(\Phi) \sum_{x \in X} \sum_{y \in \Gamma \setminus \Lambda} F(d(x,y)) 
	\end{equation} 
	holds for all $A \in \mathcal{A}_X$ and $s,t \in I$. 
\end{enumerate}
\end{cor}

We now prove a convergence result for the dynamics associated to interactions in $\mathcal{B}_F(I)$.
First, we introduce some notation and terminology associated with extensions and restrictions
of interactions, and then state the result. 

In certain applications, e.g. when considering boundary conditions, rather than considering a single interaction $\Phi: \cP_0(\Gamma) \to \cA_{\Gamma}^{\rm loc}$, we may want to consider a family of strongly continuous interactions $\{\Phi_{\Lambda}: \mathcal{P}_0( \Lambda) \times I \to \mathcal{A}_{\Lambda}^{\rm loc} \, | \, \Lambda \in \cP_0(\Gamma)\}$. In this situation, a single mapping $\Phi_{\Lambda}$ can be extended to an interaction on all of $\Gamma$ by declaring that $\Phi_{\Lambda}(Z,t) = 0$ for any $Z \in \mathcal{P}_0( \Gamma)$ with $Z \cap (\Gamma \setminus \Lambda) \neq \emptyset$. We call this new mapping the extension of $\Phi_{\Lambda}$ to $\Gamma$, and continue to denote it by $\Phi_{\Lambda}$. Similarly, given $\Phi_{\Lambda}$ and $\Lambda_0 \subset \Lambda$, we define
$\Phi_{\Lambda} \restriction_{\Lambda_0}: \mathcal{P}_0( \Lambda_0) \times I \to \mathcal{A}_{\Lambda_0}^{\rm loc}$ by 
\begin{equation}
\Phi_{\Lambda} \restriction_{\Lambda_0}(X,t) = \Phi_{\Lambda}(X,t) \quad \mbox{for all } X \in \mathcal{P}_0( \Lambda_0) \mbox{ and } t \in I
\end{equation}
We call the mapping $\Phi_{\Lambda} \restriction_{\Lambda_0}$ the restriction of $\Phi_{\Lambda}$ to $\Lambda_0$.
If the dynamics associated to $\Phi_{\Lambda} \restriction_{\Lambda_0}$ exists, we will denote it by $\tau_{t,s}^{\Phi_{\Lambda}, \Lambda_0}$.
Of course, if $\Lambda_0$ is finite, such a dynamics always exists and it is generated by the time-dependent Hamiltonian
\begin{equation}
H_{\Lambda_0}^{\Phi_{\Lambda}}(t) = \sum_{X \subset \Lambda_0} \Phi_{\Lambda}(X,t) \, .
\end{equation}

We now introduce the notion of local convergence in $F$-norm.

\begin{defn} \label{def:lcfnorm} Let $(\Gamma, d)$ and $\mathcal{A}_{\Gamma}$ be a quantum lattice system, and $I \subset \mathbb{R}$ be an interval. We say that a sequence of interactions $\{ \Phi_n \}_{n \geq 1}$ \emph{converges locally in $F$-norm} to $\Phi$ if there exists an $F$-function, $F$, such that:
\begin{enumerate}
	\item[(i)]$\Phi_n \in \mathcal{B}_F(I)$ for all $n \geq 1$,
	\item[(ii)] $\Phi \in \mathcal{B}_F(I)$,
	\item[(iii)] For any $\Lambda \in \mathcal{P}_0( \Gamma)$ and each $[a,b] \subset I$,
	\begin{equation}
	\lim_{n \to \infty} \int_a^b \| ( \Phi_n - \Phi) \restriction_{\Lambda} \|_F(t) \, dt  = 0 \, . 
	\end{equation}
\end{enumerate}
Moreover, if $F$ is an $F$-function for which (i)-(iii) are satisfied, we say that $\Phi_n$ \emph{converges locally in $F$-norm to $\Phi$ with respect to $F$}.
\end{defn}

Recall that a strongly continuous interaction $\Phi \in \mathcal{B}_F(I)$ with $\| \Phi \|_F : I \to [0, \infty)$ given by
\begin{equation}
\| \Phi \|_F(t) = \sup_{x,y \in \Gamma} \frac{1}{F(d(x,y))} \sum_{\stackrel{X \in \mathcal{P}_0( \Gamma):}{x,y \in X}} \| \Phi(X, t) \| 
\end{equation}
is locally bounded. In this situation, Theorem~\ref{thm:existbd} demonstrates that there 
exists a co-cycle of automorphisms of $\mathcal{A}_{\Gamma}$, which we denote by $\tau_{t,s}^{\Phi}$ and refer to as the dynamics associated to $\Phi$. 
Note that we have taken the self-adjoint on-sites terms to be identically zero, i.e. $H_z=0$ for all $z \in \Gamma$; see comments following the proof of Theorem~\ref{thm:exist_iv_dyn}.

\begin{thm}\label{thm:exist_iv_dyn}
Let $\{ \Phi_n \}_{n \geq 1}$ be a sequence of time-dependent interactions on $\Gamma$ with
$\Phi_n$ converging locally in the $F$-norm to $\Phi$ with respect to $F$. 
\begin{enumerate}
\item[(i)] If for every $[a,b] \subset I$,
\be \label{sup_bd}
\sup_{n \geq 1} \int_a^b \| \Phi_n \|_F(t) \, dt <  \infty \, ,
\ee
then, for any $X \in \mathcal{P}_0( \Gamma)$, $A\in \cA_X$, and $s,t\in I, s\leq t$, we have convergence of the dynamics:
\be
\lim_{n \to \infty}\Vert \tau_{t,s}^{\Phi_n}(A) -  \tau^{\Phi}_{t,s}(A)\Vert =0.
\ee
Moreover, the convergence is uniform for $s,t$ in compact intervals, and the dynamics is continuous:
\be\label{dyn_continuous}
\Vert \tau^{\Phi}_{t,s}(A) - A\Vert \leq 2|X| \Vert A\Vert \Vert F\Vert \int_s^t \Vert \Phi\Vert_F(r) dr.
\ee
\item[(ii)]
If in addition, for all $\Lambda \in \mathcal{P}_0( \Gamma)$, and $r\in I$, we also have pointwise local convergence: 
\be
\lim_{n \to \infty} \| ( \Phi_n - \Phi) \restriction_{\Lambda} \|_F(r)=0,
\ee
and uniform boundedness of the interactions on compact intervals $I_0\subset I$:
\be
 \sup_n \sup_{t\in I_0} \Vert \Phi_n\Vert_F(t)<\infty,
\ee
then, the generators converge uniformly for $t$ in compacts: for all compact $I_0\subset I$ and $A\in \cA^{\rm loc}_\Gamma$, we have
\be\label{convergence_derivations}
 \lim_n \sup_{t\in I_0} \Vert\delta_t^{(n)}(A) -\delta_t(A)\Vert =0,
\ee
where
$$
\delta_r^{(n)}(A) = \sum_{Z\in  \mathcal{P}_0( \Gamma)} [\Phi_n(Z,r),A],  \quad  \delta_r(A)
=\sum_{Z\in  \mathcal{P}_0( \Gamma)} [\Phi(Z,r),A].
$$
\end{enumerate}
\end{thm}

\begin{proof}
(i). Let $X \in \mathcal{P}_0( \Gamma)$ and $[a,b] \subset I$ be fixed.
For any $\Lambda \in \mathcal{P}_0(\Gamma)$ with $X \subset \Lambda$, define
the restricted interactions $\Phi_n \restriction_{\Lambda}$ and $\Phi \restriction_{\Lambda}$.
By Theorem~\ref{thm:existbd}, a dynamics may be associated to each of these interactions and the estimate
\bea \label{3_terms}
\Vert \tau_{t,s}^{\Phi_n}(A) -  \tau^{\Phi}_{t,s}(A)\Vert & \leq & \Vert \tau_{t,s}^{\Phi_n}(A) -  \tau^{\Phi_n, \Lambda}_{t,s}(A)\Vert \nonumber \\
& \mbox{ } & \quad +  \Vert \tau_{t,s}^{\Phi_n, \Lambda}(A) -  \tau^{\Phi, \Lambda}_{t,s}(A)\Vert  + \Vert \tau_{t,s}^{\Phi, \Lambda}(A) -  \tau^{\Phi}_{t,s}(A)\Vert 
\eea
holds for all $A \in \mathcal{A}_X$ and $s,t \in [a,b]$. 

An application of Corollary~\ref{cor:continuity} (iii), shows that
\begin{equation} \label{first}
\Vert \tau_{t,s}^{\Phi_n}(A) -  \tau^{\Phi_n, \Lambda}_{t,s}(A)\Vert \leq \frac{2 \| A \|}{C_F} e^{2 I_{t,s}(\Phi_n)} I_{t,s}(\Phi_n) \sum_{x \in X} \sum_{y \in \Gamma \setminus \Lambda} F(d(x,y))
\end{equation}
and similarly 
\begin{equation} \label{last}
\Vert \tau_{t,s}^{\Phi}(A) -  \tau^{\Phi, \Lambda}_{t,s}(A)\Vert \leq \frac{2 \| A \|}{C_F} e^{2 I_{t,s}(\Phi)} I_{t,s}(\Phi) \sum_{x \in X} \sum_{y \in \Gamma \setminus \Lambda} F(d(x,y)).
\end{equation}
For the middle term above, we apply Theorem~\ref{thm:continuity} (i) to find that
\begin{equation} \label{middle}
 \Vert \tau_{t,s}^{\Phi_n, \Lambda}(A) -  \tau^{\Phi, \Lambda}_{t,s}(A)\Vert \leq \frac{2 \| A \| |X| \| F \| }{C_F} e^{2 \min(I_{t,s}(\Phi_n), I_{t,s}(\Phi))} I_{t,s}((\Phi_n - \Phi) \restriction_{\Lambda}).
\end{equation}

By assumption (\ref{sup_bd}), it is clear that $\sup_n I_{t,s}( \Phi_n)$ is finite for all $s, t \in [a,b]$. In this case, for any $\epsilon >0$,
the estimates in (\ref{first}) and (\ref{last}) can be made arbitrarily small for all $n$ sufficiently large (for example, less than $\epsilon /3$) with a sufficiently large, but finite choice of $\Lambda \subset \Gamma$. For any such choice of $\Lambda$, the bound (\ref{middle}) can be made equally small with large $n$ by using local convergence in $F$-norm. 

To prove the bound \eq{dyn_continuous}, we use the convergence established above, the existence of the thermodynamic limit (Theorem \ref{thm:existbd}), and
the differentiability of the finite-volume dynamics, as follows:
\beann
\Vert \tau^\Phi_{t,s} (A) - A\Vert &=& \lim_\Lambda \lim_n \Vert \tau^{\Phi_n,\Lambda}_{t,s} (A) - A\Vert\\
&\leq& \limsup_{\Lambda} \lim_n\sum_{Z\subset \Lambda}\Vert \int_s^t \tau^{\Phi_n,\Lambda}_{t,s} ([\Phi_n(Z,r),A]) dr\Vert\\
&\leq&  \limsup_{\Lambda} \limsup_n 2 |X| \Vert A\Vert \Vert F\Vert \int_s^t  \Vert \Phi_n\restriction_\Lambda\Vert_F(r) dr\\
&\leq&  2 |X| \Vert A\Vert \Vert F\Vert \left[ \limsup_\Lambda \int_s^t \Vert \Phi\restriction_\Lambda\Vert_F (r) dr
+ \limsup_{\Lambda} \limsup_n \int_s^t \Vert (\Phi - \Phi_n)\restriction_\Lambda\Vert_F (r) dr\right].
\eeann
The second term between the square brackets vanishes because the interactions converge locally in $F$-norm and the first term gives the desired estimate.

(ii) The interactions $\Phi_n(t)$ and $\Phi(t)$ have finite $F$-norm. Hence the corresponding derivations,
$\delta^{(n)}_t$ and $\delta_t$, are well-defined on $\cA^{\rm loc}_\Gamma$. For any $\Lambda\in \mathcal{P}_0( \Gamma)$ we then have
\be
\delta^{(n)}_t (A) -  \delta_t(A) = \sum_{Z\subset \Lambda} [\Phi_n(Z,t) - \Phi(Z,t), A] 
+ \sum_{Z, Z\cap\Gamma\setminus\Lambda\neq\emptyset} [\Phi_n(Z,t) - \Phi(Z,t), A] .
\ee
Therefore, applying \eq{RclosetoX} with $R=0$ to the first term and using a similar argument for the second term, we get
\be\label{convergence_derivation}
\Vert \delta^{(n)}_t (A) -  \delta_t(A) \Vert \leq 2|X| \Vert A\Vert \Vert F\Vert \| ( \Phi_n - \Phi) \restriction_{\Lambda} \|_F(t)
+ 2|X| \Vert A\Vert \Vert  \Phi_n - \Phi \Vert_F(t) \sum_{x\in X, y\in \Gamma\setminus\Lambda} F(x,y).
\ee
The first term on the RHS vanishes in the limit $n\to\infty$. Therefore
\be
\limsup_n \Vert \delta^{(n)}_t (A) -  \delta_t(A) \leq  2\Vert |X| \Vert A\Vert (\sup_n \Vert \Phi_n -\Phi\Vert_F(t))  \sum_{x\in X, y\in \Gamma\setminus\Lambda} F(x,y).
\ee
By taking $\Lambda \uparrow \Gamma$, the convergence of the generators now follows. The estimate \eq{convergence_derivation} shows that the convergence
is uniform for $t$ in a compact interval.
\end{proof}

The dynamics considered in the proof of Theorem~\ref{thm:exist_iv_dyn} above corresponds to the one whose existence is
established in the proof of Theorem~\ref{thm:existbd} in the special case that the on-sites $H_z =0$ for all $z \in \Gamma$.
By going to the interaction picture, as is done in the proof of Theorem~\ref{thm:continuity}, it is clear that the
convergence results \eq{dyn_continuous} and \eq{convergence_derivations} hold in the case of arbitrary self-adjoint on-site terms.

Theorem \ref{thm:exist_iv_dyn} establishes sufficient conditions for the convergence of the sequence of co-cycles $\tau_{t,s}^{\Phi_n}$ to the co-cycle 
$\tau_{t,s}^\Phi$, as well as the convergence of the generators $\delta^{(n)}_t$ to densely defined derivations $\delta_t$. These conditions are by no means 
necessary, but will serve our purposes well.

We may now ask whether the dynamics satisfies additional properties and, in particular, whether it is differentiable with derivative given by the derivation $\delta_t$. The following
theorem addresses this question.

\begin{thm}\label{thm:differentiable_dynamics}
For all $t$ in an interval $I$, let $\Phi(t)\in \cB_F$ be interactions such that $t\mapsto \Phi(Z,t) $ is norm-continuous for all $Z\in \cP_0(\Gamma)$, and
such that $\Vert \Phi(t)\Vert_F $ is bounded on all compact intervals $I_0\subset I$. Let $\tau_{t,s}$ denote the strongly continuous dynamics generated by $\Phi(t)$. Define, for all $t\in I$,
\be
\delta_t(A)
=\sum_{Z\in  \mathcal{P}_0( \Gamma)} [\Phi(Z,t),A], A \in \cA^{\rm loc}_\Gamma.
\ee
Then, $t\to \delta_t(A)$ is norm-continuous for all $A \in \cA^{\rm loc}_\Gamma$, and for all $s,t\in I, s<t$,
\be
\frac{d}{dt} \tau_{t,s}(A) = i  \tau_{t,s}(\delta_t (A)), A \in \cA^{\rm loc}_\Gamma.
\ee
\end{thm}
\begin{proof}
First, note that the conditions of Theorem \ref{thm:exist_iv_dyn}, parts (i) and (ii), are satisfied for the sequence $\Phi_n = \Phi\restriction_{\Lambda_n}$, associated to 
any sequence of increasing and absorbing finite volumes $\Lambda_n$. Therefore, we have
\be\label{dyn_continuous_M}
\Vert \tau_{t,s}(A) - A\Vert \leq 2|X| \Vert A\Vert \Vert F\Vert |t-s| \sup_{r\in I_0} \Vert \Phi(r)\Vert_F ,  \text{ for all } s,t\in I_0
\ee
and
\be
\delta_t(A) = \lim_n \sum_{Z\subset \Lambda_n} [\Phi(Z,t),A],  \quad  A\in \cA^{\rm loc}_\Gamma.
\ee
To prove the continuity of $\delta_t(A)$ as a function of $t$, note that for $s,t\in I, X\in  \cP_0(\Gamma), A\in \cA_X$, and $\Lambda\subset\Gamma$, we have
\beann
\Vert \delta_t(A) - \delta_s(A)\Vert &=&\left\Vert \sum_{\substack{Z\in  \cP_0(\Gamma)\\Z\cap X\neq\emptyset}} [\Phi(Z,t)- \Phi(Z,s), A]\right\Vert \\
&\leq& 2\Vert A\Vert \left[\sum_{\substack{Z\subset\Lambda\\ Z\cap X\neq\emptyset}}\Vert \Phi(Z,t)- \Phi(Z,s)\Vert
+  \!\!\!\!\!\!\!\! \sum_{\substack{Z\in  \cP_0(\Gamma)\\Z\cap(\Gamma\setminus \Lambda)\neq\emptyset, Z\cap X\neq\emptyset}}\!\!\!\!\!\!\!\!
\Vert \Phi(Z,t)\Vert + \Vert \Phi(Z,s)\Vert\right]\\
&\leq& 2\Vert A\Vert \left[\sum_{\substack{Z\subset\Lambda\\ Z\cap X\neq\emptyset}}\Vert \Phi(Z,t)- \Phi(Z,s)\Vert
+2  \sup_{r\in I_0} \Vert \Phi(r)\Vert_F\sum_{x\in X}  \sum_{y\in \Gamma \setminus \Lambda} F(d(x,y))\right].
\eeann
Continuity follows from this estimate by first choosing $\Lambda$ large enough to make the second term small and then, for that $\Lambda$, using the fact
that the first term is a finite sum of continuous functions that vanish for $t=s$.

To prove differentiability, we first show that the finite-volume derivatives converge to the desired limit, uniformly on compact intervals. Consider
\be
\tau_{t,s}(\delta_t(A)) - \tau^{(n)}_{t,s}(\delta^{(n)}_t(A)) =
\tau_{t,s}(\delta_t(A) - \delta^{(n)}_t(A)) + (\tau_{t,s} - \tau_{t,s}^{(n)} )(\delta^{(n)}_t(A)).
\ee
For the first term on the RHS we use the uniformity of the convergence of the derivation obtained in Theorem \ref{thm:exist_iv_dyn} (ii), 
and the boundedness of $\tau_{t,s}$, and we estimate the second term using Corollary~\ref{cor:continuity}~(iii). This produces
\be
\Vert(\tau_{t,s} - \tau_{t,s}^{(n)} )(\delta^{(n)}_t(A))\Vert
\leq 
\frac{2 \| \delta_t^{(n)} (A) \|}{C_F} e^{2 I_{t,s}(\Phi)} I_{t,s}(\Phi) \sum_{x \in X} \sum_{y \in \Gamma \setminus \Lambda_n} F(d(x,y)).
\ee
Since $\| \delta_t^{(n)} (A) \|$ is uniformly bounded in $t$ and $n$, the RHS vanishes as $n\to \infty$, uniformly for $s,t$ in any compact interval $I_0$.
This shows the uniform convergence of the derivative, i.e. for all $X\in\cP_0(\Gamma)$ and $A\in \cA_X$
\be
\lim_n \sup_{s,t\in I_0} \Vert \tau_{t,s}(\delta_t(A)) - \tau^{(n)}_{t,s}(\delta^{(n)}_t(A))\Vert =0.
\ee
By the continuity of $\tau_{t,s}(\delta_t(A))$ in $t$ and the usual argument using the fundamental theorem of calculus, it now follows that $\tau_{t,s}(A) $ is 
differentiable in $t$ with derivative given by $i\tau_{t,s}(\delta_t(A)) $.
\end{proof}

We conclude this section with a few comments.

It is clear how to modify Definition~\ref{def:lcfnorm} in such a way to describe sequences that are locally
Cauchy in $F$-norm. Given any such Cauchy sequence which also satisfies (\ref{sup_bd}), an $\epsilon /3$-argument almost
identical to the one in the proof of Theorem~\ref{thm:exist_iv_dyn} shows that the corresponding dynamics converge to a co-cycle 
of automorphisms of $\mathcal{A}_{\Gamma}$. 
With this understanding, one sees that Theorem~\ref{thm:exist_iv_dyn} implies Theorem~\ref{thm:existbd}.
In fact, let $\Phi \in \mathcal{B}_F(I)$ and take $\{ \Lambda_n \}_{n \geq 1}$ to be an 
increasing, exhaustive sequence of finite subsets of $\Gamma$. Define the sequence of interactions $\Phi_n = \Phi \restriction_{\Lambda_n}$
and extend $\Phi_n$ to all of $\Gamma$ as indicated above. In this case, it is clear that $\{ \Phi_n \}_{n \geq 1}$ 
is locally Cauchy in the $F$-norm defined by $F$ and moreover, (\ref{sup_bd}) holds. Thus the corresponding
dynamics converge. Since the sequence $\{ \Phi_n \}_{n \geq 1}$ also converges locally to $\Phi$ in the
$F$-norm defined by $F$, we know what the generator of the limiting dynamics is by Theorem~\ref{thm:exist_iv_dyn}. 
In this manner we recover the fact that the limiting dynamics is independent of the increasing, exhaustive sequence of finite-volumes.
Moreover, we also see that this limiting dynamics is invariant under a class of finite-volume boundary conditions.

As a final comment we note that one can easily find conditions under which the Duhamel formula \eq{norm_bd_dyn_der2} is also inherited 
by the infinite-volume dynamics. As this will not be needed in this work, we do not discuss this further. 

\section{Local approximations} \label{sec:str-loc}

In the previous section, we proved a Lieb-Robinson bound for the finite volume dynamics generated by an interaction $\Phi\in \cB_F(I)$. Such bounds 
provide an estimate for the speed of propagation in a quantum lattice system. More specifically, such bounds can be used to show that while the support 
of a local observable evolved under the Heisenberg dynamics is non-local, at any fixed time $t$, the observable essentially acts as the identity
outside of a finite region of space. It is often useful to approximate these dynamically evolved observables by strictly local observables. It is further desirable that 
the operation of taking these local approximations has good continuity properties. This is the topic of this section. 

%
%

\subsection{Local approximations of observables} \label{sec:laac}

We first review how the support of a local observable can be identified using commutators. For any Hilbert space $\cH$, the algebra $\cB(\cH)$ has a trivial center; in the case of a finite-dimensional Hilbert space this is known as Schur's Lemma. A first generalization of this fact is that for any two Hilbert spaces $\cH_1$ and $\cH_2$, the commutant of $\cB(\cH_1)\otimes \idty$ in $\cB(\cH_1\otimes\cH_2)$  is given by $\idty\otimes \cB(\cH_2)$ (see, e.g., \cite[Chapter 11]{kadison:1997a}). Given this, and the structure of the quantum models introduced in Section~\ref{sec:spatialstructure}, one concludes the following: given $\Lambda \in \mathcal{P}_0( \Gamma)$ and $X\subset\Lambda$, if $A \in \cA_\Lambda$ satisfies $[A,B]=0$ for all $B\in \cA_{\Lambda\setminus X}$, then $A\in \cA_X$. In other words, vanishing commutators can be used to identify the support of local observables. If the commutator $[A,B]$ is small but not necessarily vanishing (in norm), the following lemma, which is proved in \cite{nachtergaele:2013}, shows that $A$ can be well-approximated (up to error $\epsilon$) by an observable $A'\in \cA_X$.

\begin{lem}\label{lem:otamp10}
	Let $\cH_1$ and $\cH_2$ be complex Hilbert spaces. There is a completely positive linear map $\bbE: \cB(\cH_1\otimes\cH_2)\to\cB(\cH_1)$
	with the following properties:
	\begin{enumerate}
		\item[(i)] For all $A\in\cB(\cH_1)$, $\bbE(A\otimes\idty)=A$;
		\item[(ii)] Whenever $A\in\cB(\cH_1\otimes\cH_2)$ satisfies the commutator bound
		$$
		\left\Vert [A,\idty\otimes B]\right\Vert \leq \epsilon \Vert B\Vert\, \mbox{ for all } B\in\cB(\cH_2),
		$$
		$\bbE(A)$ satisfies the estimate 
		\be
		\Vert \bbE(A)\otimes \idty -A\Vert\leq \epsilon ;
		\label{epsilon_bound}
		\ee
		\item[(iii)] For all $C,D\in \cB(\cH_1)$ and $A\in \cB(\cH_1\otimes\cH_2)$, we have
		\be
		\bbE((C\otimes\idty)A(D\otimes \idty))=C\bbE(A)D.
		\label{condex}\nonumber
		\ee
	\end{enumerate}
\end{lem}

A completely positive linear map $\bbE$ with the properties (i) and (iii) is called a {\em conditional expectation}  (see, e.g., \cite[Section 9.2]{petz:2008}).
If $\cH_2$ is finite-dimensional, one can take $\bbE=\id\otimes \tr$,
where $\tr$ denotes the normalized trace over $\cH_2$. In this case, it is straightforward to verify the 
properties listed in the lemma (see, e.g., \cite{bravyi:2006a,nachtergaele:2006}). For general $\cH_2$, 
a normalized trace does not exist, but using Lemma \ref{lem:otamp10} it is easy to show that, at the cost of a factor 2 in the RHS of \eq{epsilon_bound}, we can replace $\bbE$ with $\id\otimes \rho$ for an arbitrary state $\rho$ on $\cB(\cH_2)$. This is the content of the following lemma from \cite{nachtergaele:2013}. 

\begin{lem} \label{lem:piest} 
	Let $\mathcal{H}_1$ and $\mathcal{H}_2$ be two complex Hilbert spaces and $\rho$ a state on $\cB(\cH_2)$. The map
	$\Pi_\rho=\id\otimes \rho$ satisfies properties (i) and (iii) of Lemma \ref{lem:otamp10}. Moreover, if 
	$A \in \mathcal{B}( \mathcal{H}_1 \otimes \mathcal{H}_2)$ is such that there is an $\epsilon \geq 0$ for which
	\begin{equation} \label{piest_com}
	\| [ A, \idty \otimes B ] \| \leq \epsilon  \| B \| \quad \mbox{for all } B \in \mathcal{B}( \mathcal{H}_2) ,
	\end{equation} 
	then
	\begin{equation}
	\| \Pi_{\rho}(A) \otimes \idty- A \| \leq 2 \epsilon \, .
	\label{2epsilon_bound}
	\end{equation}
\end{lem}
\begin{proof}
	It is clear that $\Pi_\rho$ satisfies properties (i) and (iii) in Lemma \ref{lem:otamp10}. 
	Moreover, since $\Vert \Pi_\rho\Vert =1$, we also have that for any $A\in\mathcal{B}( \mathcal{H}_1 \otimes \mathcal{H}_2)$
	\bea
	\Vert \Pi_\rho(A)\otimes\idty - \bbE(A)\otimes\idty\Vert &=& \Vert \Pi_\rho(A) - \bbE(A)\Vert\nonumber\\
	&=& \Vert \Pi_\rho(A-\bbE(A)\otimes\idty)\Vert\nonumber\\
	&\leq& \epsilon
	\eea
	where we have used (\ref{epsilon_bound}). The estimate \eq{2epsilon_bound} now follows:
	\be
	\| \Pi_{\rho}(A) \otimes \idty- A \| \leq  
	\Vert \Pi_\rho(A)\otimes\idty - \bbE(A)\otimes\idty\Vert + \Vert \bbE(A)\otimes \idty -A\Vert
	\leq 2 \epsilon \, .
	\ee
\end{proof}

Note that the state $\rho$ in Lemma \ref{lem:piest} is not required to be normal, i.e. defined by a density matrix. However, in applications it will often be useful to take $\rho$  normal (or locally normal in the case of infinite systems, see Section~\ref{sec:laqlmaps}). For normal $\rho$, we can give an explicit expression for $\Pi_\rho(A)$ in terms of
its matrix elements.  Since $\rho$ is given by a density matrix, there is a countable set of orthonormal vectors $\xi_n\in\cH_2$ and positive numbers $\rho_n$ with $\sum_n \rho_n = 1$, so that $\rho(A')=\sum_{n\geq 1} \rho_n \langle\xi_n, A'\xi_n\rangle$ for all $A'\in \cB(\cH_2)$. The matrix elements of $\Pi_\rho(A)$ are then given by:
\be
\langle \eta, \Pi_\rho(A)\phi\rangle 
= \sum_{n\geq 1} \rho_n \langle\eta\otimes\xi_n, A(\phi\otimes \xi_n)\rangle, \mbox{ for all } A \in\cB(\cH_1\otimes \cH_2),\ \eta,\phi\in\cH_1.
\label{matrixelementsPi}\ee

A number of further comments are in order. First, although the mapping $\Pi_{\rho}$ depends on the
state $\rho$, the `error' estimate in \eq{2epsilon_bound} is independent of $\rho$. 
Next, if $\mathcal{H}_2$ is finite dimensional, then $\rho$ can be taken to be the normalized trace 
and we already know that the factor of two in \eq{2epsilon_bound} is not needed. The bound for 
$\Pi_\rho$ therefore appears to be non-optimal. The map $\bbE$ from Lemma \ref{lem:otamp10} is only known to be 
bounded (specifically, $\Vert \bbE \Vert =1$) and hence continuous with respect to the operator norm topology. As such, it is not guaranteed that $\bbE$ is continuous when both the domain and co-domain are endowed with the strong operator topology. However, by choosing a {\em normal} state $\rho$, we 
get a map $\Pi_{\rho}$ that is continuous on bounded subsets of the domain when both $\cB(\cH_1\otimes \cH_2)$ and 
$\cB(\cH_1)$ are endowed with their strong (or weak) operator topologies. The case of the strong operator topology is 
the content of the next proposition. The case of the weak topology follows by a similar argument.

Recall that $\cK: \cB(\cH_1) \to \cB(\cH_2)$ is continuous on bounded subsets with both its domain and co-domain considered with the strong operator topology if given any bounded net $A_\alpha\in \cB(\cH_1)$ that converges strongly to $A\in \cB(\cH_1)$, the net $\cK(A_\alpha) \in \cB(\cH_2)$ converges strongly to $\cK(A) \in \cB(\cH_2)$.

\begin{prop}\label{prop:strongcontinuity}
	Let $\mathcal{H}_1$ and $\mathcal{H}_2$ be two complex Hilbert spaces, and  $\rho$ any normal state on $\cB(\cH_2)$. The following maps, when restricted to arbitrary bounded subsets of their domain, are continuous when both the domain and codomain are equipped with the strong operator topology:
	\begin{enumerate}
		\item[(i)]
		$\Pi_\rho=\id\otimes\rho:\cB(\cH_1\otimes\cH_2) \to \cB(\cH_1)$ 
		\item[(ii)]
		$\tilde\Pi_\rho:  \cB(\cH_1\otimes\cH_2) \to \cB(\cH_1\otimes\cH_2), A\mapsto \Pi_\rho(A) \otimes\idty$.
	\end{enumerate}
\end{prop}
\begin{proof}
	(i) To prove that $\Pi_\rho$ is continuous on bounded sets, WLOG, let $\{ A_\alpha \mid \alpha \in I\}$ be a net in the unit ball of $\cB(\cH_1\otimes\cH_2)$ that converges to $A\in \cB(\cH_1\otimes\cH_2)$ in the strong operator topology, i.e. for all $\psi\in \cH_1\otimes\cH_2$ the net $A_\alpha \psi$ converges to $A \psi$ with respect to the Hilbert space norm. Since $\Vert A_\alpha\Vert \leq 1$ for all $\alpha \in I$, we necessarily have that $\Vert A\Vert \leq 1$. Let $\{\xi_n, n\geq 1\}$ denote the orthonormal set of eigenvectors  of $\rho$ corresponding to its non-zero eigenvalues $\rho_n$.
	Let $\phi\in\cH_1$ with $\Vert \phi\Vert =1$. We use \eq{matrixelementsPi} to show that $\Pi_\rho(A_\alpha)\phi\to\Pi_\rho(A)\phi$. Note that 
	\beann
	\Vert \Pi_\rho(A - A_\alpha)\phi\Vert  &=&  \sup_{\eta\in\cH_1, \Vert \eta\Vert =1} \vert\langle \eta,\Pi_\rho(A-A_\alpha)\phi\rangle\vert\\
	&=& \sup_{\eta\in\cH_1, \Vert \eta\Vert =1} \sum_{n\geq 1} \rho_n \vert \langle \eta\otimes\ \xi_n , (A-A_\alpha) (\phi\otimes\xi_n)\rangle\vert.
	\eeann
	For any $\epsilon >0$, choose $N$ so that  $\sum_{n>N} \rho_n < \epsilon/4$, and pick $\alpha_0\in I$ so that for all $\alpha \geq \alpha_0$ and any $n=1,\ldots, N$,
	$$\Vert (A-A_\alpha) (\phi\otimes\xi_n)\Vert \leq \epsilon/(2N).$$  Then for any $\alpha \geq \alpha_0$,
	\be
	\Vert \Pi_\rho(A - A_\alpha)\phi\Vert  \leq \sum_{n\leq N} \Vert (A-A_\alpha) (\phi\otimes\xi_n)\Vert + 
	\sum_{n>N} 2\rho_n < \epsilon,
	\label{splittheterms}\ee
	which establishes the desired continuity of $\Pi_\rho$.
	
	(ii)  Let $\{ A_{\alpha} \, | \alpha \in I\}$ be a net in the unit ball of $\cB( \cH_1 \otimes \cH_2)$ that converges to $A \in \cB( \cH_1 \otimes \cH_2)$. By (i), we know that $B_{\alpha} = \Pi_\rho(A_{\alpha})$ converges to $B = \Pi_\rho(A)$ in the strong operator topology on $\cB(\cH_1)$. By Proposition~\ref{prop:st_cont_prods}(ii), see also \eqref{sc_tensor} and the preceding discussion, it follows that $\{ B_\alpha\otimes \idtyty \mid \alpha \in I\}$ strongly converges to $B\otimes \idtyty$ in $\cB(\cH_1\otimes \cH_2)$.
\end{proof}

%
%

\subsection{Application to quantum lattice models} \label{sec:laqlmaps}

We now extend the results of the previous subsection to infinite quantum lattice systems on $\Gamma$. In this setting states cannot be defined in terms of a single density matrix. Moreover,
as explained below, we will want to define a consistent family of conditional expectations with values in $\cA_\Lambda$, for all $\Lambda\in\cP_0(\Gamma)$. To this end, we consider a {\em locally normal product state} $\rho$: i.e. for each site $x \in \Gamma$ we fix a normal state $\rho_x$ on $\mathcal{A}_x = \mathcal{B}( \mathcal{H}_x)$, and take the unique state $\rho$ on $\cA_\Gamma$ such that $\rho_{\restriction \cA_\Lambda} = \bigotimes_{x \in \Lambda} \rho_x$ for all finite $\Lambda\subset\Gamma$. Then, given
$X \subset \Lambda \in \cP_0(\Gamma)$, we define conditional expectations  $\Pi_{\rho}^{X, \Lambda} : \mathcal{A}_{\Lambda} \to \mathcal{A}_X$ similar to those in Lemma~\ref{lem:piest} by
\begin{equation} \label{pi_loc_norm_state}
\Pi_{\rho}^{X, \Lambda}(A) = ({\rm id}_X \otimes \rho_{\Lambda \setminus X})(A) \quad \mbox{for all } A \in \cA_{\Lambda} \,.
\end{equation}
Here, as before, we have taken ${\rm id}_X$ as the identity map on $\mathcal{A}_X$.
In our applications the dependence of these maps on $\rho$ is of minor consequence. Moreover, it will be convenient to view these maps as elements of 
$\mathcal{B}( \cA_{\Lambda})$. For these reasons, we suppress the dependence on $\rho$ and define $\Pi_X^{\Lambda} : \mathcal{A}_{\Lambda} \to \cA_{\Lambda}$ by
\begin{equation} \label{PiXLambda}
\Pi_X^{\Lambda}(A) = \Pi_{\rho}^{X, \Lambda}(A) \otimes \idty_{\Lambda \setminus X}.
\end{equation}
For fixed $X$, these projections are compatible in the sense that if $X \cup \Lambda' \subset \Lambda$ and $\Lambda \in \cP_0(\Gamma)$, then
\begin{equation} \label{compat}
\Pi_X^{\Lambda}(A \otimes \idty) = \Pi_{X\cap \Lambda'}^{\Lambda'}(A) \otimes \idty \quad \mbox{for all } A \in \cA_{\Lambda'} .
\end{equation}
We summarize this and other consistency relations in Proposition~\ref{prop:compatibility} below. First, however, we describe how, given a fixed finite volume $X$, one can extend the maps $\Pi_X^\Lambda$, $\Lambda \in \cP_0(\Gamma)$, to an operator $\Pi_X$ on $\cA^{\rm loc}_{\Gamma}$ (and consequently , $\cA_\Gamma$). 

For any $\Lambda' \subseteq \Lambda$, recall that we can identify $\cA_{\Lambda'}$ as a sub-algebra of $\cA_{\Lambda}$, and so we can write \eq{compat} as $\Pi_X^{\Lambda}\restriction_{\cA_{\Lambda'}}= \Pi_{X\cap \Lambda'}^{\Lambda'}$. In particular, if $X\subseteq \Lambda' \subseteq \Lambda$, one has that $\Pi_X^{\Lambda}\restriction_{\cA_{\Lambda'}}= \Pi_{X}^{\Lambda'}$, from which we see that the following map $\Pi_X : \cA_{\Gamma}^{\rm loc} \to \cA_{\Gamma}^{\rm loc}$ is well-defined:
\begin{equation}\label{PiX}
\Pi_X(A) := \Pi_X^{{\rm supp}(A)\cup X}(A).
\end{equation}
Since this map is bounded, in fact of norm one, $\Pi_X$ has a unique extension to $\cA_{\Gamma}$ which we also denote by $\Pi_X$. We note that $\Pi_X(A)=A $ if $A\in\cA_X.$
We refer to the family of conditional expectations  $\Pi^\Lambda_X$, respectively $\Pi_X$,  $X\in\cP_0(\Gamma)$, as a {\em localizing family}.
By construction, the finite volume local approximations $\Pi_X^\Lambda$ all satisfy the conditions of Lemma~\ref{lem:piest}. The following corollary shows 
that the results of Lemma~\ref{lem:piest} also extend to $\Pi_X: \cA_{\Gamma} \to \cA_X$.

\begin{cor}\label{cor:PiX}
	Let $X\in \cP_0(\Gamma)$ and $\Pi_X: \cA_\Gamma\to\cA_X$ be the extension of the map defined in \eqref{PiX}. Suppose $\epsilon \geq 0$ and $A\in \cA_\Gamma$ are such that
	\be\label{small_commutator}
	\Vert [A,B]\Vert \leq \epsilon \Vert B\Vert, \mbox{ for all } B\in \cA^{\rm loc}_{\Gamma \setminus X}.
	\ee
	Then 
	\[
	\Vert \Pi_X(A) - A\Vert \leq 2\epsilon.
	\]
\end{cor}
\begin{proof}
	Let $A\in \cA_\Gamma$. Then for any $\delta >0$ there exists $\Lambda\in\cP_0(\Gamma)$ and $A^\prime\in\cA_\Lambda$ such that
	$\Vert A-A^\prime\Vert < \delta$. From \eq{small_commutator} it follows that
	\[
	\Vert [A^\prime,B] \Vert \leq (\epsilon +2\delta) \Vert B\Vert, \mbox{ for all } B\in \cA^{\rm loc}_{\Lambda\setminus X}.
	\]
	Since $\Pi_X(A') = \Pi_X^\Lambda(A')$, Lemma 4.2 implies
	\[
	\Vert \Pi_X(A^\prime)-A^\prime\Vert \leq 2(\epsilon + 2\delta).
	\]
	Therefore,
	\[
	\Vert \Pi_X(A)-A\Vert \leq \Vert \Pi_X(A^\prime)-A^\prime\Vert  +\Vert (\Pi_X-\id)(A-A^\prime) \Vert
	\leq 2(\epsilon + 2\delta) + 2\delta,
	\]
	and since $\delta >0$ is arbitrary, the result follows.
\end{proof}

We now state several consistency properties of the finite and infinite volume conditional expectations. To facilitate the statement of the properties, we use $\Pi_X^\Gamma$ to denote $\Pi_X: \cA_\Gamma \to \cA_X$.

\begin{prop} \label{prop:compatibility}
	Fix a locally normal product state $\rho$ on $\cA_\Gamma$, and let $X, \, Y, \, \Lambda' \in \cP_0(\Gamma)$ and $\Lambda \in \cP_0(\Gamma) \cup \{\Gamma\}$ be such that $X$, $Y$ and $\Lambda'$ are all subsets of $\Lambda$. The following properties hold for the localizing maps defined with respect to $\rho$:
	\begin{enumerate}
		\item[(i)] If $A\in \cA_X$, then $\Pi_X^\Lambda(A) = A$.
		\item[(ii)] $\Pi_X^\Lambda \restriction_{\cA_{\Lambda'}} = \Pi_{X\cap \Lambda'}^{\Lambda'}$.
		\item[(iii)] $\Pi_X^\Lambda \circ \Pi_Y^\Lambda = \Pi_Y^\Lambda \circ \Pi_X^\Lambda = \Pi_{X\cap Y}^\Lambda.$
		\item[(iv)] If $X\subseteq \Lambda'$, then $\Pi_X^{\Lambda'}\circ \Pi_{\Lambda'}^\Lambda = \Pi_X^{\Lambda}$.
		\item[(v)] $\Pi_X^{\Lambda}(A)^* = \Pi_X^{\Lambda}(A^*)$ for all $A\in \cA_{\Lambda}$.
	\end{enumerate}
\end{prop}

The proofs of these properties for $\Lambda \in \cP_0(\Gamma)$ are all elementary and follow from the definition of $\Pi_{X,\Lambda}^\rho$, and the fact that $\rho\restriction_{\Lambda}$ is a product state. The statements for $\Lambda = \Gamma$ follow from taking finite volume limits $\Lambda'\uparrow \Gamma$ of $\Pi_X^{\Lambda'}(A)$ for $A\in\cA_\Gamma^{\rm loc}$, and using the norm bound $\|\Pi_X\|\leq 1$ to extend to $A\in \cA_\Gamma$ in the usual manner.

An immediate consequence of Proposition \ref{prop:compatibility}(i) is that 
\be \label{PiXLimit}
\lim_{\Lambda_n \uparrow \Gamma} \Pi_{\Lambda_n}(A) = A
\ee 
for any sequence of increasing and absorbing finite volumes $\Lambda_n$, and any $A\in \cA_\Gamma^{\rm loc}$. Since $\cA_\Gamma^{\rm loc}$ is dense in $\cA_{\Gamma}$, \eqref{PiXLimit} extends to any $A\in\cA_{\Gamma}$.

With the aid of the maps $\Pi_X^\Lambda$, we construct {\it local decompositions} of any observable $A\in\cA_{\Lambda}$. 
Let $X \subset \Gamma$ be finite. For any $n \geq 0$, denote by $X(n) \subset \Gamma$ the set
\begin{equation} \label{defxn}
X(n) = \{ y \in \Gamma : \mbox{there exists } x \in X \mbox{ with } d(x,y) \leq n \}. 
\end{equation}
Note that $X(0) = X$. For finite $\Lambda \subset \Gamma$ with $X \subset \Lambda$ and each integer $n \geq 0$, we define $\Delta_{X(n)}^{\Lambda} : \mathcal{A}_{\Lambda} \to \mathcal{A}_{\Lambda}$ by
\begin{equation} \label{loc_dec}
\Delta_{X(0)}^{\Lambda} = \Pi_{X}^{\Lambda} \quad \mbox{and} \quad \Delta_{X(n)}^{\Lambda} = \Pi_{X(n) \cap \Lambda}^{\Lambda} - \Pi_{X(n-1) \cap \Lambda}^{\Lambda}
\end{equation} 
for any $n \geq 1$. Note that, in contrast to the maps $\Pi^\Lambda_X$, $\Delta^\Lambda_{X(n)}$ does not only depend on the set $X(n)$, but on $X$ and $n$ separately. 
This slight abuse of notation will not lead to confusion.
As discussed above, $\Delta_{X(n)}^{\Lambda}$ has range contained in $\mathcal{A}_{X(n) \cap \Lambda}$. Moreover, as a difference of two projections, they satisfy
$\| \Delta_{X(n)}^{\Lambda} \| \leq 2$.

Of course, as discussed above, one can also extend these bounded linear maps to $\mathcal{A}_{\Gamma}$.
In fact, for each finite $X \subset \Gamma$ and any $n \geq 0$ the maps $\Delta_{X(n)} : \mathcal{A}_{\Gamma} \to \mathcal{A}_{\Gamma}^{\rm loc}$
are defined by
\begin{equation} \label{DeltaXn}
\Delta_{X(0)} = \Pi_X \quad \mbox{and} \quad \Delta_{X(n)} = \Pi_{X(n)} - \Pi_{X(n-1)}
\end{equation}
for any $n \geq 1$. We note again that the range of $\Delta_{X(n)}$ is contained in $\mathcal{A}_{X(n)}$ regarded as a sub-algebra of $\mathcal{A}_{\Gamma}^{\rm loc}$. 

A typical use of the local decompositions is as follows. Fix $\Lambda\subset\cP_0(\Gamma)$, $A \in \mathcal{A}_{\Lambda}$, and $X\subset\Lambda$,
and denote by $N$ the smallest integer $n$ for which $X(n) \cap \Lambda = \Lambda$. Clearly, $N$ depends on $X$ and $\Lambda$, and $\Delta_{X(n)}^{\Lambda} = 0$ for any $n > N$. Then, one can write
\begin{equation}
A = \sum_{n \geq 0} \Delta_{X(n)}^{\Lambda}(A) = \sum_{n=0}^N \Delta_{X(n)}^{\Lambda}(A) 
\label{telescopicA}\end{equation}
where this telescopic sum has terms with explicit, local support. 

For a quasi-local observable $A\in \cA_\Gamma$ and $X\in \cP_0(\Gamma)$, the conditional convergence of the infinite-volume analog of \eqref{telescopicA}, namely
\be \label{CondConvergence}
A = \sum_{n \geq 0}\Delta_{X(n)}(A), 
\ee
follows from noticing that $\Pi_{X(N)}(A) = \sum_{n=0}^N\Delta_{X(n)}(A)$, and invoking \eqref{PiXLimit}. In Section~\ref{sec:qlm_gen}, we will discuss situations
in which \eqref{CondConvergence} converges absolutely. 
The remainder of this section is concerned with continuity properties and basic estimates for the local approximations $\Pi_X$.

\subsubsection{Continuity of local approximations} \label{sec:cont_l_a}

Given finite sets $X \subset \Lambda \subset \Gamma$, Proposition~\ref{prop:strongcontinuity}(ii) implies that the projection map $\Pi_X^{\Lambda}$ preserves continuity in the strong operator topology. In particular, if $t \mapsto A(t) \in \cA_{\Lambda}$ is strongly continuous for all $t$ in an interval $I\subseteq \bR$, then $ t \mapsto \Pi_X^{\Lambda}(A(t)) \in \cA_{\Lambda}$ is also
strongly continuous. In applications, we will be interested in a sequence of strongly continuous functions $t\mapsto A_{\Lambda_n}(t)\in\cA_{\Lambda_n}$, with $\Lambda_n \uparrow \Gamma$, that converges to a bounded map $t\mapsto A(t) \in \cA_{\Gamma}$. It will then be desirable that the localizing projections $\Pi_Y(A(t))$, $Y\in \cP_0(\Gamma)$, also satisfy 
certain continuity properties. 

While we do not have the standard von Neumann algebra setting where the notion of locally normal is more natural, it is convenient to define a similar notion in our setting with $C^*$-algebras without reference to a representation.

\begin{defn}\label{def:locally_normal_maps}
	A linear map $\cK:\cA_\Gamma^{\rm loc}\to\cA_\Gamma$ is called {\em locally normal} if there exists an increasing, exhaustive sequence $\{ \Lambda_n \}_{n \geq 0}$
	of finite subsets of $\Gamma$ and corresponding bounded linear transformations 
	$\cK^{\Lambda_n}\in\cB(\cA_{\Lambda_n})$ with the following properties:
	\begin{enumerate}
		\item[(i)] For all $n$, $\cK^{\Lambda_n}:\cA_{\Lambda_n} \to \cA_{\Lambda_n}$ is continuous on bounded subsets 
		when both its domain and co-domain are considered with the strong operator topology;
		\item[(ii)] {\em Local uniform convergence} of $\cK^{\Lambda_n}$ to $\cK$: For all $X \subset \Gamma$ finite and any $\epsilon>0$, 
		there exists $N$ such that for all $n\geq N$ we have
		\be
		\Vert \cK(A) - \cK^{\Lambda_n}(A)\Vert \leq \epsilon\Vert A\Vert, \mbox{ for all } A\in\cA_X.
		\ee
	\end{enumerate}
\end{defn}

Note that {\em local uniform convergence} implies that $\cK^{\Lambda_n}$ converges strongly to $\cK$. However, since $N$ is allowed to depend on $X$, this convergence 
is in general not uniform in $\cP_0(\Gamma)$. If $\cH_{\Lambda_n}$ is finite-dimensional, property (i) is automatically satisfied. 
Let us now consider an example satisfying Definition~\ref{def:locally_normal_maps}.

\begin{ex} \label{ex_der_sec_4} Consider a quantum lattice system comprised of $(\Gamma, d)$ and $\cA_{\Gamma}$.
Let $F$ be an $F$-function on $(\Gamma, d)$ and $\Phi \in \mathcal{B}_F$ be an interaction. The map
$\mathcal{K} : \cA_{\Gamma}^{\rm loc} \to \cA_{\Gamma}$ given by
\begin{equation}
\mathcal{K}(A) = \sum_{Z \in \mathcal{P}_0( \Gamma)} [ \Phi(Z), A ] \quad \mbox{for any } A \in \cA_{\Gamma}^{\rm loc} 
\end{equation}
is locally normal in the sense of Definition~\ref{def:locally_normal_maps}.
In fact, let $\{ \Lambda_n \}_{n \geq 0}$ be any sequence of non-empty, finite subsets of $\Gamma$ that are increasing and
exhaustive. For each $n \geq 0$, define $\mathcal{K}^{\Lambda_n} : \cA_{\Lambda_n} \to \cA_{\Lambda_n}$ by setting
\begin{equation} \label{fvol_der_4}
\mathcal{K}^{\Lambda_n}(A) = \sum_{Z \subset \Lambda_n}  [ \Phi(Z), A ] \quad \mbox{for any } A \in \cA_{\Lambda_n} \, . 
\end{equation} 
Fix $n \geq 0$, let $X \subset \Lambda_n$ and $A \in \cA_X$. One checks that
\begin{equation} \label{fvol_der_4_simp}
\mathcal{K}^{\Lambda_n}(A) = \sum_{\stackrel{Z \subset \Lambda_n:}{Z \cap X \neq \emptyset}}  [ \Phi(Z), A ]
\end{equation}
and therefore,
\begin{equation} \label{fvol_der_4_est}
\| \mathcal{K}^{\Lambda_n}(A)  \| \leq \sum_{x \in X} \sum_{z \in \Lambda_n} \sum_{\stackrel{Z \subset \Lambda_n:}{x,z \in Z}} \| [ \Phi(Z), A ] \| \leq 2 \| \Phi \|_F \| F \| |X| \| A \| 
\end{equation}
holds for any $A \in \cA_X$, where we have used that $\Phi \in \mathcal{B}_F$. Taking $X = \Lambda_n$, one sees that $\mathcal{K}^{\Lambda_n} \in \mathcal{B}( \cA_{\Lambda_n})$. With $n \geq 0$ fixed again, let $Z \subset \Lambda_n$. It is clear that $A \mapsto [ \Phi(Z), A ]$ satisfies Definition~\ref{def:locally_normal_maps} (i) on $\cA_{\Lambda_n}$. As a finite sum of such terms, it is clear that $\mathcal{K}^{\Lambda_n}$ satisfies Definition~\ref{def:locally_normal_maps} (i) as well. Lastly, let $X \in \mathcal{P}_0( \Gamma)$, $A \in \mathcal{A}_X$, and $N \geq 1$ be sufficiently large
so that $X \subset \Lambda_N$. Again, one checks that for any $n \geq N$
\begin{equation}
\mathcal{K}(A) - \mathcal{K}^{\Lambda_n}(A) = \sum_{\stackrel{Z \in \mathcal{P}_0( \Gamma):}{Z \cap X \neq \emptyset, Z \cap ( \Gamma \setminus \Lambda_n) \neq \emptyset}} [ \Phi(Z), A] 
\end{equation}
and therefore,
\begin{equation}
\| \mathcal{K}(A) - \mathcal{K}^{\Lambda_n}(A) \| \leq 2 \| A \| \| \Phi \|_F \sum_{x \in X} \sum_{y \in \Gamma \setminus \Lambda_n} F(d(x,y)) \, .
\end{equation} 
Since $|X|< \infty$, Definition~\ref{def:locally_normal_maps} (ii) holds as $F$ is summable.
\end{ex}

A simple consequence of Definition~\ref{def:locally_normal_maps} (i) is the following: for each $n \geq 0$,  $t\mapsto \cK^{\Lambda_n} (A(t))$ is strongly continuous if $t\mapsto A(t)$ is strongly continuous. The next lemma establishes that the same property holds for the composition $\Pi_Y(\cK(A(t)))$ for any $Y\in\cP_0(\Gamma)$.

\begin{lem} \label{lem:picont}
	Let $X,Y\in\cP_0(\Gamma)$, $\Pi_Y:\cA_\Gamma\to \cA_Y$ be the extension of the map defined in \eq{PiX}, 
	and  $\cK:\cA_\Gamma^{\rm loc}\to\cA_\Gamma$ be a locally normal map. Then, for every strongly continuous map $t\mapsto A(t)\in\cA_X$ defined on an interval $I\subset\Rl$, the function $t\mapsto \Pi_Y(\cK(A(t)))\in\cA_Y$, is also strongly continuous.
\end{lem}
\begin{proof}
	For $\psi\in\cH_Y$, define $f(t) = \Pi_Y (\cK(A(t)))\psi$. We will prove that $f:I\to\cH_Y$ is continuous by showing that on compact intervals it is the uniform limit of a sequence of continuous functions. Let $\cK_{\Lambda_n}$ be a sequence of maps of the type described in Definition~\ref{def:locally_normal_maps}.  For $n$ large enough so that $X \subset \Lambda_n$, define $f_n(t) = \Pi_Y^{\Lambda_n} (\cK^{\Lambda_n}(A(t)))\psi$. Since $\cK$ is locally normal, each $f_n$ is continuous; here we are using Propostion~\ref{prop:strongcontinuity}(i) and that $t\mapsto \cK^{\Lambda_n}(A(t))$ is strongly continuous by Definition~\ref{def:locally_normal_maps}(i). Using compatibility, see Proposition~\ref{prop:compatibility}(ii), it is clear that $f_n(t) = \Pi_Y (\cK^{\Lambda_n}(A(t)))\psi$. Now for any compact set $J \subset I$, the estimate
	\bea
	\sup_{t\in J}\Vert f_n(t) - f(t)\Vert
	&\leq& \Vert \Pi_Y\Vert \sup_{t\in J} \Vert \cK^{\Lambda_n}(A(t)) -\cK(A(t))\Vert \Vert \psi\Vert\nonumber\\
	&\leq & \epsilon \Vert\psi\Vert\,\sup_{t\in J} \Vert A(t)\Vert,
	\eea
	follows from $\Vert \Pi_Y\Vert =1$ and local uniform convergence. 
	Since $A(t)$ is locally bounded and $J$ is compact, this proves the claim.
\end{proof}

Two comments are in order. First, by Proposition \ref{prop:strongcontinuity}(ii), for any $Z\in\cP_0(\Gamma)$ such that $Y\subset Z$, $t \mapsto \Pi_Y(\cK(A(t)))$ considered as a map into $\cA_Z$ is also strongly continuous. Second, if $t\mapsto A(t)$ is, in fact, continuous in the norm topology on $\cA_X$ (in particular, if $\dim \cH_X < \infty$), and $\cK$ is bounded, the result of the Lemma~\ref{lem:picont} is trivial since the bounded linear map $\Pi_Y \circ \cK$ preserves the norm-continuity.

We will also encounter one-parameter families of locally normal maps, $\{\cK_s \mid s\in I\}$, that are strongly continuous in $s$ and {\em uniformly} locally normal in the sense of the following definition. 

\begin{defn}\label{def:uniformly_locally_normal_maps}
	Let $I\subset \Rl$ be an interval. A family of linear maps $\cK_s:\cA_\Gamma^{\rm loc}\to\cA_\Gamma$, $s\in I$, is called a {\em strongly continuous family of uniformly locally normal maps} if there exists an increasing, exhaustive sequence $\{ \Lambda_n \}_{n \geq 0}$
	of finite subsets of $\Gamma$ and families of bounded linear maps
	$\cK^{\Lambda_n}_s\in\cB(\cA_{\Lambda_n})$ strongly continuous in $s$, with the following properties:
	\begin{enumerate}
		\item[(i)] For all $n$ and $s$, $\cK^{\Lambda_n}_s:\cA_{\Lambda_n} \to \cA_{\Lambda_n}$ is continuous on bounded subsets 
		when both its domain and co-domain are considered with the strong operator topology, and this continuity is 
		uniform for $s\in I$. 
		\item[(ii)] {\em Uniform local convergence} of $\cK^{\Lambda_n}_s$ to $\cK_s$: For all $X \subset \Gamma$ finite and any $\epsilon>0$, 
		there exists $N$ such that for all $n\geq N$ we have
		\be
		\Vert \cK_s(A) - \cK^{\Lambda_n}_s(A)\Vert \leq \epsilon\Vert A\Vert, \mbox{ for all } A\in\cA_X, \mbox{ and all } s\in I.
		\label{uniformlocalconvergence}\ee
	\end{enumerate}
\end{defn}

In (i), uniform for $s\in I$, means that given a bounded net $\{A_\alpha\}_{\alpha\in I}$ converging strongly to $A$ and $\epsilon >0$, there exists a choice of $\alpha_0\in I$,
independent of $s$, so that
\[\|\cK^{\Lambda_n}_s(A_\alpha)\psi -\cK^{\Lambda_n}_s(A)\psi \| < \epsilon, \mbox{ for all } \alpha \geq \alpha_0. \]

For families $\cK_s$, $s\in I$ with $I$ an infinite interval, the uniformity asked for in part (ii) of this definition will typically not hold and one is led to consider subfamilies parametrized by $s\in I_0\subset I$, for compact intervals $I_0$. Also note that the properties of a strongly continuous family of uniformly locally normal maps imply
that $s\to \cK_s$ is strongly continuous by the usual $\epsilon/3$ argument. We have not assumed, however, that the maps $\cK_s$ are bounded. In general, $\cK_s$ is only locally
bounded and cannot be extended to all of $\cA_\Gamma$.

We now discuss two examples. The first is for a model with uniformly bounded on-sites, while 
the second does not require this assumption.  

\begin{ex} \label{ex_dyn_sec_4} Consider a quantum lattice system comprised of $(\Gamma, d)$ and $\cA_{\Gamma}$.
Let $F$ be an $F$-function on $(\Gamma, d)$, $I \subset \mathbb{R}$ be an interval, and $\Phi \in \mathcal{B}_F(I)$ be a
strongly continuous interaction. For each $s_0 \in I$ and any compact $I_0 \subset I$, we claim that
$\mathcal{K}_t : \cA_{\Gamma}^{\rm loc} \to \cA_{\Gamma}$ given by
\begin{equation} \label{ex_ivdyn}
\mathcal{K}_t(A) = \tau_{t, s_0}(A) \quad \mbox{for any } A \in \cA_{\Gamma}^{\rm loc} \mbox{ and } t \in I_0 \,,
\end{equation}
is a strongly continuous family of uniformly locally normal maps in the sense of Definition~\ref{def:uniformly_locally_normal_maps}. 
Here, the dynamics in (\ref{ex_ivdyn}) we are using is the infinite volume dynamics corresponding to $\Phi$, from Theorem~\ref{thm:existbd}, with $H_z=0$ for all $z \in \Gamma$.  

To see that this is an example of Definition~\ref{def:uniformly_locally_normal_maps}, let $\{ \Lambda_n \}_{n \geq 0}$ be any
non-empty sequence of finite subsets of $\Gamma$ that are increasing and exhaustive. For each $n \geq 0$ and any $t \in I_0$, define
$\mathcal{K}_t^{\Lambda_n} : \cA_{\Lambda_n} \to \cA_{\Lambda_n}$ by setting
\begin{equation}
\mathcal{K}_t^{\Lambda_n}(A) = \tau_{t,s_0}^{\Lambda_n}(A) \quad \mbox{for any } A \in \cA_{\Lambda_n} \, ,
\end{equation}
the finite-volume dynamics associated to $\Phi$. 

Fix $n \geq 0$, $t_0 \in I_0$, and $A \in \mathcal{A}_{\Lambda_n}$. It is clear that
\begin{equation}
\mathcal{K}_t^{\Lambda_n}(A) - \mathcal{K}_{t_0}^{\Lambda_n}(A) = \int_{t_0}^t \tau_{r, s_0}^{\Lambda_n}( [iH_{\Lambda_n}^{\Phi}(r), A] ) \, dr
\end{equation}
holds, in the strong sense, for any $t \in I_0$. Thus
\begin{equation} \label{ex_dyn_cont}
\| \mathcal{K}_t^{\Lambda_n}(A) - \mathcal{K}_{t_0}^{\Lambda_n}(A) \| \leq 2 \| F \| | \Lambda_n | \| A \| \int_{ {\rm min}(t_0,t)}^{ {\rm max}(t_0,t)} \| \Phi \|_F(r) \, dr 
\end{equation} 
and therefore, $\mathcal{K}_t^{\Lambda_n} \in \mathcal{B}( \cA_{\Lambda_n})$ is strongly continuous in $t$ for each $n \geq 0$. 
Here we have argued as in the proof of (\ref{int_dyn_est_small_R}) in Corollary~\ref{app:cor:int_bd}
using that $\| \Phi \|_F$ is locally integrable.

For each $n \geq 0$, one can show that property (i) of Definition~\ref{def:uniformly_locally_normal_maps} holds by arguing as in the proof of Proposition~\ref{properties_of_expiths}(iii), and using that $I_0$ is compact and $\| \Phi \|_F$ is locally bounded. 

Finally, we observe that (ii) is a simple consequence of Corollary~\ref{cor:continuity} (iii). 
\end{ex}

\begin{ex} \label{ex_wio_sec_4}
Consider a quantum lattice system comprised of $(\Gamma, d)$ and $\cA_{\Gamma}$.
Fix a collection of densely defined, self-adjoint on-site Hamiltonians $\{ H_z \}_{z \in \Gamma}$. 
Let $F$ be an $F$-function on $(\Gamma, d)$, $I \subset \mathbb{R}$ be an interval, and 
take $\Phi_s \in \mathcal{B}_F( \mathbb{R})$ for each $s \in I$. In this case, for any 
$w \in L^1( \mathbb{R})$ the family $\{ \mathcal{K}_s \}_{s \in I}$ of linear maps with
$\mathcal{K}_s : \cA_{\Gamma}^{\rm loc} \to \cA_{\Gamma}$ given by
 \begin{equation}
\mathcal{K}_s(A) = \int_{\mathbb{R}} \tau_t^s(A) \, w(t) \, dt \quad \mbox{for all } A \in \mathcal{A}_{\Gamma}^{\rm loc} \mbox{ and } s \in I \, 
\end{equation}
is well-defined. Here for each fixed $s \in I$, $\tau_t^{s}$ is the infinite-volume dynamics corresponding to $\Phi_s$ whose
existence is proven in Theorem~\ref{thm:existbd}.
 
We will show that, under some additional assumptions on $\Phi_s$, for each compact $I_0 \subset I$,
$\{ \mathcal{K}_s \}_{s \in I_0}$ is a strongly continuous family of uniformly locally normal maps in the sense of Definition~\ref{def:uniformly_locally_normal_maps}. These assumptions are:
\begin{enumerate}
	\item[(i)] For each $Z \in \mathcal{P}_0( \Gamma)$, $(s,t) \mapsto \Phi_s(Z,t)$ is 
	jointly strongly continuous on $I_0\times \Rl$.
	\item[(ii)]  For each $Z \in \mathcal{P}_0( \Gamma)$ and $t \in \mathbb{R}$, $s \mapsto \Phi_s(Z,t)$
	is strongly differentiable, and its derivative $(s,t) \mapsto \Phi_s'(Z,t)$ is jointly strongly continuous on $I_0\times \Rl$.
	\item[(iii)] For each $s \in I_0$, $\Phi_s' \in \mathcal{B}_F(\mathbb{R})$ and moreover, for each
	$T>0$ 
	\begin{equation}
	\sup_{s \in I_0} \int_{-T}^T \| \Phi_s \|_F(t) \, dt < \infty ,\quad \sup_{s \in I_0} \int_{-T}^T \| \Phi_s' \|_F(t) \, dt < \infty.
	\end{equation}
\end{enumerate}

To prove the above claim, choose any sequence $\{ \Lambda_n \}_{n \geq 0}$ of non-empty, increasing, exhaustive finite subsets of $\Gamma$. 
For each such $n \geq 0$ and any $s \in I_0$, define approximating maps  $\mathcal{K}_s^{\Lambda_n} : \mathcal{A}_{\Lambda_n} \to \cA_{\Lambda_n}$
by setting  
\begin{equation}
\mathcal{K}_s^{\Lambda_n}(A) = \int_{\mathbb{R}} \tau_t^{\Lambda_n, s}(A) \, w(t) \, dt \quad \mbox{for all } A \in \mathcal{A}_{\Lambda_n} \mbox{ and } s \in I_0 \, .
\end{equation}
Here, $\tau_t^{\Lambda_n, s}(A) = U_{\Lambda_n}^s(t,0)^* AU_{\Lambda_n}^s(t,0)$ is the dynamics generated by the finite volume Hamiltonian
\begin{equation}
H_{\Lambda_n}^s(t) = \sum_{z \in \Lambda_n} H_z + \sum_{Z \subset \Lambda_n} \Phi_s(Z, t) .
\end{equation} 

We first show that for each $n \geq 0$, the map $\mathcal{K}_s^{\Lambda_n} \in \mathcal{B}( \cA_{\Lambda_n})$ is
strongly continuous. In fact, the argument below demonstrates that $\mathcal{K}_s^{\Lambda_n}$ is uniformly continuous
in the operator norm on $\mathcal{B}( \cA_{\Lambda_n})$.

Fix $n \geq 0$. Let $s \in I_0$ and $A \in \cA_{\Lambda_n}$. Assumptions (i) and (ii) above guarantee that 
the strong derivative of the finite-volume dynamics satisfies
\begin{equation}
\frac{d}{ds} \tau_t^{\Lambda_n,s}(A) = i \sum_{Z \subset \Lambda_n} \int_0^t \tau_r^{\Lambda_n, s}([\Phi'_s(Z,r), \tau_{t,r}^{\Lambda_n,s}(A)]) \, dr  \, .
\end{equation}
Using Assumption (iii) and the estimate (\ref{norm_bd_dyn_der2}), for any $X \subset \Lambda_n$, each $A \in \cA_X$, and any $T>0$, there exists $M>0$
such that
\begin{equation}
\sup_{s \in I_0} \sup_{t \in [-T, T]} \left\| \frac{d}{ds} \tau_t^{\Lambda_n,s}(A) \right\| \leq 2 \| A \| \| F \| |X| M e^{2 C_F M} 
\end{equation}

Now, let $s_0 \in I_0$ and take $\epsilon >0$. Since $w \in L^1( \mathbb{R})$, it is clear that there exists $T>0$ for which
\begin{equation} \label{ep_tail_of_W}
\int_{|t| > T} |w(t)| \, dt \leq \frac{\epsilon}{4}  .
\end{equation} 
A short calculation shows that
\begin{equation}
\left\| \int_{|t| \leq T} \left( \tau_t^{\Lambda_n, s}(A) - \tau_t^{\Lambda_n, s_0}(A)\right) w(t) \, dt \right\| \leq 2 \| A \| \| F \| |X| M e^{2 C_FM} \| w \|_1 |s-s_0|  
\end{equation}
and thus for $s$ sufficiently close to $s_0$,
\begin{equation}
\| \mathcal{K}_s^{\Lambda_n}(A)  - \mathcal{K}_{s_0}^{\Lambda_n}(A)  \| \leq \epsilon \| A \| .
\end{equation}
This proves the claimed continuity of $\mathcal{K}_s^{\Lambda_n}$ as a function of $s$.

We now prove (i). Again fix $n \geq 0$. 
Let $\{ A_{\alpha} \}_{\alpha \in \mathcal{I}}$ be a bounded net in $\cA_{\Lambda_n}$ which converges in the
strong operator topology to $A$. Denote by $B = \sup_{\alpha \in \mathcal{I}} \| A_{\alpha} \| < \infty$, and note that $\| A \| \leq B$. 
Let $\epsilon >0$. Take $T>0$ as in (\ref{ep_tail_of_W}) and define $\delta >0$ by requiring that 
\begin{equation}
4 \| F \| |\Lambda_n| M e^{2 C_FM} \| w \|_1 \delta \leq \epsilon  
\end{equation} 
For this choice of $\delta>0$, compactness of $I_0$ implies that there is some $N \geq 1$ and numbers $s_1, s_2, \cdots, s_N \in I_0$ for which the balls of
radius $\delta$ centered at $s_j$ ($1 \leq j \leq N$) cover $I_0$. 

For each $1 \leq j \leq N$, it is clear that $\mathcal{K}_{s_j}^{\Lambda_n}(A_\alpha) \to 
\mathcal{K}_{s_j}^{\Lambda_n}(A)$ in the strong operator topology.
In this case, for any $\psi \in \mathcal{H}_{\Lambda_n}$ and $\epsilon$ as above, there is an $\alpha_0 \in \mathcal{I}$ for which 
\begin{equation}
\left\| \left( \mathcal{K}_{s_j}^{\Lambda_n}(A_\alpha) - \mathcal{K}_{s_j}^{\Lambda_n}(A) \right) \psi \right\| \leq \epsilon B \| \psi \| 
\end{equation}
for all $1\leq j \leq N$ whenever $\alpha \geq \alpha_0$. 

In this case, for any $s \in I_0$ there is a value of $j$ for which 
\begin{eqnarray}
\left\| \left( \mathcal{K}_{s}^{\Lambda_n}(A_\alpha) - \mathcal{K}_{s}^{\Lambda_n}(A) \right) \psi \right\| & \leq & \left\| \left( \mathcal{K}_{s}^{\Lambda_n}(A_\alpha) - \mathcal{K}_{s_j}^{\Lambda_n}(A_{\alpha}) \right) \psi \right\| \nonumber \\
& \mbox{ } & \quad + \left\| \left( \mathcal{K}_{s_j}^{\Lambda_n}(A_\alpha) - \mathcal{K}_{s_j}^{\Lambda_n}(A) \right) \psi \right\|  \nonumber \\
& \mbox{ } & \quad + \left\| \left( \mathcal{K}_{s_j}^{\Lambda_n}(A) - \mathcal{K}_{s}^{\Lambda_n}(A) \right) \psi \right\| \nonumber \\
& \leq & 3 \epsilon B \| \psi \|
\end{eqnarray}
and we have proven (i).

Lastly, we need to verify uniform local convergence of $\mathcal{K}_s^{\Lambda_n}$ to $\mathcal{K}_s$. 
Fix $X \in \mathcal{P}_0( \Gamma)$ and $A \in \cA_X$. Let $\epsilon >0$. Choose $T>0$ as in (\ref{ep_tail_of_W}).
For any $n \geq 0$ such that $X \subset \Lambda_n$, Corollary 3.6 (iii) implies 
\begin{equation}
\| \tau_t^s(A) - \tau_t^{\Lambda_n, s}(A) \| \leq \frac{2 \| A \|}{C_F} e^{2I_{t,0}(\Phi_s)} I_{t,0}( \Phi_s) \sum_{x \in X} \sum_{y \in \Gamma \setminus \Lambda_n} F(d(x,y)).
\end{equation}
Since $F$ is summable, for each $x \in X$ there exists $\Lambda_x \in \mathcal{P}_0( \Gamma)$ such that
\begin{equation}
\sum_{y \in \Gamma \setminus \Lambda_x} F(d(x,y)) \leq \epsilon \frac{e^{-2 C_F M}}{4 |X| \| w \|_1 M}.
\end{equation}
For $n \geq 0$ sufficiently large that $X \subset \Lambda_n$ and $\bigcup_{x \in X} \Lambda_x \subset \Lambda_n$, one has that
\begin{eqnarray}
\| \mathcal{K}_s(A) - \mathcal{K}_s^{\Lambda_n}(A) \| & \leq & \int_{|t| \leq T} \|  \tau_t^s(A) - \tau_t^{\Lambda_n, s}(A) \| \, |w(t)| \,dt + 2 \| A \| \int_{|t| >T} |w(t)| \, dt \nonumber \\
& \leq & 2 \| w \|_1 \| A \| e^{2C_FM} M \sum_{x \in X} \sum_{y \in \Gamma \setminus \Lambda_n} F(d(x,y)) + \epsilon \| A\| /2 \nonumber \\
& \leq & \epsilon \| A \|
\end{eqnarray}
from which the claim is proved.
\end{ex}

Our next result shows that given a strongly continuous family of uniformly locally normal maps $\cK_s$ and any $Y\in \cP_0(\Gamma)$, the map $(s,t) \mapsto \Pi_Y (\cK_s(A(t)))$, is jointly strongly continuous whenever $t\mapsto A(t)$ is strongly continuous. In particular, we have continuity on the diagonal $t=s$. The result, of course, also applies to the finite-volume setting where we can take $Y=\Lambda=\Gamma$.

\begin{lem}\label{lem:s-t_continuity} Let $I$ and $J$ be intervals, $X,Y\in\cP_0(\Gamma)$, $\Pi_Y:\cA_\Gamma\to \cA_Y$ be the extension of the map defined in \eq{PiX}, 
	and $\cK_s:\cA_\Gamma^{\rm loc}\to\cA_\Gamma$, $s\in I$ be a strongly continuous family of uniformly locally normal maps. Then, for every strongly continuous $t\mapsto A(t)\in\cA_X$, $t\in J$,
	the function $(s,t)\mapsto \Pi_Y(\cK_s(A(t)))\in\cA_Y$ is jointly strongly continuous in $t$ and $s$.
\end{lem}

\begin{proof}
	By the assumptions, there exists an increasing, exhaustive sequence $\{ \Lambda_n \}_{n \geq 0}$
	of finite subsets of $\Gamma$ and bounded linear transformations $\cK_{s}^{\Lambda_n}\in\cB(\cA_{\Lambda_n})$, strongly continuous in $s$, that approximate $\cK_s$ as in Definition \ref{def:uniformly_locally_normal_maps}. 
	
	Let $X,Y\in\cP_0(\Gamma)$, $\psi\in\cH_Y$, and $t\to A(t)\in \cA_X$ be strongly continuous. Define $f(s,t) = \Pi_Y(\cK_s(A(t)))\psi \in \cH_Y$. We prove that $f(s,t)$ is jointly continuous. Without loss of generality, we may assume that $I$ is compact. Fix $(s_0, \, t_0)$. Since
	\[
	\|f(s,t)-f(s_0,t_0)\| \leq \|f(s,t)-f(s,t_0)\| + \|f(s,t_0)-f(s_0,t_0)\|,
	\]
	the joint continuity of $f(t,s)$ can be obtained by proving the following two properties:
	\begin{enumerate}
		\item[(a)] $f(s,\,t_0)$ is continuous in $s$.
		\item[(b)] $f(s,t)$ is an equicontinuous family of functions of $t$ parameterized by $s\in I$.
	\end{enumerate}
	
	For (a), let $\epsilon >0$. Using Definition~\ref{def:uniformly_locally_normal_maps} (ii), pick $n$ so that
	\[ \|\cK_s(A(t_0)) - \cK^{\Lambda_n}_s(A(t_0))\| \leq \epsilon \; \text{for all} \; s\in I. \]
	Using that $\|\Pi_Y\|=1$, the continuity of $f(\cdot,t_0)$ follows from the strong continuity of $\cK_s^{\Lambda_n}$ as
	\[\|f(s,t_0)-f(s_0,\,t_0)\| 
	\leq 
	2\epsilon \| \psi \|
	+ \|\cK_{s}^{\Lambda_n}(A(t_0)) - \cK^{\Lambda_n}_{s_0}(A(t_0))\|\|\psi\|.\]
	
	For (b), let $\epsilon>0$. We show that there exists $\delta>0$, such that $|t-t_0|< \delta$ implies
	\be
	\| f(s,t) - f(s,t_0)\| \leq \epsilon, \mbox{ for all } s\in I.
	\label{fst_equicontinuous}\ee
	To see this, let $f_n(s,t) = \Pi_Y(\cK^{\Lambda_n}_s(A(t)))\psi$, and $J$ be a compact interval containing a neighborhood of $t_0$. Since $\|A(t)\|$ is uniformly bounded for $t\in J$, using Definition~\ref{def:uniformly_locally_normal_maps}(ii), choose $n$ so that
	\be
	\sup_{t\in J} \| f_n(s,t) - f(s,t)\| \leq \epsilon, \mbox{ for all } s\in I.
	\label{uni1}\ee
	Now consider the family of functions $f_n(s,t)$ parameterized by $s\in I$. By Definition~\ref{def:uniformly_locally_normal_maps}(i), since $\|A(t)\|$ is bounded on $J$, $\cK_s^{\Lambda_n}(A(t))$ is strongly continuous in $t$ . Since $\sup_{s\in I}\|\cK_s^{\Lambda_n}\|<\infty$ by the Uniform Boundedness Principle, it follows that 
	$\|\cK_s^{\Lambda_n}(A(t))\|$ is bounded on $I\times J$, and so 
	\[\Pi_Y(\cK_s^{\Lambda_n}(A(t))) = \Pi_Y^{\Lambda_n\cup Y}(\cK_s^{\Lambda_n}(A(t)))\]
	is strongly continuous in $t$ by Proposition~\ref{prop:strongcontinuity}. The
	argument used in Proposition~\ref{prop:strongcontinuity} shows that the strong continuity of  $\Pi_Y(\cK_s^{\Lambda_n}(A(t)))$ is uniform in $s\in I$. In particular, there is a $\delta>0$ such that $|t-t_0|<\delta$ implies	\be \label{equi1}
	\|f_n(s,t)-f_n(s,t_0)\| < \epsilon \; \text{for all} \; s\in I.
	\ee
	The equicontinuity of $f(s,t)$ follows from \eq{uni1} and \eq{equi1}.
\end{proof}

In all the proofs above, we have used results on the finite volume local approximates to obtain results for the infinite volume local approximates. However, there may be instances where one wants to work in a suitable representation of the infinite-volume algebra. We conclude this section with a result regarding the GNS representation of a locally normal state.

\begin{prop}\label{prop:strongcontinuity2}
	Let $\cH_\Lambda, \Lambda\in\cP_0(\Gamma)$, be the family Hilbert spaces defined in \eq{HLambda}, $\omega$ a locally normal state on $\cA^{\rm loc}_\Gamma$, and $\pi_\omega : \cA^{\rm loc}_\Gamma\to \cB(\cH_\omega)$ its corresponding GNS representation. The map $\pi_{\omega}\vert_{\cA_\Lambda} :\cB(\cH_\Lambda) \to \cB(\cH_\omega)$ is continuous on arbitrary bounded subsets of its domain with respect to the strong operator topology on both its domain and codomain.
\end{prop}

\begin{proof}
	Let $\{ A_\alpha \mid \alpha \in I\}$ be a net in the unit ball of
	$\cB(\cH_\Lambda)$ that converges to $A$ in the strong operator topology. To prove the claim, it is sufficient to verify that $\pi_\omega(A-A_\alpha)\psi \to 0$, for all $\psi$ in the dense subspace of $\cH_\omega$  given by vectors of the form $\pi_\omega(B)\Omega_\omega$, where $B\in \cA^{\rm loc}_\Gamma$ and $\Omega_\omega$ is the cyclic vector representing $\omega$. To this end, note that
	$$
	\Vert \pi_\omega((A-A_\alpha)B)\Omega_\omega\Vert^2 = 
	\omega(B^* (A-A_\alpha)^*(A-A_\alpha)B).
	$$
	Since $A$ and the $A_\alpha$ belong to $\cA_\Lambda$ and $B\in  \cA^{\rm loc}_\Gamma$, there exists $\Lambda^\prime
	\in \cP_0(\Gamma)$ such that $B^* (A-A_\alpha)^*(A-A_\alpha)B\in \cA_{\Lambda^\prime}$. Since $\omega$ is locally normal, its restriction to $\cA_{\Lambda^\prime}$ is given by a density matrix 
	$\rho_{\Lambda^\prime}$ on $\cH_{\Lambda^\prime}$. By writing $\rho_{\Lambda^\prime} = \sum_{n\geq 1} \rho_n \ketbra{\xi_n}$ in terms of an orthonormal set of eigenvectors $\xi_n\in \cH_{\Lambda^\prime}$, it follows that
	$$
	\Vert \pi_\omega(A-A_\alpha) \pi_\omega(B)\Omega_\omega\Vert^2 
	=\Tr \rho_{\Lambda^\prime} B^*(A-A_\alpha)^*(A-A_\alpha)B
	=\sum_{n\geq 1} \rho_n \Vert(A-A_\alpha)B\xi_n \Vert^2.
	$$
	The result follows from using the analogous arguments from the proof of Proposition \ref{prop:strongcontinuity}.
\end{proof}

\section{Quasi-local maps} \label{sec:quasilocal-maps}

In Section~\ref{sec:lrb}, we proved a Lieb-Robinson bound for the dynamics associated with a sufficiently local interaction. In addition to estimating the speed of propagation of a dynamically evolved observable, these bounds imply that the dynamics for a fixed time $t$ is quasi-local. As a result, they can be well approximated by local observables as shown in Corollary~\ref{cor:PiX}. In recent years, other quasi-local maps have played a key role in proving both locality estimates of the spectral flow \cite{bachmann:2012,hastings:2004,hastings:2005,nachtergaele:2007} and spectral gap stability of frustration-free quantum lattice models \cite{bravyi:2010,bravyi:2011,michalakis:2013,young:2016,nachtergaele:2016b,hastings:2017,moon:2018,QLBII}. While we will consider both of these topics, the former in Section~\ref{sec:spectral-flow} and the latter in \cite{QLBII}, the focus of this section is the general study of quasi-local maps, see \eqref{QLmap} below, and the investigation of the key properties that will be useful in above mentioned applications. There exists a broad range of other applications that we will not discuss here \cite{gong:2017,bachmann:2017b,bachmann:2017c,bachmann:2017d,bachmann:2018,bachmann:2018a}.

We begin by showing how to apply the techniques from Section~\ref{sec:str-loc} to obtain estimable local approximations of quasi-local maps. In Section~\ref{sec:some_ex}, we provide a number of examples that will arise in our applications, including the difference of two dynamics. We discuss compositions of two quasi-local maps in Section~\ref{sec:comp_qlm}, and prove sufficient conditions for which the composition is again quasi-local. In Section~\ref{sec:quasi_loc_ints} consider the composition of a quasi-local map with an interaction. We show that under suitable conditions such a composition can be rewritten as a local interaction. Moreover, we quantify the decay of the resulting interaction in terms of the decays of the original interaction and the quasi-local map. If the transformed interaction has sufficient decay, then the theory developed in Section~\ref{sec:lrb} applies and an infinite volume dynamics exists. We conclude in Section~\ref{sec:diff_dyn_ql_est} by returning to the example of the difference of two dynamics and proving a continuity result.

\subsection{General quasi-local maps} \label{sec:qlm_gen}

Let $(\Gamma, d)$, and $\mathcal{A}_{\Gamma}$ be a quantum lattice system defined as in Section~\ref{sec:spatialstructure}. 
A linear map $\cK: \cA_{\Gamma}^{\rm loc} \to \cA_{\Gamma}$ is said to satisfy a \emph{quasi-locality bound of order} $q \geq 0$ 
if there is $C< \infty$ and a non-increasing function $G:[0,\infty) \to [0,\infty)$ with $\lim_{r\to\infty} G(r)=0$,
such that for all $X,Y \in \mathcal{P}_0( \Gamma)$, 
and $A  \in \cA_X$, $B \in \cA_Y$
\be \label{QLmap}
\Vert [\cK(A),B]\Vert \leq  C |X|^q \Vert A\Vert \Vert B\Vert  G(d(X,Y)).
\ee 
Any linear mapping $\cK$ satisfying \eq{QLmap} will be referred to as {\it quasi-local}.
When relevant, we will denote by $C_\cK(q,G)$ the smallest constant for which \eq{QLmap} holds. 
Since the function $G$ in (\ref{QLmap}) above governs the decay of the quasi-local map, 
we may refer to $G$ as a decay function associated to $\mathcal{K}$. In this work, we will always assume that quasi-local maps are linear. However, there may be other contexts in which it might be appropriate to generalize this definition.

The dependence of  the bound in (\ref{QLmap}) on the support of the observable $A$ through 
the factor $|X|^q$ is a choice we made based on the applications we have in mind. 
However, under appropriate assumptions, most of the estimates proved in this section 
also hold for quasi-local maps with more general functions of $|X|$.   

In most of our applications, the metric space $(\Gamma,d)$ is equipped with an $F$-function $F$. In this case, one can often estimate a quasi-local map $\cK$ as follows: there is 
$C< \infty$ such that for all $X,Y \in \mathcal{P}_0( \Gamma)$, 
any $A  \in \cA_X$, and $B \in \cA_Y$ we have
\be \label{QLmap_Fbd}
\Vert [\cK(A),B]\Vert \leq  C \Vert A\Vert \Vert B\Vert G_F(X,Y) \quad \mbox{where} \quad G_F(X,Y) = \sum_{x \in X} \sum_{y \in Y} F(d(x,y)).
\ee
As we will see below, in certain estimates the bound (\ref{QLmap_Fbd}) has advantages over (\ref{QLmap}). 
A simple over-counting argument shows that
$$
G_F(X,Y) \leq |X|G(d(X,Y)), \; \text{  where  } G(r):=\sup_{x \in \Gamma}\sum_{\substack{y \in \Gamma : \\ d(x,y)\geq r}}F(d(x,y)).
$$
It follows from the uniform integrability of $F$, see \eqref{F:int}, that $G(r) \to 0$ as $r\to \infty$. When the corresponding $F$-function is weighted, i.e. $F= F_g$ as defined \eqref{weightedF}, one has that
\begin{equation}
G_{F_g}(X,Y) = \sum_{x \in X} \sum_{y \in Y} F_g(d(x,y)) = \sum_{x \in X} \sum_{y \in Y} e^{-g(d(x,y))} F(d(x,y)) \leq \| F \| |X| e^{-g(d(X,Y))}
\end{equation}
and so, in this case, an estimate of the form (\ref{QLmap_Fbd}) reduces to that of (\ref{QLmap}). For more detailed information on $F$-functions, including weighted $F$-functions, see Sections \ref{app:sec_def_F}-\ref{subsec:regsets}.

We now demonstrate an important estimate concerning quasi-local maps. For this result, we use the concepts introduced in Section~\ref{sec:laqlmaps} and in particular the localizing maps $\Pi_X$ and $\Delta_{X(n)}$, defined with respect to a locally normal product state $\rho$, as in \eqref{PiX} and \eqref{DeltaXn}, respectively.

\begin{lem} \label{lem:qlm_approx} Let $(\Gamma, d)$ and $\mathcal{A}_{\Gamma}$ be a quantum lattice system, and $\rho$ a locally normal product state on $\mathcal{A}_{\Gamma}$. 
	Let $\mathcal{K} : \mathcal{A}_{\Gamma}^{\rm loc} \to \mathcal{A}_{\Gamma}$ be a quasi-local map.  For any $X \in \mathcal{P}_0( \Gamma)$ and $n \geq 0$,
	\begin{equation} \label{qlm_sl_bd}
	\| \mathcal{K}(A) - \Pi_{X(n)}(\mathcal{K}(A)) \| \leq 2 C |X|^q \| A \| G(n) \quad \mbox{for all } A \in \mathcal{A}_X \, .
	\end{equation}
	In particular, if the decay function associated to $\mathcal{K}$ is summable, i.e. $\sum_{n \geq 0} G(n) < \infty$, then
	for any $X \in \mathcal{P}_0( \Gamma)$ and each $A \in \mathcal{A}_X$, 
	\begin{equation} \label{qlm_series_rep}
	\mathcal{K}(A) = \sum_{n=0}^{\infty} \Delta_{X(n)}( \mathcal{K}(A)) \, .
	\end{equation}
	The series on the right-hand-side above is absolutely convergent in norm with a bound that is
	uniform in the choice of locally normal product state $\rho$.  
\end{lem}

Note that the result above, of course, also applies to finite systems. In particular, for any finite $\Lambda$, the result holds for any quasi local map $\cK: \cA_\Lambda \to \cA_\Lambda$.

\begin{proof}
	To see that (\ref{qlm_sl_bd}) is true, fix $X \in \mathcal{P}_0(\Gamma)$, $A \in \mathcal{A}_X$, and $n \geq 0$. 
	Observe that for any $B \in \mathcal{A}_{\Gamma \setminus X(n)}^{\rm loc}$, the estimate
	\begin{equation}
	\Vert [\cK(A),B]\Vert \leq  C |X|^q \Vert A\Vert \Vert B\Vert  G(d(X,{\rm supp}(B))).
	\end{equation}
	follows from (\ref{QLmap}). In this case, an application of Corollary~\ref{cor:PiX} with $\epsilon = C |X|^q \| A \| G(n)$ implies (\ref{qlm_sl_bd}) as
	claimed. 
	
	Recalling (\ref{DeltaXn}), for any integer $n \geq 1$, one can write 
	\begin{equation}
	\Delta_{X(n)}( \mathcal{K}(A)) = \left( \Pi_{X(n)}( \mathcal{K}(A)) - \mathcal{K}(A) \right) - \left(  \Pi_{X(n-1)}( \mathcal{K}(A)) - \mathcal{K}(A) \right) 
	\end{equation}
	and therefore, an immediate consequence of (\ref{qlm_sl_bd}) is the estimate
	\begin{equation} \label{Delta_bd}
	\| \Delta_{X(n)}( \mathcal{K}(A)) \| \leq 4 C |X|^q \| A \| G(n-1),
	\end{equation}
	which is valid for any $X \in \mathcal{P}_0( \Gamma)$, $A \in \mathcal{A}_X$, and $n \geq 1$.
	
	Since $G(n)\to 0$ as $n\to\infty$, (\ref{qlm_sl_bd}) implies
	\begin{equation}
	\lim_{n \to \infty} \| \mathcal{K}(A) - \Pi_{X(n)}( \mathcal{K}(A)) \| = 0 .
	\end{equation}
	It is clear from (\ref{DeltaXn}) that the local decompositions have telescopic sums, i.e. for any $N \geq 1$,
	\begin{equation} \label{finite_tele}
	\sum_{n=0}^N \Delta_{X(n)}( \mathcal{K}(A)) = \Pi_{X(N)}(\mathcal{K}(A)), 
	\end{equation} 
	and thus the series on the right-hand-side of (\ref{qlm_series_rep}) is norm convergent. Note also that
	\begin{equation} \label{abs_sum_D_qlm}
	\sum_{n=0}^{\infty} \| \Delta_{X(n)}( \mathcal{K}(A)) \| \leq \| \mathcal{K}(A) \| + 4 C |X|^q \| A \| \sum_{n=0}^{\infty}G(n),
	\end{equation}
	and so the series is also absolutely convergent, with norm bound independent of $\rho$, as claimed. 
\end{proof}

If in Lemma \ref{lem:qlm_approx} we had assumed $\cK$ satisfies \eq{QLmap_Fbd} instead of \eq{QLmap}, \eq{qlm_sl_bd}
would become
\begin{equation}
\| \mathcal{K}(A) - \Pi_{X(n)}(\mathcal{K}(A)) \| \leq  2 C \Vert A\Vert G_F(X, \Gamma \setminus X(n))
\end{equation}
where the right-hand-side above is finite and non-increasing (in $n$) by the uniform summability of the $F$-function $F$. 

%
%

\subsection{Examples of quasi-local maps}\label{sec:some_ex}
In this section, we discuss a few of the most common examples of quasi-local maps (defined as in \eqref{QLmap}). In applications, quasi-local maps are often constructed as the thermodynamic limit of appropriate finite-volume maps. We first describe this class of quasi-local maps as a general example. Each of the more concrete examples we present later in this section will be of this general form.

%
%

\begin{ex}[A General Example] \label{ex:General}
Let $q \geq 0$, $C< \infty$, and $G$ be a non-increasing function $G: [0, \infty) \to [0, \infty)$ with $\lim_{r \to \infty}G(r) = 0$. 
Let $\{ \Lambda_n \}_{n=1}^{\infty}$ be a sequence of increasing and exhaustive finite subsets of $\Gamma$.
In particular, this means that $\Lambda_n \subset \Lambda_{n+1}$ for all $n \geq 1$, and given
any $x \in \Gamma$, there exists $N \geq 1$ for which $x \in \Lambda_N$.  Suppose that for each $n \geq 1$, 
there is a linear map $\mathcal{K}^n : \mathcal{A}_{\Lambda_n} \to \mathcal{A}_{\Lambda_n}$ for which:
\begin{enumerate}
	\item[(i)] Given any sets $X,Y \subset \Lambda_n$ the bound
	\begin{equation} \label{fv_qlm_est}
	\| [ \mathcal{K}^n(A), B ] \| \leq C |X|^q \| A \| \| B \| G(d(X,Y))
	\end{equation}
	holds for all observables $A \in \mathcal{A}_X$ and $B \in \mathcal{A}_Y$.
	\item[(ii)] For each finite $X \subset \Gamma$ and any $\epsilon >0$, there is an $N \geq 1$ for which
	\begin{equation} \label{locally_Cauchy}
	\| \mathcal{K}^n(A) - \mathcal{K}^m(A) \| \leq \epsilon \| A \| \quad \mbox{for any } n , m \geq N \mbox{ and all } A \in \mathcal{A}_X \, .
	\end{equation}
\end{enumerate}
In this case, the local Cauchy assumption in part (ii) above, implies that a linear map 
$\mathcal{K} : \mathcal{A}_{\Gamma}^{\rm loc} \to \mathcal{A}_{\Gamma}$ is well-defined by setting
\begin{equation} \label{qlm_tdl}
\mathcal{K}(A) = \lim_{n \to \infty} \mathcal{K}^n(A) \quad \mbox{ for each } A \in \mathcal{A}_{\Gamma}^{\rm loc}, \,
\end{equation}
with the limit above being in norm. Moreover, for any finite sets $X, Y \subset \Gamma$, there is $N \geq 1$ 
large enough so that $X \cup Y \subset \Lambda_N$ since the sequence of sets is exhaustive.  In this case, for any $A\in \cA_X$, $B\in \cA_Y$ and $n \geq N$,  
\begin{eqnarray}
\| [ \mathcal{K}(A), B ] \| & \leq & \| [ \mathcal{K}^n(A), B ] \| + \| [ \left( \mathcal{K}(A) - \mathcal{K}^n(A) \right) , B ] \| \nonumber \\
& \leq & C |X|^q \| A \| \| B \| G(d(X,Y)) + 2 \| B \| \| \mathcal{K}(A) - \mathcal{K}^n(A) \| .
\end{eqnarray}
Here we have used the uniform quasi-locality estimate in (i), see (\ref{fv_qlm_est}).
As the final term above vanishes in the limit as $n \to \infty$, it is clear that the map $\mathcal{K}$ 
defined in (\ref{qlm_tdl}) above is quasi-local, see (\ref{QLmap}), and satisfies the same uniform quasi-local estimates as the collection $\{\cK^n\}_{n\geq 1}.$  
\end{ex}

For later applications, we now describe a variant of the estimate in Lemma~\ref{lem:qlm_approx} which is 
particularly relevant to quasi-local maps of the type from Example~\ref{ex:General}. 

\begin{cor}\label{cor:DeltaEst}
Let $\rho$ be a locally normal product state on $\mathcal{A}_{\Gamma}$, and suppose that $\cK$ and $\{\cK^n\}_{n=1}^\infty$ are maps that satisfy the uniform quasi-local and local Cauchy conditions from \eqref{fv_qlm_est} and \eqref{locally_Cauchy}. For any $X \in \cP_0(\Gamma)$, $m \geq 1$, and $n \geq 1$ large enough so that $X(m) \subset \Lambda_n$, the bound
\begin{equation} \label{iv_Ddifbd}
\| \Delta_{X(m)}( \mathcal{K}(A)) - \Delta_{X(m)}^{\Lambda_n}( \cK^n(A)) \| \leq \min \left\{ 2 \| \cK(A) - \cK^n(A) \|, 8 C |X|^q \| A \| G(m-1) \right\}
\end{equation}
holds for any $A \in \mathcal{A}_X$ where $\Delta_{X(m)}^{\Lambda_n} : \mathcal{A}_{\Lambda_n} \to \mathcal{A}_{X(m)}$ and $\Delta_{X(m)} : \mathcal{A}_{\Gamma} \to \mathcal{A}_{X(m)}$ are the local decomposition maps defined in Section~\ref{sec:laqlmaps} (see \eqref{loc_dec} and \eqref{DeltaXn}).
\end{cor}
The bound above is particularly useful as it expresses the decay of the quantity on the left-hand-side
above in both large $n$ and $m$.
\begin{proof} The proof uses two separate estimates. First, using consistency of the local
decompositions, see Proposition~\ref{prop:compatibility}(ii), it is clear that 
\begin{equation}
\Delta_{X(m)}( \mathcal{K}(A)) - \Delta_{X(m)}^{\Lambda_n}( \cK^n(A)) = \Delta_{X(m)}( \mathcal{K}(A) - \cK^n(A)) 
\end{equation}  
and so the first part of the estimate holds since $\Pi_{X(m)}$ is a norm one map. Note that this establishes that the LHS of \eqref{iv_Ddifbd} decays in $n$.
For the second part of the argument, we use that each term on the LHS of (\ref{iv_Ddifbd}) can be estimated using
Lemma~\ref{lem:qlm_approx}. In fact, 
\begin{equation*}
\| \Delta_{X(m)}( \mathcal{K}(A)) - \Delta_{X(m)}^{\Lambda_n}( \cK^n(A)) \| \leq \| \Delta_{X(m)}( \mathcal{K}(A)) \|  + \| \Delta_{X(m)}^{\Lambda_n}( \cK^n(A)) \| \leq
8 C |X|^q \| A \| G(m-1)
\end{equation*}
as all maps considered satisfy the same quasi-local bound.
\end{proof}

%
%

\begin{ex} \label{ex_der_ql}
Here we return to Example~\ref{ex_der_sec_4} and show that it has the form of the general example
discussed in Example~\ref{ex:General}. Consider a quantum lattice system comprised of $(\Gamma, d)$ and $\cA_{\Gamma}$. Let $F$ be an $F$-function on $(\Gamma, d)$ and $\Phi \in \mathcal{B}_F$ be an interaction. As in \eqref{fvol_der_4}, let $\{ \Lambda_n \}_{n \geq 0}$ be an increasing, exhaustive sequence of finite subsets of $\Gamma$ and define
$\mathcal{K}^n : \mathcal{A}_{\Lambda_n} \to \mathcal{A}_{\Lambda_n}$ by setting
\begin{equation}
\mathcal{K}^n(A) = \sum_{Z \subset \Lambda_n} [ \Phi(Z), A] \, .
\end{equation}

To see that (\ref{fv_qlm_est}) holds, let $X,Y \in \mathcal{P}_0( \Gamma)$. For any $A \in \cA_X$ and
$n \geq 0$ large enough so that $X \subset \Lambda_n$, one has that (\ref{fvol_der_4_simp}) holds
and therefore, 
\begin{equation} \label{lb_derivation}
\| \mathcal{K}^n (A) \| \leq \sum_{x \in X} \sum_{z \in \Lambda_n} \sum_{\stackrel{Z \subset \Lambda_n:}{x,z \in Z}} \| [ \Phi(Z), A] \| 
\leq 2 \| A \| \| \Phi \|_F G_F(X, \Lambda_n)
\end{equation}
where we use the notation $G_F$ of (\ref{QLmap_Fbd}).
As any $F$-function is uniformly integrable, the bound above is uniform in the finite volume since
$G_F(X, \Lambda_n) \leq |X| \| F \| $. For $n \geq 0$ large enough so that $X \cup Y \subset \Lambda_n$, for each $A \in \mathcal{A}_X$, and any $B \in \mathcal{A}_Y$, a simple commutator bound then yields:
\begin{equation}
\| [ \mathcal{K}^n(A), B] \| \leq 2 \| \mathcal{K}^n(A) \| \| B \| \leq 4 \| \Phi \|_F \| F \| |X| \| A \|  \|B \| \, .
\end{equation}
When $X \cap Y = \emptyset$, a better estimate is achieved by observing that
\begin{equation}
\| [ \mathcal{K}^n(A), B] \| \leq \sum_{\stackrel{Z \subset \Lambda_n:}{Z \cap X \neq \emptyset, Z \cap Y \neq \emptyset}} \| [[ \Phi(Z), A], B ] \| 
\leq 4 \| \Phi \|_F \| A \| \| B \| G_F(X,Y). 
\end{equation}
This gives a quasi-locality estimate (uniform in the finite volume) of the type in (\ref{fv_qlm_est}).

To see (\ref{locally_Cauchy}), fix a finite set $X \subset \Gamma$, $A \in \cA_X$, and take $m \leq n$ with
$m$ large enough so that $X \subset \Lambda_m$. One checks that 
\begin{equation}
\mathcal{K}^n(A) - \mathcal{K}^m(A) = \sum_{\stackrel{Z \subset \Lambda_n:}{Z \cap X \neq \emptyset, Z \cap ( \Lambda_n \setminus \Lambda_m) \neq \emptyset}} [ \Phi(Z), A ] 
\end{equation}
and therefore, 
\begin{equation}
\| \mathcal{K}^n (A) - \mathcal{K}^m(A) \| 
\leq 2 \| A \| \| \Phi \|_F G_F(X, \Lambda_n \setminus \Lambda_m).
\end{equation}
Since $F$-functions are integrable, the locally-Cauchy estimate (\ref{locally_Cauchy}) follows.  

Under appropriate assumptions, one can generalize this example to prove quasi-locality estimates for mappings of the form
\begin{equation}
\mathcal{K}^{\Lambda}(A) = \sum_{j=1}^N C_j A D_j \quad \mbox{with } C_j, D_j \in \mathcal{A}_{\Lambda}.
\end{equation}
For example, the Lindblad generator of a quantum dynamical semigroup is of this form. See \cite{poulin:2010,nachtergaele:2011}.
\end{ex}
%
%

\begin{ex}[The Dynamics Associated to $\Phi \in \mathcal{B}_F(I)$] \label{ex:qlm_dyn_ex}
Let $(\Gamma, d)$ and $\mathcal{A}_{\Gamma}$ be a quantum lattice system, and
$I \subset \mathbb{R}$ be an interval. Given an $F$-function, $F$, recall that a strongly continuous interaction $\Phi: \cP_0(\Gamma) \times I \to \cA_{loc}$
belongs to $\mathcal{B}_F(I)$ if the function $\| \Phi \|_F : I \to [0, \infty)$ defined by 
\begin{equation} \label{Phi_td_F_norm_ex}
\| \Phi \|_F (t)  = \sup_{x,y \in \Gamma} \frac{1}{F(d(x,y))} \sum_{\stackrel{Z \in \mathcal{P}_0( \Gamma):}{x,y \in Z}} \| \Phi(Z, t) \|
\end{equation}
is locally bounded on $I$, see \eqref{tdintnorm}. Fix $\Phi \in \mathcal{B}_F(I)$ and let $\{ H_z \}_{z \in \Gamma}$ be a collection of densely defined, self-adjoint on-site Hamiltonians. For any $\Lambda \in \mathcal{P}_0( \Gamma)$, considered the finite-volume Hamiltonians
\begin{equation} \label{td_hams_ex}
H_{\Lambda}^{(\Phi)}(t) = \sum_{z \in \Lambda} H_z + \sum_{Z \subset \Lambda} \Phi(Z,t) \quad \mbox{for any } t \in I
\end{equation} 
and associated to them the dynamics $\{ \tau_{t,s}^{\Lambda} \}_{s,t \in I}$, 
defined in \eqref{ubdyn}. Theorem~\ref{thm:lrb2} demonstrates that these dynamics are quasi-local maps. In fact, this result shows
that for any $X,Y \subset \Lambda$ with $X \cap Y = \emptyset$, 
\begin{equation} \label{unbd_lrb_ex}
\| [ \tau_{t,s}^{\Lambda}(A), B ] \| \leq \frac{2 \| A \| \| B \|}{C_F} \left( e^{2 I_{t,s}( \Phi)} - 1 \right) G_F(X,Y) 
\end{equation}
for all $A \in \mathcal{A}_X$, $B \in \mathcal{A}_Y$, and $s,t \in I$. Here the number $C_F$ is the convolution constant associated 
to the $F$-function $F$ on $\Gamma$, see (\ref{F:conv}), and
\begin{equation}
I_{t,s}( \Phi) = C_F \int_{{\rm min}(t,s)}^{{\rm max}(t,s)} \| \Phi \|_F(r) \, dr \, .
\end{equation} 
Moreover, the bound proven in Theorem~\ref{thm:continuity} (ii) shows that  for any finite sets $X \subset \Lambda_0 \subset \Lambda$,
\begin{equation} \label{dyn_Cauchy_est_ex}
\| \tau_{t,s}^{\Lambda}(A) - \tau_{t,s}^{\Lambda_0}(A) \| \leq 2 C_F^{-1} I_{t,s}( \Phi) e^{2 I_{t,s}( \Phi)} \| A \| G_F(X, \Lambda \setminus \Lambda_0)  
\end{equation} 
holds for all $A \in \mathcal{A}_X$ and $s,t \in I$. Again, this suffices to establish that these dynamics are
locally Cauchy in the sense of (\ref{locally_Cauchy}), and so again this example is of the general form in Example~\ref{ex:General}.

\end{ex}

%
%

\begin{ex}[The Difference of Two Dynamics] \label{ex:diff_dyn_ex}
A quasi-local map that comes up in our applications related to stability in \cite{QLBII} is the difference of two dynamics.
More precisely, consider the setting from Example~\ref{ex:qlm_dyn_ex}. 
Fix a collection of densely defined, self-adjoint on-site Hamiltonians $\{ H_z \}_{z \in \Gamma}$ and 
two interactions $\Phi, \Psi \in \mathcal{B}_F(I)$. For any $\Lambda \in \mathcal{P}_0( \Gamma)$,
consider the Hamiltonians 
\begin{equation} \label{2_td_hams_ex}
H_{\Lambda}^{(\Phi)}(t) = \sum_{z \in \Lambda} H_z + \sum_{Z \subset \Lambda} \Phi(Z,t) \quad \mbox{and } \quad H_{\Lambda}^{(\Psi)}(t) = \sum_{z \in \Lambda} H_z + \sum_{Z \subset \Lambda} \Psi(Z,t)
\end{equation} 
and their corresponding dynamics, which we denote by $\tau_{t,s}^{\Lambda}$ and $\alpha_{t,s}^{\Lambda}$, respectively.
For any $s,t \in I$, a linear map $\mathcal{K}^{\Lambda}_{t,s} : \mathcal{A}_{\Lambda} \to \mathcal{A}_{\Lambda}$ is
defined by
\begin{equation} \label{def_diff_dyn_ex}
\mathcal{K}_{t,s}^{\Lambda}(A) = \tau_{t,s}^{\Lambda}(A) - \alpha_{t,s}^{\Lambda}(A).
\end{equation}
Since both of the dynamics used in the definition of this map are automorphisms, it is clear that $\| \mathcal{K}_{t,s}^{\Lambda} \| \leq 2$.
A different estimate is provided by the local norm bound in Theorem~\ref{thm:continuity}(i), namely 
\begin{equation} \label{lb_diff_dyn_ex}
\| \mathcal{K}_{t,s}^{\Lambda}(A) \| \leq 2 C_F^{-1} I_{t,s}( \Phi- \Psi) e^{2 {\rm min}(I_{t,s}( \Phi), I_{t,s}(\Psi))} \| A \| G_F(X, \Lambda) \,,
\end{equation}
where $I_{t,s}(\Phi- \Psi)$ is defined as in \eqref{ItsPhi}. The above bound better reflects the fact that this difference is small if either $\Phi$ is close to $\Psi$ in $\mathcal{B}_F(I)$, or $|s-t|$ is small.  The bound in Theorem~\ref{thm:continuity}(ii),
see also (\ref{dyn_Cauchy_est_ex}) above, allows one to establish that these maps are locally Cauchy in the sense of (\ref{locally_Cauchy}).
In Section~\ref{sec:diff_dyn_ql_est} below, we prove that these maps $\mathcal{K}^{\Lambda}_{t,s}$ are uniformly quasi-local in the sense of
(\ref{fv_qlm_est}) with constant pre-factors that once again decay if either $\Phi$ is close to $\Psi$ in $\mathcal{B}_F(I)$ or $|s-t|$ is small;
as is the case in (\ref{lb_diff_dyn_ex}). 
\end{ex}

%
%

\begin{ex}[Weighted Integrals of Quasi-local Maps] \label{ex:wi_qlm}
We briefly mention another interesting class of examples, which includes the spectral flow introduced in Section~\ref{sec:w_int_op}. 
Let $(\Gamma, d)$ and $\mathcal{A}_{\Gamma}$ be a quantum lattice system, and $\mu$ a measure 
on $\mathbb{R}$. Suppose that for each $t \in \mathbb{R}$ there is a quasi-local map
$\mathcal{K}_t : \mathcal{A}_{\Gamma}^{\rm loc} \to \mathcal{A}_{\Gamma}$ for which the map
$\mathcal{K} : \mathcal{A}_{\Gamma}^{\rm loc} \to \mathcal{A}_{\Gamma}$ given by
\begin{equation} \label{weighted_integral_ex}
\mathcal{K}(A) = \int_{\mathbb{R}} \mathcal{K}_t(A) \, d \mu(t)
\end{equation}
is well-defined. If the family $\{ \mathcal{K}_t \}_{t \in \mathbb{R}}$ is sufficiently quasi-local and
$\int_{|t| \geq x} d \mu(t)$ decays sufficiently fast as $x\to\infty$, then the mapping defined
in (\ref{weighted_integral_ex}) is quasi-local with explicit decay estimates. 

For example, let $w\in L^1(\bR)$, and $\Phi\in \cB_{F_a}(\bR)$ where we recall that $F_a$ is a weighted $F$-function of the form $F_a(r) = e^{-ar}F(r)$ with $a>0$. Let $\tau_{t}^\Lambda:=\tau_{t,0}^{\Lambda}$ be the finite-volume dynamics given in \eqref{td_heis_dyn} that is associated with the local Hamiltonians of the form
\[
H_\Lambda(t) = \sum_{X \subset \Lambda}\Phi(X,t).
\]
Then, using Lieb-Robinson bounds, one can see that the map $\cK:\cA_\Lambda \to \cA_\Lambda$ defined by
\[
\cK(A) = \int_\bR \tau_{t}^\Lambda(A)w(t) dt
\]
satisfies the following bound: for all $A\in \cA_X$, $B\in \cA_Y$ with $X\cup Y \subseteq \Lambda$ and $X\cap Y = \emptyset,$ $\|[\cK(A), \, B]\| \leq 2\|A\|\|B\| |X| G(d(X,Y)),$ where $G$ is the decreasing function
\[
G(r) = \frac{\|w\|_{L^1}\|F\|}{C_{F_a}}e^{-ar/2} + \int_{|t| > r/2v_a}|w(t)|dt,
\]
and $v_a$ is the Lieb-Robinson velocity, see the discussion following Theorem~\ref{thm:lrb}. 
To see this, note that for all $T\in \bR$,
\[
\|[\cK(A), \, B]\| \leq \int_{|t| \leq T} \|[\tau_{t}^\Lambda(A), B]\||w(t)| dt + \int_{|t| > T} \|[\tau_{t}^\Lambda(A), B]\||w(t)| dt.
\]
With the choice of $T= d(X,Y)/2v_a$, the bound is attained by applying \eqref{lr_traditional} to the integral over $|t| \leq T$, and using the trivial bound $\|[\tau_{t}^\Lambda(A), B]\| \leq 2\|A\|\|B\|$ for the integral over $|t| > T$.
\end{ex}

%
%


\subsection{Compositions of quasi-local maps}  \label{sec:comp_qlm}

In applications, we find it useful to recognize certain mappings as the composition of quasi-local maps. 
When $\Gamma$ is finite, these compositions are well-defined and estimates, as indicated below, 
follow readily. For sets $\Gamma$ of infinite cardinality, more care must be taken 
when defining such compositions. This section discusses two classes of examples  
where these compositions are well-defined and we describe the estimates that follow. 

It will be convenient to make an additional assumption on the metric space $(\Gamma, d)$.
We say that $(\Gamma, d)$ is {\it $\nu$-regular} if the cardinality of balls in $\Gamma$ grows at most polynomially, i.e. there
exist non-negative $\kappa$ and $\nu$ for which
\begin{equation} \label{nu_ball_bd}
\sup_{x \in \Gamma} |b_x(n) | \leq \kappa n^{\nu} \quad \mbox{for any } n \geq 1 \, .
\end{equation}
More comments about $\nu$-regular metric spaces $( \Gamma, d)$ can be found in 
Appendix~\ref{app:sec_def_F}. Under this assumption, given any $X \in \mathcal{P}_0(\Gamma)$ and $n > 0$, the cardinality 
of $X(n)$, the inflation of $X$ defined in \eq{defxn}, satisfies the following rough estimate:
\begin{equation} \label{infl_bd}
| X(n) | \leq \sum_{x \in X} |b_x(n) | \leq \kappa n^{\nu} |X|. 
\end{equation}

Let $(\Gamma, d)$ and $\mathcal{A}_{\Gamma}$ be a quantum lattice system on a $\nu$-regular metric space. We will say that a linear map $\mathcal{K} : \mathcal{A}_{\Gamma}^{\rm loc} \to \mathcal{A}_{\Gamma}$ 
is {\it locally bounded} if there are non-negative numbers $p$ and $B$ for which
\begin{equation} \label{lb_est_comp}
\| \mathcal{K}(A) \| \leq B |X|^{p} \| A \|  \quad \mbox{for all } A \in \mathcal{A}_X \mbox{ with } X \in \mathcal{P}_0(\Gamma) \, .
\end{equation}
More general growth in $X$, i.e. the support of $A$, could be allowed, but the above {\it moment 
condition} covers all of the applications we have in mind. 
As discussed in Section~\ref{sec:qlm_gen}, we say that a linear map 
$\mathcal{K} : \mathcal{A}_{\Gamma}^{\rm loc} \to \mathcal{A}_{\Gamma}$ 
is {\it quasi-local} if there are non-negative numbers $q$ and $C$ as well as a non-increasing function $G$,
$G :[0, \infty) \to [0, \infty)$, with $\lim_{r \to \infty} G(r) = 0$  for which given any $X,Y \in \mathcal{P}_0(\Gamma)$,
\begin{equation} \label{ql_est_comp}
\| [ \mathcal{K}(A), B ] \| \leq C |X|^{q} \| A \| \| B \| G(d(X,Y)) \quad \mbox{for all } A \in \mathcal{A}_X \mbox{ and } B \in \mathcal{A}_Y \, .
\end{equation}
We will refer to $C, q$, and $G$ as the {\em parameters} of the quasi-local map.

We first consider compositions of linear maps for the following situation. Suppose that $\mathcal{K}_1 : \mathcal{A}_{\Gamma}^{\rm loc} \to \mathcal{A}_{\Gamma}$
is locally bounded and quasi-local in that it satisfies both (\ref{lb_est_comp}) and (\ref{ql_est_comp}). Furthermore, assume  
that $\mathcal{K}_2 : \mathcal{A}_{\Gamma} \to \mathcal{A}_{\Gamma}$ is linear, bounded, and quasi-local. So, in particular, 
there are non-negative numbers $q_2$ and $C_2$ and a decay function $G_2$ for which
the analogue of (\ref{ql_est_comp}) holds for $\mathcal{K}_2$. 
In many applications, the mapping $\mathcal{K}_2$ arises as the unique linear
extension of a bounded, quasi-local map 
$\tilde{\mathcal{K}}_2 : \mathcal{A}_{\Gamma}^{\rm loc} \to \mathcal{A}_{\Gamma}$. In this situation, we can define the composition $\mathcal{K} : \mathcal{A}^{\rm loc}_{\Gamma} \to \mathcal{A}_{\Gamma}$ in the usual way, i.e.
\begin{equation} \label{def_comp_1}
\mathcal{K}(A) = \mathcal{K}_2( \mathcal{K}_1(A)) \quad \mbox{ for all } A \in \mathcal{A}_{\Gamma}^{\rm loc} \, .
\end{equation}
Moreover, any such map satisfies the following estimate.

\begin{lem} \label{lem:comp1}
	Let $(\Gamma, d)$ be $\nu$-regular, $\mathcal{K}_1 : \mathcal{A}_{\Gamma}^{\rm loc} \to \mathcal{A}_{\Gamma}$ 
	be a locally bounded, quasi-local map, and $\mathcal{K}_2 : \mathcal{A}_{\Gamma} \to \mathcal{A}_{\Gamma}$ be a bounded,	quasi-local map. For $i=1$ or $2$, denote by $B_i,C_i, p_i,q_i, G_i$ the corresponding parameters from \eqref{lb_est_comp} and \eqref{ql_est_comp}. Then, the following hold for the composition $\mathcal{K} =\cK_2\circ\cK_1$:
	\begin{enumerate}
		\item [(i)] $\mathcal{K}$ is locally bounded: for any $A \in \mathcal{\cA}_X$ with $X\in \cP_0(\Gamma),$
		\begin{equation} \label{comp_1_lb_est}
		\| \mathcal{K}(A) \| \leq \tilde{B} |X|^p \| A \|,
		\end{equation}
		where one may take $\tilde{B} = B_1 \| \mathcal{K}_2 \|$ and $p =p_1$.
		\item[(ii)] For any $A \in \mathcal{A}_X$ and $B \in \mathcal{A}_Y$ where $X, Y \in \mathcal{P}_0( \Gamma)$, 
		\begin{equation} \label{comp_1_ql_est}
		\| [ \mathcal{K}(A), B] \| \leq  \| A \| \| B \| \min \{2\tilde{B}|X|^p, C |X|^q G(d(X,Y))  \},
		\end{equation}
		where the numbers $\tilde{B}$ and $p$ may be taken as in (\ref{comp_1_lb_est}), one may take $q =p_1 +q_2$, and
		\begin{equation} \label{comp_1_dec_fun}
		C = \max( \kappa^{q_2} B_1 C_2, \, 4 C_1 \| \mathcal{K}_2 \|  ) \quad \mbox{and} \quad
		G(r) = (r/2)^{q_2 \nu} G_2(r/2) + G_1(r/2).
		\end{equation}
	\end{enumerate}
\end{lem}
We note that if the function $G$ described in (\ref{comp_1_dec_fun}) above is non-increasing and
satisfies $\lim_{r \to \infty} G(r) = 0$, then the above estimates show that $\mathcal{K}$ is 
quasi-local. 

\begin{proof}
	To prove (i), note that given any $X \in \mathcal{P}_0( \Gamma)$,
	\begin{equation}
	\| \mathcal{K}(A) \| \leq \| \mathcal{K}_2 \| \| \mathcal{K}_1(A) \| \leq B_1 \| \mathcal{K}_2 \|  |X|^{p_1} \| A\|  \quad \mbox{for all } A \in \mathcal{A}_X \, .
	\end{equation}
	This proves (\ref{comp_1_lb_est}).
	
	The proof of (ii) follows from two observations. 
	First, the bound in (\ref{comp_1_lb_est}) implies a rough estimate on the commutator for any
	$X,Y \in \mathcal{P}_0( \Gamma)$. In fact, whenever $A \in \mathcal{A}_X$ and $B \in \mathcal{A}_Y$, one has
	\begin{equation} \label{bad_est}
	\| [ \mathcal{K}(A), B ] \| \leq 2 \| \mathcal{K}(A) \| \| B \| \leq 2 \tilde{B} |X|^p \| A \|  \| B \| .
	\end{equation}
	
	Next, we note that when $d(X,Y)>0$, we obtain better estimates.
	Under this additional constraint, set $m = d(X,Y)/2$. For any locally normal product state $\rho$ on $\mathcal{A}_{\Gamma}$,
	the estimate
	\begin{equation} \label{loc_bd_comp}
	\| \mathcal{K}_1(A) - \Pi_{X(m)}( \mathcal{K}_1(A)) \| \leq 2 C_1 |X|^{p_1} \| A \|  G_1(m) 
	\end{equation}
	follows from an application of Lemma~\ref{lem:qlm_approx}. Denoting by $A_m = \Pi_{X(m)}( \mathcal{K}_1(A))$,
	we find that
	\begin{equation} \label{comp_com}
	\| [ \mathcal{K}(A), B] \| \leq \| [ \mathcal{K}_2 ( A_m ), B ] \| + \| [ \mathcal{K}_2 ( \mathcal{K}_1(A) - A_m  ), B ] \|   . 
	\end{equation} 
	Since $\mathcal{K}_2$ is quasi-local, we have that
	\begin{eqnarray}
	\| [ \mathcal{K}_2 (A_m), B ] \| & \leq & C_2 |X(m)|^{q_2} \| A_m \|  \| B \|  G_2(d(X(m),Y)) \nonumber \\
	& \leq & \kappa^{q_2} B_1 C_2  \| A \| \| B \| |X|^{p_1+q_2} (d(X,Y)/2)^{ q_2 \nu} G_2(d(X,Y)/2)
	\end{eqnarray}
	where we have used the local bound for $\mathcal{K}_1$, (\ref{infl_bd}), and the fact that 
	$d(X,Y) \leq 2 d(X(m),Y)$. The second term in (\ref{comp_com}) above satisfies 
	\begin{eqnarray}
	\| [ \mathcal{K}_2 ( \mathcal{K}_1(A) - A_m  ), B ] \|  & \leq & 2 \| \mathcal{K}_2 \| \| \mathcal{K}_1(A) - A_m \| \| B \| \nonumber \\
	& \leq & 4 C_1 \| \mathcal{K}_2 \| \| A \| \| B \| |X|^{p_1} G_1(m) .
	\end{eqnarray}
	This completes the proof of (\ref{comp_1_ql_est}).
\end{proof}

For some of our applications, the estimates proven in Lemma~\ref{lem:comp1} do not suffice. 
Briefly, some information is lost when estimating the outer mapping $\mathcal{K}_2$ with
a rough norm bound. Due to this, we consider more general compositions in the proposition below.
First, we introduce some notation. Let $G: [0, \infty) \to [0, \infty)$ and $m \geq 0$, we say that
$G$ has a {\it finite $m$-th moment} if
\begin{equation}
\sum_{n=0}^{\infty} (1+n)^m G(n) < \infty \, .
\end{equation}

\begin{prop} \label{prop:gen_comp_qlms} 
	Let $(\Gamma, d)$ be a $\nu$-regular metric space. For $i=1,2$, let $\mathcal{K}_i : \mathcal{A}_{\Gamma}^{\rm loc} \to \mathcal{A}_{\Gamma}$ be locally bounded, quasi-local maps. Suppose that $G_1$, the decay function in (\ref{ql_est_comp}) associated to $\mathcal{K}_1$, has a finite $\nu p_2$-th moment. Then, for any choice of locally normal product state $\rho$ on $\mathcal{A}_{\Gamma}$,
	the composition $\mathcal{K}^{\rho} : \mathcal{A}_{\Gamma}^{\rm loc} \to \mathcal{A}_{\Gamma}$ given by
	\begin{equation} \label{gen_comp_qlm_def}
	\mathcal{K}^{\rho}(A) = \sum_{n=0}^{\infty} \mathcal{K}_2( \Delta_{X(n)}^{\rho}( \mathcal{K}_1(A))) \quad \mbox{for all } A \in \mathcal{A}_{\Gamma}^{\rm loc} \mbox{ with } {\rm supp}(A) = X \, 
	\end{equation}  
	is well-defined. The series above is absolutely convergent and may be estimated uniformly in
	the choice of locally normal product state $\rho$. In fact, the mapping $\mathcal{K}^{\rho}$ is independent of
	the choice of $\rho$.
\end{prop}

\begin{proof}
	Fix a locally normal product state $\rho$ on $\mathcal{A}_{\Gamma}$. 
	Lemma~\ref{lem:qlm_approx} shows that for each $X \in \mathcal{P}_0( \Gamma)$ and
	any $A \in \mathcal{A}_X$ we have that
	\begin{equation}
	\mathcal{K}_1(A) = \sum_{n=0}^{\infty} \Delta_{X(n)}^{\rho}( \mathcal{K}_1( A))
	\end{equation}
	and the series above is absolutely convergent. To obtain this 
	series representation, it is only required that the decay function associated to $\mathcal{K}_1$ is
	summable, see e.g. (\ref{abs_sum_D_qlm}) which is independent of the choice of $\rho$.

	We now claim that under the additional finite moment condition,  
	for each $X \in \mathcal{P}_0( \Gamma)$ and any $A \in \mathcal{A}_X$, 
	the series defining $\mathcal{K}^{\rho}(A)$ in (\ref{gen_comp_qlm_def}) is also absolutely convergent. In fact, the bound
	\begin{eqnarray} \label{abs_conv_comp}
	\sum_{n=0}^{\infty} \| \mathcal{K}_2( \Delta_{X(n)}^{\rho}( \mathcal{K}_1(A))) \| & \leq & B_2 \sum_{n=0}^{\infty} |X(n)|^{p_2} \| \Delta_{X(n)}^{\rho}( \mathcal{K}_1(A)) \|   \nonumber \\ 
	& \leq &  B_1 B_2 |X|^{p_1+p_2} \| A \| + 4 C_1 B_2 \kappa^{p_2} |X|^{p_2+q_1} \| A \| \sum_{n=1}^{\infty} n^{\nu p_2} G_1(n-1),
	\end{eqnarray}
	 can be obtained as follows. For the first inequality above, we use that $\cK_2$ is locally bounded. For the second inequality, we first partition the sum into $n=0$ and $n>0$. For $n=0$, we use that $\Delta_{X} = \Pi_X$ and  $\cK_1$ is locally bounded. For $n>0$, we apply \eqref{Delta_bd} using the quasi-locality of $\cK_{1}$, and invoke \eqref{infl_bd}.
	
	Now, let $\rho_1$ and $\rho_2$ be any two locally normal product states on $\mathcal{A}_{\Gamma}$.
	We show that for each fixed $X \in \mathcal{P}_0( \Gamma)$ and any $\epsilon >0$, one can estimate
	\begin{equation} \label{est_diff_comp}
	\| \mathcal{K}^{\rho_1}(A) - \mathcal{K}^{\rho_2}(A) \| \leq \epsilon \| A \| \quad \mbox{for all } A \in \mathcal{A}_X
	\end{equation}
	and hence prove that the mapping $\mathcal{K}^{\rho}$ is independent of the choice of $\rho$.
	
	By the absolute convergence proven in (\ref{abs_conv_comp}) and the finite moment condition, 
	it is clear that for any $\epsilon >0$ there is some $N \geq 1$, independent of $\rho$, for which
	\begin{equation} \label{def_N_comp}
	4 C_1 B_2 \kappa^{p_2} |X|^{p_2+q_1} \sum_{n=N}^{\infty} (1+n)^{\nu p_2} G_1(n) < \epsilon / 3.
	\end{equation} 
	For $N$ as above, we write
	\begin{eqnarray}
	\mathcal{K}^{\rho_1}(A) - \mathcal{K}^{\rho_2}(A) & = & \sum_{n=0}^N \left( \mathcal{K}_2 ( \Delta_{X(n)}^{\rho_1}( \mathcal{K}_1(A))) - \mathcal{K}_2 ( \Delta_{X(n)}^{\rho_2}( \mathcal{K}_1(A))) \right) \nonumber \\
	& \mbox{ } & \quad + \sum_{n=N+1}^{\infty} \left( \mathcal{K}_2 ( \Delta_{X(n)}^{\rho_1}( \mathcal{K}_1(A))) - \mathcal{K}_2 ( \Delta_{X(n)}^{\rho_2}( \mathcal{K}_1(A))) \right).
	\end{eqnarray}
	Based on   of $N$, it follows from (\ref{abs_conv_comp}) that 
	\begin{equation}
	\left\| \sum_{n=N+1}^{\infty} \left( \mathcal{K}_2 ( \Delta_{X(n)}^{\rho_1}( \mathcal{K}_1(A))) - \mathcal{K}_2 ( \Delta_{X(n)}^{\rho_2}( \mathcal{K}_1(A))) \right) \right\| < \frac{2 \epsilon}{3} \| A \| .
	\end{equation}
	Using linearity of $\mathcal{K}_2$ and the telescopic properties of the sums, see (\ref{finite_tele}), we also have that
	\begin{align}
	\Bigg\| \sum_{n=0}^N \Big( \mathcal{K}_2 & ( \Delta_{X(n)}^{\rho_1}( \mathcal{K}_1(A))) - \mathcal{K}_2 ( \Delta_{X(n)}^{\rho_2}( \mathcal{K}_1(A))) \Big)  \Bigg\|\\
	 & = 
	\left\| \mathcal{K}_2 \left( \Pi_{X(N)}^{\rho_1}(\mathcal{K}_1(A))  -  \Pi_{X(N)}^{\rho_2}(\mathcal{K}_1(A)) \right) \right\| \nonumber \\
	& \leq   B_2 |X(N)|^{p_2} \| \Pi_{X(N)}^{\rho_1}(\mathcal{K}_1(A)) - \Pi_{X(N)}^{\rho_2}(\mathcal{K}_1(A)) \| \nonumber \\
	& \leq  4 C_1 B_2 \kappa^{p_2} |X|^{p_2+q_1} N^{ \nu p_2} G_1(N) \| A \| \nonumber \\
	& <  \frac{\epsilon}{3} \| A \|
	\end{align} 
	where we have used the local bound for $\mathcal{K}_2$ and inserted (and removed) $\mathcal{K}_1(A)$ to apply Lemma~\ref{lem:qlm_approx}. 
	The final bound above is clear from (\ref{def_N_comp}). The claim in (\ref{est_diff_comp}) now follows. 
\end{proof}

{F}rom Proposition~\ref{prop:gen_comp_qlms}, we now have conditions under which there is a well-defined composition
of two locally bounded, quasi-local maps. The next lemma provides local bounds and quasi-local estimates for the
resulting composition.

\begin{lem} \label{lem:comp2} 
	Let $(\Gamma, d)$ be a $\nu$-regular metric space. For $i=1,2$, let $\mathcal{K}_i: \mathcal{A}_{\Gamma}^{\rm loc} \to \mathcal{A}_{\Gamma}$ be 
	locally bound, quasi-local maps. Suppose that $G_1$, the decay function in (\ref{ql_est_comp}) associated to $\mathcal{K}_1$, has a finite $\nu p_2$-th moment
	and let $\mathcal{K} : \mathcal{A}_{\Gamma}^{\rm loc} \to \mathcal{A}_{\Gamma}$ denote the composition from (\ref{gen_comp_qlm_def}).
	\begin{enumerate}
		\item[(i)] $\mathcal{K}$ is locally bounded: for any $A\in \cA_{X}$ with $X\in \cP_0(\Gamma)$
		\begin{equation} \label{comp_2_lb_est}
		\| \mathcal{K}(A) \| \leq \tilde{B} |X|^p \| A \|,
		\end{equation}
		where one may take $p = p_2 + \max\{p_1, q_1\}$ and
		\begin{equation} \label{lb_order+const}
		 \tilde{B} = B_2 \left( B_1 + 4 C_1 \kappa^{p_2} \sum_{n=0}^{\infty}(1+n)^{\nu p_2} G_1(n) \right) \, .
		\end{equation}
		\item[(ii)] For any $A \in \cA_X$ and $B \in \cA_Y$ where $X, Y \in \mathcal{P}_0( \Gamma)$, one has that 
		\begin{equation} \label{comp_2_ql_est}
		\| [ \mathcal{K}(A), B] \| \leq \| A \| \| B \| \min\left\{ 2 \tilde{B} |X|^p, \, C |X|^q G(d(X,Y)) \right\} 
		\end{equation}
		where one may take $C = {\rm max}\{\kappa^{q_2}B_1 C_2, \, 8\kappa^{p_2}C_1 B_2\}$, $q = \max\{p_1, q_1\} + \max\{p_2, q_2\}$, and 
		\begin{equation} \label{comp_2_dec_fun}
		G(r) = (r/2)^{q_2 \nu} G_2(r/2) + \sum_{n= \lfloor r/2 \rfloor }^{\infty}(1+n)^{\nu p_2} G_1(n).
		\end{equation}
	\end{enumerate}
\end{lem}

Again we stress that the bounds above demonstrate that the composition is quasi-local if the function $G$ in (\ref{comp_2_dec_fun}) is non-increasing with $\lim_{r \to \infty}G(r) = 0$.

\begin{proof}
	The bound (\ref{comp_2_lb_est}) is a consequence of (\ref{abs_conv_comp}) found in the proof of Proposition~\ref{prop:gen_comp_qlms}. To prove (\ref{comp_2_ql_est}), we argue as in the proof of Lemma~\ref{lem:comp1}(ii). We need only consider the case when $d(X,Y)>0$, and therein we set $m = \lfloor d(X,Y)/2 \rfloor$. For $A \in \cA_X$, we write
	\begin{equation}
	\mathcal{K}(A) = \sum_{n=0}^{\infty} \mathcal{K}_2( \Delta_{X(n)}( \mathcal{K}_1(A))) = \mathcal{K}_2(\Pi_{X(m)}( \mathcal{K}_1(A))) + \sum_{n=m+1}^{\infty} \mathcal{K}_2( \Delta_{X(n)}( \mathcal{K}_1(A))).
	\end{equation}
	Here we have used an expansion as in (\ref{gen_comp_qlm_def}), and the telescopic property (\ref{finite_tele}). Moreover, we have dropped the dependence of $\rho$ from the notation since \eqref{gen_comp_qlm_def} is invariant under the choice of locally normal product state by Proposition~\ref{prop:gen_comp_qlms}. The estimate
	\begin{equation}
	\| [ \mathcal{K}(A), B ] \| \leq \left\| \left[ \mathcal{K}_2(\Pi_{X(m)}( \mathcal{K}_1(A))) , B \right] \right\| 
	+ \sum_{n=m+1}^{\infty} \left\| \left[ \mathcal{K}_2 \left( \Delta_{X(n)}( \mathcal{K}_1(A)) \right), B \right] \right\|
	\end{equation} 
	readily follows. 
	
	As $\mathcal{K}_2$ is quasi-local, it is clear that
	\begin{eqnarray}
	\left\| \left[ \mathcal{K}_2(\Pi_{X(m)}( \mathcal{K}_1(A))) , B \right] \right\|  & \leq & C_2 |X(m)|^{q_2} \| \mathcal{K}_1(A) \| \| B \| G_2(d(X(m), Y)) \nonumber \\
	& \leq & B_1 C_2 \kappa^{q_2} |X|^{p_1+q_2} \| A \| \| B \| (d(X,Y)/2)^{q_2 \nu} G_2(d(X,Y)/2)
	\end{eqnarray}
	
	To estimate the remaining term, for each $n \geq m+1$ we find
	\begin{eqnarray}
	\left\| \left[ \mathcal{K}_2 \left( \Delta_{X(n)}( \mathcal{K}_1(A)) \right), B \right] \right\|  & \leq &  2  \left\| \mathcal{K}_2 \left( \Delta_{X(n)}( \mathcal{K}_1(A)) \right) \right\|  \| B \|\nonumber \\
	& \leq & 2  B_2 |X(n)|^{p_2} \|  \Delta_{X(n)}( \mathcal{K}_1(A)) \| \| B \| \nonumber \\
	& \leq & 8 C_1 B_2 \kappa^{p_2} |X|^{p_2+q_1} \| A \|  \| B \|  n^{p_2 \nu} G_1(n-1)
	\end{eqnarray}
	where we have used (\ref{Delta_bd}). This proves (\ref{comp_2_ql_est}). 
\end{proof}

%
%

\subsection{Quasi-local transformations of interactions} \label{sec:quasi_loc_ints}

Important applications of quasi-local maps arise in the classification of gapped ground state phases \cite{bachmann:2012, bachmann:2014a, bachmann:2015a, ogata:2016, ogata:2016a, ogata:2017} and recent proofs of stability of the spectral gap \cite{bravyi:2010,bravyi:2011,michalakis:2013,nachtergaele:2016b,moon:2018,QLBII}. In these proofs, key insights come from analyzing the composition of a quasi-local map with an interaction, $\cK \circ \Phi : \cP_0(\Gamma) \to \cA_{\Gamma}$. It is important to note that such maps are not necessarily interactions themselves, as the image lies in the quasi-local algebra, $\cA_{\Gamma}$, rather than the algebra of local observables, $\cA_{\Gamma}^{\rm loc}$. In our applications, the interaction and quasi-local map often depend on an auxiliary parameter and we allow for this in our construction and results. In what follows, we provide a general framework under which one can construct a {\em bona fide} interaction from such a composition and derive estimates that determine conditions under which these transformed interactions have a finite $F$-norm.

We begin with a general description of transformed interactions in Section~\ref{sec:trans_ints}. In Section~\ref{sec:qli_fvr}, 
we prove estimates on these transformed interactions in finite volume. In Section~\ref{sec:qli_tl}, we give conditions under which the
finite-volume results proven in Section~\ref{sec:qli_fvr} extend to the thermodynamic limit. A concrete application of these results will be given in Section \ref{sec:ql_sf}, where we show that the spectral flow automorphisms can be realized as the dynamics generated by a time-dependent interaction with good decay properties.

\subsubsection{Transformed finite-volume Hamiltonians} \label{sec:trans_ints} 

To investigate the spectral properties of a given Hamiltonian $H$, it is often convenient to 
work with a unitarily equivalent Hamiltonian $H' = U^*HU$. When the original Hamiltonian
is a sum of local terms, the strict locality of these terms is typically not preserved under 
the mapping $H \mapsto H'$. In recent applications, most notably the proof of stability, it is 
shown that locality based arguments, such as Lieb-Robinson bounds, still apply to $H'$ if the automorphism implemented by the unitary $U$ 
is sufficiently quasi-local.

In this section, we discuss this situation more generally. Specifically, we analyze the transformation of
a given interaction by a quasi-local map. Briefly, we first argue that the composition of a quasi-local
map with an interaction can, using the methods of Section~\ref{sec:laqlmaps}, still be realized as an interaction with 
strictly local terms. Moreover, we show that the spatial decay associated to this new interaction can be estimated in terms of the decays of the original interaction and the quasi-local map. 
Finally, we discuss quasi-locality estimates for the dynamics of this
transformed interaction.  

To establish some notation, let us first consider a simple, time-independent case in finite volume.
As before, fix a quantum lattice system comprised of $(\Gamma, d)$ and $\cA_{\Gamma}$.
Let $\Phi$ be an interaction on $\cA_{\Gamma}$ and recall that for any finite $\Lambda \subset \Gamma$,
we denote by
\begin{equation} \label{f_vol_ham}
H_{\Lambda}^{\Phi} = \sum_{X \subset \Lambda} \Phi(X)
\end{equation} 
the finite-volume Hamiltonian generated by $\Phi$. Our goal here is to analyze the transformation 
of this local Hamiltonian $H_{\Lambda}^{\Phi}$ by a linear map $\mathcal{K} : \mathcal{A}_{\Lambda} \to \mathcal{A}_{\Lambda}$. 
In particular, we consider 
\begin{equation} \label{trans_ham}
\mathcal{K}( H_{\Lambda}^{\Phi}) = \sum_{X \subset \Lambda} \mathcal{K}(\Phi(X)) \, .
\end{equation}
Generally, the map $\mathcal{K}$ will not preserve locality, and 
in such cases, each term in (\ref{trans_ham}) will be global in the sense that
$\supp(\mathcal{K}(\Phi(X)))=\Lambda$ for each $X \subset \Lambda$. For this reason, the sum
on the right-hand-side of (\ref{trans_ham}) does not represent an interaction in the sense defined in Section~\ref{subsec:lrb}.

Using the methods of Section~\ref{sec:laqlmaps}, one can rewrite the right-hand-side of 
(\ref{trans_ham}) as a sum of strictly local terms. To see this, fix a locally normal product state 
$\rho$ on $\cA_\Gamma$. In this finite-volume context, we only use the restriction of $\rho$ to $\cA_{\Lambda}$ and again refer to it as $\rho$. In terms of $\rho$, we have defined 
local decompositions with respect to any $X \subset \Lambda$ and each $n \geq 0$ 
as the maps $\Delta^\Lambda_{X(n)} : \cA_{\Lambda} \to \cA_{\Lambda}$ given by (\ref{loc_dec}).
Recall further that $\Delta^\Lambda_{X(n)}$ has range contained in $\cA_{X(n) \cap \Lambda}$
(using again the identification of the former as a sub-algebra of $\cA_{\Lambda}$), and 
moreover, $\| \Delta^\Lambda_{X(n)} \| \leq 2$.
In this case, each term $\mathcal{K}( \Phi(X))$ appearing in (\ref{trans_ham}) can be written as a finite telescopic sum as in \eq{telescopicA} by
\begin{equation}\label{finite-telescopic}
\mathcal{K}( \Phi(X)) = \sum_{n \geq 0} \Delta^\Lambda_{X(n)}( \mathcal{K}( \Phi(X))) .
\end{equation}
For any $Z \subset \Lambda$, define
\begin{equation} \label{def_comp_int}
\Psi_{\Lambda}(Z) = \sum_{n \geq 0} \sum_{\substack{X \subset Z:\\ X(n) \cap \Lambda = Z}} \Delta_{X(n)}^{\Lambda}( \mathcal{K}( \Phi(X)))
\end{equation}
with the understanding that empty sums are taken to be zero. By construction,
$\Psi_{\Lambda}(Z) \in \cA_Z$ and under the additional
assumption that $\mathcal{K}(A)^* = \mathcal{K}(A^*)$ for all $A \in \cA_{\Lambda}$, we see that
$\Psi_{\Lambda}$ is a well-defined (finite-volume) interaction in the sense of Section~\ref{subsec:lrb}. Moreover,
\begin{equation} \label{loc_comp_ham}
\mathcal{K}(H_{\Lambda}^{\Phi}) = \sum_{X \subset \Lambda} \mathcal{K}(\Phi(X))= \sum_{Z \subset \Lambda} \Psi_{\Lambda}(Z)
= H_{\Lambda}^{\Psi_{\Lambda}} \, .
\end{equation}
In words, using the notation from (\ref{f_vol_ham}), the final equalities in (\ref{loc_comp_ham}) show that the transformed Hamiltonian in (\ref{trans_ham}) may be rewritten as the Hamiltonian generated by the interaction $\Psi_{\Lambda}$.

\subsubsection{Finite-volume results} \label{sec:qli_fvr}

In this section, we give a finite-volume analysis of the transformed interactions
briefly introduced at the end of Section~\ref{sec:trans_ints}. In Section~\ref{sec:qli_tl}, we will discuss appropriate conditions under which these results will extend to the thermodynamic limit. For many
of our applications, both the interaction and the quasi-local map will be time-dependent.
As a consequence, we state and prove our estimates for families of interactions and quasi-local maps.

We make two useful observations in this section. First, we indicate 
a set of continuity assumptions under which a finite-volume transformed 
interaction corresponds to a well-defined dynamics. These assumptions will also
guarantee that the interaction which generates this transformed interaction is
strongly continuous in the sense of Section~\ref{sec:spatialstructure}. Next, 
we will show that certain decay assumptions on the initial interaction 
$\Phi$ and quasi-local map $\mathcal{K}$ lead to explicit estimates on the 
decay of the interaction $\Psi_{\Lambda}$; here we are using
the notation introduced at the end of Section~\ref{sec:trans_ints}. 
Technically, the continuity and decay assumptions are independent, however, in most 
applications, the models we consider satisfy both sets of assumptions simultaneously. 

Let $(\Gamma, d)$ and $\cA_{\Gamma}$ be a quantum lattice system, and $I \subset \mathbb{R}$ be an interval. We once again work with strongly continuous interactions $\Phi : \mathcal{P}_0( \Gamma) \times I \to \cA_{\Gamma}^{\rm loc}$; meaning that, for all $X\in \cP_0(\Gamma)$,
\begin{enumerate}
	\item[(i)] $\Phi(X,t)^* = \Phi(X,t) \in \cA_X$ for all $t \in I$.
	\item[(ii)] The map $t \mapsto \Phi(X,t)$ is continuous in the strong operator topology on $\cA_X = \mathcal{B}( \cH_X)$.
\end{enumerate}
For any $\Lambda \in \mathcal{P}_0( \Gamma)$, we define a finite-volume, time-dependent Hamiltonian associated to $\Phi$ by 
\begin{equation} \label{fv_td_ham_P}
H_{\Lambda}^{\Phi}(t) = \sum_{X \subset \Lambda} \Phi(X,t) \quad \mbox{for all } t \in I \, .
\end{equation}
From the assumptions, it is clear that $H_{\Lambda}^{\Phi}$ is also pointwise self-adjoint and strongly continuous as it is the finite sum of such terms.
 
For the remainder of this subsection, we fix a finite volume $\Lambda \in \mathcal{P}_0( \Gamma)$, and are interested in studying time-dependent transformed finite-volume Hamiltonians analogous to those consider in Section~\ref{sec:trans_ints}. Specifically, given any family of linear maps $\{ \mathcal{K}_t: \cA_{\Lambda} \to \cA_{\Lambda} \}_{t \in I}$, we consider the set of all operators of the form
\begin{equation} \label{td_trans_ham}
\mathcal{K}_t( H_{\Lambda}^{\Phi}(t)) = \sum_{X \subset \Lambda} \mathcal{K}_t(\Phi(X,t)), 
\end{equation}
and  will refer to such collections as a finite-volume family of transformed interactions. Under assumptions which guarantee that $t \mapsto \mathcal{K}_t( H_{\Lambda}^{\Phi}(t))$ is pointwise self-adjoint and strongly 
continuous, the methods of Section~\ref{sec:Dyson}, see also Section~\ref{sec:spatialstructure}, demonstrate that these transformed interactions
correspond to a dynamics. More precisely, Proposition~\ref{prop:sols} shows that
for any $s,t \in I$ the strong solution $U_{\Lambda}(t,s) \in \mathcal{A}_{\Lambda}$ of  
\begin{equation} \label{def:qluni}
\frac{d}{dt} U_{\Lambda}(t,s) = - i \mathcal{K}_t( H_{\Lambda}^{\Phi}(t)) U_{\Lambda}(t,s),
\quad U_{\Lambda}(s,s) = \idty
\end{equation}
defines a two-parameter family of unitaries, and thus a cocycle of automorphisms 
$\tau_{t,s}^{\Lambda}$ of $\mathcal{A}_{\Lambda}$ with
\begin{equation} \label{def:qldyn}
\tau_{t,s}^{\Lambda}(A) = U_{\Lambda}(t,s)^* A U_{\Lambda}(t,s) \quad \mbox{for all } A \in \cA_{\Lambda} .
\end{equation}
We refer to $\tau_{t,s}^{\Lambda}$ as the dynamics corresponding to the transformed Hamiltonian
in (\ref{td_trans_ham}). 

A main goal of this subsection is to establish assumptions under which the dynamics in (\ref{def:qldyn})
satisfies a quasi-locality estimate, also known as a Lieb-Robinson bound, see Theorem~\ref{thm:lrb}.
In order to do so, we first re-write the family of transformed interactions in (\ref{td_trans_ham}) as a sum of strictly
local terms. Fix a locally normal product state $\rho$ on $\cA_{\Gamma}$. For any $Z \subset \Lambda$
and each $t \in I$, set 
\begin{equation} \label{fv_td_qli}
\Psi_{\Lambda}(Z,t) = \sum_{n \geq 0} \sum_{\substack{X \subset Z:\\ X(n) \cap \Lambda = Z}} \Delta_{X(n)}^{\Lambda}( \mathcal{K}_t( \Phi(X,t)))
\end{equation} 
where, as in Section~\ref{sec:trans_ints}, we have made local decompositions of the global terms on the right-hand-side of (\ref{td_trans_ham}); compare with (\ref{finite-telescopic}) and (\ref{def_comp_int}). We stress that for all $t \in I$, we 
make local decompositions with respect to the same locally normal product state $\rho$. As in (\ref{loc_comp_ham}), it
is clear that for each $t \in I$, 
\begin{equation} \label{td_fv_qli_ham}
\mathcal{K}_t(H_{\Lambda}^{\Phi}(t)) = \sum_{X \subset \Lambda} \mathcal{K}_t(\Phi(X,t))= \sum_{Z \subset \Lambda} \Psi_{\Lambda}(Z,t)
= H_{\Lambda}^{\Psi_{\Lambda}}(t) \, .
\end{equation}
We now introduce a set of assumptions on the family of functions $\{ \mathcal{K}_t: \cA_{\Lambda}\to \cA_{\Lambda} \}_{t \in I}$ which
guarantee that: (i) the dynamics in (\ref{def:qldyn}) is well-defined and (ii) the mapping
$\Psi_{\Lambda} : \mathcal{P}_0( \Lambda) \times I \to \cA_{\Lambda}$ is a strongly 
continuous interaction in the sense of Section~\ref{sec:spatialstructure}.  

\begin{assumption}\label{ass:fv_qlm_cont} We assume the collection of finite-volume linear maps $\{ \mathcal{K}_t : \mathcal{A}_{\Lambda} \to \mathcal{A}_{\Lambda} \}_{t \in I}$, is a strongly continuous family of strongly continuous transformations that are compatible with the involution in the sense that:
	\begin{enumerate}
		\item[(i)] For each $t \in I$, $\mathcal{K}_t(A)^* = \mathcal{K}_t(A^*)$ for all $A \in \mathcal{A}_{\Lambda}$.
		\item[(ii)] For each  $A \in \mathcal{A}_{\Lambda}$, the function $t \mapsto \mathcal{K}_t(A)$ is norm continuous. 
		\item[(iii)]  For each $t \in I$, the map $\mathcal{K}_t : \mathcal{A}_{\Lambda} \to \mathcal{A}_{\Lambda}$ is continuous
		on bounded subsets when both its domain and co-domain are equipped with the strong operator topology and moreover, this continuity
		is uniform on compact subsets of $I$. 
	\end{enumerate}
\end{assumption}

Assumption~\ref{ass:fv_qlm_cont}(i), together with Proposition~\ref{prop:compatibility}(v), is used to ensure that the terms defined in (\ref{fv_td_qli}) are point-wise self-adjoint. This is important in defining the unitary propagator, but it plays no role in establishing various continuity properties. Next, as is discussed in Section~\ref{sec:cont_l_a}, Assumption~\ref{ass:fv_qlm_cont} (ii) and (iii)
guarantee that $t \mapsto \mathcal{K}_t(H_{\Lambda}^{\Phi}(t))$ is strongly continuous.
As such, the finite-volume dynamics associated to this transformed interaction, see (\ref{def:qluni}) and (\ref{def:qldyn}), is
well-defined. In particular, this dynamics is independent of the choice of $\rho$. 
Note further that Assumption~\ref{ass:fv_qlm_cont} (ii) and (iii) also ensure that
for each $X \subset \Lambda$, $t \mapsto \mathcal{K}_t(\Phi(X,t))$ is strongly continuous.
Given this, Proposition~\ref{prop:strongcontinuity} shows that each of the finitely many terms on the right-hand-side of
(\ref{fv_td_qli}) is strongly continuous as well, and as a result, $\Psi_{\Lambda}$ is a strongly continuous interaction.
The interaction $\Psi_{\Lambda}$, which does depend on the choice of $\rho$, will be useful in proving
a quasi-locality bound on the finite-volume dynamics in (\ref{def:qldyn}).

The goal for the remainder of this section is to quantify a quasi-locality estimate for the finite-volume 
dynamics in (\ref{def:qldyn}) in terms of decay properties of the original interaction $\Phi$ and the
finite-volume transformations $\{ \mathcal{K}_t \}_{t \in I}$. For these results, we assume that
$(\Gamma, d)$ is $\nu$-regular and equipped with an $F$-function $F$.

Let us again fix an interval $I \subset \mathbb{R}$ and a finite-volume $\Lambda \in \mathcal{P}_0( \Gamma)$.
We make the following decay assumptions on a family of finite volume transformations.

\begin{assumption} \label{ass:fv_qlm_dec} We assume that the family of finite-volume linear maps $\{ \mathcal{K}_t : \mathcal{A}_{\Lambda} \to \mathcal{A}_{\Lambda}\}_{t \in I}$ is a time-dependent family of locally bounded, quasi-local maps in the sense that:
\begin{enumerate}
	\item[(i)] There is some $p \geq 0$ and a measurable, locally bounded function 
	$B:I \to [0, \infty)$ so that given any $X \subset \Lambda$,
	\begin{equation} \label{fv_lb_K_td}
	\| \mathcal{K}_t(A) \| \leq B(t) |X|^p \| A \|    \mbox{ for all } A \in \cA_X  \mbox{ and } t \in I \, .
	\end{equation} 
	\item[(ii)] There is some $q \geq 0$, a non-increasing function $G: [0, \infty) \to [0, \infty)$ with
	$G(r) \to 0$ as $r \to \infty$, and a measurable, locally bounded function $C:I \to [0, \infty)$ for which given any sets $X, Y \subset \Lambda$, one has that 
	\begin{equation} \label{fv_ql_K_td}
	\| [ \mathcal{K}_t(A), B ] \| \leq C(t) |X|^q \| A \| \| B \|  G(d(X,Y))  \mbox{ for all } A\in \cA_X ,  B\in \cA_Y, \mbox{ and } t \in I \, .
	\end{equation}
\end{enumerate}
\end{assumption} 

For the initial interaction, we impose decay assumptions which compensate for the factors of $|X|$
found in (\ref{fv_lb_K_td}) and (\ref{fv_ql_K_td}) above. More precisely, for any time-dependent 
interaction $\Phi$ and each $m \geq 0$, we define a new interaction $\Phi_m$, which we call 
the {\it $m$-th moment of $\Phi$}, with terms
\begin{equation}
\Phi_m(X,t) = |X|^m \Phi(X,t) \quad \mbox{for all } X \in \mathcal{P}_0( \Gamma) \mbox{ and } t \in I \, .
\end{equation} 
To prove the result in Theorem~\ref{thm:qli_est} below, we will assume that the initial interaction $\Phi$
satisfies $\Phi_m \in \mathcal{B}_F(I)$ for $m = \max\{p,q\}$ with $p$ and $q$ 
as in (\ref{fv_lb_K_td}) and (\ref{fv_ql_K_td}), respectively. Recall that an interaction 
$\Phi \in \mathcal{B}_F(I)$ if and only if $\Phi: \mathcal{P}_0( \Gamma) \times I \to \mathcal{A}_{\Gamma}^{\rm loc}$ is a strongly continuous interaction and
the map $\| \Phi \|_F : I \to [0, \infty)$ given by 
\begin{equation} \label{td_int_F_norm}
\| \Phi \|_F(t) = \sup_{x,y \in \Gamma} \frac{1}{F(d(x,y))} \sum_{\stackrel{X \in \mathcal{P}_0( \Gamma):}{x,y \in X}} \| \Phi(X, t) \|,
\end{equation}
is locally bounded. An immediate consequence of (\ref{td_int_F_norm}) is that for any
finite volume $\Lambda \in \mathcal{P}_0(\Gamma)$ and any pair $x, y \in \Lambda$, the bound
\begin{equation}
\sum_{\stackrel{X \subset \Lambda:}{x,y \in X}} \| \Phi(X, t) \| \leq \| \Phi \|_F(t) F(d(x,y))  
\end{equation}
holds for all $t \in I$.  We refer to Section~\ref{sec:spatialstructure} for more details on $\mathcal{B}_F(I)$.

Finally, before we state our first result we review some notation. Recall that a non-negative function $G:[0, \infty) \to  (0,\infty)$ has a finite $m$-th moment if
\begin{equation} \label{G_m_mom}
\sum_{n=0}^{\infty} (1+n)^m G(n) < \infty \, .
\end{equation}
Note that, in this case, the tail of the series $r\mapsto \sum_{n= \lfloor r \rfloor}(1+n)^mG(n)$ is a non-negative, non-increasing function for which
\begin{equation}
\lim_{r \to \infty} \sum_{n=\lfloor r \rfloor }^{\infty} (1+n)^m G(n) = 0.
\end{equation}

We state our basic estimate on these finite-volume transformed interactions. In the statement, we make use of the quantities $p,\, q$ and $G$ from 
Assumption~\ref{ass:fv_qlm_dec}. 

\begin{thm} \label{thm:qli_est} Consider a quantum lattice system comprised of $\nu$-regular metric space $(\Gamma, d)$ and quasi-local algebra $\cA_{\Gamma}$. Let $F$ be an $F$-function on $(\Gamma,d)$, $I \subset \mathbb{R}$ be an interval, and $\Lambda \in \cP_0(\Gamma)$. Assume that $\{ \mathcal{K}_t: \cA_\Lambda\to \cA_{\Lambda} \}_{t \in I}$ is a quasi-local family of transformations satisfying Assumption~\ref{ass:fv_qlm_dec}, and $\Phi$ is a strongly continuous interaction such that $\Phi_m \in \mathcal{B}_F(I)$ for $m = \max\{p,q\}$. If the decay function $G$ associated to
the family $\{ \mathcal{K}_t \}_{t \in I}$ has a finite $2 \nu +1$ moment, then for any locally normal 
state $\rho$ and each choice of $x,y \in \Lambda$, the mapping $\Psi_{\Lambda}$ defined in (\ref{fv_td_qli}), satisfies the estimate
\begin{equation} \label{rough_qli_est}
	\sum_{\stackrel{Z \subset \Lambda:}{x,y \in Z}}  \| \Psi_{\Lambda}(Z, t) \|  \leq 
C_1(t) F(d(x,y)/3) + C_2(t) \sum_{n = \lfloor d(x,y)/3 \rfloor} (1+n)^{\nu+1} G(n) 
\end{equation} 
where the time-dependent pre-factors $C_1$ and $C_2$ may be taken as
\begin{equation}
C_1(t) = B(t) \| \Phi_p \|_F(t) + 4 \kappa^2 C(t) \| \Phi_q \|_F(t) \sum_{n=0}^{\infty} (1+n)^{2 \nu +1}G(n),
\end{equation}
and $C_2(t) = 4 \kappa \| F \| C(t) \| \Phi_q \|_F(t)$.
\end{thm}

It is clear from the statement, as well as the proof, that the estimate proven in Theorem~\ref{thm:qli_est} does not require that the mappings $\mathcal{K}_t$ satisfy Assumption~\ref{ass:fv_qlm_cont}. As indicated previously, Assumption~\ref{ass:fv_qlm_cont} is convenient because it guarantees that the mapping $\Psi_{\Lambda}$ satisfies the continuity requirements needed to be a strongly continuous interaction, as defined in the beginning of this subsection.  

\begin{proof}
	Fix $Z \subset \Lambda$ and $t \in I$. A simple norm estimate, using (\ref{fv_td_qli}), shows that
	\begin{eqnarray} \label{fv_qli_est_1}
	\| \Psi_{\Lambda}(Z,t) \| & \leq & \| \Pi_Z^{\Lambda}( \mathcal{K}_t( \Phi(Z,t))) \| + \sum_{n \geq 1} \sum_{\substack{X \subset Z:\\ X(n) \cap \Lambda = Z}} \| \Delta_{X(n)}^{\Lambda}( \mathcal{K}_t( \Phi(X,t))) \| \nonumber \\
	& \leq & B(t) |Z|^p \| \Phi(Z,t) \| + 4 C(t) \sum_{n \geq 1} G(n-1) \sum_{\substack{X \subset Z:\\ X(n) \cap \Lambda = Z}} |X|^q \| \Phi(X,t) \|  .
	\end{eqnarray}
	Here we first used that $\Delta^{\Lambda}_{Z(0)} = \Pi_Z^{\Lambda}$ and that $\| \Pi_Z^{\Lambda} \| \leq 1$, 
	see Section~\ref{sec:laqlmaps} for more details. 
	Next, we used the local bound on $\mathcal{K}_t$, i.e. (\ref{fv_lb_K_td}), for the first term on the right-hand-side above.
	For the remaining terms, we combined the quasi-local bound on $\mathcal{K}_t$, i.e. (\ref{fv_ql_K_td}), with the estimate
	(\ref{Delta_bd}) as proven in Lemma~\ref{lem:qlm_approx}.
	
	We conclude that
	\begin{eqnarray} \label{fv_qli_est_2}
	\sum_{\stackrel{Z \subset \Lambda:}{x,y \in Z}}  \| \Psi_{\Lambda}(Z, t) \| & \leq & B(t)  \sum_{\stackrel{Z \subset \Lambda:}{x,y \in Z}} |Z|^p \| \Phi(Z, t) \| \nonumber \\
	& \mbox{ } & \quad + 4 C(t) \sum_{n \geq 0} G(n) \sum_{\substack{X \subset \Lambda:\\ x,y \in X(n+1)}} |X|^q \| \Phi(X, t) \|.  
	\end{eqnarray}
	An application of Lemma~\ref{lem:weight_int_dec} completes the proof.
\end{proof}

We finish this subsection with a quasi-locality estimate for the finite-volume dynamics (\ref{def:qldyn}).
Such a result is an immediate consequence of Theorem~\ref{thm:lrb} once we obtain that
$\Psi_{\Lambda} \in \mathcal{B}_{\tilde{F}}(I)$ for some $F$-function $\tilde{F}$ on $(\Gamma, d)$.
In concrete applications, the existence of such a function $\tilde{F}$ will depend on the original $F$-function $F$ and quasi-local decay function $G$. Rather than make further assumptions on $F$ and $G$ from, e.g., Theorem~\ref{thm:qli_est} let us assume there is an $F$-function $\tilde{F}$ on $(\Gamma, d)$ for which 
\begin{equation} \label{good_F_dec}
\max\left\{ F(r/3) , \sum_{n= \lfloor r/3 \rfloor}^{\infty} (1+n)^{\nu+1} G(n)  \right\} \leq \tilde{F}(r) \quad \mbox{for all } r \geq 0 \, .
\end{equation} 
We note that in many applications the initial decay functions are weighted $F$-functions in the sense of 
Section~\ref{app:sec_weight_F}, and therefore explicit choices for $\tilde{F}$
are readily determined by manipulating the weights. 
In any case, given such a function $\tilde{F}$, the bound in (\ref{rough_qli_est}) above implies an explicit 
pointwise estimate on $\| \Psi_{\Lambda} \|_{\tilde{F}}$, see e.g. (\ref{tF_norm_int_bd}) below.

We end this subsection with the following corollary.

\begin{cor} \label{cor:fv_ql_dyn_lrb} Under the assumptions of Theorem~\ref{thm:qli_est}, 
	suppose further that the family $\{ \mathcal{K}_t \}_{t \in I}$ satisfies Assumption~\ref{ass:fv_qlm_cont}, and that $\tilde{F}$ is an $F$-function on $(\Gamma,d)$ satisfying (\ref{good_F_dec}). Then, $\Psi_{\Lambda} \in \mathcal{B}_{\tilde{F}}(I)$, and the finite-volume dynamics in (\ref{def:qldyn}) associated to $\Psi_{\Lambda}$ satisfies the following bound: given any $A \in \cA_X$, $B \in \cA_Y$ where $X,Y \subset \Lambda$ with $X \cap Y = \emptyset$,
	\begin{equation} \label{qli_cor_lrb}
	\| [ \tau_{t,s}^{\Lambda}(A), B ] \| \leq  \frac{2 \| A \| \| B \|}{C_{\tilde{F}}} \left( e^{2 I_{t,s}( \Psi_{\Lambda})} - 1 \right) \sum_{x \in X} \sum_{y \in Y} \tilde{F}(d(x,y)) 
	\end{equation}
	holds for all $s,t \in I$. Moreover, for $s \leq t$,
	\begin{eqnarray} \label{tF_norm_int_bd}
	I_{t,s}( \Psi_{\Lambda}) & \leq &  C_{\tilde{F}} \int_s^t B(r) \| \Phi_p \|_F(r) \, dr \nonumber \\
	& \mbox{ } & \quad + 4 \kappa C_{\tilde{F}} \left( \kappa \sum_{n=0}^{\infty} (1+n)^{2 \nu+1} G(n) + \| F \|  \right) \int_s^t C(r) \| \Phi_q \|_F(r) \, dr
	\end{eqnarray}
	where $I_{t,s}( \Psi_{\Lambda})$ is as in (\ref{ItsPhi}) with $F$ replaced by $\tilde{F}$. 
\end{cor}

\subsubsection{Results in infinite volume} \label{sec:qli_tl}

In this section, we show how the results of the previous section extend to the thermodynamic limit. We begin with an assumption on a collection of quasi-local maps $\{\cK_t:\cA_{\Gamma}^{\rm loc}\to \cA_{\Gamma}\}$. In essence, this definition combines the notion of uniformly locally normal from Definition~\ref{def:uniformly_locally_normal_maps} with Assumptions~\ref{ass:fv_qlm_cont} and ~\ref{ass:fv_qlm_dec}. As always, we consider a quantum lattice system comprised of $(\Gamma, d)$ and $\cA_{\Gamma}$, and 
let $I \subset \mathbb{R}$ be an interval.

\begin{assumption} \label{ass:td_qlms}  We assume that the family of linear maps $\{ \mathcal{K}_t  : \mathcal{A}_{\Gamma}^{\rm loc} \to \mathcal{A}_{\Gamma}\}_{t \in I}$, is strongly continuous,
uniformly locally normal, and uniformly quasi-local in the following sense; there is an increasing, exhaustive sequence  $\{ \Lambda_n \}_{n \geq 1}$ of finite subsets of $\Gamma$ with a family of linear maps $\{\mathcal{K}_t^{(n)} : \mathcal{A}_{\Lambda_n} \to \mathcal{A}_{\Lambda_n}\}_{t\in I}$ for each $n\geq 1$ such that:
\begin{enumerate}
	\item[(i)] For each $n \geq 1$, the family $\{\mathcal{K}_t^{(n)} : \mathcal{A}_{\Lambda_n} \to \mathcal{A}_{\Lambda_n}\}_{t\in I}$ satisfies Assumption~\ref{ass:fv_qlm_cont}.
	\item[(ii)] There is some $p \geq 0$ and a measurable, locally bounded function 
	$B:I \to [0, \infty)$ for which given any $X \in \mathcal{P}_0( \Gamma)$ and $n \geq 1$ large enough so
	that $X \subset \Lambda_n$, 
	\begin{equation} \label{lb_K_td}
	\| \mathcal{K}_t^{(n)}(A) \| \leq B(t) |X|^p \| A \|    \mbox{ for all } A \in \cA_X  \mbox{ and } t \in I \, .
	\end{equation} 
	\item[(iii)] There is some $q \geq 0$, a non-negative, non-increasing function $G$ with
	$G(x) \to 0$ as $x \to \infty$, and a measurable, locally bounded function $C:I \to [0, \infty)$ for which given any sets $X, Y \in \mathcal{P}_0( \Gamma)$ and $n \geq 1$ large enough so that $X \cup Y \subset \Lambda_n$,
	\begin{equation} \label{ql_K_td}
	\| [ \mathcal{K}_t^{(n)}(A), B ] \| \leq C(t) |X|^q \| A \| \| B \|  G(d(X,Y))  \mbox{ for all } A\in \cA_X \, , B\in \cA_Y, \mbox{ and } t \in I \, .
	\end{equation}
	\item[(iv)] There is some $r \geq 0$, a non-negative, non-increasing function $H$ with
	$H(x) \to 0$ as $x \to \infty$, and a measurable, locally bounded function $D:I \to [0, \infty)$ for which 
	given any $X \in \mathcal{P}_0( \Gamma)$ there exists $N \geq 1$ such that for $n \geq N$,
	\begin{equation} \label{local_conv_qlm}
	\| \mathcal{K}_t^{(n)}(A) - \mathcal{K}_t(A) \| \leq D(t) |X|^r \| A \| H(d(X, \Gamma \setminus \Lambda_n)) \quad \mbox{for all } A \in \mathcal{A}_X \mbox{ and } t \in I \, .
	\end{equation}
\end{enumerate}
\end{assumption} 

Before proving the theorem, we make the following comments. First, if $\{ \mathcal{K}_t \}_{t \in I}$ is a family of linear maps  
which satisfies Assumption~\ref{ass:td_qlms}, then for any compact $I_0 \subset I$, the family
$\{ \mathcal{K}_t \}_{t \in I_0}$ is clearly a strongly continuous family of uniformly locally normal
maps in the sense of Definition~\ref{def:uniformly_locally_normal_maps}. Moreover, 
conditions (ii) and (iii) of Assumption~\ref{ass:td_qlms} guarantee that the sequence of finite-volume 
approximates $\{ \mathcal{K}_t^{(n)} \}_{n \geq 1}$ satisfies Assumption~\ref{ass:fv_qlm_dec} with estimates that are uniform in $n$. In Section~\ref{sec:spectral-flow}, an explicit family of weighted integral operators of the type discussed in Example~\ref{ex_wio_sec_4} will be shown to satisfy all conditions of Assumption~\ref{ass:td_qlms}.

Let us now return to the discussion of transformed interactions. Let $I \subset \mathbb{R}$ be an interval, 
$\Phi : \mathcal{P}_0( \Gamma) \times I \to \cA_{\Gamma}^{\rm loc}$ be a strongly continuous
interaction, and $\{ \mathcal{K}_t \}_{t \in I}$ be a family of transformations satisfying 
Assumption~\ref{ass:td_qlms}. Let $\{ \Lambda_n \}_{n \geq 1}$ be the increasing, exhaustive 
sequence of finite subsets of $\Gamma$ whose existence is guaranteed by Assumption~\ref{ass:td_qlms}.
For each $n \geq 1$, we will denote by $H_{\Lambda_n}^{\Phi}(t)$ the finite-volume, time-dependent
Hamiltonian associated to $\Phi$ defined as in (\ref{fv_td_ham_P}). Let us further denote by
\begin{equation} \label{fv_transformed_ham}
\mathcal{K}_t^{(n)}(H_{\Lambda_n}^{\Phi}(t)) = \sum_{X \subset \Lambda_n} \mathcal{K}_t^{(n)}(\Phi(X,t))
\end{equation} 
the corresponding finite-volume transformed Hamiltonian. Our assumptions, 
specifically Assumption~\ref{ass:td_qlms} (i), guarantee that the transformed Hamiltonian in (\ref{fv_transformed_ham})
is still a Hamiltonian in the sense that $t \mapsto \mathcal{K}_t^{(n)}(H_{\Lambda_n}^{\Phi}(t))$ is strongly continuous and 
pointwise self-adjoint. As a result, a finite-volume dynamics may be defined by solving
\begin{equation}
\frac{d}{dt} U_n(t,s) = - i \mathcal{K}_t^{(n)}(H_{\Lambda_n}^{\Phi}(t)) U_n(t,s) \quad \mbox{with} \quad U_n(s,s) = \idty \, 
\end{equation}
and then using the corresponding unitary propagator to declare that
\begin{equation} \label{fv_app_qli_dyn}
\tau_{t,s}^{(n)}(A) = U_n(t,s)^* A U_n(t,s) \quad \mbox{for all } A \in \mathcal{A}_{\Lambda_n} \mbox{ and } t, s \in I \, ,
\end{equation}
is the finite-volume time evolution.

A main goal of this section is to show that, under appropriate decay assumptions, the finite volume dynamics 
in (\ref{fv_app_qli_dyn}) converge to a limiting dynamics as $n \to \infty$. To be more precise, let us 
introduce some further notation. Fix a locally normal product state $\rho$ on $\cA_{\Gamma}$. As in \eqref{fv_td_qli}, with respect to this fixed $\rho$, for any $n \geq 1$, $Z \subset \Lambda_n$ and $t \in I$, set 
\begin{equation} \label{fv_app_td_qli}
\Psi_n(Z,t) = \sum_{m \geq 0} \sum_{\substack{X \subset Z :\\ X(m) \cap \Lambda_n = Z}} \Delta_{X(m)}^{\Lambda_n}( \mathcal{K}_t^{(n)}( \Phi(X,t))) \, .
\end{equation} 
These finite volume maps $\Psi_n$ are constructed in such a way that 
\begin{equation} \label{fv_app_qli_ham}
\mathcal{K}_t^{(n)}(H_{\Lambda_n}^{\Phi}(t)) = \sum_{X \subset \Lambda_n} \mathcal{K}_t^{(n)}(\Phi(X,t)) = \sum_{Z \subset \Lambda_n} \Psi_n(Z,t)
= H_{\Lambda_n}^{\Psi_n}(t) 
\end{equation}
and moreover, as checked in Section~\ref{sec:qli_fvr}, under the assumptions above, each $\Psi_n$ is a strongly continuous
interaction in the sense of Section~\ref{sec:spatialstructure}. With respect to the same locally normal state $\rho$, we can
also define a map $\Psi : \mathcal{P}_0( \Gamma) \times I \to \mathcal{A}_{\Gamma}^{\rm loc}$ by setting 
\begin{equation} \label{def_iv_qli}
\Psi(Z,t) = \sum_{m \geq 0} \sum_{\substack{X \subset Z:\\ X(m) = Z}} \Delta_{X(m)}( \mathcal{K}_t( \Phi(X,t))) .
\end{equation}
Since the family of transformations $\{ \mathcal{K}_t \}$ locally satisfies Definition~\ref{def:uniformly_locally_normal_maps}, it
is clear that Lemma~\ref{lem:s-t_continuity} applies, and hence $\Psi$ is a strongly continuous interaction as well. 

In the remainder of this section, we will show that if the initial interaction $\Phi$ decays sufficiently fast, then the transformed interactions 
$\{ \Psi_n \}_{n \geq 1}$ converge locally in $F$-norm to $\Psi$ in the sense of Definition~\ref{def:lcfnorm}. Moreover, 
our assumptions will allow for an application of Theorem~\ref{thm:exist_iv_dyn} from which we will conclude that
the finite-volume dynamics in (\ref{fv_app_qli_dyn}) converge.
For ease of later reference, let us declare the relevant decay of $\Phi$ as an assumption.

\begin{assumption} \label{ass:suit_td_int} Given a $\nu$-regular metric space $(\Gamma, d)$, and a family of maps $\{\cK_t:\cA_\Gamma \to \cA_\Gamma\}_{t\in I}$ satisfying Assumption~\ref{ass:td_qlms}, we assume $\Phi$ is a strongly continuous interaction such that $\Phi_m \in \mathcal{B}_F(I)$ for $m = \max\{p,q,r\}$ where $p$, $q$, and $r$ are the numbers in Assumption~\ref{ass:td_qlms}.
\end{assumption} 

We can now state the main result of this section, for which it will be useful to review Definition~\ref{def:lcfnorm}.

\begin{thm} \label{thm:qli_conv} Consider a quantum lattice system comprised of a $\nu$-regular metric space $(\Gamma, d)$ and quasi-local algebra $\cA_{\Gamma}$. Let $I \subset \mathbb{R}$ be an interval, and $F$ be an $F$-function on $(\Gamma, d)$. Assume that $\{ \mathcal{K}_t \}_{t \in I}$ is a family of linear maps satisfying  Assumption~\ref{ass:td_qlms}, $\Phi$ is an interaction satisfying Assumption~\ref{ass:suit_td_int}, and $\rho$ is a locally normal product state on $\cA_{\Gamma}$. 
\begin{enumerate}
	\item[(i)] Suppose the quasi-local decay function $G$ from (\ref{ql_K_td}) has a finite $2 \nu +1$ moment and $\tilde{F}$ is an $F$-function on $(\Gamma, d)$ satisfying (\ref{good_F_dec}), then $\Psi \in \mathcal{B}_{\tilde{F}}(I)$.
	\item[(ii)] Suppose there is some $0< \alpha <1$ for which $G^{\alpha}$ has a finite $2 \nu +1$ moment, where $G$ is as in (\ref{ql_K_td}). Suppose also that $\tilde{F}$ is an $F$-function on $(\Gamma, d)$ satisfying (\ref{good_F_dec}) with $G$ replaced by $G^{\alpha}$. Then $\Psi \in \mathcal{B}_{\tilde{F}}(I)$ and $\Psi_n$ converges locally in $F$-norm to $\Psi$ with respect to $\tilde{F}$. 
\end{enumerate}
\end{thm}

Some comments are in order. First, under the assumptions of Theorem~\ref{thm:qli_conv}(i),
the finite-volume interactions $\Psi_n$, as defined in (\ref{fv_app_td_qli}), satisfy the assumptions of
Theorem~\ref{thm:qli_est} and hence the estimate (\ref{rough_qli_est}). In this case, for any
$F$-function $\tilde{F}$ on $(\Gamma, d)$ satisfying (\ref{good_F_dec}), the corresponding finite-volume
dynamics, i.e. the automorphisms $\tau_{t,s}^{(n)}$ defined in (\ref{fv_app_qli_dyn}), satisfy the 
quasi-locality bound proven in Corollary~\ref{cor:fv_ql_dyn_lrb}, see (\ref{qli_cor_lrb}).  
A main point of Theorem~\ref{thm:qli_conv}(i) is that both of these observations extend to
the thermodynamic limit. In fact, the assumptions of Theorem~\ref{thm:qli_conv}(i) also
guarantee that the arguments in Theorem~\ref{thm:qli_est}, and hence an analogue of 
the bound (\ref{rough_qli_est}), also apply to the infinite-volume interaction $\Psi$ as defined in
(\ref{def_iv_qli}). Here we are using that the uniform local convergence in (\ref{local_conv_qlm}) guarantees that both 
the local bound, see (\ref{lb_K_td}), and the quasi-local bound, see (\ref{ql_K_td}), extend to the limiting map $\mathcal{K}_t$, and 
in this case, Lemma~\ref{lem:qlm_approx} applies. Given this, for any $F$-function 
$\tilde{F}$ on $(\Gamma, d)$ satisfying (\ref{good_F_dec}), one concludes that $\Psi \in \mathcal{B}_{\tilde{F}}(I)$.
As a result, we can apply Theorem~\ref{thm:existbd}, where we take the case of trivial on-sites $H_z=0$ for all $z \in \Gamma$.
This then shows that there exists an infinite volume dynamics, which we denote by $\tau_{t,s}$, associated to $\Psi$.
By construction, this infinite-volume dynamics $\tau_{t,s}$ also satisfies Corollary~\ref{cor:fv_ql_dyn_lrb}. 

Theorem~\ref{thm:qli_conv}(ii) implies that, under the slightly stronger decay assumptions, the
finite-volume dynamics $\tau_{t,s}^{(n)}$ converge to the infinite-volume dynamics $\tau_{t,s}$ in
the sense given by Theorem~\ref{thm:exist_iv_dyn}. Since the interactions $\Psi_n$ are constructed using
finite-volume local decompositions, see (\ref{fv_app_td_qli}), they are not finite-volume restrictions of $\Psi$, and so an additional argument is required here. We remark that the decay assumptions in Theorem~\ref{thm:qli_conv}(ii) imply the decay assumed in Theorem~\ref{thm:qli_conv}(i). As a result, the better quasi-locality estimates for the dynamics, which follow as a consequence of the assumptions in Theorem~\ref{thm:qli_conv}(i), may be used generally. 

Next, a careful look at the proof of Theorem~\ref{thm:qli_conv}(i) below shows that we actually only require 
$\Phi_m \in \mathcal{B}_F(I)$ for $m = \max(p,q)$. The proof of Theorem~\ref{thm:qli_conv}(ii) requires
the stronger condition of Assumption~\ref{ass:suit_td_int}, namely $\Phi_m \in \mathcal{B}_F(I)$ for $m = \max(p,q,r)$.

Finally, we note that if the decay function $G$ in (\ref{ql_K_td}) is a weighted $F$-function, the arguments
below can be simplified a bit.

\begin{proof}
The proof of Theorem~\ref{thm:qli_conv}(i) is argued in the paragraphs above.

To prove Theorem~\ref{thm:qli_conv}(ii), first note that as $0 < \alpha <1$, it is clear that finiteness of the $2 \nu +1$ moment of $G^{\alpha}$ implies
finiteness of the $2 \nu +1$ moment of $G$. In this case, the estimate proven in Theorem~\ref{thm:qli_est}, see (\ref{rough_qli_est}), holds for 
each finite-volume interaction $\Psi_n$ as well as for $\Psi$. Since $G$ is non-negative and non-increasing, 
$G(n) \leq G^{1-\alpha}(0) G^{\alpha}(n)$, and therefore, given any $[a,b] \subset I$,
\begin{eqnarray}
\sup_{n \geq 1} \int_a^b \| \Psi_n \|_{\tilde{F}}(r) \, dr & \leq & \int_a^b B(r) \| \Phi_p \|_F(r) \, dr  \nonumber \\
& \mbox{ } & + 4 \kappa \left( \kappa \sum_{n=0}^{\infty} (1+n)^{2 \nu+1} G(n) + \| F \| G^{1- \alpha}(0) \right) \int_a^b C(r) \| \Phi_q \|_F(r) \, dr
\end{eqnarray}
holds for any $F$-function $\tilde{F}$ satisfying the conditions of Theorem~\ref{thm:qli_conv}(ii). 
An analogous bound holds for the infinite volume interaction $\Psi$. 

We need only show that $\Psi_n$ converges locally in $F$-norm to $\Psi$ with respect to $\tilde{F}$, see Definition~\ref{def:lcfnorm}. Let $\Lambda \in \mathcal{P}_0( \Gamma)$ and take $n \geq 1$ large enough so that $\Lambda \subset \Lambda_n$. For any $Z \subset \Lambda$ and each $t \in I$, we estimate
\be \label{psi_dif_ftoiv}
\| \Psi_n(Z, t) - \Psi(Z,t) \| \leq \Sigma_1(Z,t) + \Sigma_2(Z,t) + \Sigma_3(Z,t)
\ee
where, for the terms corresponding to $m=0$ in (\ref{fv_app_td_qli}) and (\ref{def_iv_qli}), we have set 
\be \label{sigma_1}
\Sigma_1(Z,t) = \Vert \Pi^{\Lambda_n}_Z(\cK_t^{(n)}(\Phi(Z,t))) - \Pi_Z(\cK_t(\Phi(Z,t))) \Vert \, ,
\ee
we have collected the bulk of the terms in 
\be \label{sigma_2}
\Sigma_2(Z,t) = \sum_{m \geq 1}\sum_{\substack{X \subset Z:\\ X(m) = Z}} 
\|  \Delta_{X(m)}^{\Lambda_n}( \mathcal{K}_t^{(n)}( \Phi(X,t))) -  \Delta_{X(m)}( \mathcal{K}_t( \Phi(X,t)))\| \, ,
\ee
and finally, we have denoted any boundary terms by   
\be \label{sigma_3}
\Sigma_3(Z,t) = \sum_{m \geq 1} \sum_{\substack{X \subset Z:\\ X(m) \not\subseteq Z, X(m)\cap\Lambda_n = Z}} 
\|  \Delta_{X(m)}^{\Lambda_n}( \mathcal{K}_t^{(n)}( \Phi(X,t)))\| \, .
\ee
	
%
It is now clear that
\be
\sum_{\substack{Z \subset \Lambda_0\\ x,y \in Z}} \| \Psi_n(Z, t) - \Psi(Z,t) \|  \leq \sum_{j=1}^3 
\sum_{\substack{Z \subset \Lambda_0\\ x,y \in Z}} \Sigma_j(Z, t). 
\ee
	
\noindent To complete this proof, we will show that each of the 3 sums above are bounded by a product of:
a) a measurable, locally bounded function of $t$; b) $\tilde{F}(d(x,y))$; and c) a quantity
that vanishes as $n \to \infty$. Given this, it is clear that $\Psi_n$ converges to $\Psi$ locally in $F$-norm.

Consider the first collection of terms. By consistency of the projections, 
\be \label{sig_1_bd}
\Sigma_1(Z,t) \leq \Vert \cK_t^{(n)}(\Phi(Z,t)) - \cK_t(\Phi(Z,t)) \Vert \leq D(t) |Z|^r \| \Phi(Z,t) \| H(d(Z, \Gamma \setminus \Lambda_n)).
\ee
Here, we have also applied Assumption~\ref{ass:td_qlms}(iv). Since $H$ is non-increasing, the bound
	\be
	\sum_{\substack{Z \subset \Lambda_0\\ x,y \in Z}}\Sigma_1(Z, t)  \leq D(t) \| \Phi_r \|_F(t) F(d(x,y)) H(d( \Lambda_0 , \Gamma \setminus \Lambda_n)) 
	\ee
	follows as $\Phi_r \in \mathcal{B}_F(I)$. Since $\tilde{F}$ maximizes $F$, this completes the argument for the first set of terms.

We now consider the bulk terms. An application of Corollary~\ref{cor:DeltaEst}, see (\ref{iv_Ddifbd}), yields
	\bea \label{j=2_est}
	\|  \Delta_{X(m)}^{\Lambda_n}( \mathcal{K}_t^{(n)}( \Phi(X,t))) -  \Delta_{X(m)}( \mathcal{K}_t( \Phi(X,t)))\|  &  \nonumber \\
	& \hspace{-4cm} \leq \min\left\{ 2 \| \cK_t^{(n)}(\Phi(X,t)) - \cK_t(\Phi(X,t)) \|, 8 C(t) |X|^q G(m-1) \right\} \nonumber \\
	& \hspace{-3.5cm} \leq 2 \| \Phi(X,t) \| \min\left\{ D(t) |X|^r H(d(X, \Gamma \setminus \Lambda_n)), 4 C(t) |X|^q  G(m-1) \right\}.
	\eea
To obtain an estimate with explicit decay in both $n$ and $m$, we use the naive bound $\min\{ a, b \} \leq a^{1- \alpha} b^{\alpha}$
which is valid for any $0 < \alpha < 1$ and all non-negative $a$ and $b$. If we denote by 
$d_n = d( \Lambda_0, \Gamma \setminus \Lambda_n)$ and $p' = \max(q,r)$, then the right-hand-side of
(\ref{j=2_est}) may be further estimated by
\begin{equation}
	2^{2 \alpha +1} |X|^{p'} \| \Phi(X,t) \| D(t)^{1- \alpha} H(d_n)^{1- \alpha} C(t)^{\alpha} G(m-1)^{\alpha} \, .
\end{equation}
Using this, we conclude that
	\be
	\sum_{\substack{Z \subset \Lambda_0\\ x,y \in Z}}\Sigma_2(Z, t) \leq  2^{2 \alpha +1} D(t)^{1- \alpha}C(t)^{\alpha}( t) H(d_n)^{1- \alpha} \sum_{m \geq 0} G^{\alpha}(m) \sum_{\stackrel{X \subset \Lambda_0:}{x,y \in X(m+1)}} |X|^{p'} \| \Phi(X, t) \|,
	\ee
and Lemma~\ref{lem:weight_int_dec} applies. Recalling how $\tilde{F}$ is defined, this completes the argument for 
the second collection of terms.

For the final collection of terms, each non-zero contribution must correspond to values of $m \geq 1$ large enough so that $X(m) \cap ( \Gamma \setminus \Lambda_n) \neq \emptyset$. As such, using the notation above, one checks that $m \geq d_n =d(\Lambda_0, \Gamma \setminus \Lambda_n)$.
The bound (\ref{Delta_bd}) applies to each term and thus 
	\be
	\sum_{\substack{Z \subset \Lambda_0\\ x,y \in Z}}\Sigma_3(Z, t)  \leq  4C(t) \sum_{m \geq d_n} G(m-1)
	\sum_{\stackrel{X \subset \Lambda_0:}{x,y \in X(m)}} |X|^q \| \Phi(X, t) \|.
	\ee
Exploiting again that $G = G^{1-\alpha}G^{\alpha}$ and using its non-increasing behavior, we obtain decay in $n$. 
Estimating what remains using Lemma~\ref{lem:weight_int_dec}, we have completed the proof of Theorem~\ref{thm:qli_conv}(ii).
\end{proof}

%
%
%

\subsection{Quasi-locality for the difference of two dynamics} \label{sec:diff_dyn_ql_est}

In this section, we prove a quasi-locality estimate for the difference of two dynamics
as discussed in Example~\ref{ex:diff_dyn_ex} of Section~\ref{sec:some_ex}.

\begin{thm} \label{thm:dyn_diff} Let $(\Gamma, d)$ be a $\nu$-regular metric space. Fix a collection of densely defined, self-adjoint on-site Hamiltonians $\{ H_z \}_{z \in \Gamma}$ and two time-dependent interactions $\Phi, \Psi \in \mathcal{B}_F(I)$. For any $\Lambda \in \mathcal{P}_0( \Gamma)$
	and each $t \in I$, consider the Hamiltonians 
	\begin{equation} \label{2_td_hams_ql_est}
	H_{\Lambda}^{(\Phi)}(t) = \sum_{z \in \Lambda} H_z + \sum_{Z \subset \Lambda} \Phi(Z,t) \quad \mbox{and } \quad H_{\Lambda}^{(\Psi)}(t) = \sum_{z \in \Lambda} H_z + \sum_{Z \subset \Lambda} \Psi(Z,t).
	\end{equation} 
	For any $s,t \in I$, denote by $\tau_{t,s}^{\Lambda}$ and $\alpha_{t,s}^{\Lambda}$ the dynamics corresponding to
	$H_{\Lambda}^{(\Phi)}$ and $H_{\Lambda}^{(\Psi)}$, respectively, and define
	$\mathcal{K}^{\Lambda}_{t,s} : \mathcal{A}_{\Lambda} \to \mathcal{A}_{\Lambda}$ by
	\begin{equation} \label{def_diff_dyn_ql_est}
	\mathcal{K}_{t,s}^{\Lambda}(A) = \tau_{t,s}^{\Lambda}(A) - \alpha_{t,s}^{\Lambda}(A) .
	\end{equation}
	
	If $F$ has a finite $2 \nu$-moment, i.e. $\sum_{n=0}^{\infty} (1+n)^{2 \nu} F(n)<\infty$, then for any $X,Y \subset \Lambda$ 
	\begin{equation} \label{dyn_diff_com_bd}
	\| [ \mathcal{K}_{t,s}^{\Lambda}(A), B ] \| \leq 4 C_F^{-1} I_{t,s}(\Phi - \Psi) \| A \| \| B \| \min \{ C_{t,s}^{(1)} \| F \| |X|, C_{t,s}^{(2)} G(d(X,Y)) \} 
	\end{equation}
	for any $A \in \mathcal{A}_X$, $B \in \mathcal{A}_Y$, and $t,s \in I$. Here one may take
	\begin{equation} \label{diff_dyn_min_const_est}
	C_{t,s}^{(1)} = e^{2 {\rm min}(I_{t,s}( \Phi), I_{t,s}(\Psi))} \quad \mbox{and} \quad C_{t,s}^{(2)} = \left( C_{t,s}^{(1)} - 1 \right) \left( 1 + \frac{5 \| F \|}{C_F} \right) + \kappa^2 
	\end{equation}
	and with $R = d(X,Y)$ we find that
	\begin{eqnarray}
	G(R) & = & \left( 1 + |X(R/2)| \right) G_F(X, \Lambda \setminus X(R/2)) + |X(R/2)| G_F( X(3R/8),Y)  \nonumber \\
	& \mbox{ } & \quad + |X(R/2)| \sum_{n = \lfloor R/4 \rfloor}^{\infty} (1 +n)^{2 \nu}F(n)  .
	\end{eqnarray}
\end{thm} 

\begin{proof} 
To begin, we note that for any $X,Y \in \mathcal{P}_0( \Gamma)$, the naive commutator bound
\begin{equation}
\| [ \mathcal{K}_{t,s}^{\Lambda}(A), B ] \| \leq 2 \| \mathcal{K}_{t,s}^{\Lambda}(A) \| \| B \| 
\end{equation} 
holds for any $A \in \mathcal{A}_X$, $B \in \mathcal{A}_Y$, and $s,t \in I$. In this case, the
local bound proven in Theorem~\ref{thm:continuity}(i), see also (\ref{lb_diff_dyn_ex}), provides a 
rough estimate, which is linear in $I_{t,s}(\Phi - \Psi)$. This explains the first part of the minimum in 
(\ref{dyn_diff_com_bd}). Given this, we need only consider the case of $d(X,Y) >0$. Moreover, 
as is clear from the arguments given in the proof of Theorem~\ref{thm:lrb2}, we need only 
consider the case of trivial on-sites, i.e. $H_z=0$ for all $z \in \Gamma$.
	
Let $X,Y \in \mathcal{P}_0( \Gamma)$ satisfy $d(X,Y) >0$ and, for convenience, 
assume that $s \leq t$. Writing $\mathcal{K}_{t,s}^{\Lambda}(A)$ as in (\ref{dyn_diff}), the bound
	\begin{equation} \label{comm_bd_1st_est}
	\| [\mathcal{K}^{\Lambda}_{t,s}(A), B] \| \leq \sum_{Z \subset \Lambda} \int_s^t \| [ \tau_{r, s}^{\Lambda}( [ \alpha_{t, r}^{\Lambda}(A), \Theta(Z, r)] ), B ] \| \, dr 
	\end{equation}
	follows readily; see also (\ref{dyn_diff_est1}). 
	Here, as in (\ref{diff_ham_terms}), we have denoted by $\Theta$ the time-dependent interaction with terms $\Theta(Z,r) = \Phi(Z,r) - \Psi(Z,r)$.
	
	In the estimates below, we use an argument similar to that of Theorem~\ref{thm:continuity}, see in particular \eqref{Cor7.2App}, to show that the claim holds with $e^{2I_{t,s}(\Psi)}$ replacing $C^{(1)}_{t,s}$ in the definition of $C^{(2)}_{t,s}$. However, by reordering the dynamics in (\ref{def_diff_dyn_ql_est}), (or equivalently, by considering $-\mathcal{K}_{t,s}^{\Lambda}(A)$) we see that the analogue of (\ref{comm_bd_1st_est}) holds with the roles of the dynamics $\tau_{t,s}$ and $\alpha_{t,s}$ interchanged. Since the argument given below applies equally well in this case, it will be clear that $C_{t,s}^{(2)}$ can then be expressed in terms of $C_{t,s}^{(1)}$. We now continue with our estimate of the right-hand-side of (\ref{comm_bd_1st_est}).
	
	To prove (\ref{dyn_diff_com_bd}), we first consider those terms on the right-hand-side of (\ref{comm_bd_1st_est}) corresponding to $Z \subset \Lambda$ with $d(Z,X) > d(X,Y)/2$. Since $\tau_{r,s}^{\Lambda}$ is an automorphism, the commutator bound
	\begin{equation}
	\| [ \tau_{r, s}^{\Lambda}( [ \alpha_{t, r}^{\Lambda}(A), \Theta(Z, r)] ), B ] \| \leq 2 \| B \| \| [ \alpha_{t, r}^{\Lambda}(A), \Theta(Z, r)] \| \, 
	\end{equation}
	is clear. By Theorem~\ref{thm:lrb2}, the dynamics $\alpha_{t,r}^{\Lambda}$ corresponding to 
	$H_{\Lambda}^{(\Psi)}$ satisfies a quasi-locality bound, in particular, we may estimate as in (\ref{unbd_lrb_ex}).
	In this case, an application of Corollary~\ref{app:cor:int_bd} with $R = d(X,Y)/2$ shows that 
	\begin{eqnarray} \label{dif_com_bd1}
	\sum_{\stackrel{Z \subset \Lambda:}{d(Z,X)>R}}  \int_s^t \| [ \tau_{r, s}^{\Lambda}( [ \alpha_{t, r}^{\Lambda}(A), \Theta(Z, r)] ), B ] \| \, dr 
	\leq  \hspace{4cm}  \nonumber \\  \frac{4 \| A \| \| B \|}{C_F} (e^{2 I_{t,s}(\Psi)} -1 ) I_{t,s}( \Theta) G_F(X, \Lambda \setminus X(R)) .
	\end{eqnarray}

	We need only estimate those terms on the right-hand-side of (\ref{comm_bd_1st_est}) corresponding to $Z \subset \Lambda$ with $d(Z,X) \leq d(X,Y)/2$.
	For these terms, we first make a strictly local approximation of the inner-most dynamics, i.e. $\alpha_{t,r}^{\Lambda}$. 
	Given the quasi-locality estimate (\ref{unbd_lrb_ex}) for $\alpha_{t,r}^{\Lambda}$, an application of  
	Lemma~\ref{lem:qlm_approx} shows that
	\begin{equation}
	\| \alpha_{t,r}^{\Lambda}(A) - A_R(r) \| \leq \frac{4 \| A \|}{C_F} (e^{2 I_{t,s}(\Psi)} - 1) G_F(X, \Lambda \setminus X(R))
	\end{equation}
	where we have set $A_R(r) = \Pi_{X(R)}^{\Lambda}( \alpha_{t,r}^{\Lambda}(A))$ and found an upper bound independent of $s \leq r \leq t$.
	
	In this case, for any $s \leq r \leq t$,
	\begin{eqnarray} \label{insert_loc}
	\| [ \tau_{r, s}^{\Lambda}( [ \alpha_{t, r}^{\Lambda}(A), \Theta(Z, r)] ), B ] \| & \leq & \| [ \tau_{r, s}^{\Lambda}( [ A_R(r), \Theta(Z, r)] ), B ] \| + \nonumber \\
	& \mbox{ } & \quad + \| [ \tau_{r, s}^{\Lambda}( [ \alpha_{t, r}^{\Lambda}(A) -  A_R(r), \Theta(Z, r)] ), B ] \|.
	\end{eqnarray}
	For the second term on the right-hand-side of (\ref{insert_loc}), it is clear that
	\begin{equation}
	\| [ \tau_{r, s}^{\Lambda}( [ \alpha_{t, r}^{\Lambda}(A) -  A_R(r), \Theta(Z, r)] ), B ] \| \leq \frac{16 \| A \| \| B \|}{C_F} \| \Theta(Z,r) \| ( e^{2 I_{t,s}(\Psi)} - 1) G_F(X, \Lambda \setminus X(R))
	\end{equation} 
	and therefore, the bound 
	\begin{eqnarray}  \label{dif_com_bd2}
	\sum_{\stackrel{Z \subset \Lambda:}{d(Z,X) \leq R}}  \int_s^t  \| [ \tau_{r, s}^{\Lambda}( [ \alpha_{t, \cdot}^{\Lambda}(A) -  A_R(r), \Theta(Z, r)] ), B ] \| \, dr 
	\leq  \hspace{4cm}  \nonumber \\  \frac{16 \| A \| \| B \| \| F \| }{C_F^2} (e^{2 I_{t,s}(\Psi)} -1 ) I_{t,s}( \Theta) |X(R)| G_F(X, \Lambda \setminus X(R)) 
	\end{eqnarray}
	follows from an application of Proposition~\ref{app:prop:int_bd_dist}, see (\ref{RclosetoX}).
	
	With the remaining terms, i.e. those corresponding to the first term on the right-hand-side of (\ref{insert_loc}), we find it useful to further
	sub-divide the sets $Z$ into those of {\it relative} `large' and `small' diameter. More precisely, we will estimate using
	\begin{equation} \label{splitsum}
	\sum_{\stackrel{Z \subset \Lambda:}{d(Z,X) \leq R}} \cdot \leq \sum_{x \in X(R)} \sum_{\stackrel{Z \subset \Lambda: x \in Z}{ {\rm diam}(Z) \leq R/2}} \cdot  + \sum_{x \in X(R)}\sum_{\stackrel{Z \subset \Lambda: x\in Z}{ {\rm diam}(Z) > R/2}} \cdot 
	\end{equation} 
	For the terms with `small' diameter, we apply the quasi-locality estimate for the outer
	dynamics $\tau_{t,s}^{\Lambda}$, again we use the form found in (\ref{unbd_lrb_ex}), to obtain 
	\begin{equation}
	\| [ \tau_{r,s}^{\Lambda}([A_R(r), \Theta(Z,r)]), B ] \| \leq \frac{4 \| A \| \| B \|}{C_F} \| \Theta(Z,r) \| (e^{2 I_{t,s}( \Phi)} -1) G_F(X(R) \cup Z, Y) .
	\end{equation}
	Clearly, $X(R) \cup Z \subset X(3R/2)$ for any $Z$ with $Z \cap X(R) \neq \emptyset$ and ${\rm diam}(Z) \leq R/2$. 
	In this case, 
	\begin{eqnarray}  \label{dif_com_bd3}
	\sum_{x \in X(R)} \sum_{\stackrel{Z \subset \Lambda: x \in Z}{ {\rm diam}(Z) \leq R/2}}  \int_s^t \| [ \tau_{r,s}^{\Lambda}([A_R(r), \Theta(Z,r)]), B ] \| \, dr 
	\leq  \hspace{4cm}  \nonumber \\  \frac{4 \| A \| \| B \| \| F \| }{C_F^2} (e^{2 I_{t,s}(\Phi)} -1 ) I_{t,s}( \Theta) |X(R)| G_F(X(3R/2), Y) 
	\end{eqnarray}
	follows immediately from the arguments in Proposition~\ref{app:prop:int_bd_diam}, see (\ref{smalldiambd}).
	
	The remaining terms have relatively large diameter, and so we make the naive estimate
	\begin{equation}
	\| [ \tau_{r,s}^{\Lambda}([A_R(r), \Theta(Z,r)]), B ] \| \leq 4 \| A \| \| B \| \| \Theta(Z,r) \| .
	\end{equation}
	As a consequence,
	\begin{eqnarray}  \label{dif_com_bd4}
	\sum_{x \in X(R)} \sum_{\stackrel{Z \subset \Lambda: x \in Z}{ {\rm diam}(Z) > R/2}}  \int_s^t \| [ \tau_{r,s}^{\Lambda}([A_R(r), \Theta(Z,r)]), B ] \| \, dr 
	\leq  \hspace{4cm}  \nonumber \\  \frac{4 \kappa^2 \| A \| \| B \|}{C_F} I_{t,s}( \Theta) |X(R)|  [M_{2 \nu}(F)](R/2) 
	\end{eqnarray}
	follows from Proposition~\ref{app:prop:int_bd_diam}, see (\ref{largediambd}).
	
	Collecting the estimates in (\ref{dif_com_bd1}), (\ref{dif_com_bd2}), (\ref{dif_com_bd3}), and (\ref{dif_com_bd4}), we find (\ref{dyn_diff_com_bd}) as claimed. 
\end{proof}

%
%

\section{The spectral flow}\label{sec:spectral-flow}

In this section, we consider a family of finite volume quantum lattice Hamiltonians $H_\Lambda(s)$ acting on a Hilbert space $\cH_\Lambda$ that depend smoothly on a parameter $s\in [0,1]$. We assume that the spectrum of $H_\Lambda(s)$ can be decomposed into two non-empty disjoint sets, i.e. $\spec(H_\Lambda(s)) = \Sigma_1(s)\cup \Sigma_2(s)$, where $\Sigma_1(s)$ is bounded, and the distance between $\Sigma_1(s)$ and $\Sigma_2(s)$ is greater than a positive value independent of $s$. The main goal of this section is to show that if the interaction defining $H_\Lambda(s)$ is smooth and decays sufficiently fasts, then we can use the theory of Section~\ref{sec:quasilocal-maps} to construct a \emph{quasi-local} automorphism $\alpha_s:\cA_\Lambda \to \cA_\Lambda$, which we call the \emph{spectral flow}, that maps the spectral projection of $H_\Lambda(s)$ onto $\Sigma_1(s)$ back to the spectral projection of $H_\Lambda(0)$ onto $\Sigma_1(0)$. In Section~\ref{sec:equivalence_of_phases} we use the spectral flow to discuss the classification of gapped ground state phases.  A second important application concerns models with a spectral gap above their ground states, for which $s$ parameterizes a perturbation of the system; this is the main topic we analyze in \cite{QLBII}. While both these applications are for ground states, the methods we introduce here are more general and work equally well for isolated bounded subsets anywhere in the spectrum. 

Denoting by $P(s)$ the spectral projection of $H_\Lambda(s)$ onto $\Sigma_1(s)$, the existence of an automorphism $\alpha_s$ satisfying 
\be\label{sf_property}
\alpha_s(P(s))= P(0)
\ee
is well-known. As shown by Kato in \cite{kato:1995}, under certain conditions which guarantee the smoothness of $P(s)$, the unique strong solution of
\begin{equation}
\frac{d}{ds} U^K(s) = - i [ P'(s), P(s) ] U^K(s), \quad U^K(0) = \idty
\end{equation}
is unitary and satisfies
$$
\alpha_s^K(P(s)):=U^K(s)^{*}P(s)U^K(s) = P(0).
$$ 
The automorphism studied by Kato was for a family of Hamiltonians defined on a general Hilbert space $\cH$, and so his results do not take into account the locality structure of a quantum lattice system. As a result, the automorphism induced by $U^K(s)$ is not obviously quasi-local. Hastings and Wen were the first to introduce a technique for constructing an automorphism on a quantum lattice system that both satisfies \eqref{sf_property} and is quasi-local \cite{hastings:2005}. In that work, they referred to the quasi-local automorphism as the quasi-adiabatic evolution (or continuation). It is this approach that we follow in this section. Neither name, spectral flow or quasi-adiabatic continuation, accurately and unambiguously captures the essence of this quasi-local automorphism. It suffices to say that it is a unitary dynamics with useful properties. In other works, Hastings introduced novel ways to combine particular instances of the spectral flow with quasi-locality properties of quantum lattice systems, most notably in \cite{hastings:2004}. This work inspired a string of new results in the theory of quantum lattice systems, and so it seems appropriate to refer to the generator of this spectral flow as the {\em Hastings generator}.  

\subsection{Set up and main results} \label{sec:sf_setup} We first consider a family of parameter dependent Hamiltonians on a general complex Hilbert space $\cH$ and later return to apply our results to the setting of quantum lattice systems. Specifically, we consider operators that depend on a parameter $s\in[0,1]$, and we note that the choice of interval $[0,1]$ is a matter of convenience. We begin with the following definition.

\begin{defn}\label{def:stg_diff}
Let $\cH$ be a complex Hilbert space. We say that a map $\Phi:[0,1]\to \cB(\cH)$ is \emph{strongly $C^1$} if $\Phi(s)$ is strongly differentiable for all $s\in[0,1]$, and the derivative $\Phi':[0,1]\to \cB(\cH)$ is continuous in the strong operator topology. 
\end{defn}

We consider a family of parameter dependent Hamiltonians of the form
\begin{equation}\label{sf_fv_ham}
	H(s) = H + \Phi(s), \quad s\in[0,1]
\end{equation}
where $H$ is a self-adjoint operator acting on some dense domain $\mathcal{D} \subset \mathcal{H}$, and $\Phi : [0,1] \to \mathcal{B}( \mathcal{H})$ is strongly $C^1$ and pointwise self-adjoint, i.e. $\Phi(s)^* = \Phi(s)$ for all $0 \leq s \leq 1$. Since $\Phi$ is bounded and self-adjoint, for each $s\in[0,1]$ it is clear that $H(s)$ corresponds to a well-defined, self-adjoint operator with the same dense domain $\mathcal{D} \subset \mathcal{H}$. We will refer to $\{ H(s) \}_{s \in [0,1]}$ as a smooth family of Hamiltonians on $\mathcal{H}$. 

For each $0 \leq s \leq 1$, let us denote by $\tau_t^{(s)}$ the Heisenberg dynamics corresponding to $H(s)$, i.e., 
\begin{equation} \label{sf_ham_dyn}
\tau_t^{(s)}(A) = e^{itH(s)}Ae^{-itH(s)} \quad \mbox{for all } A \in \mathcal{B}( \mathcal{H}) \mbox{ and } t \in \mathbb{R} \, .
\end{equation}
It is clear that for each $s$, this dynamics is a one-parameter family of automorphisms of $\mathcal{B}( \mathcal{H})$, and so for any real-valued function $W \in L^1( \mathbb{R})$, the mapping $D:[0,1] \to \mathcal{B}( \mathcal{H})$ given by
\begin{equation} \label{def_ds_gen}
D(s) = \int_{\mathbb{R}} \tau_t^{(s)} (\Phi'(s) ) \, W(t) \, dt
\end{equation}
is well-defined,  pointwise self-adjoint, bounded, and continuous in the strong operator topology. In this case, the methods of Section~\ref{sec:Dyson} show that the unique strong solution of
\begin{equation}\label{sf_unitary}
\frac{d}{ds}U(s) = -i D(s) U(s) \quad \mbox{with} \quad U(0) = \idtyty
\end{equation}
is well-defined, unitary, and norm-continuous. In terms of these unitaries, we can define an automorphism $\alpha_s: \cB(\cH)\to \cB(\cH)$ for each $0\leq s \leq 1$ by
\begin{equation} \label{def_alpha_gen}
\alpha_s(A) = U(s)^* A U(s).
\end{equation}
Note that here, $D(s)$, $U(s)$, and $\alpha_s$ all depend on the choice of weight function $W\in L^1(\bR)$. We will use this construction to define the spectral flow of interest. 

As described in the introduction,  we will consider the situation that the smooth family of Hamiltonians defined as in (\ref{sf_fv_ham}) has a spectrum which can be decomposed into two disjoint, non-empty sets.
This decomposition will depend on $0 \leq s \leq 1$, and we are particularly interested in cases where
the gap between these sets has a uniform lower bound. To be precise, some additional notation will be 
convenient: for any two non-empty sets $X, Y \subset \mathbb{R}$, denote by $d(X,Y)$ the distance between these sets:
\begin{equation} \label{set_dist}
d(X,Y) := \inf\{ |x-y| : x \in X \mbox{ and } y \in Y \}.
\end{equation}
\begin{assumption}\label{sf_ass:fv_uni_gap} For each $0\leq s \leq 1$, the spectrum of $H(s)$ can be partitioned into two disjoint sets $\Sigma_1(s)$ and $\Sigma_2(s)$, i.e. $\spec(H(s)) = \Sigma_1(s) \cup \Sigma_2(s)$, such that
	\begin{equation} \label{pt_wise_gap}
	\gamma' := \inf_{0\leq s \leq 1}d(\Sigma_1(s), \Sigma_2(s)) >0,
	\end{equation}
	and, moreover, there are compact intervals $I(s)$ with end-points depending smoothly on $s$, for which $\Sigma_1(s)\subset I(s)\subset (\Rl\setminus\Sigma_2(s))$ and $\mu(s) : = d(I(s), \Sigma_2(s))$ satisfies $\mu:=\inf_{0\leq s \leq 1}\mu(s)>0.$
\end{assumption}

In many concrete examples one can pick the interval $I(s)$ as the smallest interval containing $\Sigma_1(s)$, and in that case $\mu=\gamma'$. 

Given a smooth family of Hamiltonians $H(s)$ that satisfy Assumption~\ref{sf_ass:fv_uni_gap}, the spectral flow of interest depends on the choice of an auxiliary parameter $0<\gamma\leq\gamma'$. For any such $\gamma$ and $0\leq s \leq 1$, we define the spectral flow $\alpha_s^\gamma: \cB(\cH)\to \cB(\cH)$ by
\begin{equation} \label{def_alpha}
\alpha_s^{\gamma}(A) = U^{\gamma}(s)^* A U^{\gamma}(s)
\end{equation}
where $U^\gamma(s)$ is the unitary solution to \eqref{sf_unitary} for the self-adjoint operator
\begin{equation} \label{def_ds}
D^{\gamma}(s) = \int_{\mathbb{R}} \tau_t^{(s)} (\Phi'(s) ) \, W_{\gamma}(t) \, dt
\end{equation}
defined by a well-chosen $\gamma$-dependent, real-valued weight function $W_\gamma\in L^1(\bR)$. In Section~\ref{sec:ex_weight} we state the necessary conditions for choosing $W_\gamma$ and give an explicit example a weight function that satisfies these conditions. In fact, we will be able to define weight functions $W_\gamma$ for any $\gamma>0$. However, to obtain the spectral flow property, i.e. \eqref{sf_property}, one must choose $W_\gamma$ with $0<\gamma\leq\gamma'$. As discussed in the introduction, the Hamiltonian $D^\gamma(s)$ will be called a \emph{Hastings generator}. We can now state the first main result of this section.

\begin{thm} \label{thm:sf}
Let $\mathcal{H}$ be a complex Hilbert space, and $H(s)$ be a smooth family of Hamiltonians as in \eqref{sf_fv_ham} satisfying Assumption~\ref{sf_ass:fv_uni_gap}. For any $0<\gamma < \gamma'$, there is a real-valued function $W_\gamma\in L^1(\bR)$ such that the automorphism $\alpha_s^\gamma:\cB(\cH)\to\cB(\cH)$ defined as in (\ref{def_alpha}) satisfies
\begin{equation} \label{sf:followsp}
\alpha_s^\gamma(P(s)) = P(0)  \, 
\end{equation}
for all $0 \leq s \leq 1$. Here, $P(s)$ denotes the spectral projection associated to $H(s)$ onto the isolated part of the spectrum $\Sigma_1(s)$.
\end{thm}

In the context of a quantum lattice system, the novel feature of the Hastings generator is that it generates a \emph{quasi-local} family of automorphisms. This is the second main result of this section. Recall that given a quantum lattice system $(\Gamma, d)$ and $\cA_\Gamma$, the local Hamiltonians for a strongly continuous interaction $\Phi:\cP_0(\Gamma)\times[0,1] \to \cA_\Gamma^{\rm loc}$ are given by
\be\label{td_loc_hams}
H_\Lambda(s) = \sum_{X\subseteq \Lambda}\Phi(X,s), \quad \text{for all} \quad \Lambda \in \cP_0(\Gamma).
\ee
Note that if $\Phi(X,s)$ is strongly $C^1$ for all $X\subset \Lambda$ in the sense of Definition~\ref{def:stg_diff}, then the Hamiltonian $H_\Lambda(s)$ is also strongly $C^1$. In this case, for every $\Lambda \in \cP_0(\Gamma)$ we may define the finite volume Hastings generator by
\be \label{def_ds_qls}
D_\Lambda^\gamma(s) = \int_{\mathbb{R}} \tau_t^{(s)} (H_\Lambda'(s)) \, W_{\gamma}(t) \, dt,
\ee
where for each $s\in[0,1]$, $\tau_t^{(s)}$ is the Heisenberg dynamics associated to $H_\Lambda(s)$. We may now state the quasi-locality result.
\begin{thm}\label{thm:sf_ql_all}
	Consider a quantum lattice system comprised of $\nu$-regular metric space $(\Gamma, d)$ and quasi-local algebra $\cA_{\Gamma}$. Suppose that $\Phi	\in \cB_F([0,1])$ for an $F$-function of the form 
	\be
	F(r) = e^{-ar^\theta}(1+r)^{-p} \quad \text{for some} \quad a>0, \; 0<\theta\leq 1 \; \text{and} \; p>\nu+1.
	\ee 
	If $\Phi(X,s)$ is strongly $C^1$ for all $X\in \cP_0(\Gamma)$ and $\Phi_1' \in \cB_F([0,1])$ where $\Phi_1'(X,s) = |X|\Phi'(X,s)$, then for any $\gamma >0$ there is an $F$-function, $F^{(\gamma)}$, such that for any $\Lambda \in \cP_0(\Gamma)$
	\be\label{HG_Int}
	D_\Lambda^\gamma(s) = \sum_{X\subseteq \Lambda}\Psi_\Lambda(X,s)
	\ee
	for a strongly-continuous interaction $\Psi_\Lambda\in \cB_{F^{(\gamma)}}([0,1])$. Moreover, there is an interaction $\Psi\in \cB_{F^{(\gamma)}}([0,1])$ such that  $\Psi_{\Lambda_n}$ converges locally in $F$-norm to $\Psi$ with respect to $F^{(\gamma)}$ for any sequence of increasing and absorbing finite volumes $\Lambda_n \uparrow \Gamma$.
\end{thm}

We give some context for this result. Suppose that $\Phi\in\cB_{F}([0,1])$ is such that the local Hamiltonians $H_\Lambda(s)$ are strongly $C^1$. Recall that for any $\gamma >0$, the Hastings generator $D_\Lambda^\gamma(s)$, which is defined in terms of $H_\Lambda'(s)$ (see \eqref{def_ds_qls}), is strongly continuous and self-adjoint. The automorphism $\alpha_s^{\gamma, \Lambda}$ defined as in \eqref{def_alpha} can then be recognized as the Heisenberg dynamics associated to $D_\Lambda(s)$. If Theorem~\ref{thm:sf_ql_all} holds, then $D_\Lambda(s)$ is itself a local Hamiltonian associated to a strongly-continuous interaction $\Psi_\Lambda\in \cB_{F^{(\gamma)}}([0,1])$. Applying the Lieb-Robinson bound, i.e. Theorem~\ref{thm:lrb}, to $\alpha_s^{\gamma, \Lambda}$ shows that the spectral flow is quasi-local as claimed. In the proof of Theorem~\ref{thm:sf_ql_all}, we show that the norm $\|\Psi_\Lambda\|_{F^{(\gamma)}}$ is bounded from above by a constant independent of $\Lambda$, from which local $F$-norm convergence will follow. The interaction $\Psi$ then defines an infinite volume spectral flow automorphism $\alpha_s^\gamma: \cA_\Gamma \to \cA_\Gamma$ that is also quasi-local.

Note that we do not require Assumption~\ref{sf_ass:fv_uni_gap} for Theorem~\ref{thm:sf_ql_all}, and in particular, the quasi-locality result holds where the spectrum of $H_\Lambda(s)$ is or is not gapped. If, however, $H_\Lambda(s)$ satisfies Assumption~\ref{sf_ass:fv_uni_gap} with gap $\gamma_\Lambda' >0$, then the finite-volume automorphisms $\alpha_s^{\Lambda,\gamma}$ generated by $D_\Lambda^\gamma(s)$ for any $0<\gamma \leq\gamma_\Lambda'$ will both be quasi-local and satisfy (\ref{sf:followsp}). In applications to stability, one is interested in the situation that there is some sequence of finite volumes $(\Lambda_n)_{n\geq 1}$ for which both Theorem~\ref{thm:sf} and Theorem~\ref{thm:sf_ql_all} hold simultaneously and that the gaps  $\gamma_{\Lambda_n}'$ as in \eqref{pt_wise_gap} are uniformly bounded from below by a positive constant independent of $n$.

In what follows, we will typically work with a Hastings generator and spectral flow automorphism that depend on a fixed value of $\gamma$. As such, we will often suppress the dependence of $\gamma$ from our notation.

The remainder of the section is organized as follows. In Section~\ref{sec:ex_weight} we define the explicit weight function $W_\gamma$ used in our results and prove some basic decay estimates on this function. The reader can skip these details on first reading. Recall that the definition of the Hastings generator is given in terms of a specific weighted integral operator. In Section~\ref{sec:w_int_op} we define several general weighted integral operators in terms of appropriate $L^1$ functions and prove some useful properties. We use the results from this section to give the proof of Theorem~\ref{thm:sf} in Section~\ref{subsec:sp_flow}. We consider quantum lattice systems in Section~\ref{sec:ql_sf} where we show that, in this context, the weighted integral operators introduced in Section~\ref{sec:w_int_op} are quasi-local when defined using the weight functions from Section~\ref{sec:ex_weight}. We then use these results to prove Theorem~\ref{thm:sf_ql_all} (which is restated as Theorem~\ref{thm:sf_ints_conv}). We end the section by showing that there is a well-defined spectral flow automorphism in the thermodynamic limit when the conditions of Theorem~\ref{thm:sf_ql_all} hold.

%
%
%

\subsection{An explicit weight} \label{sec:ex_weight}

To write down the generator of the spectral flow dynamics, see (\ref{def_ds}) above, requires a 
{\it weight function} with certain properties. In Section \ref{sec:w_int_op} we will define a
class of transformations on the algebra of observables of the form
$$
I(A) = \int_\Rl w(t) \tau_t(A)\, dt
$$
where $w\in L^1(\Rl)$. In the following section we will make increasingly detailed assumptions
of $w$ in order to prove useful properties of the map $I$. At some point, it becomes more 
efficient to work with a specific family of functions $w$ for which the assumptions
hold. Having such a family of functions will make it possible to state explicit decay estimates that are useful for applications. As such, in this section we introduce this family of functions, for which interesting properties were already investigated in \cite{ingham:1934}, and prove some basic estimates; these will be particularly relevant in Section~\ref{sec:ql_sp_flow}. It will be clear that other functions can be used to derive similar results. The details of this section can be skipped on first reading. Its main importance is to demonstrate the existence of functions with all the desired properties.

Consider the sequence $(a_n)_{n\geq 1}$ defined by
\begin{equation} \label{seqsum}
a_n = \frac{a_1}{n \ln(n)^2} \quad \mbox{for }n\geq 2 \mbox{  and } \sum_{n=1}^{\infty} a_n = \frac{1}{2} \, .
\end{equation}
In terms of this sequence, define a function $w : \mathbb{R} \to \mathbb{R}$ by setting
\begin{equation} \label{def_w}
w(0) = c \quad \mbox{and} \quad w(t) =c \prod_{n=1}^{\infty} \left( \frac{\sin(a_n t)}{a_n t} \right)^2 \quad \mbox{if } t \neq 0 \,
\end{equation}
where $c>0$ is chosen so that
\begin{equation} \label{w-L1-norm}
\int_{\mathbb{R}} w(t) \, dt = 1 \, .
\end{equation} 
It follows from Lemma~\ref{lem:wint} below that $w\in L^1(\bR)$ and so this constant is well-defined.

It is clear that $w$ is non-negative and even. Moreover, if we denote by
$\hat{w}_{\gamma} : \mathbb{R} \to \mathbb{R}$ the unitary Fourier transform of $w$, i.e. for each $k \in \mathbb{R}$
\begin{equation}
\hat{w}_{\gamma}(k) = \frac{1}{\sqrt{2 \pi}} \int_{\mathbb{R}} e^{-ikt} w_{\gamma}(t) \, dt \, ,
\end{equation}
then it is easy to check, see e.g \cite{bachmann:2012}, that ${\rm supp}(\hat{w}) \subset [-1,1]$. 
The following lemma provides a useful estimate on $w$. 
%
%
\begin{lem} \label{lem:wint} Let $a>0$ and $p \geq 0$ be an integer. For any $x >1$ with $\ln(x) \geq \max \left\{ 9, \sqrt{ \frac{p+1}{a}} \right\}$, one has that
\begin{equation} \label{calc2}
\int_x^{\infty} \left( \frac{t}{ \ln(t)^2} \right)^p  e^{- \frac{at}{\ln(t)^2} } \, dt \leq \frac{9(p+2)}{7a} \left( \frac{x}{ \ln(x)^2} \right)^{p+1}  e^{- \frac{ax}{\ln(x)^2} } .
\end{equation}
As a consequence, there is a number $\eta \in (2/7, 1)$ for which if $x \geq e^9$, then
\begin{equation} \label{wintbd}
\int_x^{\infty} w(t) \, dt \leq \frac{27}{14} c e^4 \left( \frac{x}{\ln(x)^2} \right)^2 e^{- \frac{ \eta x}{ \ln(x)^2}}.
\end{equation}
\end{lem}
\begin{proof}
To see (\ref{calc2}), consider the change of variables: $u = at/ \ln(t)^2$. Clearly,
\begin{equation}
\frac{du}{dt} = \left( 1- \frac{2}{\ln(t)} \right) \frac{a}{\ln(t)^2} \, .
\end{equation} 
It will be convenient to take $x$ large enough so that $\ln(x)^4 \leq x$. As one readily checks, this is
the case if $x \geq e^9$; however, we note that this lower bound is not optimal. In any case, using this one also has that
\begin{equation}
\frac{du}{dt} \geq \frac{7a^2}{9u} \quad \mbox{if } t \geq e^9 \, .
\end{equation}
Consequently,
\begin{equation} \label{intest}
\int_x^{\infty} \left( \frac{t}{ \ln(t)^2} \right)^p e^{- \frac{at}{\ln(t)^2}} \, dt \leq \frac{9}{7a^{p+2}} \int_{u(x)}^{\infty} u^{p+1} e^{-u} \, du .
\end{equation}
For integers $p \geq 0$, the above integral may be bounded using
\begin{equation} \label{polyexp}
\int_k^{\infty} u^{p+1} e^{-u} \, du =  (p+1)! e^{-k} \sum_{n=0}^{p+1} \frac{k^n}{n!} \leq (p+2) k^{p+1}e^{-k}
\end{equation}
where the final inequality is valid whenever $k \geq p+1$. With the further constraint that $\ln(x) \geq \sqrt{\frac{p+1}{a}}$, the bound
\begin{equation}
p+1 \leq a \ln(x)^2 \leq u(x) 
\end{equation}
follows, using again that $\ln(x)^4 \leq x$. Now (\ref{calc2}) follows from (\ref{intest}) and (\ref{polyexp}). 

We now estimate $w$ to establish \eqref{wintbd}. Note that for any $N \geq 1$ and $t \neq 0$,
\begin{equation}
w(t) \leq c \prod_{n=1}^N \left( \frac{\sin(a_nt)}{a_n t} \right)^2 \leq \frac{c}{(a_1 t)^{2N}} \prod_{n=2}^N n \ln(n)^2 \leq c \left( \frac{\ln(N)^2}{a_1 t} \right)^{2N} (N!)^2.
\end{equation}
Using Stirling's formula, i.e. $N! \leq e N^{N + \frac{1}{2}}e^{-N}$, and choosing $N = \lfloor \frac{a_1 t}{ \ln(t)^2} \rfloor$, we find that
\begin{eqnarray}
w(t) & \leq & c e^2 \left( N \cdot \frac{\ln(N)^2}{a_1 t} \right)^{2N} N e^{-2N} \nonumber \\
& \leq & c e^2 \left(  \frac{1}{\ln(t)^2}  \cdot \ln \left( \frac{a_1 t}{ \ln(t)^2} \right)^2 \right)^{2N}  \frac{a_1 t}{ \ln(t)^2} \cdot e^{-2( \frac{a_1 t}{ \ln(t)^2} -1)} \nonumber \\
& \leq & ce^4  \frac{a_1 t}{ \ln(t)^2} \cdot e^{- \frac{ 2 a_1 t}{ \ln(t)^2}}
\end{eqnarray}
where, for the final inequality above, we used that $t$ is large enough so that both $1 \leq \ln(t)^2$ and $\ln(a_1t)^2 \leq \ln(t)^2$ hold.
Since (\ref{seqsum}) implies that $a_1 < 1/2$, both inequalities are true if $t \geq e$. As 
\begin{equation}
1 + \sum_{n=2}^{\infty} \frac{1}{n \ln(n)^2} \leq 1 + \frac{1}{2 \ln(2)^2} + \int_2^{\infty} \frac{1}{t \ln(t)^2} \, dt \leq 3.5 \, ,
\end{equation}
it is clear that $a_1 > 1/7$. Now, setting $\eta = 2 a_1$, we have found that
\begin{equation}
w(t) \leq \frac{c \eta e^4}{2} \cdot \frac{t}{ \ln(t)^2} \cdot e^{- \frac{\eta t}{ \ln(t)^2}} \quad \mbox{for all } t \geq e\, .
\end{equation}
Now (\ref{wintbd}) follows from (\ref{calc2}).
\end{proof}

For our estimates on the spectral flow, it will be convenient to rescale this weight function $w$. For any $\gamma >0$, define $w_{\gamma} : \mathbb{R} \to \mathbb{R}$
by setting 
\begin{equation} \label{def_wg}
w_{\gamma}(t) = \gamma w( \gamma t) \, .
\end{equation}
It is clear that each such $w_{\gamma}$ is non-negative, even, $L^1$-normalized, and moreover, 
\begin{equation} \label{wsupp}
{\rm supp}( \hat{w}_{\gamma}) \subset [ - \gamma, \gamma] \, .
\end{equation}
The function $W_{\gamma} : \mathbb{R} \to \mathbb{R}$ given by
\begin{equation} \label{def:Wgam1}
W_{\gamma}(x) = - \int_{- \infty}^x w_{\gamma}(t) \, dt + H(x) \quad \mbox{for } x \in \mathbb{R},
\end{equation}
where $H(x)$ is the Heavyside function (for clarity, we take $H(0) =1$) will also play a key role below. This may be re-written as 
\begin{equation} \label{def:Wgam2}
W_{\gamma}(0) = \frac{1}{2} \quad \mbox{and} \quad W_{\gamma}(x) = {\rm sgn}(x) \cdot \int_{|x|}^{\infty} w_{\gamma}(t) \, dt \quad \mbox{for } x \neq 0.
\end{equation}
Thus $W_{\gamma}$ is odd, and since $w_{\gamma}$ is normalized and even, one has that $\| W_{\gamma} \|_{\infty} \leq 1/2$.
In fact, a short calculation shows that
\begin{equation}
\| W_{\gamma} \|_1 = \int_{\mathbb{R}} |W_{\gamma}(t)| \, dt = \frac{2}{\gamma} \int_0^{\infty} t w(t) \, dt \, .
\end{equation}
It is clear from (\ref{def:Wgam1}) that the distributional derivative of $W_{\gamma}$ is
\begin{equation} \label{W-g-der}
\frac{d}{dx} W_{\gamma} (x) = - w_{\gamma}(x) + \delta_0(x) 
\end{equation}
and thus its (unitary) Fourier transform satisfies
\begin{equation} \label{FTW}
\hat{W}_{\gamma}(0) = 0 \quad \mbox{and} \quad (i k) \hat{W}_{\gamma}(k) =  - \hat{w}_{\gamma}(k) + \frac{1}{\sqrt{2 \pi}}  \quad \mbox{for all } k \neq 0.
\end{equation}
In particular, we have 
\be  \label{FTWbis}
\hat W_\gamma(k)= \frac{-i}{\sqrt{2\pi}k}, \mbox{ for } k\notin(-\gamma,\gamma).
\ee
As we will see in subsequent sections, a ``well-chosen" weight function $W_\gamma$ for defining the spectral flow as described following \eqref{def_ds} is one which satisfies \eqref{FTWbis} and has a decay estimate that is at least stretched exponential, similar to the next result.

\begin{cor} \label{cor:weight_dec} Let $\gamma >0$. If $ \gamma x \geq e^9$, then
\begin{equation} \label{w_g_dec}
\int_x^{\infty} w_{\gamma}(t) \, dt \leq \frac{27}{14} c e^4 \left( \frac{ \gamma x}{\ln( \gamma x)^2} \right)^2 e^{- \eta \frac{\gamma x}{ \ln( \gamma x)^2}}
\end{equation} 
with $c$ as in (\ref{def_w}), see also (\ref{w-L1-norm}), and $\eta \in (2/7,1)$ as in Lemma~\ref{lem:wint}.
Moreover, if $ \gamma x \geq e^9$, then
\begin{equation} \label{W_g_dec}
\int_x^{\infty} W_{\gamma}(t) \, dt \leq \frac{486}{49 \gamma \eta} c e^4 \left( \frac{ \gamma x}{\ln( \gamma x)^2} \right)^3 e^{- \eta \frac{ \gamma x}{  \ln( \gamma x)^2}}
\end{equation} 
again with $c$ and $\eta$ as above.
\end{cor}

%
%

\subsection{On weighted integrals of dynamics} \label{sec:w_int_op}

In this section, we briefly discuss some general facts about weighted integrals of
a dynamics. Such operators arise as the generator of the spectral flow, and in this
case, a number of their basic properties are relevant.

\subsubsection{Some generalities} 

Let $H$ be a densely defined self-adjoint operator on a Hilbert space $\mathcal{H}$.
Denote by $\tau_t$ the associated Heisenberg dynamics, i.e. the one parameter
family of automorphisms of $\mathcal{B}( \mathcal{H})$ given by
\begin{equation} \label{H_dyn}
\tau_t(A) = e^{itH} A e^{-itH} \quad \mbox{for any } A \in \mathcal{B}( \mathcal{H}) \mbox{ and all } t \in \mathbb{R} \, .
\end{equation}
For any $w \in L^1( \mathbb{R})$, a bounded mapping $I: \mathcal{B}( \mathcal{H}) \to \mathcal{B}(\mathcal{H})$
is defined by setting
\begin{equation} \label{def_I}
I(A) = \int_{\mathbb{R}} \tau_t(A) \, w(t) \, dt \quad \mbox{for any } A \in \mathcal{B}(\mathcal{H}) \, . 
\end{equation}
In fact, Stone's theorem guarantees that this integral is well-defined in both the weak and strong sense. We refer to the operator $I$ above as the integral of the dynamics $\tau_t$ weighted by $w$, or more briefly, as
a weighted integral operator.

Our applications will mainly concern families of these weighted integral operators.
In fact, suppose $H(s) = H + \Phi(s)$ is as described in \eqref{sf_fv_ham} and for each $0 \leq s \leq 1$, consider $I_s: \mathcal{B}( \mathcal{H}) \to \mathcal{B}(\mathcal{H})$ with
\begin{equation}
I_s(A) = \int_{\mathbb{R}} \tau_t^{(s)}(A) w(t) \, dt 
\end{equation}
here $\tau_t^{(s)}$ is the dynamics corresponding to $H(s)$, see (\ref{sf_fv_ham}) and (\ref{sf_ham_dyn}), and $w \in L^1( \mathbb{R})$ is real-valued.
The following lemma is a useful observation.

\begin{lem} \label{lem:wio_sc} Let $H$ be a densely defined self-adjoint operator on a Hilbert space $\cH$ and, for $s\in [0,1]$, let $\Phi(s)=\Phi(s)^*\in\mathcal{B}( \mathcal{H})$ 
be continuous in $s$ for the strong operator topology. Suppose $w \in L^1( \mathbb{R})$ is real-valued, and $A:[0,1] \to \mathcal{B}( \mathcal{H})$
is pointwise self-adjoint and continuous in the strong operator topology. Then, the mapping $D : [0,1] \to \mathcal{B}( \mathcal{H})$ given by
\begin{equation}
D(s) = I_s(A(s)) = \int_{\mathbb{R}} \tau_t^{(s)}(A(s)) \, w(t) \, dt 
\end{equation}
is pointwise self-adjoint and continuous in the strong operator topology.
\end{lem} 

\begin{proof}
Self-adjointness of $D(s)$, which uses that $w$ is real-valued, is clear. Set $A(s,t) \in \mathcal{B}( \mathcal{H})$ by 
\begin{equation} \label{def_ast}
A(s,t) = \tau_t^{(s)}(A(s)) = e^{itH(s)} A(s) e^{-itH(s)} \quad \mbox{for any } 0 \leq s \leq 1 \mbox{ and } t \in \mathbb{R} \, .
\end{equation}
With $s_0 \in [0,1]$ fixed, for any $0 \leq s \leq 1$, we have that
\begin{equation} \label{stbdD2}
\left\| \left( D(s) - D(s_0) \right) \psi \right\| \leq \int_{\mathbb{R}} \left\| \left( A(s,t) - A(s_0,t) \right) \psi \right\| \, | w(t)| \, dt \quad \mbox{for any } \psi \in \mathcal{H} \, .
\end{equation}
Stone's theorem guarantees that for each $0 \leq s \leq 1$, the mapping $A(s, \cdot) : \mathbb{R} \to \mathcal{B}( \mathcal{H})$ is continuous
in the strong operator topology, and so the integrand above is clearly measurable. 
We now claim that for each $t \in \mathbb{R}$, $A(\cdot, t) : [0,1] \to \mathcal{B}( \mathcal{H})$ is also continuous in the
strong operator topology. Given this, the claimed continuity of $D$ will follow from 
an application of dominated convergence. Here we are using that strong continuity of $A$ implies $\sup_{0 \leq s \leq 1} \| A (s) \| <\infty$.  

Due to the form of $A(s,t)$, see (\ref{def_ast}), we need only show that $s \mapsto e^{i t H(s)}$ is strongly
continuous for each fixed $t \in \mathbb{R}$. To see this, note that for any $\phi \in \mathcal{D}$, the common domain of all $H(s)$,
\begin{equation}
\frac{d}{dt} e^{i t H(s)} e^{-itH(s_0)} \phi = i e^{it H(s)} \left( \Phi(s) - \Phi(s_0) \right) e^{-it H(s_0)} \phi 
\end{equation}
{f}rom which the well-known Duhamel's formula is proven. As a consequence,
\begin{equation}
\left\| \left( e^{-it H(s)} - e^{-itH(s_0)} \right) \psi \right\| \leq \int_0^t \left\| \left( \Phi(s) - \Phi(s_0) \right) e^{- i u H(s_0)} \psi \right\| \, du 
\end{equation}
is valid for all $\psi \in \mathcal{H}$ and $t \geq 0$ (a similar bound holds for $t<0$). Dominated convergence applied here,
using the continuity assumption on $\Phi$, shows that $s \mapsto e^{i t H(s)}$ is
continuous in the strong operator topology for each fixed $t \in \mathbb{R}$. The proof is now completed as
described above.
\end{proof}

%
%

\subsubsection{Two particular weighted integrals} \label{sec:tp_wios}
For the applications that follow, two particular weighted integral operators play a key role.
We introduce a notation for them here and discuss some basic properties.

Generally, the set-up is as before. Let $H$ be a densely defined self-adjoint operator on a Hilbert
space $\mathcal{H}$ and denote by $\tau_t$ the corresponding dynamics, see e.g. (\ref{H_dyn}). 

For any fixed $\gamma >0$,  let $w_\gamma, \, W_\gamma\in L^1(\Rl)$ be any real-valued functions so that 
\eqref{wsupp}, \eq{FTW}, and \eq{FTWbis} hold. Define two linear maps
$\mathcal{F}, \mathcal{G} : \mathcal{B}( \mathcal{H}) \to \mathcal{B}( \mathcal{H})$ by setting
\begin{equation} \label{def_F+G}
\mathcal{F}(A) = \int_{\mathbb{R}} \tau_t(A) \, w_{\gamma}(t) \, dt \quad \mbox{and} \quad \mathcal{G}(A) = \int_{\mathbb{R}} \tau_t(A) \, W_{\gamma}(t) \, dt. 
\end{equation}
As we will see, the properties of $\mathcal{F}$ and $\mathcal{G}$ depend crucially 
on the choice of $\gamma>0$.

In the remainder of this and the next subsection (subsection \ref{subsec:sp_flow}) we do not
require the more detailed properties of $w_\gamma$ and $W_\gamma$ that we have proved for the specific
functions constructed in Section~\ref{sec:ex_weight} (see (\ref{def_w}), (\ref{def_wg}), and (\ref{def:Wgam1})).
These properties will become important later when we analyze the quasi-locality properties of the spectral flow.
In particular, in the following lemma and the proof of Theorem \ref{thm:sf}, the specific functions
defined in Section \ref{sec:ex_weight} are not required.

\begin{lem} \label{lem:F_comm_proj} Let $H$ be a densely defined, self-adjoint operator on a Hilbert space $\mathcal{H}$. Let $\gamma>0$, $w_\gamma,W_\gamma\in L^1(\Rl)$ be real-valued and satisfy 
\eqref{wsupp}, \eq{FTW} and \eq{FTWbis}, and $\mathcal{F},\mathcal{G}:\mathcal{B}( \mathcal{H}) \to \mathcal{B}( \mathcal{H})$ be as defined in (\ref{def_F+G}). Suppose that the spectrum of $H$ can be decomposed into two non-empty, 
disjoint sets $\Sigma_1$ and $\Sigma_2$,
\begin{equation}
{\rm spec}(H) = \Sigma_1 \cup \Sigma_2
\end{equation}
with $\Sigma_1$ contained in some compact set and $d( \Sigma_1, \Sigma_2) \geq \gamma$.
Denote by $P$ the spectral projection associated to $H$ onto $\Sigma_1$. Then, for any $A \in \mathcal{B}( \mathcal{H})$
\begin{equation} \label{F_comm_0}
[ \mathcal{F}(A), P ] =0
\end{equation}
and
\begin{equation} \label{G_comm}
[ \mathcal{G}(A), P ] = i \int_{\Sigma_1\times \Sigma_2} \frac{1}{ \mu - \lambda} dE_{\lambda} A dE_{\mu}  +  i \int_{\Sigma_1\times \Sigma_2} \frac{1}{ \mu - \lambda} dE_{\mu} A dE_{\lambda}.
\end{equation}
Here, $E_{\lambda}$ denotes the spectral family associated to $H$.  
\end{lem}

\begin{proof}
We first prove (\ref{F_comm_0}). In fact, we will show that each $\mathcal{F}(A)$ is diagonal with respect to $P$ in the sense that
\begin{equation} \label{cross_proj_0}
P \mathcal{F}(A) ( \idtyty - P) = ( \idtyty - P) \mathcal{F}(A) P=0 \quad \mbox{for any } A \in \mathcal{B}( \mathcal{H}) \, .
\end{equation}
Given this, one readily checks that 
\begin{equation}
[ \mathcal{F}(A), P] = (\mathcal{F}(A)P - P\mathcal{F}(A)P) - (P\mathcal{F}(A) - P\mathcal{F}(A)P)= 0
\end{equation}
as claimed.

We now calculate the left-hand-side of (\ref{cross_proj_0}). 
To do so, we will use results on double operator integrals, see e.g. \cite{birman:2003}.
In fact, using Theorem 4.1 in \cite{birman:2003}, one sees that
\begin{eqnarray}
P \mathcal{F}(A) ( \idtyty -P) & = & \int_{\mathbb{R}} P e^{itH} A e^{-itH} ( \idtyty -P) w_{\gamma}(t) \, dt \nonumber \\
& = & \int_{\mathbb{R}} \int_{\Sigma_1\times \Sigma_2} e^{it( \lambda - \mu)} w_{\gamma}(t) dE_{\lambda} A dE_{\mu} \, dt \nonumber \\
& = & \sqrt{2 \pi} \int_{\Sigma_1\times \Sigma_2}  \hat{w}_{\gamma}(\mu -\lambda) dE_{\lambda} A dE_{\mu} = 0 \, .
\end{eqnarray}
Here we have used $E_{\lambda}$ to denote the spectral family associated to $H$. Moreover, $w_\gamma\in L^1(\bR)$ is sufficient to guarantee the re-ordering of the integrals above; it is here that we apply Theorem~4.1~(iii) of \cite{birman:2003}. The final equality is due to the fact that the Fourier transform of $w_{\gamma}$ is supported in $[- \gamma, \gamma]$, see (\ref{wsupp}). The other relation in (\ref{cross_proj_0}) is proven similarly, and (\ref{F_comm_0}) follows.

Arguing as above, we find that
\begin{eqnarray}
[ \mathcal{G}(A), P]  & = & ( \idtyty - P) \mathcal{G}(A) P - P \mathcal{G}(A) ( \idtyty - P) \nonumber \\
& = &  \sqrt{2 \pi} \int_{\Sigma_1\times \Sigma_2}  \hat{W}_{\gamma}(\lambda -\mu) dE_{\mu} A dE_{\lambda}  -  
\sqrt{2 \pi} \int_{\Sigma_1\times \Sigma_2}  \hat{W}_{\gamma}(\mu -\lambda) dE_{\lambda} A dE_{\mu} 
\end{eqnarray}
The claim in (\ref{G_comm}) now follows from (\ref{FTWbis}).
\end{proof}

A useful observation for certain applications (see, e.g., \cite{bachmann:2018a,monaco:2019}) is that the map $\cG$ is a (left-) inverse 
of the Liouvillean $[H,\cdot]$ on the space of off-diagonal operators. 

\begin{prop}\label{prop:}\label{inverse_liouvillean}
Let $H$ be a densely defined, self-adjoint operator on a Hilbert space $\mathcal{H}$, and let $[H,\cdot]$ denote the generator
of the Heisenberg dynamics generated by $H$.  Let $\gamma>0$ and $\mathcal{F}$ and $\mathcal{G}$ as 
defined in (\ref{def_F+G}). Suppose that the spectrum of $H$ can be decomposed into two non-empty, disjoint sets $\Sigma_1$ and $\Sigma_2$, 
with $\Sigma_1$ compact and $d( \Sigma_1, \Sigma_2) \geq \gamma$. Let $P$ denote the spectral projection of $H$ onto $\Sigma_1$. Then,
for all $A\in \cB(\cH)$ such that $\cG(A) \in \dom [H,\cdot] $, we have 
\begin{equation}
i[H, \mathcal{G}(A)] = \mathcal{F}(A) - A.
\end{equation}
If, in addition $A$ is off-diagonal with respect to $P$, meaning $A \in  P\cB(\cH)(\idty -P) \oplus  (\idty -P)\cB(\cH)P$, we have $\cF(A)=0$ and
\begin{equation}
-i[H, \mathcal{G}(A)] = A
\end{equation} 
\end{prop}
\begin{proof}
For any $u \in \mathbb{R}$, 
\begin{equation} \label{inverse-eqn}
\tau_u(\mathcal{G}(A)) = \int_{\mathbb{R}} \tau_{t+u}(A) W_{\gamma}(t) \, dt = \int_{\mathbb{R}} \tau_y(A) W_{\gamma}(y-u) \, dy 
\end{equation}
Since, by assumption,  $\cG(A) \in \dom [H,\cdot] $, and $ \dom [H,\cdot] $ is $\tau_u$-invariant, we then have
\begin{eqnarray}\label{prove-inverse}
i[H, \tau_u(\mathcal{G}(A))] =\frac{d}{du}  \tau_u(\mathcal{G}(A))& = & \int_{\mathbb{R}} \tau_y(A) \frac{d}{du} W_{\gamma}(y-u) \, dy \nonumber \\
& = & \mathcal{F}( \tau_u(A)) - \tau_u(A)
\end{eqnarray} 
where the derivative of $W_\gamma$ is taken in the distributional sense.
Evaluation of (\ref{prove-inverse}) at $u=0$ results in:
\begin{equation}\label{step1}
i[H, \mathcal{G}(A)] = \mathcal{F}(A) - A
\end{equation}
If $A \in  P\cB(\cH)(\idty -P)$ (or $A \in  (\idty- P)\cB(\cH)P$), then $\cF(A) \in P\cB(\cH)(\idty -P)$ (or $\cF(A) \in  (\idty- P)\cB(\cH)P$), and hence in either case, by (\ref{cross_proj_0}), 
we have $\cF(A)=0$. With this, \eq{step1} becomes
\begin{equation}
-i[H, \mathcal{G}(A)] = A
\end{equation}
\end{proof}

In applications to quantum spin systems, either finite or infinite, the domain condition on $\cG(A)$ in this proposition is quite generally satisfied due to the quasi-locality 
properties of both $\cG$ and the generator of the Heisenberg dynamics See, e.g., the discussion of the domain of the generator of the dynamics in the proof of Theorem \ref{thm:semi-continuity_gap}.

%
%

\subsection{The proof of Theorem~\ref{thm:sf}} \label{subsec:sp_flow}

The goal of this section is to complete the proof of Theorem~\ref{thm:sf}.
Let us recap our progress so far.

Let $H(s)=H+\Phi(s)$ be as defined in \eqref{sf_fv_ham}. For any $\gamma>0$, a map
$D : [0,1] \to \mathcal{B}( \mathcal{H})$ is defined by
\begin{equation}
D(s) = \int_{\mathbb{R}} \tau_t^{(s)}(\Phi'(s)) \, W_{\gamma}(t) \, dt,
\end{equation}
where $\tau_t^{(s)}$ is the dynamics associated
to $H(s)$ as in (\ref{sf_ham_dyn}) and $W_{\gamma}$ is the particular weight function defined in (\ref{def:Wgam1}).
By Lemma~\ref{lem:wio_sc}, $D(s)$ is pointwise self-adjoint and continuous in the strong operator topology. 
In this case, for any $0 \leq s \leq 1$, an automorphism $\alpha_s$ of $\mathcal{B}( \mathcal{H})$ is defined by
setting
\begin{equation}
\alpha_s(A) = U(s)^* A U(s) \quad \mbox{for any } A \in \mathcal{B}( \mathcal{H}),
\end{equation}
where the unitary $U(s)$ is the unique strong solution of 
\begin{equation}
\frac{d}{ds}U(s) = -i D(s) U(s) \quad \mbox{with} \quad U(0) = \idtyty \, .
\end{equation}

The proof of Theorem~\ref{thm:sf} is completed by showing that if $H(s)$ satisfies
Assumption~\ref{sf_ass:fv_uni_gap} for some $\gamma >0$, then the automorphisms $\alpha_s$
introduced above satisfy (\ref{sf:followsp}), i.e.
\begin{equation} \label{sf:followsp_2}
\alpha_s(P(s)) = P(0) \quad \mbox{for all } 0 \leq s \leq 1 \, .
\end{equation}

\begin{proof}[Proof of Theorem~\ref{thm:sf}:] 
As discussed above, we need only verify (\ref{sf:followsp_2}).
A formal calculation shows that
\begin{equation}\label{formal_calc}
\frac{d}{ds} \alpha_s(P(s)) = \alpha_s \left( i [D(s), P(s)] + \frac{d}{ds} P(s) \right)
\end{equation}
in the sense of strong derivatives.
Since $\alpha_0(P(0)) = P(0)$, we need only prove that
\begin{equation} \label{dspecproj}
\frac{d}{ds} P(s) = -i [D(s), P(s)] \, .
\end{equation}

It is well-known, see \cite{kato:1995}, that spectral projections can be determined through a contour integral of the resolvent, i.e. 
\begin{equation}
P(s) = -\frac{1}{2 \pi i} \int_{\eta(s)} R(z,s) \, dz ,
\end{equation}
where $R(z,s) = (H(s) - z)^{-1}$ is the resolvent of $H(s)$ and $\eta(s)$ is any contour in the complex plane that 
encircles the interval $I(s)$, as described in Assumption~\ref{sf_ass:fv_uni_gap}. From this representation, it is clear that 
strong differentiability of $P$ follows from strong differentiability of $R(z, \cdot)$, and so the formal calculation in \eqref{formal_calc} is well-defined. Now, note that for any fixed $s_0 \in [0,1]$
the gap assumption allows for a choice of contour $\eta(s)$ which is independent of $s$ in a neighborhood of $s_0$.  With such a contour 
one checks that
\begin{equation}
\frac{d}{ds} P(s) = \frac{1}{2 \pi i} \int_{\eta(s)} R(z,s) \Phi'(s) R(z,s) \, dz .
\end{equation}
As $P(s)$ is a strongly differentiable family of orthogonal projections, one can
also verify that
\begin{equation}
P(s) \frac{d}{ds} P(s) P(s) = ( \idtyty - P(s)) \frac{d}{ds}P(s) ( \idtyty - P(s) ) = 0 .
\end{equation}
We conclude that
\begin{equation} \label{dP=ints}
\frac{d}{ds} P(s) = \frac{1}{2 \pi i} \int_{\eta(s)} A(s,z) \Phi'(s) B(s,z) \, dz + \frac{1}{2 \pi i} \int_{\eta(s)} B(s, \overline{z})^* \Phi'(s) A(s, \overline{z})^* \, dz ,
\end{equation}
where we have set
\begin{equation}
A(s,z) = P(s) R(z,s) \quad \mbox{and} \quad B(s,z) = R(z,s) ( \idtyty - P(s)) .
\end{equation}

To simplify the integrals on the right-hand-side of (\ref{dP=ints}), we again appeal to the 
formalism of double operator integrals. In fact, let us denote by, $E_{\lambda}^{(s)}$, the spectral family associated to 
the self-adjoint operator $H(s)$. One checks that 
\begin{eqnarray}
\frac{1}{2 \pi i} \int_{\eta(s)} A(s,z) \Phi'(s) B(s,z) \, dz  & = & 
\int_{\Sigma_1(s)} \int_{\Sigma_2(s)} \frac{1}{2 \pi i} \int_{\eta(s)} \frac{1}{\lambda -z} \frac{1}{\mu - z} \, dz \, dE^{(s)}_{ \lambda} \, \Phi'(s) \, dE^{(s)}_{\mu} \nonumber \\ 
& = & \int_{\Sigma_1(s)} \int_{\Sigma_2(s)} \frac{1}{\mu - \lambda} \, dE^{(s)}_{ \lambda} \, \Phi'(s) \, dE^{(s)}_{\mu},
\end{eqnarray}
where again, the re-ordering of the integrals appearing above is justified by Theorem 4.1 (iii) in \cite{birman:2003}.
Here specifically, the required integrability condition on the contour is readily verified using 
Assumption~\ref{sf_ass:fv_uni_gap}. Applying similar arguments to the
second term in (\ref{dP=ints}), we find that
\begin{equation}
\frac{d}{ds} P(s) = \int_{\Sigma_1(s)} \int_{\Sigma_2(s)} \frac{1}{\mu - \lambda} \, dE^{(s)}_{ \lambda} \, \Phi'(s) \, dE^{(s)}_{\mu} + \int_{\Sigma_2(s)} \int_{\Sigma_1(s)} \frac{1}{\mu - \lambda} \, dE^{(s)}_{ \mu} \, \Phi'(s) \, dE^{(s)}_{\lambda}.
\end{equation}

On the other hand, the right-hand-side of (\ref{dspecproj}) is clearly given by 
\begin{equation}
 -i [D(s), P(s)]  = [ \mathcal{G}^{(s)}(-i \Phi'(s)), P(s)] \quad \mbox{for any } 0 \leq s \leq 1 \, .
\end{equation}
Here we have used the notation $\mathcal{G}^{(s)}$ for the weighted integral operator, see (\ref{def_F+G}),
defined with respect to the parameter dependent dynamics, $\tau_t^{(s)}$.
Using Lemma~\ref{lem:F_comm_proj}, in particular (\ref{G_comm}) with $A = -i \Phi'(s)$,  the equality claimed in (\ref{dspecproj}) is now clear.
This completes the proof of Theorem~\ref{thm:sf}.
\end{proof}

%
%

\subsection{Quasi-locality of the spectral flow} \label{sec:ql_sf}

For the remainder of this section, let us assume that $(\Gamma, d)$ is a $\nu$-regular metric space, in the sense 
of \eq{nu_ball_bd}, and $\cH_x$ is the complex Hilbert space of the quantum system at $x\in\Gamma$. 
We start by considering a finite system corresponding to $\Lambda\in\cP_0(\Gamma)$. 
Recall that for any $X \subset \Lambda$, we denote by $\mathcal{H}_X = \bigotimes_{x \in X} \mathcal{H}_x$ and
$\mathcal{A}_X = \mathcal{B}( \mathcal{H}_X)$.

This section is divided into two parts. First, in Section~\ref{sec:ql_wios}, we prove quasi-locality estimates for
the two weighted integral operators introduced in Section~\ref{sec:tp_wios}. Then, in Section~\ref{sec:ql_sp_flow}, 
we establish quasi-locality bounds for the spectral flow constructed in the proof of Theorem~\ref{thm:sf}.
 
\subsubsection{Quasi-locality for two weighted integral operators} \label{sec:ql_wios}
In Section~\ref{sec:tp_wios}, we introduced two particular weighted integral operators that will
appear frequently in our applications. We now demonstrate that, under certain additional conditions,
each of these weighted integral operators satisfies an explicit quasi-locality estimate in the sense of 
Section~\ref{sec:quasilocal-maps}.

Let us assume that there is a one-parameter family of automorphisms of $\cA_\Lambda$, 
which we denote by $\tau_t$, that satisfies a quasi-locality estimate. More precisely, suppose that there are positive 
numbers $C$ and $v$ as well as a non-negative, non-decreasing function $g$ for which: given any $X,Y \subset \Lambda$,
\begin{equation} \label{sf_lrb}
\| [ \tau_t(A), B] \| \leq C \| A \| \| B \| |X| e^{v|t| - g(d)} 
\end{equation}
for all $A \in \mathcal{A}_X$, $B \in \mathcal{A}_Y$, and $t \in \mathbb{R}$. Here $d=d(X,Y)$ is the distance
between the sets $X$ and $Y$. As is discussed in Section~\ref{sec:lrb}, such a bound is known for the 
dynamics generated by a short range Hamiltonian; it is, e.g., a consequence of the Lieb-Robinson bounds in Theorem~\ref{thm:lrb}.
In order to prove the quasi-local bounds below, we need only know (\ref{sf_lrb}) and that $g(d)$ becomes sufficiently 
large (see (\ref{def_d*})).
In applications, we typically have \eq{sf_lrb} with 
\begin{equation} \label{g_grows}
\lim_{d \to \infty} g(d) = + \infty 
\end{equation}

In terms of these automorphisms $\tau_t$, for each $\gamma >0$ define $\mathcal{F}, \mathcal{G} : \mathcal{B}( \mathcal{H}) \to \mathcal{B}(\mathcal{H})$
by
\begin{equation} \label{def_sf_F+G}
\mathcal{F}(A) = \int_{\mathbb{R}} \tau_t(A) w_{\gamma}(t) \, dt \quad \mbox{and} \quad \mathcal{G}(A) = \int_{\mathbb{R}} \tau_t(A) W_{\gamma}(t) \, dt
\end{equation}
for any $A \in \mathcal{B}( \mathcal{H})$; compare with (\ref{def_F+G}). Here again $w_{\gamma}$ and $W_{\gamma}$ are the specific weight functions
introduced in Section~\ref{sec:ex_weight}. 

Before we state our first result, recall that $w_{\gamma}(t) = \gamma w( \gamma t)$ and therefore 
$\| w_{\gamma}\|_{\infty} \leq \gamma c$ with $c$ 
the $L^1$-normalization of $w$, see (\ref{def_w}). Moreover, Corollary~\ref{cor:weight_dec}, see specifically (\ref{w_g_dec}), 
demonstrates that there is an $\eta \in (2/7,1)$ for which given any $x \geq \gamma^{-1} e^9$ the bound
\begin{equation} \label{wt_est}
\int_x^{\infty} w_{\gamma}(t) \, dt \leq \frac{27}{14} c e^4 f_{\gamma}(x)^2 e^{- \eta f_{\gamma}(x)} 
\end{equation} 
holds. Here, for any $b>0$, we have introduced the subadditive, non-decreasing function 
\begin{equation} \label{sub_exp_f}
f_b(x) = \left\{ \begin{array}{cc} \frac{e^2}{4} & \mbox{if } 0 \leq x \leq b^{-1} e^2, \\ 
\frac{bx}{\ln(bx)^2} & \mbox{if } x \geq b^{-1} e^2. \end{array} \right.
\end{equation}
See Section \ref{sec:3_common_weights} for a discussion of the properties of $f_b$.

Our quasi-locality estimate on the weighted integral operator $\mathcal{F}$ follows.

\begin{lem} \label{lem:gen_F_ql_est} Let $\tau_t$ be a family of automorphisms of $\mathcal{B}(\mathcal{H})$ satisfying
(\ref{sf_lrb}) with (\ref{g_grows}). Let $\gamma >0$ and take $\mathcal{F}$ to be the weighted integral operator defined in (\ref{def_sf_F+G}). 
For any $0 < \epsilon < 1$ and all $X,Y \subset \Lambda$ the bound
\begin{equation}
\| [ \mathcal{F}(A), B ] \| \leq 2 \| A \| \| B \| |X| G_{\mathcal{F}}^{\epsilon}(d(X,Y))
\end{equation}
holds for all $A \in \mathcal{A}_X$ and $B \in \mathcal{A}_Y$. Here
\begin{equation} \label{gen_F_ql_dec_est}
G_{\mathcal{F}}^{\epsilon}(d) = \left\{ \begin{array}{cl} 1 \,  & \mbox{if } 0 \leq d \leq d_{\epsilon}^* \\ 
\min\left\{ 1, c\left( \frac{C \gamma}{v} + \frac{27}{7} e^4 f_{\gamma_{\epsilon}}(g(d))^2 \right) e^{- \eta f_{\gamma_{\epsilon}}(g(d))} \right\}\,  & \mbox{otherwise}, \end{array} \right.
\end{equation}
where $d_{\epsilon}^*$ is the smallest value of $d$ for which
\begin{equation} \label{ld_ge}
\max\left[ 9, \sqrt{ \frac{ \eta \gamma_{\epsilon}}{ \epsilon}} \right] \leq \ln( \gamma_{\epsilon} g(d)) \quad \mbox{where } \gamma_{\epsilon} = \frac{(1 - \epsilon) \gamma}{v} \, . 
\end{equation}
\end{lem}
It can be verified that the function $G_{\mathcal{F}}^{\epsilon}(d)$ given in \eq{gen_F_ql_dec_est} is monotone and strictly decreasing when $|G_{\mathcal{F}}^{\epsilon}(d)|<1$.
\begin{proof}
Let $X,Y \subset \Lambda$. Since $w_{\gamma}$ is $L^1$-normalized, it is clear that
\begin{equation}
\| [ \mathcal{F}(A), B] \| \leq 2 \| A \| \| B \| \quad \mbox{for all } A \in \mathcal{A}_X \mbox{ and } B \in \mathcal{A}_Y \, .
\end{equation}
In applications, this bound is best when $d=d(X,Y)$ is small.

When $d=d(X,Y)$ is sufficiently large, see below, a different estimate holds.
In fact, let $T \geq 0$ and estimate 
\begin{equation} \label{simp_int_bd}
\| [ \mathcal{F}(A), B ] \| \leq \int_{|t| \leq T} \| [ \tau_t(A), B ] \| w_{\gamma}(t) \, dt + \int_{|t| > T} \| [ \tau_t(A), B ] \| w_{\gamma}(t) \, dt \, .
\end{equation}
For the first term above, we ignore the weight and use the locality bound for the dynamics, i.e. (\ref{sf_lrb}). For the
second term, we ignore the dynamics and use the estimate on the weight, see (\ref{wt_est}) above.
{F}rom these, we obtain the bound
\begin{equation} \label{F_ql_bd_1}
\| [ \mathcal{F}(A), B ] \| \leq 2 c \| A \| \| B \| |X| \left(  \frac{C \gamma}{v} e^{vT-g(d)} + \frac{27}{7} e^4 f_\gamma(T)^2e^{- \eta f_{\gamma}(T)} \right) 
\end{equation}
It is important to note that Corollary~\ref{cor:weight_dec}, summarized in (\ref{wt_est}) above, has a constraint, and so (\ref{F_ql_bd_1}) is only valid if $\gamma T \geq e^9$.
For any $0 < \epsilon < 1$, choose
$T = \frac{(1- \epsilon)}{v} g(d)$. In this case, we find that
\begin{equation} \label{F_ql_bd_2}
\| [ \mathcal{F}(A), B ] \| \leq 2 c \| A \| \| B \| |X| \left( \frac{C \gamma}{v} e^{- \epsilon g(d)} + \frac{27}{7} e^4 f_{\gamma_{\epsilon}}(g(d))^2e^{- \eta f_{\gamma_{\epsilon}}(g(d))} \right) 
\end{equation}
whenever $\gamma_{\epsilon} g(d) \geq e^9$ and $\gamma_{\epsilon}$ is as in (\ref{ld_ge}). Since $\lim_{d \to \infty}g(d) = \infty$, it is clear that
(\ref{F_ql_bd_2}) can be estimated as in (\ref{gen_F_ql_dec_est}) for sufficiently large $d$. The relation
\begin{equation} \label{def_d*}
- \epsilon g(d) + \eta f_{\gamma_{\epsilon}}(g(d)) = - \epsilon g(d) \left[ 1 - \frac{ \eta \gamma_{\epsilon}}{ \epsilon \ln(\gamma_{\epsilon} g(d))^2} \right]
\end{equation}
is used when defining $d_{\epsilon}^*$ as above. This completes the proof. 
%
\end{proof}

Depending on the application one has in mind, more decay of the function governing the locality of the 
dynamics, specifically $e^{-g(d)}$, may be needed. For example, many applications require certain
moments of the function $G_{\mathcal{F}}^{\epsilon}$ to be finite. Let us make two observations in
this regard.

{\it On polynomial decay:} Let us consider a family of automorphisms $\tau_t$ with a locality estimate of the form (\ref{sf_lrb}).
If the non-decreasing function $g$ is of the form
\begin{equation} \label{log_g}
g_q(x) = q \ln(1+x) \quad \mbox{for some } q >0
\end{equation} 
then (for fixed $t$) the locality bound decays like a power-law. In this case, Lemma~\ref{lem:gen_F_ql_est} holds, however, the resulting
decay function, see (\ref{gen_F_ql_dec_est}), has no finite moments. In fact, for any positive numbers $a$, $b$, $c$, and $d$,
one readily checks that
\begin{equation}
\lim_{x \to \infty} (1+x)^a e^{- b f_c(g_d(x))} = + \infty
\end{equation}
where the functions $f_c$ and $g_d$ are as defined in (\ref{sub_exp_f}) and (\ref{log_g}) respectively. As we will see in \cite{QLBII}, 
this lack of moments restricts the known proofs of stability of the spectral gap to perturbations
that decay faster than any polynomial. We do not believe that arguments in \cite{michalakis:2013} can be extended to obtain a uniform lower 
bound for the spectral gap in the case of perturbations with only power-law decay, contrary to the claim made in that work.
To see why, note that the proof of Lemma~\ref{lem:gen_F_ql_est} depends on the choice of $T \geq 0$, see e.g. (\ref{simp_int_bd}).
This choice must be made in such a way that both terms on the right-hand-side of (\ref{F_ql_bd_1}) decay.
In order for the first term to decay, $vT - g(d) <0$ and so one must take $T < v^{-1}g(d)$. As the function
$f_{\gamma}(T)$ is increasing for large $T$, the most decay one can obtain from the second term is when $T = v^{-1}g(d)$. If $g$ is logarithmic as discussed above, then even this choice has no finite moments.

{\it On stretched-exponential decay:} Let us consider a family of automorphisms $\tau_t$ with a locality estimate of the form (\ref{sf_lrb}).
If the non-decreasing function $g$ is of the form
\begin{equation} \label{sub_exp_g}
g(r) \geq a r^{\theta}  \quad \mbox{for some } a >0 \mbox{ and } 0< \theta \leq 1
\end{equation} 
then all moments of the decay function, see (\ref{gen_F_ql_dec_est}), are finite. In fact, for any $\delta >0$,
there is a number $C_{\delta}$ for which $\ln(x) \leq C_{\delta} x^{\delta}$ whenever $x \geq 1$. In this case,
for any $b>0$,
\begin{equation}
f_b(g(r)) = \frac{bg(r)}{\ln(bg(r))^2} \geq C_{\delta}^{-2}(bg(r))^{1-2 \delta} \geq \frac{(ab)^{1 - 2 \delta}}{C_{\delta}^2} r^{\theta(1-2 \delta)} 
\end{equation}
and therefore, for any $\delta < 1/2$, the function in (\ref{gen_F_ql_dec_est}) decays at least as fast as a stretched exponential, see Section~\ref{sec:3_common_weights}
more on this terminology.

The following lemma is the analogue of Lemma~\ref{lem:gen_F_ql_est} applicable to $\mathcal{G}$.

\begin{lem} \label{lem:gen_G_ql_est} 
Let $\tau_t$ be a family of automorphisms of $\mathcal{B}(\mathcal{H})$ satisfying
(\ref{sf_lrb}) with (\ref{g_grows}). Let $\gamma >0$ and take $\mathcal{G}$ to be the weighted integral operator defined in (\ref{def_sf_F+G}). 
For any $0 < \epsilon < 1$ and all $X,Y \subset \Lambda$ the bound
\begin{equation} \label{G_ql_est}
\| [ \mathcal{G}(A), B ] \| \leq 2 \| A \| \| B \| |X| G_{\mathcal{G}}^{\epsilon}(d(X,Y))
\end{equation}
holds for all $A \in \mathcal{A}_X$ and $B \in \mathcal{A}_Y$. Here
\begin{equation} \label{gen_G_ql_dec_est}
G_{\mathcal{G}}^{\epsilon}(d) = \left\{ \begin{array}{cl} \| W_{\gamma} \|_1 \,  & \mbox{if } 0 \leq d \leq d_{\epsilon}^* \\ 
\min\left\{ \| W_{\gamma} \|_1, \left( \frac{C }{2v} + \frac{243}{49 \gamma \eta} c e^4 f_{\gamma_{\epsilon}}(g(d))^3 \right) e^{- \eta f_{\gamma_{\epsilon}}(g(d))} \right\} \,  & \mbox{otherwise}, \end{array} \right.
\end{equation}
and $d_{\epsilon}^*$ is as defined in Lemma~\ref{lem:gen_F_ql_est}.
\end{lem}
The proof of this lemma is almost identical to that of Lemma~\ref{lem:gen_F_ql_est}
except that one uses the estimate (\ref{W_g_dec}) from Corollary~\ref{cor:weight_dec}, instead of (\ref{w_g_dec}).
We also use that $\| W_{\gamma} \|_{\infty} \leq 1/2$. 

Here too, it can be verified that the function $G_{\mathcal{G}}^{\epsilon}(d)$ given in \eq{gen_G_ql_dec_est} is monotone and strictly decreasing when $|G_{\mathcal{G}}^{\epsilon}(d)|<\| W_{\gamma} \|_1$.

%
%

\subsubsection{Quasi-locality of the spectral flow automorphism} \label{sec:ql_sp_flow}

We will now consider the spectral flow in the thermodynamic limit. In order to derive explicit estimates
useful for applications, we work with $F$-functions of the $\nu$-regular metric space  $(\Gamma, d)$ 
of the form $F(r) = e^{-g(r)}F_0(r)$, where $g$ is non-decreasing and subadditive, and 
\begin{equation} \label{sf_base_F}
F_0(r) = \frac{1}{(1+r)^{\xi}},
\end{equation}
for a suitable $\xi>0$. As is shown in the Appendix (Section~\ref{sec:2_comm_ex}), any choice of 
$\xi > \nu +1$ will define an $F$-function on a $\nu$-regular $(\Gamma, d)$.  In the case
$\Gamma = \mathbb{Z}^{\nu}$, $\xi > \nu$ is sufficient. We will say that
$F$ is a weighted $F$-function on $(\Gamma, d)$ with base $F_0$.

Let us now introduce the models we consider through an assumption.

\begin{assumption} \label{sf_ass:qlm_ham} There is a collection $\{ H_x \}_{x \in \Gamma}$ of
densely defined, self-adjoint on-site Hamiltonians. For each $0 \leq s \leq 1$, there
is an interaction $\Phi(s)$ on $\cA_{\Gamma}$ for which 
\begin{enumerate}
\item[(i)] For each $X \in \mathcal{P}_0( \Gamma)$, $\Phi(X,s)^* = \Phi(X,s) \in \mathcal{A}_X$ for all $0 \leq s \leq 1$.
\item[(ii)] For each $X  \in \mathcal{P}_0( \Gamma)$, $\Phi(X, \cdot) : [0,1] \to \mathcal{A}_X$ strongly $C^1$ in the sense of Definition~\ref{def:stg_diff}.
\item[(iii)] $F$ is a weighted $F$-function on $(\Gamma, d)$ with base $F_0$ as in (\ref{sf_base_F}) and there is a 
bounded, measurable function $\| \Phi \|_{1,1} :[0,1] \to [0, \infty)$ for which given any $x,y \in \Gamma$, the estimate
\begin{equation} \label{sf_base_int_bd}
\sum_{\stackrel{X \in \mathcal{P}_0( \Gamma):}{x,y \in X}} \left( \| \Phi(X, s) \| + |X| \| \Phi'(X, s) \| \right) \leq \| \Phi \|_{1,1}(s) F(d(x,y)) 
\end{equation}
holds for all $0 \leq s \leq 1$. 
\end{enumerate}
\end{assumption}

Under Assumption~\ref{sf_ass:qlm_ham}, given any $\Lambda \in \mathcal{P}_0( \Gamma)$, 
\begin{equation} \label{sf_fv_ham_mod}
H_{\Lambda}(s) = H_{\Lambda} + \Phi_{\Lambda}(s) \quad \mbox{with} \quad H_{\Lambda} = \sum_{x \in \Lambda} H_x \quad
\mbox{and} \quad \Phi_{\Lambda}(s) = \sum_{X \subset \Lambda} \Phi(X,s) 
\end{equation}
is a well-defined self-adjoint operator on $\mathcal{H}_{\Lambda}$ for each $0 \leq s \leq 1$. If we denote by
$\mathcal{D}_x \subset \mathcal{H}_x$ the dense domain of the on-site Hamiltonian $H_x$, then each $H_{\Lambda}(s)$
has the same dense domain $\mathcal{D}_{\Lambda} = \bigotimes_{x \in \Lambda} \mathcal{D}_x \subset \mathcal{H}_{\Lambda}$. 
We stress here that in our applications the number $0 \leq s \leq 1$ plays the role of a parameter. In this case, the 
finite-volume Hamiltonians $H_{\Lambda}(s)$ do not depend on {\it time} $t$, and thus 
using functional calculus, for each $0 \leq s \leq 1$, the dynamics corresponding to these self-adjoint operators is simply given by
\begin{equation} \label{sf_fv_dyn_mod}
\tau_t^{\Lambda, s}(A) = e^{itH_{\Lambda}(s)} A e^{- it H_{\Lambda}(s)} \quad \mbox{for any } A \in \mathcal{A}_{\Lambda} \mbox{ and } t \in \mathbb{R} \, .
\end{equation} 
Interactions depending on the time $t$ itself, as in the model (\ref{time+para_ham}) discussed in Section~\ref{sec:class_ubd_ham},
can be accommodated without  difficulty.

Assumption~\ref{sf_ass:qlm_ham} implies that $\Phi(s) \in \mathcal{B}_F$ for each $0 \leq s \leq 1$. Therefore,
an application of Theorem~\ref{thm:existbd} shows that there is an infinite volume dynamics defined by
\begin{equation} \label{sf_iv_dyn_def}
\tau_t^s(A) = \lim_{\Lambda \to \Gamma} \tau_t^{\Lambda, s}(A) \quad \mbox{for each } A \in \cA_{\Gamma}^{\rm loc} \mbox{ and } t \in \mathbb{R} \, .
\end{equation}

In terms of this infinite volume dynamics, and any $\gamma>0$, we now define a family of linear maps 
$\{ \mathcal{K}_s^{\gamma} : \cA_{\Gamma}^{\rm loc} \to \cA_{\Gamma} \}_{s \in [0,1]}$ by setting
\begin{equation} \label{sf_iv_wio}
\mathcal{K}_s^{\gamma}(A) = \int_{\mathbb{R}} \tau_t^s(A) \, W_{\gamma}(t) \, dt \quad \mbox{for any } A \in \cA_{\Gamma}^{\rm loc} \, .
\end{equation}
Here $W_{\gamma}$ is the weight function introduced in (\ref{def:Wgam1}) of Section~\ref{sec:ex_weight}. Often, we will regard $\gamma$ as fixed
and drop it from our notation. 

The following proposition relates the quantities introduced above to the methods discussed in Section~\ref{sec:qli_tl}.

\begin{prop} \label{prop:sf_wio_sat} Consider a quantum lattice system comprised of a $\nu$-regular metric space $(\Gamma, d)$ and $\cA_{\Gamma}$.  Suppose Assumption~\ref{sf_ass:qlm_ham} holds with a weighted $F$-function $F(r) =e^{-g(r)}F_0(r)$ for which
$\lim_{r \to \infty} g(r) =\infty$ and $F_0$ is as in (\ref{sf_base_F}). Then, for any $\gamma>0$, the family of maps $\{ \mathcal{K}_s^{\gamma} \}_{s \in [0,1]}$, as defined in
(\ref{sf_iv_wio}) above, satisfies the conditions of Assumption~\ref{ass:td_qlms}.
\end{prop}

\begin{proof}
As is clear from the statement, the results we prove below hold for any $\gamma >0$. 
For convenience of presentation, we now fix such a value $\gamma>0$ and
then suppress it in our notation.

Let $\{ \Lambda_n \}_{n \geq 1}$ be any increasing, exhaustive sequence of (non-empty) finite subsets of $\Gamma$.
For each $n \geq 1$, define a family of maps $\{ \mathcal{K}_s^{(n)} \}_{s \in [0,1]}$, with 
$\mathcal{K}_s^{(n)} : \cA_{\Lambda_n} \to \cA_{\Lambda_n}$ for all $0 \leq s \leq 1$, by setting
\begin{equation} \label{sf_wio_fv}
\mathcal{K}_s^{(n)}(A) = \int_{\mathbb{R}} \tau_t^{\Lambda_n, s}(A) W_{\gamma}(t) \, dt \quad \mbox{for any } A \in \cA_{\Lambda_n} \, .
\end{equation}
Here we have used the finite volume dynamics, see (\ref{sf_fv_ham_mod}) and (\ref{sf_fv_dyn_mod}), with $\Lambda = \Lambda_n$.
We will show that any such choice of $\{ \Lambda_n \}_{n \geq 1}$ determines a sequence of maps 
$\mathcal{K}_s^{(n)}$, as defined in (\ref{sf_wio_fv}), which satisfies all four conditions in Assumption~\ref{ass:td_qlms}.
 
The first part of Assumption~\ref{ass:td_qlms} requires
we show that for each $n \geq 1$, the finite volume families of maps $\{ \mathcal{K}_s^{(n)} \}_{s \in [0,1]}$ satisfy Assumption~\ref{ass:fv_qlm_cont}. Since $W_{\gamma} \in L^1( \mathbb{R})$ is real-valued, 
Assumption~\ref{ass:fv_qlm_cont}~(i) is easily verified. We note that integrability of $W_{\gamma}$ is 
a consequence of the estimate (\ref{W_g_dec}) in Corollary~\ref{cor:weight_dec}.
To check the remaining parts of Assumption~\ref{ass:fv_qlm_cont}, we recall that
properties of maps with the form (\ref{sf_wio_fv}) were discussed in Example~\ref{ex_wio_sec_4}.
Assumption~\ref{sf_ass:qlm_ham} guarantees that the methods of
Example~\ref{ex_wio_sec_4} apply, and the remaining details are readily checked.

For our applications, the simple bound
\begin{equation}
\| \mathcal{K}_s^{(n)}(A) \| \leq \| W_{\gamma} \|_1 \| A \| 
\end{equation}
suffices, and thus Assumption~\ref{ass:td_qlms} (ii) is trivially satisfied with $p=0$ and $B(s) = \| W_{\gamma} \|_1$.

The uniform quasi-locality estimate in Assumption~\ref{ass:td_qlms} (iii) can be seen as follows. 
By Assumption~\ref{sf_ass:qlm_ham} (iii), $\Phi(s) \in \mathcal{B}_F$ for each $0 \leq s \leq 1$. 
As a result, given any $n \geq 1$, the model's finite-volume dynamics, i.e. $\tau_t^{\Lambda_n, s}$, satisfies a Lieb-Robinson bound as in Theorem~\ref{thm:lrb2}. In fact, for any $X, Y \subset \Lambda_n$, with $X \cap Y \neq \emptyset$, and any $A \in \cA_X$ and $B \in \cA_Y$, the bound
\begin{eqnarray}
\| [ \tau_t^{\Lambda_n, s}(A), B ] \| & \leq & \frac{2 \| A \| \| B \|}{C_F} \left( e^{2 C_F \| \Phi(s) \|_F |t|} - 1 \right) \sum_{x \in X} \sum_{y \in Y} F(d(x,y)) \nonumber \\
& \leq & C \| A \| \| B \| |X| e^{v|t| - g(d(X,Y))}
\end{eqnarray}  
holds for all $t \in \mathbb{R}$ and $s \in [0,1]$. Here we have estimated the weighted $F$-function $F(r) = e^{-g(r)}F_0(r)$ and set
\begin{equation} \label{def-lrb-v}
C = \frac{2}{C_F} \| F_0 \|  \quad \mbox{and} \quad v = 2 C_F \sup_{0 \leq s \leq 1} \| \Phi (s) \|_F \, .
\end{equation}

In this case, Lemma~\ref{lem:gen_G_ql_est} applies. We 
have that for any $0< \epsilon <1$ and $A$ and $B$ as above, the bound
\begin{equation} \label{sf_fv_wio_qlm_dec_est}
\| [\mathcal{K}_s^{(n)}(A), B] \| \leq 2 \| A \| \| B \| |X| G^{\epsilon}(d(X,Y)) 
\end{equation}
holds with decay function $G^{\epsilon}$ as in (\ref{gen_G_ql_dec_est}). In principle, here we have used that
$\lim_{r \to \infty} g(r) = \infty$, although one only needs that it becomes sufficiently large. Note further that
the bound in (\ref{sf_fv_wio_qlm_dec_est}) above is uniform in $n \geq 1$ as well as $0 \leq s \leq 1$, and
we have proven Assumption~\ref{ass:td_qlms} (iii).

We now demonstrate that Assumption~\ref{ass:td_qlms} (iv) holds. 
Fix $X \in \mathcal{P}_0( \Gamma)$ and $A \in \cA_X$. For $n \geq 1$ sufficiently large, $X \subset \Lambda_n$
and thus for any $T>0$, the following estimate holds:
\begin{equation}
\| \mathcal{K}_s(A) - \mathcal{K}_s^{(n)}(A) \| \leq \int_{|t| \leq T} \| \tau_t^s(A) - \tau_t^{\Lambda_n, s}(A) \| |W_{\gamma}(t)| \, dt + 2 \| A \| \int_{|t|>T} |W_{\gamma}(t)| \, dt .
\end{equation}
Appealing to the continuity result in Corollary~\ref{cor:continuity}, see specifically (\ref{sub_diff_bd_iv}), for any $t \in [-T,T]$, 
\begin{eqnarray}
\| \tau_t^s(A) - \tau_t^{\Lambda_n, s}(A) \| & \leq & \frac{\| A \|}{C_F} v|t| e^{v|t|} \sum_{x \in X} \sum_{y \in \Gamma \setminus \Lambda_n} F(d(x,y)) \nonumber \\
& \leq &  \frac{ v \| F_0 \|}{C_F} |X| \| A \| |t| e^{vT - g(d(X, \Gamma \setminus \Lambda_n))}
\end{eqnarray}
where $v$ is as in (\ref{def-lrb-v}), and in this particular application,
\begin{equation}
2 I_{t,0}(\Phi(s)) = 2C_F \| \Phi (s)\|_F |t| \leq v |t| \, ,
\end{equation}
see e.g. (\ref{ItsPhi}). Thus if $T>0$ is sufficiently large (e.g. $\gamma T \geq e^9$), then by 
Corollary~\ref{cor:weight_dec}, 
\begin{equation}
\int_{|t| > T} |W_{\gamma}(t)| \, dt \leq \frac{972}{49 \gamma \eta} c e^4 \left( f_{\gamma}(T) \right)^3 e^{-\eta f_{\gamma}(T)} \, ,
\end{equation}
where $f_\gamma(t)$ is as defined in (\ref{sub_exp_f}). Arguing now as in the proof of Lemma~\ref{lem:gen_F_ql_est}, let
$0< \epsilon <1$, set $d_n = d(X, \Gamma \setminus \Lambda_n)$, and take $T$ to be defined by $vT = (1- \epsilon) g(d_n)$.  We have proven that
there are positive numbers $C_1$ and $C_2$ for which 
\begin{equation}
\| \mathcal{K}_s(A) - \mathcal{K}_s^{(n)}(A) \| \leq 2 \| A \| |X| \left( C_1 + C_2 f_{\gamma_{\epsilon}}(g(d_n))^3 \right)e^{- \eta f_{\gamma_{\epsilon}}(g(d_n))}
\end{equation}
for $n \geq 1$ sufficiently large. For example, one may take
\begin{equation}
C_1 = \frac{v \| F_0 \|}{2C_F} \int_{\mathbb{R}} |t| |W_{\gamma}(t)| \, dt \quad \mbox{and} \quad C_2 = \frac{972}{49 \gamma \eta} c e^4 \, .
\end{equation}
This completes Assumption~\ref{ass:td_qlms} (iv), and so Proposition~\ref{prop:sf_wio_sat} is proven.
\end{proof}

We can now introduce the spectral flow for the class of models under consideration. 
As above, all comments below are valid for any choice of $\gamma>0$, which is suppressed in the notation.
The spectral flow automorphism can be defined for any choice of $\gamma >0$ and, under the general assumptions above, the 
spectral flow is quasi-local with $\gamma$-dependent estimates. It is only for the special relation with the
spectral projection $P(s)$ as in \eq{sf:followsp} that $\gamma$ needs to be a lower bound for the gap in the spectrum
as described in Assumption \ref{sf_ass:fv_uni_gap}.

Given Proposition~\ref{prop:sf_wio_sat}, we know that the family of maps $\{ \mathcal{K}_s \}_{s \in [0,1]}$ satisfies
Assumption~\ref{ass:td_qlms}. By Assumption~\ref{sf_ass:qlm_ham}, $\Phi'$ is a well-defined interaction on
$\mathcal{A}_{\Gamma}$, and moreover, $\Phi'$ is a suitable initial interaction in the sense of 
Assumption~\ref{ass:suit_td_int}, in particular, $\Phi_1' \in \mathcal{B}_F([0,1])$ as is clear from (\ref{sf_base_int_bd}). 
In this case, we are in a position to apply Theorem~\ref{thm:qli_conv}; we first introduce the relevant notation. 
Let $\{ \Lambda_n \}_{n \geq 1}$ be a sequence of (non-empty) increasing and exhaustive finite subsets of $\Gamma$. For each $n \geq 1$ and any 
$0 \leq s \leq 1$, consider the transformed (bounded) Hamiltonian 
\begin{equation} \label{fv-sf-gen}
\mathcal{K}_s^{(n)}(H^{\Phi'}_{\Lambda_n}(s)) = \sum_{X \subset \Lambda_n} \mathcal{K}_s^{(n)}( \Phi'(X,s)) 
\end{equation}
which acts on $\mathcal{H}_{\Lambda_n}$; compare with (\ref{fv_transformed_ham}). As in
Section~\ref{sec:qli_tl}, see the unique strong solution of 
\begin{equation}
\frac{d}{ds} U_n(s) = -i \mathcal{K}_s^{(n)}(H^{\Phi'}_{\Lambda_n}(s)) U_n(s) \quad \mbox{with} \quad U_n(0) = \idty
\end{equation}
can be used to define a family of automorphisms of $\cA_{\Lambda_n}$ by setting
\begin{equation} \label{fv_sf_dyn}
\alpha_s^{(n)}(A) = U_n(s)^* A U_n(s) \quad \mbox{for all } A \in \cA_{\Lambda_n} \mbox{ and } s \in [0,1] \,,
\end{equation}
see specifically \eqref{fv_app_td_qli} and \eqref{fv_app_qli_dyn}. We will refer to the automorphisms $\alpha_s^{(n)}$ as the finite-volume spectral flow dynamics. To estimate the quasi-locality of $\alpha_s^{(n)}$, we fix a locally normal
product state $\rho$ on $\cA_{\Gamma}$ and proceed as in Section~\ref{sec:qli_tl}. Consider the
finite-volume, $s$-dependent interaction $\Psi_n$ with terms $\Psi_n(Z,s)$, for $Z \subset \Lambda_n $ and $0 \leq s \leq 1$, defined by
\begin{equation} \label{sf_fv_qli_def}
\Psi_n(Z,s) = \sum_{m \geq 0} \sum_{\stackrel{X \subset Z:}{X(m) \cap \Lambda_n = Z}} \Delta_{X(m)}^{\Lambda_n} ( \mathcal{K}_s^{(n)}(\Phi'(X,s)))
\end{equation}
and further, the corresponding infinite-volume interaction $\Psi$ given by 
\begin{equation} \label{sf_iv_qli_def}
\Psi(Z,s) = \sum_{m \geq 0} \sum_{\stackrel{X \subset Z:}{X(m) = Z}} \Delta_{X(m)}( \mathcal{K}_s(\Phi'(X,s))) \quad \mbox{for any } Z \in \mathcal{P}_0( \Gamma) \mbox{ and each } s \in [0,1] \, , 
\end{equation}
again, one should compare with (\ref{fv_app_td_qli}) and (\ref{def_iv_qli}). 

The main result of this subsection is as follows.
\begin{thm} \label{thm:sf_ints_conv}  Consider a quantum lattice system comprised of a $\nu$-regular metric space $(\Gamma, d)$ and $\cA_{\Gamma}$. Suppose that Assumption~\ref{sf_ass:qlm_ham} holds with a weighted $F$-function of the form $F(r) =e^{-g(r)}F_0(r)$ where $F_0(r)$ is as in (\ref{sf_base_F}) and
\begin{equation} \label{g_growth}
g(r) \geq a r^{\theta}
\end{equation}
for some $a>0$ and $0< \theta \leq 1$. Then, $\Psi_n$ converges locally in $F$-norm to $\Psi$ with respect to an $F$-function on $(\Gamma, d)$. Here $\Psi_n$ and $\Psi$ are as defined in (\ref{sf_fv_qli_def}) and (\ref{sf_iv_qli_def}) above.
\end{thm}

We make some comments and point out two corollaries before proving the theorem. 

First, in principle, one can do better than the growth assumption in (\ref{g_growth}); in fact,
one needs only that there is some $0< \epsilon <1$ for which the decay function $G^{\epsilon}$, see (\ref{sf_fv_wio_qlm_dec_est}), satisfies
the conditions of Theorem~\ref{thm:qli_conv} (ii). As can be seen from our comments in Section~\ref{sec:ql_wios} after 
Lemma~\ref{lem:gen_F_ql_est}, any weight function $g$ satisfying (\ref{g_growth}) corresponds to such a decay function $G^{\epsilon}$ which satisfies
the conditions of Theorem~\ref{thm:qli_conv} (ii). However, no weight function $g$ which grows proportional to a 
logarithm, see (\ref{log_g}), corresponds to a decay function 
$G^{\epsilon}$ satisfying the conditions of Theorem~\ref{thm:qli_conv}~(ii).

Next, let us say some more about the $F$-function whose existence plays a crucial role in
the proof of Theorem~\ref{thm:sf_ints_conv}. Given the decay of the initial interaction $\Phi$, see (\ref{sf_base_int_bd}), it is proven
in Proposition~\ref{prop:sf_wio_sat} that the weighted integral operators used in defining the generator of the spectral flow, see (\ref{sf_wio_fv})
and subsequently (\ref{fv-sf-gen}), satisfy the quasi-locality estimate (\ref{sf_fv_wio_qlm_dec_est}). The corresponding decay function $G$
(take $\epsilon = 1/2$ for convenience) has the following form: there exist positive numbers $C_1$, $C_2$, and $d_*$ for which
\begin{equation}
G(d) \sim \left\{ \begin{array}{cl} C_1 & 0 \leq d \leq d_* , \\
C_2 \,  {\rm exp} \left[- \eta f_{\frac{\gamma}{2v}}(g(d)) \right] & d > d_*  \end{array} \right.
\end{equation} 
where we stress that the function $g$ above corresponds to the weight in the $F$-function governing the
decay of the initial interaction $\Phi$. It should be clear that any $F$-function governing the decay of $\Psi$ 
(and similarly $\Psi_n$) will decay no faster than this $G$. Our estimates show that: there are positive numbers $C_1'$, $C_2'$, and $d_*'$ 
for which $\Psi, \Psi_n \in \mathcal{B}_{\tilde{F}}([0,1])$ with 
\begin{equation} \label{tilde-F-form}
\tilde{F}(d) = \begin{cases}
C_1' & 0 \leq d \leq d_*'  \\
\frac{C_2' }{(1+d)^{\xi}}\exp\left[ - \eta' f_{\frac{\gamma}{2v}}( \tilde{g}(d)) \right] & d > d_*' \end{cases}
\end{equation} 
where $\eta'$ is any number strictly less that $\eta$ and the function $\tilde{g}$ may be taken as
$\tilde{g}(d) = \tilde{a} d^{\theta}$ with the same value of $\theta$ (as $g$) but, in general, 
a smaller value of $a$. As is discussed in Section~\ref{app:sec_weight_F},
any function with the form (\ref{tilde-F-form}) is an $F$-function on $(\Gamma, d)$.  
To obtain the local convergence of $\Psi_n$ to $\Psi$, we will need to modify the above
$F$-function slightly, but all relevant estimates on $\Psi$ and $\Psi_n$, see for example 
Corollary~\ref{cor:sf_est_1} and Corollary~\ref{cor:sf_est_2} below, will be made with respect to the function
$\tilde{F}$ described above.  For more details on this, see the discussion following the
statement of Theorem~\ref{thm:qli_conv}.

Finally, our proof of Theorem~\ref{thm:sf_ints_conv} guarantees that Theorem~\ref{thm:exist_iv_dyn} applies and so
\begin{equation} \label{sf-dyn-as-limit}
\lim_{n \to \infty} \alpha_s^{(n)}(A) = \alpha_s(A) \quad \mbox{for any } A \in \mathcal{A}_{\Gamma}^{\rm loc} \mbox{ and } s \in [0,1] \, .
\end{equation}
Here the limit is in norm and the quantity, $\alpha_s$, is the well-defined, infinite-volume
dynamics associated to $\Psi \in \mathcal{B}_{\tilde{F}}([0,1])$, with terms as in (\ref{sf_iv_qli_def}), whose existence is
guaranteed by Theorem~\ref{thm:existbd}.

For ease of later reference, we now state two corollaries providing explicit estimates on the quasi-locality of
the spectral flow. 
\begin{cor} \label{cor:sf_est_1} Under the assumptions of Theorem~\ref{thm:sf_ints_conv}, for any $X,Y \in \mathcal{P}_0( \Gamma)$ with 
$X \cap Y = \emptyset$, the bound 
\begin{equation}
\| [ \alpha_s(A), B] \|  \leq  \frac{2 \| A \| \| B \| }{C_{\tilde{F}}} \left( e^{2 I_{s,0}( \Psi)} - 1 \right) \sum_{x \in X} \sum_{y \in Y} \tilde{F}(d(x,y))
\end{equation}
holds for any $A \in \cA_X$, $B \in \cA_Y$, and $0 \leq s \leq 1$. Here $\tilde{F}$ may be taken as in (\ref{tilde-F-form}).
\end{cor}
Since $\Psi \in \mathcal{B}_{\tilde{F}}([0,1])$, the above is an immediate consequence of Corollary~\ref{cor:continuity} (i).
Our estimates actually show that $\sup_{0 \leq s \leq 1} \| \Psi \|_{\tilde{F}}(s) < \infty$ and, therefore, we have
\begin{equation} \label{sf-psi-int-F-bd}
I_{s,0}(\Psi) = C_{\tilde{F}} \int_0^s \| \Psi \|_{\tilde{F}}(r) \, d r \leq s C_{\tilde{F}} \sup_{0 \leq s \leq 1} \| \Psi \|_{\tilde{F}}(s) \, .
\end{equation}
Given (\ref{sf-psi-int-F-bd}), it is clear that the bound in Corollary~\ref{cor:sf_est_1} may be
further estimated with a linear dependence on $s$. This observation is useful in some applications.

The following corollary is a direct application of Theorem~\ref{thm:exist_iv_dyn} (i).

\begin{cor} \label{cor:sf_est_2} Under the assumptions of Theorem~\ref{thm:sf_ints_conv}, for any $X \in \mathcal{P}_0( \Gamma)$, the bound 
\begin{equation} \label{alpha-I}
\| \alpha_s(A) - A \| \leq  2  |X| \| A \| \| \tilde{F} \| \int_0^s \| \Psi \|_{\tilde{F}}(r) \, dr  
\end{equation}
holds for all $A \in \cA_X$ and $0 \leq s \leq 1$. Here $\tilde{F}$ may be taken as in (\ref{tilde-F-form}).
\end{cor}
 
\begin{proof}[Proof of Theorem~\ref{thm:sf_ints_conv}:]
It is clear that, under the decay assumptions (\ref{g_growth}), Proposition~\ref{prop:sf_wio_sat} holds. 
In this case, for each $\gamma>0$, the family of maps $\{ \mathcal{K}_s^{\gamma} \}_{s \in [0,1]}$
satisfies Assumption~\ref{ass:td_qlms}, and moreover, $\Phi'$ is a suitable initial interaction in the sense
of Assumption~\ref{ass:suit_td_int}. As before, we will suppress the dependence on $\gamma>0$ in what follows.

For any $x,y \in \Gamma$ and $n \geq 1$ large enough so that $x,y \in \Lambda_n$, an application of
Theorem~\ref{thm:qli_est} shows 
\begin{equation} \label{sf_base_int_est_1}
\sum_{\stackrel{Z \subset \Lambda_n:}{x,y \in Z}} \| \Psi_n(Z,s) \| \leq C_1 F(d(x,y)/3) + 
C_2 \sum_{m = \lfloor d(x,y)/3 \rfloor}^{\infty} (1+m)^{\nu +1} G^{\epsilon}(m)
\end{equation}
where $\Psi_n$ is the finite-volume interaction in (\ref{sf_fv_qli_def}), $F$ is the weighted $F$-function
governing the decay of $\Phi$ as in Assumption~\ref{sf_ass:qlm_ham} (iii), and $G^{\epsilon}$ is the 
decay function associated to the family $\{ \mathcal{K}_s^{(n)} \}_{s \in [0,1]}$ as in (\ref{sf_fv_wio_qlm_dec_est}), see also \eqref{gen_G_ql_dec_est}. Our estimates show that one may take 
\begin{equation}
C_1 =  \left( \| W_{\gamma} \|_1 + 8 \kappa^2 \sum_{m=0}^{\infty} (1+ m)^{2 \nu +1}G^{\epsilon}(m) \right) \sup_{0 \leq s \leq 1} \| \Phi \|_{1,1}(s)
\end{equation}
and $C_2 = 8 \kappa \| F \| \sup_{0 \leq s \leq 1} \| \Phi \|_{1,1}(s) $. As we have argued before, the analogue of (\ref{sf_base_int_est_1}) also holds with the interaction $\Psi$ replacing $\Psi_n$ on the left-hand-side and the right-hand-side unchanged.

It is also clear that the decay assumption (\ref{g_growth}) guarantees that for any $0< \delta < \eta$,
\begin{equation}
C_{\delta} = \sum_{m=0}^{\infty} (1+m)^{\nu + 1 + \xi} \left( \frac{C}{2 v} + \frac{243}{49 \gamma \eta} c e^4 f_{\gamma_{\epsilon}}(g(m))^3 \right) e^{- \delta f_{\gamma_{\epsilon}}(g(m))} < \infty
\end{equation}
As such, for $d = d(x,y)$ sufficiently large, we may further estimate the right-hand-side of (\ref{sf_base_int_est_1}) by 
\begin{equation}
C_1 \frac{e^{-g(d/3)}}{(1+ d/3)^{\xi}} + C_2 C_\delta \frac{e^{-(\eta-\delta) f_{\gamma_{\epsilon}}( g(\lfloor d/3 \rfloor))} }{(1+ \lfloor d/3 \rfloor)^{\xi}}.
\end{equation}
For large $d$, the second term above dominates. Using the
facts provided in Section~\ref{subsec:regsets}, it is clear that an $F$-function, $\tilde{F}$, of the form in (\ref{tilde-F-form}) bounds this quantity. For any such $\tilde{F}$, our estimates are uniform with respect to $s$, and so we have proven that 
\begin{equation}
\sup_{0 \leq s \leq 1} \| \Psi \|_{\tilde{F}}(s) < \infty \, .
\end{equation}
In fact, the same uniform estimate also holds for the finite-volume interactions $\Psi_n$. 

Since the form of the decay function $G^{\epsilon}$, as in (\ref{sf_fv_wio_qlm_dec_est}), 
is explicit, see (\ref{gen_G_ql_dec_est}), it is clear that
$\sqrt{G^{\epsilon}}$ has a finite $2 \nu +1$ moment. Arguing as above, a similar, yet different, 
$F$-function $\hat{F}$ can be produced which satisfies the assumptions of 
Theorem~\ref{thm:qli_conv} (ii) with $\alpha =1/2$; in fact, any choice of $0< \alpha <1$
suffices. In this case, $\Psi_n$ converges locally in $F$-norm to $\Psi$ with respect to this
function $\hat{F}$. As indicated previously, for the estimates in Corollary~\ref{cor:sf_est_1}
and Corollary~\ref{cor:sf_est_2}, one can use the original $F$-function $\tilde{F}$
having the form (\ref{tilde-F-form}).
\end{proof}

\section{Automorphic equivalence of gapped ground state phases}\label{sec:equivalence_of_phases}

\subsection{Uniformly gapped curves and automorphic equivalence} In this section, we use the spectral flow to study gapped ground state phases of a quantum lattice model $(\Gamma, d)$ and $\cA_\Gamma$. As in the previous sections, we will discuss both finite and infinite volume systems and take the thermodynamic limit along a sequence of increasing and absorbing finite volumes. To this end, we will consider the following set up for this section.

Throughout this section, let $(\Gamma,d)$ be a fixed $\nu$-regular metric space with a weighted $F$-function 
of the form $F(r)= e^{-g(r)}F_0(r)$, where $F_0$ is an $F$-function for $(\Gamma,d)$ of the form \eq{sf_base_F} and $g$ is a non-negative, non-decreasing, subadditive function bounded below by $ar^\theta$ for some $\theta\in(0,1]$. In addition, we consider a fixed sequence of increasing and absorbing finite volumes $\Lambda_n \uparrow\Gamma$, and with the convention that we always take the thermodynamic limit with respect to a subsequence of this sequence. We will use the notation $\cB_F^1([0,1])$ to denote the space of differentiable curves of interactions $\Phi(s)\in \cB_F$, $s\in[0,1]$, satisfying  Assumption \ref{sf_ass:qlm_ham}. At each $x\in \Gamma$ we may have a densely defined self-adjoint $H_x$, but these
         we regard as fixed. Specifically, we only consider here relations between models with different interactions $\Phi(s)$ but with the same $\{H_x\mid x\in\Gamma\}$.

For simplicity we will assume that the finite-volume Hamiltonians for
the models parametrized by $s\in [0,1]$, are defined by
\be
H_{\Lambda}(s)=\sum_{x\in\Lambda} H_x  +\sum_{X\subset\Lambda} \Phi(X,s).
\label{hamcurve}\ee

Within the context described above, we now introduce the notion of a uniformly gapped curve of models or, equivalently, a curve of uniformly gapped interactions for which we use the notation 
$E_\Lambda(s)=\inf\spec H_{\Lambda}(s)$ to denote the ground state energy of $H_\Lambda(s)$.

\begin{defn}\label{def:uniformlygapped}
Let $\gamma >0$. A curve of interactions $\Phi\in\cB^1_F([0,1])$ is called {\em uniformly gapped with gap $\gamma$}, if there exists a non-negative sequence $(\delta_n)_n$, with $\lim_n \delta_n =0$, 
such that for all $n\geq 1$ and $s\in [0,1]$
\be
\spec H_{\Lambda_n}(s) \subset [E_{\Lambda_n}(s), E_{\Lambda_n}(s)+\delta_n] \cup [E_{\Lambda_n}(s)+\delta_n+\gamma,\infty),
\label{uniformlygapped}\ee
where $H_{\Lambda_n}(s)$ is the finite volume Hamiltonian defined in \eq{hamcurve}.
\end{defn}

We can leave $\gamma$ unspecified and call the curve simply {\em uniformly gapped} if there exists $\gamma>0$ such that it is uniformly gapped with gap $\gamma$. 

It is well-known that the spectral gap generally depends on the boundary conditions. Our choice to define the path of Hamiltonians \eqref{hamcurve} with respect to a single interaction leads to boundary conditions that are not necessarily the most general ones of interest; studying all possible cases at once would lead to quite onerous notation, 
	which we want to avoid in this discussion. Suffice it to note that everything in this section could be generalized to the situation where we have boundary conditions expressed by a sequence $\Phi_n\in \cB_F^1([0,1])$, by requiring that $\Phi_n$ converges
	locally (in a uniform version of Definition \ref{def:lcfnorm}) in a suitable norm to some $\Phi\in\cB_F^1([0,1])$. In this case, $\Phi_n$ is then used to define the Hamiltonian \eq{hamcurve} on the volume $\Lambda_n$ for each $n$. For example, this can be used to study finite systems in $\Ir^\nu$ with periodic boundary conditions. Even without considering $n$-dependent interactions, the present set-up allows one to study the effects of certain boundary conditions. For example, by replacing $\Gamma$ by a subset $\Gamma_0\subset\Gamma$ and different sequences of finite volumes $\Lambda_n$, models defined with the same interaction $\Phi$ may show different behavior. An example of this is discussed in detail for a class of so-called PVBS models in \cite{bachmann:2015,bishop:2016a}. There, $\Gamma_0$ is the half-space in $\Gamma=\Ir^\nu$ defined by an arbitrary hyperplane. For these models, the spectral gap is shown to depend non-trivially on the orientation of the hyperplane.

We use the notion of a uniform gap to define a relation $\sim$ on $\cB_F$ as follows.

\begin{defn}\label{def:equivalence_gapped_phases}
For $\Phi_0,\Phi_1\in\cB_F$, we say that $\Phi_0$ and $\Phi_1$ are {\em equivalent}, denoted by $\Phi_0\sim\Phi_1$, if there 
exists a uniformly gapped curve $\Phi\in\cB^1_F([0,1])$ such that $\Phi(0)=\Phi_0$ and $\Phi(1)=\Phi_1$.
\end{defn}

In the physics literature, two models $\Phi_0$ and $\Phi_1$ are said to be in the same {\it gapped 
ground state phase} if $\Phi_0\sim\Phi_1$ \cite{chen:2010,chen:2011}. Studying curves of models has proved to be fruitful also in mathematical studies \cite{bachmann:2012,bachmann:2012bf,bachmann:2014a,bachmann:2014c,bachmann:2015a}.
In this section we explore some essential properties of models that belong to the same gapped ground state phase. First, however, we show that the relation $\sim$, used to define this notion is indeed an equivalence relation.

\begin{prop}\label{prop:equivalent_gapped_phases}
The relation $\sim$ defined in Definition \ref{def:equivalence_gapped_phases} is an equivalence relation on $\cB_F$.
\end{prop}
\begin{proof}
The defining properties of reflexivity and symmetry of an equivalence relation follow by considering constant curves $\Phi(s)= \Phi_0$, for all $s\in [0,1]$, and reversed curves $\Phi^{-1} (s) = \Phi(1-s)$. 
For transitivity, consider two curves $\Phi^{(1)}(s), \, \Phi^{(2)}(s)\in \cB_F^1([0,1])$ such that $\Phi^{(1)}(1) = \Phi^{(2)}(0)$, and define $\Phi(s) \in \cB_F^1([0,1])$ by
\[
\Phi(s) =\begin{cases}
\Phi^{(1)}(\sin(\pi s)) & s\leq 1/2 \\
\Phi^{(2)}(1-\sin(\pi s)) & s >1/2
\end{cases}
\]
Here, the re-parameterization of $s$ in the piecewise definition is chosen only to ensure the differentiability of $\Phi$ at $s=1/2.$ Other re-parameterizations will also work. Transitivity follows from setting $\delta_n=\max (\delta^{(1)}_n,\delta^{(2)}_n)$, and $\gamma = \min (\gamma^{(1)},\gamma^{(2)})$, where $\delta^{(i)}_n$ and $\gamma^{(i)}$, $i=1,2$, refer to
the sequences and the gap for the two curves.
\end{proof}

Note that, without loss of generality, we can assume that the sequence $(\delta_n)$ in Definition \ref{def:uniformlygapped} is non-increasing. It  is also easy to see that for uniformly gapped  $\Phi$, the spectral projection $P_n(s)$ of $H_{\Lambda_n}(s)$ associated with the interval 
$[E_{\Lambda_n}(s), E_{\Lambda_n}(s)+\delta_n]$ becomes independent of the choice of sequence $(\delta_n)$ for large $n$ in the sense that for any two sequences $(\delta_n)$ and $(\delta_n^\prime)$ for which \eq{uniformlygapped} holds, the spectral projections associated with the intervals $[E_{\Lambda_n}(s), E_{\Lambda_n}(s)+\delta_n]$ and $[E_{\Lambda_n}(s), E_{\Lambda_n}(s)+\delta_n^\prime]$ coincide for sufficiently large $n$.

Let $\Phi$ be uniformly gapped. Then, consider the collection of states of $\cA_{\Lambda_n}$ supported on the spectral subspace of $H_{\Lambda_n}(s)$ 
associated with the intervals $[E_{\Lambda_n}(s), E_{\Lambda_n}(s)+\delta_n]$. More precisely, define 
$$
\cS_n(s) = \{ \omega\in\cS(\cA_{\Lambda_n}) \mid \omega (P_n(s)) =1\}.
$$
Here, for any complex $C^*$-algebra $\cA$ with unit $\idty$, $\cS(\cA)$ denotes the state space of $\cA$, that is the set of positive linear functionals on $\cA$ with $\omega(\idty)=1$. The remarks in the previous paragraph show that $\cS_n(s)$ becomes independent of the choice of the sequence $(\delta_n)$ for large $n$. Therefore, it is possible to define
$$
\cS(s) = \{\omega\in\cS(\cA_\Gamma) \mid \exists (n_k) \mbox{ increasing and } \omega_k \in \cS_{n_k}(s) \mbox{ s.t. } \lim_k \omega_k(A) = \omega(A), 
\forall  A \in \cA_\Gamma^{\rm loc}\}.
$$
In this sense, $\cS(s)$ is the set of all weak$-*$ limits of states in $\cS_n(s)$. It can be shown that the elements $\omega \in\cS(s)$ are ground states of the infinite-volume model defined by the dynamics $\tau_t$ obtained as the thermodynamic limit of the model with Hamiltonians \eq{hamcurve} \cite{QLBII}. The prime example to keep in mind is that the set  $\cS_n(s)$ also consists of ground states for a Hamiltonian $H_{\Lambda_n}(s)$ that has a uniform lower-bound, denoted by $\gamma >0$, separating the ground state energy from the rest of the spectrum. However, it will be interesting to consider the slightly more general set-up we have introduced above.

Let us now fix a uniformly gapped curve $\Phi\in\cB_F^1([0,1])$ with gap $\gamma$. As an application of Theorem~\ref{thm:sf_ints_conv} and the comments following it, we have
strongly continuous spectral flow automorphisms $\alpha_{s}^{(n)}$ for the curve of finite-volume on $\Lambda_n$, and $\alpha_{s}$ for the infinite system on $\Gamma$. Here, the uniform gap of the curve plays the role of the parameter $\gamma$ in the construction of the spectral flow. Moreover, Theorem~\ref{thm:sf_ints_conv}, see specifically \eq{sf-dyn-as-limit}, establishes that $\alpha_{s}$ is the strong limit of $\alpha_{s}^{(n)}$, and the convergence of this limit is uniform for $s\in [0,1]$. Moreover, we can use the spectral flow $\alpha_s$ to construct a co-cycle of automorphisms $\alpha_{t,s} :=\alpha_s^{-1}\circ \alpha_t$, for all $t,s\in[0,1]$. We can similarly define a collection of finite volume co-cycles, $\alpha_{t,s}^{(n)}$.

Our next goal is to show that the spectral flow co-cycle establishes a close relationship between the sets $\cS(s)$ for different values of $s$, which we refer to
as {\em automorphic equivalence of gapped ground state phases}. Using the definition of $\alpha_{t,s}^{(n)}$ above, Theorem \ref{thm:sf} establishes the following relationships between the spectral projections $P_n(s)$:
\be\label{projectionflow}
\alpha^{(n)}_{t,s}(P_n(t)) = P_n(s), \mbox{ for all }s,t\in [0,1].
\ee
As an immediate consequence we have that $\omega_n\circ \alpha_{t,s}^{(n)} \in \cS_n(t)$ for any $\omega_n \in \cS_n(s)$ as
\[\omega_n\circ \alpha_{t,s}^{(n)}(P_n(t)) =\omega_n(P_n(s)) = 1.\]
Since $\alpha_{t,s}^{(n)}$ is an automorphism, it is
invertible. In fact, its inverse is given by  $\alpha_{s,t}^{(n)}$. As such, we see that composition with $\alpha_{t,s}^{(n)}$
defines a bijection between the sets $\cS_n(s)$ and $\cS_n(t)$. Explicitly
\be
\mathcal{S}_n(t)=\{ \omega_n \circ \alpha_{t,s}^{(n)} \mid \omega_n\in \mathcal{S}_n(s)\}.
\label{finite-volume_bijection}\ee
The next theorem extends this bijection to the thermodynamic limit. The quasi-local properties of the spectral flow established in Section \ref{sec:spectral-flow} play an important role both at the technical and the conceptual level. Technically, they are the main ingredient in establishing the convergence of the thermodynamic limit. Conceptually, the fact that $\alpha_s$ is a dynamics generated by a short-range interaction shows that local properties of the states at different values of $s$
are related by a `natural,' finite-time, unitary evolution.

\begin{thm} \label{thm:automorphic_equivalence}
For all $s,t\in [0,1]$, the spectral flow automorphism $\alpha_{s,t}$ provides a bijection between the sets $\cS(s)$ and $\cS(t)$
by composition:
\begin{equation}
\mathcal{S}(t)=\mathcal{S}(s)\circ\alpha_{t,s}
\label{infinite_gs_relation}
\end{equation}
\end{thm}

\begin{proof}
This is a direct consequence of~\eq {finite-volume_bijection}, Theorem~\ref{thm:sf_ints_conv} and the Lemma~\ref{lem:convergence-of-states} below.
\end{proof}

\begin{lem}\label{lem:convergence-of-states}
Let $(\alpha_n)_n$ be a strongly convergent sequence of automorphisms of a $C^*$-algebra
$\mathcal{A}$, converging to $\alpha$ and let $(\omega_n)_n$ be a sequence of states
on $\mathcal{A}$. Then the following are equivalent:
\begin{enumerate}
\item[(i)] $\omega_n$ converges to $\omega$ in the weak-$*$ topology;
\item[(ii)] $\omega_n\circ\alpha$ converges to $\omega\circ\alpha$ in the weak-$*$ topology;
\item[(iii)] $\omega_n\circ\alpha_n$ converges to $\omega\circ\alpha$ in the weak-$*$ topology.
\end{enumerate}
\end{lem}
\begin{proof}
(i)$\Leftrightarrow$(ii) follows immediately from the fact that $\alpha$ and $\alpha^{-1}$
are automorphisms. Now if (ii) holds, the limit of the second term in the RHS of
\begin{equation*}
\vert(\omega_n\circ\alpha_n)(A)-(\omega\circ\alpha)(A)\vert \leq \vert\omega_n(\alpha_n(A)-\alpha(A)) \vert + \vert\omega_n(\alpha(A))-\omega(\alpha(A)) \vert\,,
\end{equation*}
vanishes. So does the limit of the first term since
\begin{equation*}
\vert\omega_n(\alpha_n(A)-\alpha(A)) \vert \leq \Vert\omega_n\Vert \Vert\alpha_n(A)-\alpha(A) \Vert \longrightarrow 0
\end{equation*}
as $\omega_n$ is a state. Therefore, (iii) holds. A similar argument implies (iii)$\Rightarrow $(ii).
\end{proof}

Recall the following essential property of the spectral flow automorphisms with parameter $\gamma>0$ constructed in Section~\ref{sec:spectral-flow} for a family of Hamiltonians $H_\Lambda(s)$. Suppose that $P(s)$ is the spectral projection of $H_\Lambda(s)$ associated to a bounded interval $[a(s), \, b(s)]$ (with $a(s)$ and $b(s)$ differentiable) that is gapped from the rest of the spectrum by $\gamma>0$, i.e.
\[
(a(s)-\gamma,a(s))\cap \spec(H_\Lambda(s))
=(b(s), b(s)+\gamma)\cap \spec(H_\Lambda(s))=\emptyset \; \text{  for all  } \; s\in [0,1].
\]
Then by Theorem~\ref{thm:sf} the spectral flow $\alpha_{t,s}$ with parameter $\gamma$ associated with $H_\Lambda(s)$ once again maps $P(t)$ to $P(s)$.  In the discussion above we focused on gapped ground state phases, for which the relevant part of the spectrum is at the bottom. We will describe examples in some detail in Section \ref{sec:examples_gapped_curves}. There is no reason, however, why a similar strategy could not be employed to study states supported in spectral subspaces associated with an isolated part elsewhere in the spectrum. For example, an isolated band of excited states could also be studied with the help of spectral automorphisms. This new and largely unexplored territory seems promising to us.

We conclude this section with the following result regarding the continuity of the spectral gap above the ground state energy of the GNS Hamiltonian $H_{\omega_s}$ of
an infinite volume ground state $\omega_s\in \cS(s)$ for the case of quantum spin systems, i.e., the single-site Hilbert spaces $\cH_x$ are finite-dimensional and the dynamics is generated
by an interaction $\Phi\in\cB_F^1([0,1])$ for a suitable $F$-function $F$. The restriction to the case of quantum spin systems is because we rely
on some well-known properties of the dynamics and, in particular, its generator in that case. The finite-dimensionality of the single-site Hilbert spaces is not essential, but the boundedness of the interactions is, in addition to the general setup described at the beginning of this section, including Assumption \ref{sf_ass:qlm_ham}.

\begin{thm}\label{thm:semi-continuity_gap}
Consider a quantum spin model defined by a uniformly gapped curve of interactions $\Phi\in\cB_F^1([0,1])$, with gap $\gamma >0$.
Fix  $s_0\in [0,1]$ and let $\{H_{\omega_s}\}_{s\in[0,1]}$ denote the set of GNS Hamiltonians associated to the states $\omega_s\in \mathcal{S}(s)$ of 
the form $\omega_s = \omega_{s_0} \circ \alpha_{s,s_0}$ for some $\omega_{s_0}\in\cS(s_0)$, and with the spectral flow $\alpha_{s,s_0}$ corresponding 
to the parameter $\gamma$.  If for all $s\in[0,1]$, $\ker H_{\omega_s}$ is one-dimensional, then
	\[
	\gamma(s) := \sup\{\delta>0 \; : \;\spec(H_{\omega_s})\cap(0,\delta) = \emptyset\}
	\]
	 is a upper-semicontinuous function of $s$.
\end{thm}
\begin{proof}
	Recall that for each $s$, the infinite volume dynamics $\tau_t^{(s)}$ is a strongly continuous group of automorphisms of $\mathcal{A}_\Gamma$ generated by a closed operator $\delta^{(s)}$, i.e. $\tau_t^{(s)}=e^{it\delta^{(s)}}$, and that $\mathcal{A}^{\rm loc}_\Gamma$ is a core for $\delta^{(s)}$. 
	
First, we show that for all $t,s,s_0\in I$, $\alpha_{t,s} (\mathcal{A}^{\rm loc}_\Gamma)$ is a core for $\delta^{(s_0)}$. Since $\mathcal{A}^{\rm loc}_\Gamma$ is a core and $\alpha_{t,s}$
is an automorphism, we only need to show that  $\alpha_{t,s}(A) \in\dom(\delta^{(s_0)})$, for all $A\in \mathcal{A}^{\rm loc}_\Gamma$. Let $X = \mathrm{supp}(A)$ and $\Pi_{X(n)}$ be as in \eqref{PiX}. Then by \eqref{PiXLimit}
\[
\alpha_{t,s}(A) = \lim_{n\to\infty}\Pi_{X(n)}(\alpha_{t,s}(A)).
\]
The result follows from showing that $\delta^{(s_0)}(\Pi_{X(n)}(\alpha_{t,s}(A)))$ is convergent. Using the telescopic property of $\Delta_{X(n)}$ and linearity of $\delta^{(s_0)}$, it follows that for any $n\geq 0$
\[
\delta^{(s_0)}(\Pi_{X(n)}(\alpha_{t,s}(A))) = \sum_{m=0}^n \delta^{(s_0)}(\Delta_{X(m)}(\alpha_{t,s}(A))),
\]
which is absolutely convergent by Proposition~\ref{prop:gen_comp_qlms}.

	Now, let $\gamma(s)$ denote the spectral gap of $H_{\omega_s}$ and pick any $s_0\in [0,1]$. By the variational principle,
	\begin{equation}\label{gap_by_VP}
	\gamma(s) = \inf \{ \omega_s(A^* \delta^{(s)}(A)) \mid A \in \mathcal{C} , \; \omega_s(A)=0, \;\omega_s(A^*A)=1\},
	\end{equation}
	where $\mathcal{C} $ is any core for $\delta^{(s)}$. Using that $\omega_s = \omega_{s_0}\circ \alpha_{s,s_0}$ and $\mathcal{C}_s = \alpha_{s_0,s}(\mathcal{A}^{\rm loc}_\Gamma)$ is a core for  $\delta^{(s)}$, we have the following identity:
	\be\label{core_equality}
	\{ A\in \mathcal{C}_s \mid \omega_s(A)=0, \;\omega_s(A^*A)=1\} = \{ \alpha_{s,s_0} (A) \mid A \in \mathcal{A}^{\rm loc}_\Gamma, \omega_{s_0}(A)=0,\omega_{s_0}(A^*A)=1\}.
	\ee
	For any $A\in \cA_{\Gamma}^{\rm loc}$, define the function $f_A(s) = \omega_{s_0} (A^* \alpha_{s,s_0}\circ \delta^{(s)} \circ \alpha_{s_0,s} (A))$, and consider the family of functions,
	$$
	\mathcal{F} = \{f_A \mid A \in \mathcal{A}^{\rm loc}_\Gamma, \omega_{s_0}(A)=0,\omega_{s_0}(A^*A)=1\}.
	$$
	Using \eqref{core_equality}, the expression for the gap in \eq{gap_by_VP} can be rewritten in terms of the family $\mathcal{F}$ as
	\begin{equation}\label{gap_by_F}
	\gamma(s) = \inf \{ f(s) \mid f\in \mathcal{F}\}.
	\end{equation}
	
	The result follows from showing that all $f\in \mathcal{F}$ are continuous. This can be seen by expressing the operator $\alpha_{s,s_0}\circ \delta^{(s)} \circ \alpha_{s_0,s}$ as the generator
	of the dynamics for a new $s$-dependent interaction $\Psi\in \cB_{\tilde F}(I)$. Using the continuity of automorphisms  and $\alpha_{s,s_0}\circ \alpha_{s_0,s}={\rm id}$, we have
	\begin{equation}
	\alpha_{s,s_0}\circ \delta^{(s)} \circ \alpha_{s_0,s}(A) = \lim_n [\alpha_{s,s_0}(H_{\Lambda_n}(s)),  A ].
	\end{equation}
	Theorem~\ref{thm:qli_est} and Corollary~\ref{cor:sf_est_1} imply the existence of an $F$-function $\tilde F$, see \eqref{good_F_dec}, and a strongly continuous $\Psi\in \cB_{\tilde F}(I)$ such that the RHS
	is the generator determined by $\Psi(s)$, which is again locally bounded and quasi-local.  It follows that the map $s\mapsto \alpha_{s,s_0}\circ \delta^{(s)} \circ \alpha_{s_0,s}(A) $ is continuous for each finite $X\subset\Gamma$ and $A\in\mathcal{A}_X$. Therefore, $f(s) = \omega_{s_0} (A^* \alpha_{s,s_0}\circ \delta^{(s)} \circ \alpha_{s_0,s} (A))$ defines a
	continuous function.
\end{proof}

As far as we are aware, for all models that satisfy the conditions of the theorem, the gap appears to be {\em continuous} in the parameter, not just semicontinuous. In particular, the gap is continuous when perturbation theory applies. This raises the question whether one indeed has continuity of the spectral gap as long as it is strictly positive, or whether additional assumptions are needed for continuity. Needless to say, the gap is not always stable and so should not be expected to be continuous in general on a domain where it vanishes at some points.

\subsection{Automorphic equivalence with symmetry}

In the previous section we introduced the classification of gapped ground state phases
through equivalence classes of interactions for which there exists an interpolation by a 
uniformly gapped curve. We showed that within each equivalence class the sets of ground
states are mapped into each other by an automorphism with good quasi-locality properties
(the spectral flow derived from the uniformly gapped curve of interactions interpolating between
the models). Implicit in this description is the idea that any curve of interactions interpolating
between two models in {\em distinct} phases (different equivalence classes) must contain at
least one point where the gap vanishes. Such points are called quantum critical points and
one says that a quantum phase transition occurs in the system \cite{sachdev:1999}.

Physical systems often have symmetries that play an important role. In the description 
of certain phenomena, it may be essential that a certain symmetry be present in the model. 
This led to the concept of symmetry protected gapped phases \cite{chen:2013,duivenvoorden:2013b} 
due to the observation that if one only allows curves of interactions that all possess a given symmetry, a finer classification of gapped ground state phases may arise. A nice example of this are the $\Ir_2\times\Ir_2$ protected phases of the spin-1 chain \cite{pollmann:2010,chen:2011,tasaki:2018,ogata:2019}. In general, the equivalence classes break up into subclasses if a restricted set of uniformly gapped curves of interactions is used to define the
equivalence relation. This prompts us to revisit the notion of automorphic equivalence in the
presence of symmetry.

A symmetry is usually specified by the action of a group $G$ (as automorphisms) on the algebra of observables
of the system. Although there are interesting symmetries that do not fit the
framework of group representations by automorphisms, such as dualities and quantum group 
symmetries, we limit the discussion here to that setting. In general, we use the label
$G$ to specify the presence of a certain symmetry. So, we will consider spaces of interactions
$\cB_F^G\subset \cB_F$ and of curves of interactions $\cB_F^{1,G}([0,1])\subset \cB^1_F([0,1])$.
To be clear, in this context $G$ stands for the full specification of the symmetry including its action on the system, not just the abstract group.

Here are four important classes of symmetries:

\begin{enumerate}
\item[(i)] {\em Local symmetries} are described by automorphisms $\beta$ of $\cA_\Gamma$ with the property that
they leave the single-site algebras, $\cA_{\{x\}} = \cB(\cH_x)$, invariant. Specifically, we assume the restrictions of
$\beta$ to  $\cA_{\{x\}}$, $x\in\Gamma$, are inner automorphisms given in terms of a unitary $U_x\in\cB(\cH_x)$:
$\beta(A) = U_x^*AU_x$ for $A\in\cB(\cH_x)$. This type of symmetries are sometimes called {\em gauge symmetries} 
because gauge symmetries are of this form. Thus, any local symmetry $\beta$ is determined by a family of unitaries 
$\{U_x\in\cB(\cH_x)
\mid x\in\Gamma\}$. We say that $\beta$ is a symmetry of $\Phi$ if $\beta(\Phi(X)) =\Phi(X)$, for all $X\in\cP_0(\Gamma)$. 
It is easy to see that this implies that $\beta$ commutes with the dynamics $\tau_t$ generated by $\Phi$: $\beta\circ \tau_t=\tau_t\circ\beta$. 
If $\Phi$ depends on a parameter $s$ or on the time $t$, the symmetry 
condition is assumed to hold pointwise in $s$ and/or $t$. The set of all local symmetries form a group under the law of composition of automorphisms. It is often useful to consider the (projective) representations of this group, $G$, given by the local unitaries $U_x(g), g\in G$.

\item[(ii)] {\em Lattice symmetries} are, in general, described by a bijection $R:\Gamma\to\Gamma$. It is usually important that $R$ preserves the local structure
of $(\Gamma, d)$, e.g., one requires that $R$ is isometric: $d(R(x),R(y))= d(x,y)$, $x,y\in\Gamma$. Examples include translations of lattices such as $\Ir^\nu$, 
and reflection symmetries satisfying $R^2=\id$. If we assume that $\cH_{R(x)} \cong \cH_x$, $R$ can be lifted to an automorphism of $\cA_\Gamma$ as follows.
Denote by $i_x:\cB(\cH_x)\to\cA_{\{x\}}$ the natural isomorphism (or a well-chosen one) and define the automorphism $\beta_R$ of $\cA_\Gamma$, by putting
\be
\beta_R(A) = i_{R(x)}\circ i_x^{-1}(A), \; \text{  for all  } \; A \in\cA_{\{x\}} \; \text{  and  } \; x\in \Gamma.
\ee
The symmetry of the interaction is expressed by the property $\Phi (R(X)) = \beta_R(\Phi(X))$. In the case of lattice translations this yields a 
representation of $(\Ir^m,+)$ on $\cA_\Gamma$, i.e., for $a\in\Ir^m, R(x) = x+a$ denotes the action of translations on $\Gamma$, and $X+a
=\{ x+a\mid x\in X\}$. Correspondingly, $\Phi$ is called translation invariant if $\beta_a(\Phi(X))=\Phi(X+a)$, for all $a\in \Ir^m$.

\item[(iii)] {\em Time-reversal symmetry} is expressed as a local symmetry (discussed in (i)) given by an anti-automorphism, implemented on each site by an anti-unitary
transformation. The latter are, in general, the composition of a unitary transformation and a complex conjugation. Besides taking into account the anti-linearity, time reversal symmetry can be treated in the same way as linear local symmetries.

\item[(iv)] {\em Chiral symmetry} is described by a unitary, say $C$, that anti-commutes with the Hamiltonian. So, at each point in the curve of Hamiltonians
we have $C^* H(s) C = -H(s)$. For the dynamics this implies that for all $s\in [0,1]$, $t\in \Rl$, and $A\in \cA$,
\be\label{chiral_dyn}
C^* \tau_t^{H(s)} ( A) C = \tau_t^{-H(s)}(C^* A C) =  \tau_{-t}^{H(s)}(C^* A C).
\ee
\end{enumerate}

It should be noted that the basic types of symmetries can be combined. For example, some models are invariant under a combined lattice reflection 
and time-reversal transformation, without possessing either of these symmetries separately.

We assume the same setup as described in the beginning of Section~\ref{sec:equivalence_of_phases} of a fixed $\nu$-regular metric space 
$(\Gamma,d)$ with a specified weighted $F$-function $F$. Let $G$ denote the symmetries under consideration. The fixed family of on-site Hamiltonians
$H_x, x\in\Gamma$, is assumed to have the symmetry $G$ as if it were a zero-range interaction. For example, if $U_x\in\cB(\cH_x)$ describes a local
unitary symmetry, we assume that the domain of $H_x$ is invariant under $U_x$ and that $H_x$ and $U_x$ commute. Or, as another example, 
if $\Gamma = \Ir^\nu$ and the symmetry is the full translation invariance of the lattice, $H_x$ is assumed to be the same self-adjoint operator at
each site $x$. 

For the interactions, let $\cB_F^G\subset\cB_F$ denote the space of interactions with finite $F$-norm that possess the symmetry $G$, and 
\be
\cB_F^{1,G}([0,1]) = \{\Phi\in \cB_F^1([0,1]) \, : \,  \Phi(s) \in \cB_F^G \; \text{for all} \; s\in [0,1]\}
\ee
Definitions \ref{def:uniformlygapped} and \ref{def:equivalence_gapped_phases} of uniformly gapped curves and the equivalence relation now carry over the
situation with a symmetry $G$ in the obvious way, as does the proof of the analogue of Proposition~\ref{prop:equivalent_gapped_phases}. 
The resulting equivalence classes are called {\em symmetry protected phases}. Since the uniformly gapped curves with symmetry are a special case of the general situation,
Theorem~\ref{thm:automorphic_equivalence} applies and the spectral flow automorphism establishes a bijection between the sets of states $\cS(s)$ along the curve. 

For the study of the stability of gapped ground state 
phases with symmetry breaking we present in our subsequent work \cite{QLBII}, it will be important that the automorphisms $\alpha_{t,s}$ commute with the automorphisms $\beta_g, g\in G$, representing the symmetry on $\cA_\Gamma$. Moreover, it will be desirable that the interaction $\Psi(s)$ generating $\alpha_{t,s}$ and its finite-system analogues all have the symmetry. There are a few subtleties that merit further discussion concerning the construction of a spectral flow with the desired symmetry properties.

As mentioned above, it is important that both the spectral flow $\alpha_{t,s}$ and its generating interaction $\Psi(s)$ respect the symmetry of the initial interaction $\Phi(s)$. Recall that the conditional expectations $\Pi_X$ from \eqref{PiX} play a crucial role in the quasi-locality properties of $\alpha_{t,s}$ and the definition of $\Psi(s)$, see Corollary~\ref{cor:sf_est_1}, \eqref{sf_iv_qli_def} and Section~\ref{sec:laqlmaps}. In the presence of a local symmetry $\beta$, it is useful to choose the locally normal product state in the definition of the conditional expectations $\Pi_X$ that is $\beta$-invariant, meaning $\rho_x(A) = \rho_x(\beta(A)), A\in \cA_{\{x\}}$ 
or, equivalently, $\beta(\rho_x)=\rho_x$,  see \eqref{pi_loc_norm_state}. This requirement guarantees that if $A$ is invariant under $\beta$, then so is $\Pi_X(A)$, i.e. 
\be \label{invariant_comm}
\beta(\Pi_X(A))= \Pi_X(\beta(A))= \Pi_X(A).
\ee

If $\dim \cH_x <\infty$, then a $\beta$-invariant locally normal product state always exists. For example, setting $\rho_x$ to be the tracial state will produce a $\beta$-invariant state. Given any $\Phi(s)$ with a local symmetry and any symmetric $\rho$, it is easy to see using \eqref{invariant_comm} that the Hastings interactions $\Psi_n$ defined in \eq{sf_fv_qli_def} and, consequently, the spectral flow $\alpha^{(n)}_s$ derived from $\Phi(s)$ both inherit this symmetry. The same holds true for the corresponding infinite volume objects, $\Psi$ and $\alpha_s$.
In particular $\alpha_s$ commutes with any local symmetry automorphism $\beta$ that leaves $\Phi'(s)$ invariant for all $s\in[0,1]$.

For infinite-dimensional $\cH_x$, a symmetric normal state on $\cH_x$ may or may not exist. One may have to relax either the normality or the symmetry requirement. Which of the
two is more relevant would depend on the situation at hand but for the type of applications we are considering here, it is important to use normal states. In the case of a gauge symmetry described by a compact Lie group, constructing symmetric normal states is not a problem. However, even when such a state 
does not exist, the symmetry of the spectral flow is restored in the thermodynamic limit. This follows from the 
observation that although the infinite-volume interaction $\Psi(s)$ depends on the choice of the locally normal state $\rho$ used in its construction, the infinite-volume flow
is the thermodynamic limit of automorphisms generated by self-adjoint operators that commute with the symmetry. This is apparent, e.g., from the expression \eq{loc_comp_ham} in which
$H^{\Phi}_\Lambda$ is to be replaced by $H_\Lambda^\prime(s)$, and $\cK$ is $\mathcal{G}_s$ defined in \eq{def_F+G}.

In the presence of a lattice symmetry such as translation invariance, it makes sense to pick a translation invariant product state $\rho$ to define the 
conditional expectations $\Pi_X$. This is obviously always possible and yields a covariant family of conditional expectations, meaning that
$\beta_a\circ \Pi_X = \Pi_{X+a} \circ \beta_a$, where for $a\in\Ir^m$, $X+a$ denotes the action of the translations on $\Gamma$ and $\beta_a$ denotes 
the corresponding action on $\cA_\Gamma$. Finite subsystems defined on quotient lattices, e.g. $\bZ^n/ (\mathbf{L}\Ir^m)$ with $n\geq m$, can have the corresponding
quotient symmetry $(\Ir^m\mod\mathbf{L})$, which is equivalent to considering the system with periodic boundary conditions. In general, finite systems will 
not have an exact translation symmetry but, again, the symmetry is recovered in the thermodynamic limit.

The case of time-reversal symmetry can be treated in the same way as local unitary symmetries. Due to the oddness of the function $W_\gamma$ in \eqref{def:Wgam2} the 
Hastings interaction $\Psi(s)$ changes sign under time-reversal. Since the time-reversal automorphism is anti-linear, however, this is exactly the 
requirement for $\alpha^{(n)}_s$ and $\alpha_s$ to commute with it.

The case of a fixed chiral symmetry $C$ along the curve of Hamiltonians $H(s) = H + \Phi(s)$, implies that $\Phi^\prime(s)$ anti-commutes with $C$. 
Using again the oddness of the function $W_\gamma$, and the property \eq{chiral_dyn}, it is straightforward to check that $C$ then {\em commutes} with the 
generator of the spectral flow, i.e., it is a symmetry of the spectral flow.

\begin{thm} \label{thm:automorphic_equivalence_G}
Let $\{\beta_g\mid g\in G\}$ be the automorphisms on $\cA_\Gamma$ representing symmetries of the system of the type described in (i-iv) above. 
Then, for any uniformly gapped curve of interactions  $\Phi\in \cB_F^{1,G}([0,1])$, there exists a strongly continuous co-cycle of spectral flow automorphisms 
$\alpha_{t,s}, s,t\in [0,1]$, such that
\begin{equation}
\mathcal{S}(t)=\mathcal{S}(s)\circ\alpha_{t,s},
\label{Ginfinite_gs_relation}
\end{equation}
and 
$$
\beta_g\circ\alpha_{t,s} = \alpha_{t,s}\circ\beta_g, \mbox{ for all } s,t\in [0,1], g\in G.
$$
\end{thm}

The list of types of symmetries we have discussed here is not exhaustive. For example, another type of symmetry relevant
for applications are duality symmetries. We postpone the discussion of those to \cite{QLBII}, where we will study the 
stability of gapped ground state phases. 

\subsection{Examples of uniformly gapped curves}\label{sec:examples_gapped_curves}
The construction of the automorphisms $\alpha_{t,s}$ assumes the existence of 
a uniform lower bound for the spectral gap above the ground states along the 
curve of models in the sense of Definition \ref{def:uniformlygapped}. Establishing
a uniform bound for the gap is generally a very hard problem. Fortunately, there are
a good number of interesting examples where the existence of a positive uniform lower 
bound can be proved. 

The largest variety of examples is found as a result of perturbing models for which
the ground state and the existence of a spectral gap above it are known. We will review the state of
the art of perturbative results of this type in \cite{QLBII}. For this reason, we limit
ourselves here to citing a few works that illustrate the broad range of examples
that exist in the literature: some exactly solvable models such as the anisotropic $XY$ 
chain \cite{lieb:1961}, quantum perturbations of classical spin models
\cite{matsui:1990,kennedy:1992}, perturbations of the AKLT chain \cite{affleck:1988,yarotsky:2006} 
and similar models \cite{szehr:2015}, perturbations of simple models with topological order
in the ground state such as the Toric Code Model \cite{bravyi:2010}, general perturbations of 
frustration-free models satisfying a Local Topological Order Condition
\cite{michalakis:2013}, and perturbations of quasi-free Fermion systems \cite{de-roeck:2018}.

Other interesting examples for which explicit lower bounds for the gap can be obtained 
and classes of models for which the equivalence classes can be explicitly determined 
are the frustration-free spin chains with finitely correlated ground states, also known as 
matrix product states \cite{fannes:1992,nachtergaele:1996,bachmann:2012bf,bachmann:2014c,ogata:2016,ogata:2016a,ogata:2017}.
Allowing for general perturbations of such models typically leads to splitting of the 
degenerate ground states found in the frustration-free model. The so-called Kennedy triplet
of `excited' states of the spin-$1$ Heisenberg antiferromagnetic chain of even length
can be regarded as an example of this phenomenon \cite{kennedy:1999}. In general,
sufficiently small perturbations of one-dimensional frustration-free models with a gap 
above the ground state will have a group of eigenvalues near the bottom of the spectrum
separated by a gap (uniform in the size of the system) from the rest of the spectrum. The
associated eigenstates all converge to ground states in the thermodynamic limit. Both
statements are proved in \cite{moon:2018}. 

We postpone a more comprehensive discussion of examples of models with
distinct gapped ground state phases until after the presentation of the stability 
results of gapped phases in \cite{QLBII}.

 %
 %

 \section{Appendix} \label{sec:appendix}

This section collects a number of facts about the decay bounds used throughout this paper. 
In general, we will assume that $\Gamma$ is a countable set equipped with a metric, and we
denote this metric by $d$.  A good example to keep in mind is $\Gamma = \mathbb{Z}^{\nu}$ 
with the $\ell^1$-metric. When necessary, we will also assume $\Gamma$ is $\nu$-regular, see \eqref{Gballbd}.

When considering the Heisenberg dynamics associated to a Hamiltonian, our 
quasi-locality estimates require a {\it short-range} assumption on the corresponding 
interaction. For general sets $\Gamma$, which need not have the structure of a lattice,
a sufficient condition for the existence of a dynamics in the thermodynamic limit can be
expressed in terms of a norm on the interaction. We have found it convenient 
to express the decay of interactions with distance by a so-called $F$-function,
which we discuss below. Depending on the application one has in mind, more explicit forms of decay,
again expressed in terms of a family of $F$-functions, is convenient. These are by no means the
only ways to express decay assumptions for interaction. If generality is not the concern, 
one can easily re-express decay into a more suitable form for the case at hand, say, e.g., for 
systems with pair interactions only. In this appendix our goal is to
briefly summarize various notions of decay which occur frequently in the main text.    

\subsection{On $F$-functions} \label{app:sec_def_F}

%
%
%

Let $(\Gamma, d)$ be a countable metric space. When $\Gamma$ is finite, most notions introduced
below are trivial, and for that reason we will mainly consider the situation where $\Gamma$ has infinite cardinality.
We will say that $\Gamma$ is equipped with an $F$-function if there is a 
non-increasing function $F: [0, \infty) \to (0, \infty)$ for which:

\noindent (i) {\it $F$ is uniformly integrable}:
\begin{equation} \label{Fint}
\| F \| : = \sup_{y \in \Gamma} \sum_{x \in \Gamma} F(d(x,y)) < \infty
\end{equation}

\noindent

\noindent (ii) {\it $F$ satisfies a convolution condition}: 
\begin{equation} \label{Fconv}
C_F:= \sup_{x,y \in \Gamma} \sum_{z \in \Gamma} \frac{F(d(x,z)) F(d(z,y))}{F(d(x,y))} < \infty.
\end{equation}

Any function $F$ satisfying (\ref{Fint}) and (\ref{Fconv}) will be called an {\it $F$-function} on $\Gamma$.
We note that an immediate consequence of (\ref{Fconv}) is that for any pair $x,y \in \Gamma$, we have the bound
\begin{equation} \label{Fconvest}
\sum_{z \in \Gamma} F(d(x,z)) F(d(z,y)) \leq C_F F(d(x,y)).
\end{equation}
The constant $C_F$ enters into a number of our estimates. We say that an $F$-function on
$\Gamma$ is {\it normalized} if $C_F =1$. Of course, for any $F$-function $F$, the function $\tilde{F} = C_F^{-1} F$ defines a new $F$-function on $\Gamma$ for which $C_{\tilde{F}} =1$. 

Note that if $\Gamma$ is equipped with an $F$-function $F$, then 
\begin{equation}
\sup_{x \in \Gamma} | b_x(n) | \leq \| F \| F(n)^{-1} \quad \mbox{for any } n >0 \, ,
\label{growth_of_balls}\end{equation}
where the left-hand-side above is a uniform estimate on the cardinality of the ball of radius $n$
centered at $x \in \Gamma$. The above follows immediately from the estimate
\begin{equation} \label{F-regular}
F(n) |b_x(n)| \leq \sum_{y \in b_x(n)} F(d(x,y)) \leq \| F \| .
\end{equation}
This estimate also demonstrates that the existence of an $F$-function guarantees that $\Gamma$ is uniformly,
locally finite. 

Moreover, if $\Gamma$ is infinite, the existence of an $F$-function implies the diameter of $\Gamma$ is infinite. In this situation, if $\{ \Lambda_n \}_{n \geq 1}$ is an increasing, exhaustive sequence of finite
subsets of $\Gamma$ (i.e. $\Lambda_n \subset \Lambda_{n+1}$ for all $n \geq 1$ and
$\Lambda_n \uparrow \Gamma$), then for any finite $X \subset \Gamma$, 
\begin{equation} \label{Xaway}
d(X, \Gamma \setminus \Lambda_n) \to \infty \quad \mbox{as } n \to \infty \, .
\end{equation}
This follows by observing that for any $m \geq 1$, 
\begin{equation}
X(m) : = \bigcup_{x \in X} b_x(m)
\end{equation}
is a finite subset of $\Gamma$. Since $\{\Lambda_n\}_{n\geq 1}$ is absorbing, there is an $N \geq 1$ for which $X(m) \subset \Lambda_N$. Since $\Gamma$ has infinite cardinality, the set $\Gamma \setminus \Lambda_N$ is non-empty. It immediately follows that $d(x,y) \geq m$ for all $x \in X$ and $y \in \Gamma \setminus \Lambda_N$, from which (\ref{Xaway}) follows. 

\subsubsection{Two common examples of $F$-functions} \label{sec:2_comm_ex}

First, many well-studied quantum spin models are defined on the hypercubic lattice 
$\Gamma = \mathbb{Z}^{\nu}$ for some integer $\nu \geq 1$. 
For concreteness, consider $\Ir^\nu$ equipped with the $\ell^1$-metric 
\begin{equation}
d(x,y) = |x-y| = \sum_{j=1}^{\nu} |x_j -y_j|.
\end{equation}
Other translation invariant metrics can be treated similarly. For any $\epsilon >0$, the function 
\begin{equation} \label{latticeF}
F(r) = \frac{1}{(1+r)^{\nu + \epsilon}} \quad \mbox{for all } r \geq 0 \, ,
\end{equation}
is an $F$-function on $\Gamma = \mathbb{Z}^{\nu}$. Integrability follows from
\begin{equation}
\| F \| = \sum_{z \in \mathbb{Z}^{\nu}} \frac{1}{(1+|z|)^{\nu + \epsilon}} < \infty.
\end{equation}
Moreover, for any metric space $(\Gamma, d)$: if $p \geq 1$, the bound
\begin{eqnarray} \label{poly_F_bd}
(1+ d(x,y))^p & \leq & (1+ d(x,z) + 1 + d(z,y))^p \nonumber \\
& \leq & 2^{p-1} (1+ d(x,z))^p + 2^{p-1}(1+d(z,y))^p
\end{eqnarray}
holds for all $x,y,z \in \Gamma$, since the function $t \mapsto t^p$ is (midpoint) convex. In this case,
the function defined in (\ref{latticeF}) satisfies (\ref{Fconv}) with 
\begin{equation}
C_F \leq 2^{\nu + \epsilon} \| F \|.
\end{equation}

Next, we note that for many of our results it is not necessary that $\Gamma$ has the structure of a lattice. 
We will say that a metric space $(\Gamma, d)$ is {\em $\nu$-regular} if there exists $\nu>0$ and $\kappa < \infty$ 
for which
 \begin{equation} \label{Gballbd}
 \sup_{x \in \Gamma} |b_x(n)| \leq \kappa n^{\nu} \quad \mbox{for all } n \geq 1 \, .
 \end{equation}
Here for any $x \in \Gamma$ and $n \geq 0$, $b_x(n)$ is the ball of radius $n$ centered at $x$ and
 $| \cdot |$ denotes cardinality. From \eq{growth_of_balls} we see that if $\Gamma$ has an $F$-function for which $F(r) \leq C r^{-\nu}$, then 
$\Gamma$ is $\nu$-regular.

 If $(\Gamma, d)$ is $\nu$-regular, then for any 
 $\epsilon >0$, the function 
  \begin{equation} \label{baseF}
 F(r) = \frac{1}{(1+r)^{\nu + 1 + \epsilon}} \quad \mbox{for all } r \geq 0 ,
 \end{equation}
 is an $F$-function on $\Gamma$. 
 
To see that this is the case, we need only check uniform integrability, i.e. (\ref{Fint}),  
as an argument using (\ref{poly_F_bd}) shows that this $F$ satisfies (\ref{Fconv}).
Fix $x \in \Gamma$. Set $B_x(1) = b_x(1)$ and $B_x(n) = b_x(n) \setminus b_x(n-1)$ for any
$n \geq 2$. It is clear then that (\ref{Gballbd}) implies
\begin{eqnarray}
\sum_{y \in \Gamma} \frac{1}{(1+d(x,y))^{\nu + 1 + \epsilon}} & = & \sum_{n=1}^{\infty} \sum_{y \in B_x(n)} \frac{1}{(1+d(x,y))^{\nu + 1 + \epsilon}} \nonumber \\
& \leq & \sum_{n=1}^{\infty} \frac{|b_x(n)|}{n^{\nu + 1 + \epsilon}} \nonumber \\
& \leq & \sum_{n=1}^{\infty} \frac{\kappa}{n^{1 + \epsilon}} < \infty
\end{eqnarray}
uniformly in $x$; hence (\ref{Fint}) holds. A computation similar to \eqref{poly_F_bd} shows that $C_F<\infty$.

We can combine the above discussion with \eqref{F-regular} to prove the following result.

\begin{prop}\label{nu-reg-F-function} Let $(\Gamma, d)$ be a countable metric space. 
	\begin{enumerate}
		\item[(i)] If $(\Gamma, d)$ is $\nu$-regular then $F(r)=(1+r)^{-(\nu+1+\epsilon)}$ is an $F$ function of $(\Gamma,d)$ for all $\epsilon >0$.
		\item[(ii)] If $F(r)=(1+r)^{-(\nu+\epsilon)}$ is an $F$-function of $(\Gamma,d)$ for all $\epsilon >0$, then $\Gamma$ is $\nu$-regular.
	\end{enumerate} 
\end{prop}

\begin{proof}
Part (i) follows immediately from the discussion after \eqref{baseF}. For part (ii), suppose that $F(r)=(1+r)^{-(\nu+\epsilon)}$ is an $F$-function of $(\Gamma, d)$ for all $\epsilon >0$. Fix $\epsilon>0$. Then by \eqref{growth_of_balls}, for any $n\geq 1$ and $x\in\Gamma$,
\be
|b_x(n)| \leq \|F\|(1+n)^{\nu+\epsilon} \leq \|F\|(2n)^{\nu+\epsilon}.
\ee
Taking the infimum over $\epsilon>0$ shows that $(\Gamma,d)$ is $\nu$-regular with $\kappa = 2^\nu\|F\|$.
\end{proof}

%
%
%
%

\subsection{On weighted $F$-functions} \label{app:sec_weight_F}

For certain applications, it is convenient to consider families of $F$-functions of a specific form, which we call {\em weighted $F$-functions}. 

Let $(\Gamma, d)$ be a metric space equipped with an $F$-function $F$ as described in
Section~\ref{app:sec_def_F}. Let $g: [0, \infty) \to [0, \infty)$ 
be a non-negative, non-decreasing, sub-additive function, i.e.
\begin{equation} \label{sub-add}
g(r+s) \leq g(r) + g(s) \quad \mbox{for any } r, s \geq 0 \, .
\end{equation}
Corresponding to any such $g$, the function
\begin{equation}
F_g(r) := e^{-g(r)} F(r) \quad \mbox{for all } r \geq 0 \, ,
\end{equation}
is an $F$-function on $\Gamma$. In fact, since $g$ is non-negative, $F_g$ satisfies (\ref{Fint}) 
with $\| F_g \| \leq \|F \|$. Moreover, since $g$ is non-decreasing
and sub-additive, one also has that
\begin{equation}
g(d(x,y)) \leq g( d(x,z) + d(z,y)) \leq g(d(x,z)) + g(d(z,y)) \quad \mbox{for all } x,y,z \in \Gamma \, .
\end{equation} 
Thus (\ref{Fconv}) holds with $C_{F_g} \leq C_F$. 

We may refer to $F$ as the {\it base} $F$-function associated to $F_g$; note that $F_0 = F$ for $g=0$ . 
The function $g$ induces a factor $r \mapsto e^{- g(r)}$ which is often referred to as a {\it weight}.
We may also loosely refer to $g$ as a weight and similarly $F_g$ as a {\it weighted $F$-function}.
One readily checks that sums, non-negative scalar multiples, and compositions of weights are also weights; in the sense that 
if $g_1$ and $g_2$ are both non-negative, non-decreasing, sub-additive functions,
then so too are $g_1+g_2$, $a g_1$ (for $a \geq 0$), and $g_1 \circ g_2$.

In certain applications, it is useful to introduce a one-parameter family of weighted $F$-functions by
taking a base $F$-function $F$ on $\Gamma$, fixing a weight $g$, and associating to any $a \geq 0$, the function $g_a(r) = ag(r)$, for which
$F_{g_a}(r) = e^{-ag(r)}F(r)$. When $g$ is understood, we often just write $F_a:=F_{g_a}$ to describe this family of weighted $F$-functions. For example, if $F(r) = (1+r)^{-p}$ and $g(r) =r$, then $F_a(r) = e^{-ar}(1+r)^{-p}$ is the family of weighted $F$-functions defined in Section~\ref{sec:spatialstructure}.

As a motivation for introducing these weights, consider again $\Gamma = \mathbb{Z}^{\nu}$. As discussed in the previous subsection, the polynomially decaying function $F$ in (\ref{latticeF}) is an $F$-function associated to $\Gamma = \mathbb{Z}^{\nu}$. Such an $F$-function is appropriate for interactions $\Phi$ with terms that decay polynomially with the diameter of their support: $\|\Phi(X)\| \leq C(1+\diam(X))^{-\nu+\epsilon}$. However, we are typically interested in interactions whose terms decay much faster, in particular exponentially fast. One readily checks that for any $a >0$, the exponential function $g(r) = e^{-a r}$ fails to satisfy (\ref{Fconv}) and as such is not an $F$-function on $\Gamma = \mathbb{Z}^{\nu}$, but does satisfy the criterion to be a weight. Since exponential functions often govern the decay of our interactions, it is convenient that one can obtain an exponentially decaying $F$-function on $\Gamma = \mathbb{Z}^{\nu}$ by making an appropriate choice of weight. 

Before moving on to discussing several useful weights, we point out one added benefit of these functions. In the situation were we do not assume to have a weighted $F$-function and given $X, \, Y \in \cP_0(\Gamma)$, we will often use the simple bound
\be
\sum_{x \in X} \sum_{y \in Y} F(d(x,y)) \leq \vert X\vert\, \Vert F\Vert
\ee
when applying a Lieb-Robinson bound or quasi-locality bound, see \eqref{lrbmin} and \eqref{QLmap_Fbd}. For weighted $F$-functions, however, the following is also frequently used:
\be\label{weighted_decay}
\sum_{x \in X} \sum_{y \in Y} F_g(d(x,y)) \leq \vert X\vert\, \Vert F\Vert e^{-g(d(X, Y))}.
\ee
Here, one typically is considering a quantum lattice system defined on a (large) finite volume $\Lambda$, and in the situation that $Y = \Lambda \setminus X(n)$ then the RHS of \eqref{weighted_decay} decays as $e^{-g(n)}$.

\subsubsection{Three common weights}\label{sec:3_common_weights}
With an eye towards our specific applications, we now introduce three particular classes of weights.

First, let $\mu \in [0,1]$. The function $g: [0, \infty) \to [0, \infty)$ given by $g(r) =r^{\mu}$ is non-negative, non-decreasing, and sub-additive
in the sense of (\ref{sub-add}). The constant function $g(r) = 1$ corresponding to $\mu=0$ is of minor interest, however, the 
choice of $\mu =1$ generates exponentially decaying weights. When $0 < \mu <1$, the function $e^{-r^\mu}$   
is often called a {\it stretched exponential}. 

Next, we provide an example {\it between} exponential and stretched exponential decay. As we will show, for any $p>0$, the function 
\begin{equation}
g(r) = \left\{ \begin{array}{cc} \frac{e^p}{p^p} & \mbox{if } 0 \leq r \leq e^p \\ \frac{r}{\ln(r)^p} & r \geq e^p \end{array} \right.
\end{equation}
is non-negative, non-decreasing, and sub-additive.
In our applications of the spectral flow, see e.g. Section~\ref{sec:spectral-flow}, the choice of $p=2$ is particularly
relevant. 

Note that at $r = e^p$ the non-constant part of $g$ has a zero derivative, and for $r > e^p$, $g$ is strictly increasing. 
That motivates this particular choice of cut-off. Also, it is easy to see that this function is
sub-additive  by taking cases: Let $r,s \geq 0$. Consider (i) $r+s \leq e^p$ and (ii) $r+s >e^p$. Both cases are easy to see.
For the second use that
\begin{equation}
g(r+s) = \frac{r}{\ln(r+s)^p} + \frac{s}{\ln(r+s)^p}.
\end{equation}
Note that
\begin{equation}
\mbox{ if } r \geq e^p , \mbox{  then} \quad \frac{r}{\ln(r+s)^p} \leq \frac{r}{\ln(r)^p}  \quad \mbox{whereas} \quad \mbox{ if } r \leq e^p , \mbox{ then} \quad \frac{r}{\ln(r+s)^p} \leq \frac{e^p}{\ln(e^p)^p}
\end{equation}
the latter fact using that $r+s>e^p$.

Lastly, the function $g :[0, \infty) \to [0, \infty)$ given by
\begin{equation}
g(r) = \ln(1+r)
\end{equation}
is clearly non-negative and non-decreasing. Since for any $r,s \geq 0$, we have that
\begin{equation}
1 + r + s \leq (1+r)(1+s) 
\end{equation}
and sub-additivity of $g$ readily follows. Starting with a base $F$-function, as in (\ref{latticeF}) or (\ref{baseF}), 
a proper scaling of $F$ by this weight allows for arbitrary power-law decay.

 \subsection{Simple transformations of $F$-functions} \label{subsec:regsets}

In certain applications, it is convenient to know that various decaying functions are in fact
$F$-functions. Quantities of interest can be estimated in terms of 
translations or re-scalings of known $F$-functions. For $\nu$-regular $\Gamma$, these modifications preserve the basic properties of an $F$-function. The following two propositions show that suitably defined truncations, shifts, and dilations of $F$-functions are again $F$-functions.

\begin{prop} \label{prop:mod_F_1} 
	Let $(\Gamma, d)$ be a $\nu$-regular metric space with an $F$-function $F$. For any $a \geq 0$ and any choice of $c \geq F(a)$, the function $\tilde{F} : [0, \infty) \to (0, \infty)$ defined by
	setting
	\begin{equation} \label{F_with_step}
	\tilde{F}(r) = \left\{ \begin{array}{ll} c & \mbox{if } 0 \leq r \leq a, \\
	F(r) & \mbox{if } r >a,
	\end{array} \right.
	\end{equation}
	is an $F$-function on $(\Gamma, d)$. In fact, 
	\begin{equation} \label{step_F_est}
	\| \tilde{F} \| \leq c \kappa a^{\nu} + \| F \| \quad \mbox{and} \quad C_{\tilde{F}} \leq \max\{c, F(0)\} \frac{cC_F}{F(a)^2}  
	\end{equation}
\end{prop}
\begin{proof}
	Fix $x \in \Gamma$. Note that
	\begin{equation}
	\sum_{y \in \Gamma} \tilde{F}(d(x,y)) = c |b_x(a)| + \sum_{y \in \Gamma \setminus b_x(a) } F(d(x,y))
	\end{equation}
	and the first bound in (\ref{step_F_est}) follows from $\nu$-regularity. 
	
	To see the second bound, note that 
	\begin{equation}
	\tilde{F}(r) \leq \frac{c}{F(a)} F(r) \quad \Rightarrow \quad \sum_{z \in \Gamma} \tilde{F}(d(x,z)) \tilde{F}(d(z,y)) \leq \left( \frac{c}{F(a)} \right)^2 C_F F(d(x,y))
	\end{equation}
	for all sites $x,y \in \Gamma$. By considering the cases $d(x,y)\geq a$ and $d(x,y) <a$ separately, one can show
	\[
	cF(d(x,y)) \leq \max\{c,\, F(0)\}\tilde{F}(d(x,y)), 
	\]
	from which the second bound in (\ref{step_F_est}) follows.
\end{proof} 

\begin{prop} \label{prop:mod_F_2} Let $(\Gamma, d)$ be a $\nu$-regular metric space, and $p \geq 1$ be such that
	\begin{equation}
	F_0(r) = \frac{1}{(1+r)^p}, \; \text{  for all  } \; r\geq 0
	\end{equation}
	is an $F$-function on $(\Gamma, d)$.
	\begin{enumerate}
		\item[(i)] If $F(r)=e^{-g(r)}F_0(r)$ is a weighted $F$-function on $(\Gamma, d)$, then for any $\epsilon>0$ the function defined by
		$\tilde{F}(r) = F(\epsilon r)$ is an $F$-function on $(\Gamma, d)$. Moreover, 
		\begin{equation} \label{scaled_F_est}
		\| \tilde{F} \| \leq \max\{ 1, \epsilon^{-p} \} \| F_0 \| \quad \mbox{and} \quad C_{\tilde{F}} \leq 2^p \| \tilde{F} \| \, .
		\end{equation}
		\item[(ii)]If $F(r)=e^{-g(r)}F_0(r)$ is a weighted $F$-function on $(\Gamma, d)$ then for any $a>0$ the function defined by
		\begin{equation}
		\tilde{F}(r) = \left\{ \begin{array}{ll} F(0) & 0 \leq r \leq a, \\
		F(r-a) & r > a, \end{array} \right.
		\end{equation} 
		is an $F$-function on $(\Gamma, d)$. In fact, 
		\begin{equation} \label{translate_F_est}
		\| \tilde{F} \| \leq \frac{\max\{ 1, F(0) \}}{F(a)}  \| F \| \quad \mbox{and} \quad C_{\tilde{F}} \leq \left( \frac{\max\{ 1, F(0) \}}{F(a)}  \right)^2 C_F \, .
		\end{equation}
	\end{enumerate}
\end{prop}
\begin{proof}
	To see (i), first note that $r \mapsto g( \epsilon r)$ is non-negative, non-decreasing, and sub-additive. In this case, we need only verify that $r \mapsto (1+ \epsilon r)^{-p}$ is an $F$-function. For any $\epsilon >0$, one has that
	\begin{equation}
	\min\{1, \epsilon\} (1+ r) \leq (1 + \epsilon r) \leq \max\{ 1, \epsilon\} (1+r) 
	\end{equation}
	holds for all $r \geq 0$. The first bound in (\ref{scaled_F_est}) is then clear, and the second bound
	follows as in (\ref{poly_F_bd}).
	
	To see (ii), a short calculation shows that
	\begin{equation}
	F(r) \leq \tilde{F}(r) \leq \frac{\max\{1, F(0) \}}{F(a)} F(r) \quad \mbox{for all } r \geq 0 \, .
	\end{equation}
	The first inequality above is trivial since $F$ is non-increasing. The second follows from sub-additivity of $g$, namely $g(r) \leq g(a) + g(r-a)$ for any $r >a$, as well as the fact that $F_0(r-a) \leq F_0(r)/F_0(a)$. The bounds in (\ref{translate_F_est}) readily follow.
\end{proof}

A simple consequence of Proposition~\ref{prop:mod_F_2} is the following. Under the assumptions of
Proposition~\ref{prop:mod_F_2}, let $F$ be a weighted $F$-function with base $F_0$. For any
$r>1$, the bound
\begin{equation}
F(r) \leq F(\lfloor r \rfloor) \leq F(r -1)
\end{equation}
is clear, and this bound extends to all $r \geq 0$ using the function in Proposition~\ref{prop:mod_F_2} (ii) above.
Thus $F(\lfloor r \rfloor)$ is an $F$-function as well.

%
%

\subsection{Basic interaction bounds} \label{app:sec:BasicBounds}

In the main text, we frequently use a number of basic estimates concerning interactions that are expressed using $F$-functions.
Here we collect a few results for later reference.

We begin by recalling some of the basic notation associated with interactions.
Let $(\Gamma, d)$ be a metric space equipped with an $F$-function $F$ as discussed in Section~\ref{app:sec_def_F}.
Let $\mathcal{P}_0( \Gamma)$ denote the set of finite subsets of $\Gamma$.
We say that a mapping $\Phi$ is an interaction on $\Gamma$ if $\Phi : \mathcal{P}_0( \Gamma) \to \mathcal{A}_{\Gamma}^{\rm loc}$
with the property that $\Phi(X)^* = \Phi(X) \in \mathcal{A}_X$ for every $X \in \mathcal{P}_0( \Gamma)$. 
If $\Phi$ is an interaction on $\Gamma$, we write that $\Phi \in \mathcal{B}_F$ if and only if
\begin{equation} \label{pmomphiF}
\| \Phi \|_F : = \sup_{x,y \in \Gamma} \frac{1}{F(d(x,y))} \sum_{\stackrel{X \in \mathcal{P}_0( \Gamma):}{x,y \in X}} \| \Phi(X) \| < \infty.
\end{equation}
A basic consequence of \eqref{pmomphiF} is that for all $x, y\in \Gamma$,
\begin{equation}\label{Fsum}
\sum_{\substack{X \in P_0(\Gamma) : \\ x,y\in X}} \|\Phi(X)\| \leq \|\Phi\|_F F(d(x,y)).
\end{equation} 
%
%

\subsubsection{Estimates based on distance}

We first provide a basic $F$-norm estimate based on summing interaction terms whose distance from a specific set is given. Recall that if
$X \subset \Gamma$ and $n \geq 0$, the set $X(n) \subset \Gamma$ is defined as
\begin{equation} \label{set_xn}
X(n) = \{ y \in \Gamma : d(X,y) \leq n \} = \bigcup_{x\in X}b_x(n).
\end{equation}

\begin{prop} \label{app:prop:int_bd_dist} Let $( \Gamma, d)$ be a metric space equipped with an $F$-function $F$. 
Let $\Phi$ be an interaction on $\Gamma$ with $\Phi \in \mathcal{B}_F$.
For any $X \in \mathcal{P}_0( \Gamma)$ and each $R \geq 0$, 
\begin{equation} \label{RclosetoX}
\sum_{\stackrel{Z \in \mathcal{P}_0( \Gamma):}{d(Z,X) \leq R}} \| \Phi(Z) \| \leq \| \Phi \|_F \sum_{x \in X(R)} \sum_{y \in \Gamma} F(d(x,y)) \leq |X(R)| \| F \| \| \Phi \|_F.
\end{equation}
Moreover, 
\begin{equation} \label{RfarfromX}
\sum_{\stackrel{Z \in \mathcal{P}_0( \Gamma):}{d(Z,X) > R}} \| \Phi(Z) \| \sum_{x \in X} \sum_{z \in Z} F(d(x,z)) \leq C_F \| \Phi \|_F \sum_{x \in X} \sum_{y \in \Gamma \setminus X(R) }F(d(x,y)) .
\end{equation}
\end{prop} 
\begin{proof}
Given the uniform integrability of $F$, i.e. \eqref{Fint}, the second estimate in (\ref{RclosetoX}) is clear given the first. To see the first, note that by over-counting,
\begin{equation} \label{overcounting}
\sum_{\stackrel{Z \in \mathcal{P}_0( \Gamma):}{d(Z,X) \leq R}} \| \Phi(Z) \| \leq \sum_{x \in X(R)} \sum_{y \in \Gamma} \sum_{\stackrel{Z \in \mathcal{P}_0( \Gamma):}{x,y \in Z}} \| \Phi(Z) \|.
\end{equation} 
The estimate in (\ref{RclosetoX}) now follows from \eqref{Fsum} and \eqref{Fint}.

To prove (\ref{RfarfromX}), note that again by over-counting,
\begin{eqnarray} \label{large_R_est_1}
\sum_{\stackrel{Z \in \mathcal{P}_0( \Gamma):}{d(Z,X) > R}} \| \Phi(Z) \| \sum_{x \in X} \sum_{z \in Z} F(d(x,z)) & \leq &
\sum_{x \in X} \sum_{z \in \Gamma} F(d(x,z)) \sum_{y \in \Gamma \setminus X(R)} \sum_{\stackrel{Z \in \mathcal{P}_0( \Gamma):}{z,y \in Z}}  \| \Phi(Z) \|  \nonumber \\
& \leq &  \| \Phi \|_F \sum_{x \in X}  \sum_{y \in \Gamma \setminus X(R)} \sum_{z \in \Gamma} F(d(x,z)) F(d(z,y)).
\end{eqnarray} 
Thus, (\ref{RfarfromX}) follows using the convolution condition on $F$.
\end{proof}

A simple corollary of these bounds follows. To state it requires that we introduce two
notions. First, we describe {\it compatible $F$-functions}. Let $(\Gamma, d)$ be a metric space 
equipped with two $F$-functions denoted by $F_1$ and $F_2$.
We will say that $F_1$ and $F_2$ are {\it compatible} if there is a positive number $C_{1,2}$ and
a non-increasing function $F_{1,2}: [0, \infty) \to (0, \infty)$ for which: given any $x,y \in \Gamma$, 
\begin{equation} \label{compat_Fs}
\sum_{z \in \Gamma} F_1(d(x,z)) F_2(d(z,y)) \leq C_{1,2} F_{1,2}(d(x,y)) .
\end{equation}
Next we briefly describe time-dependent interactions (see Section~\ref{sec:spatialstructure} for more details). 
Let $I \subset \mathbb{R}$ be an interval. We say that $\Phi : \mathcal{P}_0( \Gamma) \times I \to 
\mathcal{A}_{\Gamma}^{\rm loc}$ is a {strongly continuous interaction} on $\Gamma$ if for each $X\in \cP_0(\Gamma)$:
\begin{enumerate}
	\item[(i)] $\Phi(X,t)^* = \Phi(X,t) \in \mathcal{A}_X$ for all $t \in I$;
	\item[(ii)] The map $\Phi(X, \cdot) : I \to \mathcal{A}_X$ is continuous in
	the strong operator topology.
\end{enumerate}
Moreover, we will write $\Phi \in \mathcal{B}_F(I)$ if $\Phi(\cdot,t) \in \cB_F$, for all $t\in I$, and $\| \Phi(t) \|_F$, which we sometimes 
write as $\| \Phi\|_F(t)$,  is a locally bounded 
function of $t$.

%

We now state a corollary of Proposition~\ref{app:prop:int_bd_dist}.

\begin{cor} \label{app:cor:int_bd} Let $(\Gamma, d)$ be a metric space equipped with two compatible $F$-functions $F_1$ and $F_2$.
For an interval $I \subset \mathbb{R}$, suppose there is a co-cycle of automorphisms $\{ \alpha_{t,s} \}_{t,s \in I}$ of $\mathcal{A}_{\Gamma}$, which satisfy a Lieb-Robinson bound, i.e.
for any disjoint $X, Y \in \mathcal{P}_0( \Gamma)$, each $A \in \mathcal{A}_X, B \in \mathcal{A}_Y$,
and $s,t \in I$, 
\begin{equation} \label{lrb_cor_bd}
\| [ \alpha_{t,s}(A), B] \| \leq \frac{2 \| A \| \| B \|}{C_{F_1}} D^{\alpha}(t,s) \sum_{x \in X} \sum_{y \in Y} F_1(d(x,y)) 
\end{equation}
with a pre-factor $D^{\alpha}(t,s)$ that increases as $|t-s|$ increases. Let $\Psi$ be a time-dependent interaction
with $\Psi \in \mathcal{B}_{F_2}(I)$. Given $R \geq 0$ and $s,t \in I$ with $s \leq t$, one has that for any
 $A \in \mathcal{A}_X$ with $X \in \mathcal{P}_0(\Gamma)$, 
\begin{equation} \label{int_dyn_est_small_R}
\sum_{\stackrel{Z \in \mathcal{P}_0( \Gamma):}{d(Z,X) \leq R}} \int_s^t \| [ \alpha_{t,r}(A), \Psi(Z,r) ] \| \, dr \leq 2 \| A \| 
\int_s^t \| \Psi \|_{F_2}(r) \, dr \sum_{x \in X(R)} \sum_{y \in \Gamma} F_2(d(x,y)) 
\end{equation}
and, 
\begin{eqnarray} \label{int_dyn_est_large_R}
\sum_{\stackrel{Z \in \mathcal{P}_0( \Gamma):}{d(Z,X) > R}}  \int_s^t \| [ \alpha_{t,r}(A), \Psi(Z,r) ] \| \, dr 
 & \leq & \nonumber \\
&  \mbox{ } & \hspace{-3cm} \frac{2 \| A \| C_{1,2}}{C_{F_1}}  D^{\alpha}(t,s)  \int_s^t \| \Psi \|_{F_2}(r) \, dr \sum_{x \in X} \sum_{y \in \Gamma \setminus X(R)} F_{1,2}(d(x,y)).
\end{eqnarray}
\end{cor}

\begin{proof}
To prove (\ref{int_dyn_est_small_R}), note that for any $s \leq r \leq t$, the simple bound
\begin{equation} 
\| [ \alpha_{t,r}(A), \Psi(Z,r) ] \| \leq 2 \| A \| \| \Psi(Z,r) \| 
\end{equation} 
holds. 
Estimating as in (\ref{RclosetoX}) yields (\ref{int_dyn_est_small_R}) as claimed.

To see (\ref{int_dyn_est_large_R}), take some $s \leq r \leq t$ and note that (\ref{lrb_cor_bd}) applies. In fact, applying the Lieb-Robinson bound we find 
\begin{equation}
\| [  \alpha_{t, r}(A), \Psi(Z,r) ] \| \leq \frac{2 \| A \| \| \Psi(Z,r) \|}{C_{F_1}} D^{\alpha}(t,s) \sum_{x \in X} \sum_{z \in Z} F_1(d(x,z)) .
\end{equation}
Here, we use that, by our assumptions, $D^{\alpha}(t,r)\leq D^{\alpha}(t,s)$. Then, arguing as in the proof of \eqref{RclosetoX}, see in particular \eqref{large_R_est_1}, we find
\begin{equation}
\sum_{\stackrel{Z \in \mathcal{P}_0( \Gamma):}{d(Z,X) > R}} \| \Psi(Z, r) \| \sum_{x \in X} \sum_{z \in Z} F_1(d(x,z)) \leq  \| \Psi \|_{F_2}(r) 
\sum_{x \in X} \sum_{y \in \Gamma \setminus X(R)} \sum_{z \in \Gamma} F_1(d(x,z))F_2(d(z,y)).
\end{equation} 
Using compatibility, i.e. (\ref{compat_Fs}), the bound in (\ref{int_dyn_est_large_R}) now follows upon integration.
\end{proof}

%
%

\subsubsection{Estimates based on diameter}

In some situations, it is more convenient to form decay arguments based on the diameters of sets, rather than the distance between sets. For these cases, we will further assume that $(\Gamma, d)$ is $\nu$-regular so that (\ref{Gballbd}) holds. 

Before we state our first result, let us introduce some convenient notation.
Our estimates will often be in terms of moments of
certain decay functions. To this end, let $G: [0, \infty) \to (0, \infty)$ be a
{\it decay function}. 
For any $p \geq 0$ and each $m \geq 0$, set 
\begin{equation}\label{pmom}
M_p^G(m) = \sum_{n= \lfloor m \rfloor }^{\infty} (1+n)^p G(n)
\end{equation}
whenever the sum on the right-hand-side above is finite. 
We will refer to $M_p^G(0)$ as the {\it $p$-th moment of $G$}. The notation
\begin{equation}
M_p^G \circ M_q^G(m) = \sum_{n= \lfloor m \rfloor }^{\infty} (1+n)^p \sum_{n'=n}^{\infty}(1+n')^q G(n') 
\end{equation}
will be used for iterated moments. A rough estimate involving exchanging the order of the summations shows
\begin{equation} \label{it_mom_bd}
M_p^G \circ M_q^G(m) \leq M_{p+q+1}^G(m). 
\end{equation}

We now state our first result, compare with Proposition~\ref{app:prop:int_bd_dist}.

\begin{prop}  \label{app:prop:int_bd_diam} Let $(\Gamma, d)$ be a $\nu$-regular metric space equipped with an $F$-function $F$, and $\Phi$ be an interaction on $\Gamma$ with $\Phi \in \mathcal{B}_F$. Then, for any $R \geq 0$,
\begin{equation} \label{smalldiambd}
\sup_{x \in \Gamma} \sum_{\stackrel{Z \in \mathcal{P}_0( \Gamma):}{x \in Z; {\rm diam}(Z) \leq R}} \| \Phi(Z) \| \leq \| \Phi \|_F \| F \| .
\end{equation}
If, in addition, $F$ has a finite $2\nu$-th moment, i.e. $M_{2 \nu}^F(0)<\infty$, then  
\begin{equation} \label{largediambd}
\sup_{x \in \Gamma} \sum_{\stackrel{Z \in \mathcal{P}_0( \Gamma):}{x \in Z; {\rm diam}(Z) >R}} \| \Phi(Z) \| \leq \kappa^2 \| \Phi \|_F M_{2 \nu}^F(R).
\end{equation}
\end{prop}
\begin{proof}
For any fixed $x \in \Gamma$, note that
\begin{equation} \label{simple_close_diam_est}
\sum_{\stackrel{Z \in \mathcal{P}_0( \Gamma):}{x \in Z; {\rm diam}(Z) \leq R}} \| \Phi(Z) \| \leq \sum_{y \in b_x(R)} 
\sum_{\stackrel{Z \in \mathcal{P}_0( \Gamma):}{x, y \in Z}} \| \Phi(Z) \| \leq \| \Phi \|_F \sum_{y \in b_x(R)} F(d(x,y)).
\end{equation}
Using this bound, (\ref{smalldiambd}) follows from uniform integrability of $F$. 

To see (\ref{largediambd}), again fix $x \in \Gamma$ and see that 
\begin{eqnarray}
\sum_{\substack{Z \in \mathcal{P}_0( \Gamma):\\ x \in Z; \, {\rm diam}(Z) >R}}  \| \Phi(Z) \| & \leq & \sum_{n= \lfloor R \rfloor}^{\infty} 
\sum_{\substack{Z \in \mathcal{P}_0( \Gamma):\\x \in Z;\, n \leq {\rm diam}(Z) <n+1}}  \| \Phi(Z) \| \nonumber \\
& \leq & \sum_{n= \lfloor R \rfloor}^{\infty} \sum_{\substack{w,z \in b_x(n+1): \\ n \leq d(w,z) < n+1}} 
\sum_{\substack{Z \in \mathcal{P}_0( \Gamma):\\ w,z \in Z}} \| \Phi(Z) \| \nonumber \\
& \leq & \| \Phi \|_F  \sum_{n= \lfloor R \rfloor}^{\infty} \sum_{\substack{w,z \in b_x(n+1): \\ n \leq d(w,z) < n+1}}  F(d(w,z)) \nonumber \\
& \leq & \kappa^2 \| \Phi \|_F  \sum_{n= \lfloor R \rfloor}^{\infty} (1+n)^{2 \nu} F(n) 
\end{eqnarray}
where, for the last line above, we used that $F$ is non-increasing and over-estimated using (\ref{Gballbd}). This proves (\ref{largediambd}).
\end{proof}

In some arguments, we encounter moments of interactions. More precisely, 
let $\Phi$ be an interaction on $\Gamma$. For any $p \geq 0$, the 
mapping $\Phi_p : \mathcal{P}_0( \Gamma) \to \mathcal{A}_{\Gamma}^{\rm loc}$ given by
\begin{equation}
\Phi_p(X) = |X|^p \Phi(X) \, 
\end{equation}
also defines an interaction on $\Gamma$. We refer to $\Phi_p$ as the {\it $p$-th moment of $\Phi$}. 
The next lemma provides a basic estimate for interactions of this type.

\begin{lem} \label{lem:pmom} 
Let $(\Gamma, d)$ be a $\nu$-regular metric space equipped with an $F$-function $F$, and $\Phi$ be an interaction on $\Gamma$ with $\Phi \in \mathcal{B}_F$. If  $p \geq 0$ and $M_{(p+2) \nu}^F(0) < \infty$, then the $p$-th moment of $\Phi$ satisfies
\begin{equation} \label{pnormbd}
\sum_{\substack{X \in \mathcal{P}_0( \Gamma):\\x,y \in X}} \| \Phi_p(X) \| \leq \kappa^{p+2} \| \Phi \|_F M_{(p+2) \nu}^F(d(x,y))  \, .
\end{equation} 
\end{lem}

\begin{proof} 
Fix $x,y \in \Gamma$. Set $m = \lfloor d(x,y) \rfloor$ and note that 
\begin{equation}
\sum_{\substack{X \in \mathcal{P}_0( \Gamma):\\x,y \in X}} \| \Phi_p(X) \|  \leq  \sum_{n \geq m} \sum_{\substack{X \in \mathcal{P}_0(\Gamma): x,y \in X\\n \leq {\rm diam}(X) < n+1}} |X|^p \| \Phi(X) \|
\end{equation}
Now, if $x \in X$ and ${\rm diam}(X) < n+1$, then (\ref{Gballbd}) guarantees that
\begin{equation}
|X| \leq |b_x(n+1)| \leq \kappa (n+1)^{\nu} \quad \mbox{and therefore} \quad |X|^p \leq \kappa^p (n+1)^{p \nu} 
\end{equation}
Moreover, by over-counting, it is clear that
\begin{equation}
\sum_{\substack{X \in \mathcal{P}_0(\Gamma): x,y \in X\\n \leq {\rm diam}(X) < n+1}} * \leq \sum_{\substack{w,z \in b_x(n+1)\\n \leq d(w,z) < n+1}} \sum_{\substack{X \in \mathcal{P}_0( \Gamma):\\w,z \in X}} *
\end{equation}
where $*$ represents a non-negative quantity. We conclude that
\begin{eqnarray}
\sum_{\substack{X \in \mathcal{P}_0( \Gamma):\\x,y \in X}} \| \Phi_p(X) \| &  \leq &  \kappa^p \sum_{n \geq m} (n+1)^{ p \nu} \sum_{\substack{w,z \in b_x(n+1):\\n \leq d(w,z) < n+1}} \sum_{\substack{X \in \mathcal{P}_0(\Gamma):\\z,w \in X}}  \| \Phi(X) \| \nonumber \\
& \leq & \kappa^p \| \Phi \|_F \sum_{n \geq m} (n+1)^{p \nu} \sum_{\substack{w,z \in b_x(n+1):\\n \leq d(w,z) < n+1}} F(d(w,z)) \nonumber \\
& \leq & \kappa^{p+2} \| \Phi \|_F \sum_{n \geq m} (n+1)^{(p+2) \nu} F(n) \, ,
\end{eqnarray}
which proves (\ref{pnormbd}).
\end{proof}

When considering a weighted $F$-function $F(r) = e^{-g(r)}F_0(r)$, one can often use the weight with \eqref{pnormbd} to prove that the $p$-th moment of $\Phi$ has a finite $F$-norm. This allows us to apply the Lieb-Robinson bound theory to $\Phi_p$.

\begin{cor}\label{cor:pth-mom-F-norm}
	Let $p\geq 0$ and $F(r) = e^{-g(r)}F_0(r)$ be a weighted $F$-function on a $\nu$-regular metric space $(\Gamma,d)$. If $M_{(p+2) \nu}^{-ag}(0) < \infty$ for some $0<a\leq 1$, then $\Phi_p\in B_{\tilde{F}}$ for any $\Phi\in \cB_F$ where $\tilde{F}$ is the $F$-function 
	\be
	\tilde{F}(r) = e^{(a-1)g(\lfloor r\rfloor)}F_0(\lfloor r\rfloor).
	\ee
\end{cor}

\begin{proof}
	This is an immediate consequence of \eqref{pnormbd}, and the bound
\beann
M_{(p+2) \nu}^F(k) & = & \sum_{n=\lfloor k\rfloor}^\infty (1+n)^{(p+2)\nu}e^{-g(n)}F_0(n) \\
& \leq &
e^{(a-1)g(\lfloor k\rfloor)}F_0(\lfloor k\rfloor)\sum_{n=0}^\infty (1+n)^{(p+2)\nu}e^{-ag(n)} \\
& = & M_{(p+2) \nu}^{-ag}(0) \tilde{F}(k)
\eeann
for all $k\geq 0$.
\end{proof}

\subsubsection{An estimate on weighted sums}

\begin{lem}\label{lem:weight_int_dec}
Let $(\Gamma, d)$ be a $\nu$-regular metric space equipped with an $F$-function $F$, and $\Phi$ be an interaction on $\Gamma$ with $\Phi \in \mathcal{B}_F$.
If $G: [0, \infty) \to (0, \infty)$ satisfies $M_{2 \nu +1}^G(0) < \infty$, then for any $x,y \in \Gamma$
\begin{equation} \label{wmbd}
\sum_{n=0}^{\infty} G(n) \sum_{\substack{X \in \mathcal{P}_0( \Gamma):\\x,y \in X(n+1)}} \| \Phi(X) \| \leq \kappa \| \Phi \|_F \left( \kappa M_{2 \nu +1}^G(0) F(d(x,y)/3) + \| F \| M_{\nu+1}^G(d(x,y)/3) \right) .
\end{equation}
\end{lem}

\begin{proof}
For each $X \in \mathcal{P}_0( \Gamma)$, the sets $X(n)$, see e.g. (\ref{set_xn}), are increasing and therefore, if
 $x,y \in X(n)$ for some $n \geq 1$, then  $x,y \in X(m)$ for all $m \geq n$. With this in mind, we write
\begin{equation}
\sum_{n=0}^{\infty} G(n) \sum_{\substack{X \in \mathcal{P}_0( \Gamma):\\x,y \in X(n+1)}} \| \Phi(X) \|  = 
\sum_{m=0}^{\infty} \sum_{\substack{X \in \mathcal{P}_0( \Gamma):\\x,y \in X(m+1)}} \chi_{m+1}(X) \| \Phi(X) \| \sum_{n=m}^{\infty} G(n)   
\end{equation} 
where we have inserted the characteristic function $\chi_{m+1}$ defined by
\begin{equation}
\chi_{m+1}(X) = \left\{ \begin{array}{cl} 0 & \mbox{if } x,y \in  X(m) ,  \\ 1 & \mbox{otherwise.} \end{array} \right.
\end{equation}
Note that $x,y \in X(m+1)$ means there exist $w, z \in X$ such that $w \in b_x(m+1)$ and $z \in b_y(m+1)$. Using this, the definition of $M_0^G(m)$ (see \eqref{pmom}), and the $F$-norm of $\Phi$, the estimate
\begin{eqnarray} \label{Sig2bd}
\sum_{n=0}^{\infty} G(n) \sum_{\substack{X \in \mathcal{P}_0( \Gamma):\\x,y \in X(n+1)}} \| \Phi(X) \|  & \leq & \sum_{m=0}^{\infty} M_0^G(m) \sum_{\substack{w \in b_x(m+1)\\z \in b_y(m+1)}} \sum_{\substack{X \in \mathcal{P}_0( \Gamma):\\w,z \in X }}   \| \Phi(X) \|  \nonumber \\
& \leq & \| \Phi \|_F  \sum_{m=0}^{\infty} M_0^G(m) \sum_{\substack{w \in b_x(m+1)\\z \in b_y(m+1)}} F(d(w,z))
\end{eqnarray}
readily follows.

We now divide the final series on the right-hand-side of (\ref{Sig2bd}) into sums of large and small $m$. More precisely, for any fixed $0 < \epsilon <1$, set $m_0 = m_0( \epsilon, x,y)$ to be the largest integer $m\geq -1$ satisfying
\begin{equation}
2(m+1) \leq (1- \epsilon)d(x,y) \, .
\end{equation}
For any $0 \leq m \leq m_0$, $w \in b_x(m+1)$ and $z \in b_y(m+1)$, on has $\epsilon d(x,y) \leq d(w, z)$ as
\begin{eqnarray}
d(x,y) \leq d(x,w) + d(w, z) + d(z, y) & \leq & 2(m+1) + d(w,z) \nonumber \\
& \leq & (1- \epsilon)d(x,y) + d(w,z).
\end{eqnarray} 
	
In this case, the first few terms may be estimated as 
\begin{eqnarray}
\sum_{m=0}^{m_0} M_0^G(m) \sum_{\substack{w \in b_x(m+1)\\z \in b_y(m+1)}} F(d(w,z)) 
& \leq & \kappa^2 F( \epsilon d(x,y))  \sum_{m=0}^{m_0} (1+m)^{2 \nu} M_0^G(m) \nonumber \\ 
& \leq & \kappa^2 M_{2 \nu +1}^G(0) F( \epsilon d(x,y))
\end{eqnarray}
where in the last equality we have used (\ref{it_mom_bd}).
	
For the remaining terms, uniform integrability of $F$ implies
\begin{equation}
\sum_{m \geq m_0 +1} M_0^G(m) \sum_{\substack{w \in b_x(m+1)\\z \in b_y(m+1)}} F(d(w,z)) \leq \kappa \| F \| M_{\nu}^G \circ M_0^G(m_0+1) .
\end{equation}
Using the definition of $m_0$ and again applying \eqref{it_mom_bd}, the bound claimed in (\ref{wmbd}) follows from choosing $\epsilon = 1/3$. 
\end{proof}

 \section*{Acknowledgements}

We would like to thank Valentin Zagrebnov for illuminating discussions about the non-autonomous Cauchy problems that
arise in quantum dynamical systems. We also thank Martin Gebert for reading an early version of this paper and asking good
questions and Derek Robinson for several useful remarks and informative comments about the history of the subject.
All three authors wish to thank the Departments of Mathematics of the University 
of Arizona and the University of California, Davis, for extending their kind hospitality to us and for the stimulating atmosphere they offered 
during several visits back and forth over the years it took to complete this project. BN also acknowledges the support of a CRM-Simons 
Professorship for a stay at the Centre de Recherches Math\'ematiques (Montr\'eal) during Fall 2018, which created the perfect 
circumstances to complete this paper.
 
\providecommand{\bysame}{\leavevmode\hbox to3em{\hrulefill}\thinspace}
\providecommand{\MR}{\relax\ifhmode\unskip\space\fi MR }
\providecommand{\MRhref}[2]{%
  \href{http://www.ams.org/mathscinet-getitem?mr=#1}{#2}
}
\providecommand{\href}[2]{#2}

\end{document}